\documentclass{article}

\usepackage{amsfonts}
\usepackage{amsmath}
\usepackage{amsthm}
\usepackage{arxiv}
\usepackage{booktabs}
\usepackage{doi}
\usepackage[mathscr]{eucal}
\usepackage[T1]{fontenc}
\usepackage{graphicx}
\usepackage{hyperref}
\usepackage[utf8]{inputenc}
\usepackage{mathrsfs}
\usepackage{microtype}
\usepackage{nicefrac}
\usepackage{url}
\usepackage{xcolor}

\numberwithin{equation}{section}

\newtheorem{Def}{Definition}
\newtheorem{Thm}[Def]{Theorem}

\newtheorem{Lemma}[Def]{Lemma}
\newtheorem{Remark}[Def]{Remark}
\newtheorem{Corollary}[Def]{Corollary}

\DeclareMathOperator{\csch}{csch}

\title{Cosmic Accelerations Characterize the Instability of the Critical Friedmann Spacetime}

\author{ \href{https://orcid.org/0000-0001-9255-6281}{\includegraphics[scale=0.06]{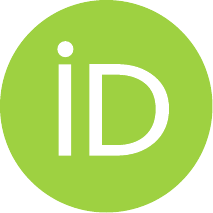}\hspace{1mm}Christopher Alexander}\\
Department of Mathematics\\
University College London\\
London, WC1H 0AY\\
United Kingdom\\
\texttt{christopher.alexander@ucl.ac.uk}\\
\And
\href{https://orcid.org/0000-0002-6907-1101}{\includegraphics[scale=0.06]{orcid.pdf}\hspace{1mm}Blake Temple} \\
Department of Mathematics\\
University of California\\
Davis, CA 95616\\
United States\\
\texttt{temple@math.ucdavis.edu}\\
\And
{Zeke Vogler} \\
Department of Mathematics\\
University of California\\
Davis, CA 95616\\
United States\\
\texttt{zekius@math.ucdavis.edu}
}

\renewcommand{\headeright}{}
\renewcommand{\undertitle}{}
\renewcommand{\shorttitle}{Cosmic Accelerations Characterize the Instability of the Critical Friedmann Spacetime}

\hypersetup{
	pdftitle={Cosmic Accelerations Characterize the Instability of the Critical Friedmann Spacetime},
	pdfsubject={gr-qc, math-ph},
	pdfauthor={Christopher Alexander},
	pdfkeywords={General Relativity, Instability, Cosmology, Dark Energy},
}

\begin{document}

\maketitle

\begin{abstract}
	We give a definitive characterization of the instability of the pressureless ($p=0$) critical ($k=0$) Friedmann spacetime to smooth radial perturbations. We use this to characterize the global accelerations away from $k\leq0$ Friedmann spacetimes induced by the instability in the underdense case. The analysis begins by incorporating the Friedmann spacetimes into a mathematical analysis of smooth spherically symmetric solutions of the Einstein field equations expressed in self-similar coordinates $(t,\xi)$ with $\xi=\frac{r}{t}<1$, conceived to realize the critical Friedmann spacetime as an unstable saddle rest point $SM$. We identify a new maximal asymptotically stable family $\mathcal{F}$ of smooth outwardly expanding solutions which globally characterize the evolution of underdense perturbations. Solutions in $\mathcal{F}$ align with a $k<0$ Friedmann spacetime at early times, generically introduce accelerations away from $k<0$ Friedmann spacetimes at intermediate times and then decay back to the same $k<0$ Friedmann spacetime as $t\to\infty$ uniformly at each fixed radius $r>0$. We propose $\mathcal{F}$ as the maximal asymptotically stable family of solutions into which generic underdense perturbations of the unstable critical Friedmann spacetime will evolve and naturally admit accelerations away from Friedmann spacetimes within the dynamics of solutions of Einstein's original field equations, that is, without recourse to a cosmological constant or dark energy.
\end{abstract}


This material is based upon work supported by EPSRC Project: EP/S02218X/1

\vfill

\pagebreak

\tableofcontents

\vfill

\pagebreak

\section{Introduction}\label{S1}

In our 2017 announcement in Proceedings of the Royal Society A \cite{SmolTeVo},\footnote{Authors dedicate this paper to our former collaborator and long-time friend Joel Smoller and acknowledge our use of unpublished notes which were the basis for \cite{SmolTeVo} and represent the point of departure for the present paper.} Smoller, Temple and Vogler introduced the STV-PDE, a version of the perfect fluid Einstein field equations for spherically symmetric spacetimes. These were obtained by starting with the spacetime metric in standard Schwarzschild coordinates (SSC) and then re-expressing it using the self-similar variable $\xi=\frac{r}{t}$ in place of $r$, assuming zero cosmological constant and assuming $|\xi|<1$ to keep $\xi$ as a spacelike coordinate. Since $\xi=1$ is a measure of the distance of light travel since the Big Bang in a Friedmann spacetime,\footnote{By a Friedmann spacetime, we mean a Friedmann--Lemaître--Robertson--Walker (FLRW) spacetime. Also recall that the curvature parameter $k$ in Friedmann spacetimes can always be scaled to one of the values $k=-1,0,1$. The $k=0$ spacetime is unique and $k=\pm1$ spacetimes each describe a one parameter family of distinct spacetimes depending on the single parameter $\Delta_0$ defined below in (\ref{DeltaDef}). Thus we refer to $k<0$ Friedmann spacetimes by $k$ or $\Delta_0$. Unless a different equation of state is specified, our use of the term Friedmann always assumes a dust ($p=0$) solution of the Einstein field equations.} we view $|\xi|<1$ as valid out to approximately the Hubble radius, a measure of the distance across the visible Universe \cite{smolteMemoirs}. The STV-PDE were conceived to represent the pressureless ($p=0$) critical ($k=0$) Friedmann spacetime as an unstable saddle rest point\footnote{We refer to a \emph{rest point} of a PDE as a time independent solution depending only on $\xi$. The character of a rest point of the STV-PDE, that is, unstable, stable and so on, is determined by the character it exhibits in the approximating STV-ODE obtained by the expansion of solutions in powers of $\xi$, as described later.}, which we denote by $SM$ for \emph{Standard Model}. This is based on the important realization that the metric components and fluid variables of the $p=0$, $k=0$ Friedmann spacetime in SSC can be expressed as a function of $\xi$ alone in an appropriate time gauge (see \cite{SmolTeVo} and Theorem \ref{SelfSimExpandxi} below). The character of a rest point is difficult to disentangle in coordinate systems, such as comoving coordinates, where it appears time dependent, especially so for PDE. To analyze rest points of PDE requires a procedure of finite approximation and this was manifested in \cite{SmolTeVo} by the observation that solutions of the STV-PDE which are smooth at the center of symmetry can be developed into a regular expansion in even powers of $\xi$ with time dependent coefficients. This generates a nested sequence of autonomous $2n\times2n$ ODE which close\footnote{This asymptotic expansion does not close at order $n$ when $p\neq0$ \cite{SmolTeVo}.} in the time dependent coefficients at every order $n\geq1$. We name the resulting $2n\times2n$ system of ODE the STV-ODE of order $n$ and observe that at each order, the STV-ODE is autonomous in the log-time variable $\tau=\ln t$ (so $0<t<\infty$ and $-\infty<\tau<\infty$). Moreover, the phase portrait of the resulting autonomous system at each order contains the unstable saddle rest point $SM$, together with the stable degenerate rest point $M$. $M$, for \emph{Minkowski}, is the limit of the time asymptotic decay of solutions as $t\to\infty$.\footnote{The presence of a single rest point $M$ which characterizes the late time dynamics of solutions is a serendipitous simplification inherent in the choice of self-similar coordinates.}

The STV-ODE are nested in the sense that higher order solutions provide strict refinements of solutions determined at lower orders. We prove that the STV-ODE are linear inhomogeneous ODE at every order $n>1$ in the sense that the coefficients of the highest order terms are always of lower order. Our analysis shows that the eigenvalue structure of the rest points $SM$ and $M$ determine the character of the phase portrait of the STV-ODE at every order and the essential character of all the phase portraits can be deduced from the phase portraits at orders $n=1$ and $n=2$. Our analysis establishes that $k<0$ and $k>0$ Friedmann spacetimes correspond to unique solution trajectories which lie in the unstable manifold of $SM$ at all orders of the STV-ODE, and general higher order solutions of the STV-ODE agree exactly with a Friedmann spacetime at order $n=1$. In particular, we prove that the phase portraits of the STV-ODE of order $n=1$ and $n=2$ characterize the instability of $k\leq0$ Friedmann spacetimes: The $k=0$ Friedmann spacetime is unstable with a codimension one unstable manifold, while the $k<0$ Friedmann spacetimes are locally unstable at $SM$ but are globally asymptotically stable in the sense that all underdense perturbations of $SM$ tend to the same rest point $M$ as $t\to\infty$. Moreover, this remains true at all orders $n>2$ and a smooth solution of the STV-PDE will lie in the unstable manifold of $SM$ at all orders $n\geq1$ of the STV-ODE if and only if it lies in the unstable manifold of $SM$ at order $n=2$.

The existence of a second positive eigenvalue of $SM$ at order $n=2$ (the first being at order $n=1$) implies the existence of a one parameter family of nontrivial solutions of the STV-ODE of order $n=2$ which exist within the unstable manifold of $SM$ but diverge from Friedmann spacetimes at that order. The existence of a second negative eigenvalue at $SM$ at order $n=2$ establishes that the unstable manifold of $SM$ is a codimension one space of trajectories, so solutions of the STV-ODE are generically not within the unstable manifold of $SM$. We prove that at order $n=2$ all solutions in the unstable manifold of $SM$ exit tangent to the trajectory associated with Friedmann spacetimes but a unique positive eigenvalue smaller than the leading order eigenvalue enters at order $n=3$. Since the eigenvector associated with the smallest eigenvalue dominates at rest point $SM$ in backward time, it follows that generic solutions in the unstable manifold of $SM$ exit tangent to a new eigendirection, different from Friedmann, at all orders $n=3$ and above.

The stable and unstable manifolds of $SM$ together with the stable manifold of $M$ characterize the global phase portraits of the STV-ODE at all orders. An analysis of the phase portraits lead to the introduction of a new family $\mathcal{F}$ of solutions of the STV-PDE that extend the $k<0$ Friedmann spacetimes to a maximal asymptotically stable family of solutions into which underdense perturbations of the unstable $k=0$ Friedmann spacetime will globally evolve and generically accelerate away from Friedmann spacetimes early on. Thus $\mathcal{F}$ globally characterizes the instability of the $p=0$, $k=0$ Friedmann spacetime to smooth radial underdense perturbations. We prove that solutions in $\mathcal{F}$ align with a $k<0$ Friedmann spacetime at early times, introduce accelerations away from $k<0$ Friedmann spacetimes at intermediate times and then decay back to the same $k<0$ Friedmann spacetime as $t\to\infty$ (at each fixed $r>0$). The existence of positive eigenvalues at $SM$, with eigensolutions tending to $M$ as $t\to\infty$ at every order $n\geq1$ of the STV-ODE, demonstrates the global instability of the $k=0$ Friedmann spacetime to perturbation at every order. On the other hand, the decay of solutions in $\mathcal{F}$ to the rest point $M$ as $t\to\infty$ establishes the global large time asymptotic stability of all $k<0$ Friedmann spacetimes. However, the existence of a second positive eigenvalue at $SM$ at order $n=2$ demonstrates the instability of $k<0$ Friedmann spacetimes to perturbation within the unstable manifold of $SM$ at early times, implying that solutions in $\mathcal{F}$ generically accelerate away from Friedmann spacetimes at intermediate times before asymptotic stability brings them back to a $k<0$ Friedmann spacetime via decay to $M$ as $t\to\infty$ (for fixed $r>0$). The existence of positive eigenvalues of $SM$ at every order implies that a similar instability of $k<0$ Friedmann spacetimes occurs at $t=0$ within the unstable manifold of $SM$ at every order $n\geq3$ of the STV-ODE as well. We conclude that solutions in $\mathcal{F}$ characterize both the instability of the $p=0$, $k=0$ Friedmann spacetime to smooth radial underdense perturbations and characterize the accelerations away from Friedmann spacetimes at intermediate times, both within the dynamics of Einstein's original field equations, that is, without recourse to a cosmological constant or dark energy.

\subsection{Introduction to the Family of Spacetimes $\mathcal{F}$}

We argued in \cite{SmolTeVo} that solutions of the $p=0$ STV-PDE which are smooth at the center of symmetry can be expanded in even powers of $\xi$ by Taylor's theorem (see Section \ref{S3}) and from the $n^{th}$ order term we obtain an approximation which solves the STV-ODE of order $n$. In the present paper we go on to prove that, for each such solution, there exists a solution dependent time translation $t\to t-t_*$ of the SSC time coordinate $t$, which we call \emph{time since the Big Bang}, such that making the SSC gauge transformation to time since the Big Bang has the effect of eliminating the leading order negative eigenvalue at $SM$. Moreover, this gauge transforms every solution to either the rest point $SM$ or one of the two trajectories in the unstable manifold of $SM$ at $n=1$. As a result of this, every solution agrees exactly with a $k<0$, $k=0$ or $k>0$ Friedmann spacetime in the phase portrait of the STV-ODE at leading order $n=1$, see Figure \ref{Figure1}. There is an important distinction to make here: The STV-ODE are autonomous in log-time $\tau=\ln t$, so translation in $\tau$ maps solutions to physically different solutions which traverse the same trajectory of the STV-ODE, but translation in $t$ is a gauge freedom of the SSC metric ansatz which maps trajectories of the STV-ODE to different trajectories which represent the same physical solution. Thus choosing time since the Big Bang eliminates a physical redundancy in the solution trajectories of the STV-ODE. When time since the Big Bang is imposed, there are only three remaining trajectories in the leading order $n=1$ phase portrait of the STV-ODE: The unstable rest point $SM$, the underdense (left) component of the unstable manifold and the overdense (right) component of the unstable manifold, see Figure \ref{Figure1}. The underdense component of the unstable manifold of $SM$ at $n=1$ is the unique trajectory which takes $SM$ to $M$ and corresponds to $k<0$ Friedmann spacetimes, with the value of $k$ determined by translation in $\tau$ along this unique trajectory. The unique trajectory exiting $SM$ in the opposite overdense direction is the component of the unstable manifold of $SM$ corresponding to $k>0$ Friedmann spacetimes. The three (including the fixed point $SM$) trajectories of the $n=1$ phase portrait, after time since the Big Bang is imposed, is diagrammed in Figure \ref{Figure2}.

In this paper we identify and study the entire subset of solutions $\mathcal{F}$ of the STV-PDE defined by the condition that, when time $t$ is taken to be time since the Big Bang, the resulting solution agrees with an underdense $k<0$ Friedmann solution in the leading order $n=1$ STV-ODE phase portrait. That is, the solution projects to the unique trajectory in the unstable manifold of $SM$ which takes $SM$ to $M$ at order $n=1$, parameterized by $\tau-\tau_0$ for some log-time translation constant $\tau_0=\ln t_0$ of the autonomous STV-ODE. We identify $\mathcal{F}$ as a new maximal asymptotically stable family of outwardly expanding solutions of the STV-PDE defined by the condition that solutions agree with a $k<0$ Friedmann spacetime in the leading order phase portrait of the expansion in even powers of $\xi$ when the SSC time is translated to time since the Big Bang associated with each solution. The main purpose of this paper is to demonstrate and characterize the instability of the $k=0$ Friedmann spacetime and the large time asymptotic stability and early time instability of the $k<0$ Friedmann spacetimes within the family of solutions $\mathcal{F}$.

We first characterize the forward time dynamics and asymptotic stability of solutions in $\mathcal{F}$ by proving that every solution which decays to the rest point $M$ as $t\to\infty$ in the phase portrait of the STV-ODE at order $n=1$, that is, every solution in $\mathcal{F}$, also decays as $t\to\infty$ to a corresponding unique stable rest point $M$ in the phase portrait of the STV-ODE at every order $n\geq1$. This characterizes the forward time dynamics of solutions in $\mathcal{F}$ because it implies that every solution in $\mathcal{F}$ decays, as $t\to\infty$, to a $k<0$ Friedmann spacetime faster than it decays to Minkowski space as $t\to\infty$ at each fixed radii $r>0$ at every order $n\geq1$. The asymptotic decay of solutions in $\mathcal{F}$ as $t\to\infty$ immediately implies uniform bounds on solutions of the STV-ODE for all time $t>t_0>0$ and $n\geq1$ in terms of bounds on the initial data at $t=t_0>0$ alone. The STV-ODE are linear in the highest order variables when lower order solutions are interpreted as known inhomogeneous terms, so solutions of the STV-ODE exist for all time $0<t<\infty$ at every order. For the purposes of asymptotic analysis, we formally assume convergence of solutions of the STV-ODE up to order $n$, with errors at order $n$ estimated by bounds at order $n+1$ according to Taylor's theorem, an assumption justified rigorously by simply restricting to an appropriate space of analytic solutions.\footnote{The convergence of solutions in $\mathcal{F}$ in the limit $n\to\infty$ for $|\xi|<1$, with estimates provided by Taylor's theorem, follows directly from mild assumptions on the growth rate of the initial data, due to the fact that all solutions lie on bounded trajectories which tend to $M$ as $t\to\infty$ at every order $n\geq1$.}

The backward time dynamics $t\to0$ ($\tau\to-\infty$) of solutions in $\mathcal{F}$ is determined at each order $n\geq1$ by the instability of the $k=0$ Friedmann spacetime, that is, by the eigenvalues of the saddle rest point $SM$ as it is represented in the STV-ODE phase portrait at each order $n\geq1$. By definition, each solution in $\mathcal{F}$ agrees at leading order ($n=1$) with a $k<0$ Friedmann spacetime represented as the unique trajectory in the unstable manifold of $SM$ which tends to $M$ as $t\to\infty$ and to $SM$ as $t\to0$. To understand the backward time dynamics of solutions in $\mathcal{F}$ at higher orders $n\geq2$, we demonstrate that $k<0$ Friedmann solutions correspond to a single trajectory in the unstable manifold of $SM$ in the phase portrait of the STV-ODE at every order $n\geq1$, with the value of $k$ determined by log-time translation at order $n=1$. We then prove that two new eigenvalues of the rest point $SM$ appear at each new order $n\geq2$ and all of them are positive except one negative eigenvalue which appears at order $n=2$. From this we conclude that the family $\mathcal{F}$ decomposes into two essential components: The underdense component of the unstable manifold of $SM$, consisting of trajectories which connect $SM$ to $M$ at every order $n\geq1$, and solutions which tend to $M$ in forward time but do not tend to $SM$ in backwards time. Because of the presence of the order $n=2$ negative eigenvalue at $SM$, solutions in $\mathcal{F}$ will generically not tend to $SM$ in backward time, but rather will follow the one-dimensional stable manifold of $SM$, the unique trajectory which emerges from $SM$ in backward time as $t\to0$ ($\tau\to-\infty$). Thus the Big Bang limit $t\to0$ of a generic solution in $\mathcal{F}$ is not self-similar like $SM$ beyond the leading order, but rather, generically displays a non-self-similar character at the Big Bang for all orders $n=2$ and above. This provides an important example of the self-similarity hypothesis, that solutions which approach a singularity exhibit self-similarity to leading order \cite{carrco}. However, in this case, such solutions are generically not self-similar beyond leading order.

To reiterate, each spacetime in $\mathcal{F}$ agrees with a unique $k<0$ Friedmann spacetime in the phase portrait of the STV-ODE at leading order $n=1$ (all the way into the Big Bang at $t=0$) and agrees with the same $k<0$ Friedmann spacetime in the limit $t\to\infty$ (for each fixed $r>0$), but generically accelerates away from this $k<0$ Friedmann spacetime at intermediate times in the phase portraits of the STV-ODE of order $n\geq2$. Since each member of the family $\mathcal{F}$ by definition contains the component of the unstable manifold of $SM$ which contains the leading order behavior of $k<0$ Friedmann spacetimes imposed at $n=1$, $\mathcal{F}$ represents a maximal extension of the one parameter family of $k<0$ Friedmann spacetimes to a robust asymptotically stable family of spacetimes which exist within the family of smooth spherically symmetric spacetimes. Since the family $\mathcal{F}$ is the full space of solutions into which underdense perturbations of $SM$ evolve, $\mathcal{F}$ characterizes the instability of $SM$ to underdense perturbations. Because a negative eigenvalue of $SM$ does not exist above order $n=2$, it follows that a solution in $\mathcal{F}$ lies in the unstable manifold of $SM$ if and only if its projection onto the $n=2$ phase portrait of the STV-ODE lies in the unstable manifold of $SM$, so the unstable manifold of $SM$ at $n=2$ characterizes the unstable manifold of $SM$ at all orders. A main goal of this paper is to prove that, even though $\mathcal{F}$ is asymptotically stable, solutions in $\mathcal{F}$ generically accelerate away from $k<0$ Friedmann spacetimes at early times due to an instability at $t=0$. 

To establish the instability of $k<0$ Friedmann spacetimes at $t=0$, it suffices to demonstrate that there are multiple solutions of the $STV-ODE$ which agree with $k<0$ Friedmann spacetimes at order $n=1$ and in the limit $t\to0$ but which diverge from $k<0$ Friedmann spacetimes at $t>0$ for some $n\geq2$. Imposing time since the Big Bang, every solution in $\mathcal{F}$ lies on the unique trajectory in the unstable manifold of $SM$ which takes $SM$ to $M$ and hence every solution agrees with a $k<0$ Friedmann solution at order $n=1$. Thus to establish the instability of $k<0$ Friedmann spacetimes at order $n=2$, it suffices to prove that: (1) Solution trajectories corresponding to $k<0$ Friedmann lie in the unstable manifold of $SM$ and (2) there exist other solutions in $\mathcal{F}$ in the unstable manifold of $SM$ which diverge from $k<0$ Friedmann at $t>0$. For this, we use known exact formulas to prove that the solution trajectory of $k<0$ Friedmann spacetimes lies in the unstable manifold of $SM$ at order $n=2$ with time translation differentiating $k$. Since a solution in $\mathcal{F}$ lies in the unstable manifold of $SM$ at every order if and only if it lies in the unstable manifold of $SM$ order $n=2$, we can conclude that $k<0$ Friedmann spacetimes lie on a single trajectory in the unstable manifold of $SM$ at every order $n\geq2$ as well. Next we use exact formulas to prove that the trajectory corresponding to $k<0$ Friedmann spacetimes at order $n=2$ is an eigensolution of the eigenvalue $\lambda_{A1}$ which enters at order $n=1$ (this is also shown to order $n=3$ in Section \ref{S11.5}). Thus to establish the instability of $k<0$ Friedmann spacetimes it suffices only to prove that a second positive eigenvalue $\lambda_{B2}\neq\lambda_{A1}$ emerges at $SM$ at order $n=2$. From this, the early time instability of the $k<0$ Friedmann spacetimes is established in the phase portrait of the STV-ODE of order $n=2$, diagrammed in Figure \ref{Figure3}, even though the whole family $\mathcal{F}$ is asymptotically stable. We discuss the phase portrait at order $n=2$ in Section \ref{S1.6}.

At this stage it is convenient to introduce a refined notation that distinguishes the family $\mathcal{F}$ of solutions of the STV-PDE from the solutions of the STV-ODE which approximate solutions in $\mathcal{F}$ up to arbitrary order. To this end, we define $\mathcal{F}_n$ as the set of solutions of the STV-ODE of order $n$ which satisfy the property that the trajectory reduces at order $n=1$ to the unique trajectory which connects $SM$ to $M$, but not necessarily at higher orders. We also define $\mathcal{F}_n'\subset\mathcal{F}_n$ to be the subset of solutions which lie in the unstable manifold of the rest point $SM$ in STV-ODE of order $n$. Note that by definition $\mathcal{F}_1'=\mathcal{F}_1$. Finally, let $\mathcal{F}'\subset \mathcal{F}$ denote the set of solutions of the STV-PDE whose $n^{th}$ order approximation lies in the unstable manifold of rest point $SM$ in the phase portrait of the STV-ODE of order $n$ for every $n\geq1$. We may sometimes refer to $\mathcal{F}_n$ as $\mathcal{F}$ at order $n$ and $\mathcal{F}_n'$ as the stable manifold of $SM$ in $\mathcal{F}$ at order $n$.

Having set out the main elements, we now discuss them in detail.

\subsection{The STV-PDE}

In Section \ref{S6} we give a new derivation of the Einstein field equations $G=\kappa T$ in what we call SSCNG coordinates. These coordinates are standard Schwarzschild coordinates (SSC), that is, where the metric takes the form
\begin{align}
    ds^2 = -B(t,r)dt^2 + \frac{dr^2}{A(t,r)} + r^2d\Omega^2,\label{SSCintro}
\end{align}
in addition to possessing a normal gauge (NG) when expressed in the variables $(t,\xi)$. An arbitrary spherically symmetric spacetime can generically\footnote{That is, under the condition $\frac{\partial C(t,r)}{\partial r}\neq0$, where $C(t,r)d\Omega^2$ is the angular part of the metric, see \cite{SmolTeVo}.} be transformed to SSC metric form by defining $r$ so $r^2d\Omega^2$ is the angular part of the metric and then constructing a time coordinate $t$, complementary to $r$, such that the metric is diagonal in $(t,r)$-coordinates \cite{smolteMemoirs,wein}. The SSC metric form is invariant under the gauge freedom of arbitrary time transformation $t\to\phi(t)$, so to set the gauge, we impose the condition that $B(t,0)=1$, that is, geodesic (proper) time at $r=0$ \cite{SmolTeVo}. We refer to this as the normal gauge (NG).

The STV-PDE use density and velocity variables:
\begin{align*}
    z &=\frac{\kappa\rho r^2}{1-v^2}, & w &= \frac{v}{\xi},
\end{align*}
coupled to the SSC metric components $A$ and $D=\sqrt{AB}$. Recall that the STV-PDE are not the Einstein field equations in $(t,\xi)$ coordinates but rather the Einstein field equations with an SSC metric form expressed in terms of $z$ and $w$ with independent variables $(t,\xi)$ \cite{SmolTeVo}. We extend the derivation of the STV-PDE to equations of state of the form $p=\sigma\rho$, with $\sigma$ constant, in Theorem \ref{thmsigma} below. However, our concern in this paper is with the case $p=\sigma=0$, applicable to late time Big Bang Cosmology.\footnote{In the Standard Model of Cosmology, the pressure drops precipitously to zero at about 10,000 years after the Big Bang, an order of magnitude before the uncoupling of matter and radiation \cite{Peacock}}. 

\begin{Thm}[Special case of Theorem \ref{thmsigmainitial2}]\label{STV-PDE}
    Assume the equation of state $p=0$. Then the perfect fluid Einstein field equations with an SSC metric are equivalent to the following four equations in unknowns $A(t,\xi)$, $D(t,\xi)$, $z(t,\xi)$ and $w(t,\xi)$:
    \begin{align}
        \xi A_\xi &= -z + (1-A),\label{AeqnxiIntro}\\
        \xi D_\xi &= \frac{D}{2A}\big(2(1-A)-(1-\xi^2w^2)z\big),\label{DeqnxiIntro}\\
        tz_t + \xi\big((-1+Dw)z\big)_\xi &= -Dwz,\label{zeqnxiIntro}\\
        tw_t + \xi(-1+Dw)w_\xi &= w - D\bigg(w^2+\frac{1}{2\xi^2}(1-\xi^2w^2)\frac{1-A}{A}\bigg).\label{weqnxiIntro}
    \end{align}
\end{Thm}

Equations (\ref{AeqnxiIntro})--(\ref{weqnxiIntro}) are the STV-PDE. The NG is imposed by defining a time transformation $t\to\phi(t,r)$ which sets $B=1$ at $r=0$ \cite{SmolTeVo}. From here on, when we refer to a solution of the STV-PDE we always assume the NG gauge is imposed. Note that there is one remaining freedom left in this gauge, this being the time translation freedom $t\to t-t_0$ for some constant $t_0$.

\subsection{The STV-ODE}

The STV-ODE are derived by expanding smooth solutions of the STV-PDE (\ref{AeqnxiIntro})--(\ref{weqnxiIntro}) in powers of $\xi$ with time dependent coefficients, collecting like powers of $\xi$ and assuming the NG gauge. The assumption of smoothness at the center implies that the non-zero coefficients in the expansion occur only for even powers $\xi^{2n}$ and this significantly reduces the solution space of the Einstein field equations by disentangling solutions smooth at the center from the larger generic solution space. Fortuitously, the resulting equations in the fluid variables $z$ and $w$ uncouple from the equations for the metric components $A$ and $D$, leading to the expansions:
\begin{align}
    z(t,\xi) &= z_2(t)\xi^2 + z_4(t)\xi^4 + \dotsc + z_{2n}(t)\xi^{2n} + \dots,\label{zfirst1}\\
    w(t,\xi) &= w_0(t) + w_2(t)\xi^2 + \dotsc + w_{2n-2}(t)\xi^{2n-2} + \dots,\label{wfirst1}
\end{align}
which close in $(z_{2k},w_{2k-2})$ for $1\leq k\leq n$ at every order $n\geq1$ when $p=0$. We prove in Theorem (\ref{ThmSTVequations}) that the STV-ODE of order $n$ closes to form a $2n\times2n$ autonomous system
\begin{align}
    t\dot{\boldsymbol{U}} = \boldsymbol{F}_n(\boldsymbol{U})\label{nbynsystemFirst}
\end{align}
in unknowns
\begin{align*}
    \boldsymbol{U} := (z_2,w_0,\dots,z_{2n},w_{2n-2})^T,
\end{align*}
such that the leading order variables are determined by an inhomogeneous $2\times2$ system of the form
\begin{align*}
    \frac{d}{d\tau}\boldsymbol{v}_n = \left(\begin{array}{cc}
    (2n+1)(1-w_0)-1 & -(2n+1)z_2\\
    -\frac{1}{4n+2} & (2n+2)(1-w_0)-1
    \end{array}\right)\boldsymbol{v}_n + \boldsymbol{q}_n
\end{align*}
where $\boldsymbol{v}_n=(z_{2n},w_{2n-2})^T$ and $\boldsymbol{q}_n$ involves only lower order terms which are determined by the STV-ODE of order $n-1$. We refer to the $2n\times2n$ system of equations (\ref{nbynsystemFirst}) as the \emph{STV-ODE of order $n$}.

It follows that the STV-ODE are nested in the sense that each STV-ODE of order $n\geq2$ contains as a sub-system the STV-ODE of order $k$ for all $1\leq k\leq n-1$. Thus the self-similar formulation decouples solutions at every order in the sense that one can solve for solutions up to order $n-1$ and use the STV-ODE at order $n$ to solve for $(z_{2n}(t),w_{2n-2}(t))$ from arbitrary initial conditions $(z_{2n}(0),w_{2n-2}(0))$. Assuming lower order solutions $k\leq n-1$ are fixed, the STV-ODE of order $n$ turn out to be linear in the highest order terms $(z_{2n},w_{2n-2})$. For approximations up to orders $n=2$ we can use the approximation $z=\kappa\rho r^2$ in place of $z=\frac{\kappa\rho r^2}{1-v^2}$ because this incurs errors of the order $O(\xi^6)$ given that $v^2=O(\xi^2)$. The derivation of the general STV-ODE of order $n$ is carried out carefully in following sections but to set up the picture and highlight the main results we first describe the phase portraits of the STV-ODE which emerge at orders $n=1$ and $n=2$ of this expansion.\footnote{The STV-ODE at order $n=1$ and $n=2$ were introduced in \cite{SmolTeVo} without detailed derivation. Here we derive the STV-ODE up to order $n=3$ and derive the form of the STV-ODE at all orders $n\geq4$ together with an explicit algorithm for computing them. Note that Figure \ref{Figure1} is take from \cite{SmolTeVo} with the modification that in this paper, the coordinate system is centered on $(z_2,w_0)=0$, while coordinates were centered at $SM$ in \cite{SmolTeVo}. We also record a correction to the $z_4$ equation incorrectly expressed in equation (3.33) of \cite{SmolTeVo}. This error occurred at fourth order in $\xi$ and did not affect the results claimed in \cite{SmolTeVo}, see (\ref{z2first2})--(\ref{w2first2}) below.} Unanticipated by the authors ahead of time, it turns out that the global character of the phase portrait of the STV-ODE at any order $n\geq3$ emerges from orders $n=1$ and $n=2$.

\subsection{The STV-ODE Phase Portrait of Order $n=1$} 

A calculation shows that the STV-ODE of order $n=1$ is the $2\times2$ system:
\begin{align}
    t\dot{z}_2 &= 2z_2 - 3z_2w_0,\label{z2first1}\\
    t\dot{w}_0 &= -\frac{1}{6}z_2 + w_0 - w_0^2.\label{w0first1}
\end{align}
Using $\frac{d}{d\tau}=t\frac{d}{dt}$, system (\ref{z2first1})--(\ref{w0first1}) converts to an autonomous $2\times2$ system of ODE in $\tau=\ln t$. It is easy to verify that the system admits three rest points: The source $U=(0,0)$, the unstable saddle rest point $SM=(\frac{4}{3},\frac{2}{3})$ and the degenerate stable rest point $M=(0,1)$. The rest point $U$ plays no role once time since the Big Bang is imposed, the rest point $M$ describes the asymptotics of solutions tending to Minkowski space as $t\to\infty$ and the coordinates of rest point $SM$ are precisely the first two terms in the self-similar expansion of the critical ($k=0$) Friedmann spacetime, viewed here as the Standard Model due to the central role it has played in the history of Cosmology. A calculation gives the eigenvalues and eigenvectors of rest points $M$ and $SM$ as:
\begin{align*}
    \lambda_M &= -1, & \boldsymbol{R}_M &= \left(\begin{array}{cc}
    0\\
    1
    \end{array}\right),\\
    \lambda_{A1} &= \frac{2}{3}, & \boldsymbol{R}_{A1} &= \left(\begin{array}{cc}
    -9\\
    \frac{3}{2}
    \end{array}\right),\\
    \lambda_{B1} &= -1, & \boldsymbol{R}_{B1} &= \left(\begin{array}{cc}
    4\\
    1
    \end{array}\right),
\end{align*}
respectively. The phase portrait for system (\ref{z2first1})--(\ref{w0first1}) is diagrammed in Figure \ref{Figure1}. The two components of the unstable manifold of $SM$ correspond to the two trajectories associated with the positive eigenvalue $\lambda_{A1}=\frac{2}{3}$, the underdense component being the trajectory which connects $SM$ to $M$ and the overdense component leaving $SM$ in the opposite direction, see Figure \ref{Figure1}. The trajectory connecting $U$ to $SM$ is the underdense trajectory in the stable manifold of $SM$ corresponding to the negative eigenvalue $\lambda_{B1}=-1$ and the overdense stable trajectory emerges opposite to this at $SM$. Using classical formulas for Friedmann spacetimes which set the time of the Big Bang to $t=0$, we verify that the underdense trajectory connecting $SM$ to $M$ corresponds to $k<0$ Friedmann spacetimes and the overdense trajectory in the unstable manifold of $SM$ corresponds to $k>0$ spacetimes. We obtain an exact formula for the trajectory connecting $SM$ to $M$ from an expansion of such formulas for Friedmann spacetimes in powers of $\xi$. The variable $\Delta_0$, which parameterizes the one-parameter family of Friedmann spacetimes under scalings that set $k=-1,0,1$, is given by $\Delta_0=\ln t_0$, so the entire one-parameter family of Friedmann spacetimes at order $n=1$ consists of $SM$ together with the two trajectories in its unstable manifold, parameterized by the time translation freedom $\tau-\tau_0$, which determines the value of $\Delta_0$ and thereby determines a unique solution in the Friedmann family \cite{wein}.

Now the SSC metric ansatz with NG still leaves one gauge freedom yet to be set, namely, the freedom to impose time translation $t\to t-t_0$. The time translation freedom of the SSC system leaves open an unresolved redundancy in solutions of the the STV-ODE in the sense that time translation maps each trajectory of the STV-ODE of order $n=1$ to a different trajectory which represents the same physical solution. We show that for each trajectory of the system (\ref{z2first1})--(\ref{w0first1}), there exists a unique time translation $t\to t-t_0$, which we call \emph{time since the Big Bang}, which converts that trajectory either to $SM$ or to one of the two trajectories in the unstable manifold of $SM$. In particular, referring to the phase portrait depicted in Figure \ref{Figure1} and making the gauge transformation to time since the Big Bang, the trajectories in the stable manifold of $SM$, that is, the one taking $U$ to $SM$ and the trajectory opposite it at $SM$, go over to $SM$, whereas all the trajectories above these, that is, trajectories in the domain of attraction of $M$, go over to the underdense portion of the unstable manifold of $SM$ corresponding to $k<0$ Friedmann. Trajectories below the stable manifold of $SM$ go over to the overdense portion of the unstable manifold of $SM$, corresponding to $k>0$ Friedmann spacetimes. From this it follows that imposing the solution dependent time since the Big Bang has the effect of eliminating the negative eigenvalue $\lambda_{B1}=-1$ together with the rest point $U$, and we can, without loss of generality, restrict our analysis to the space of solutions which agree with $SM$ or a trajectory in its unstable manifold, in the phase portrait of the solution at $n=1$. Since these trajectories agree with the Friedmann spacetimes, we conclude that, under appropriate change of time gauge, all smooth solutions of the STV-PDE agree with a Friedmann spacetime at leading order in the STV-ODE. Our purpose here is to study the space $\mathcal{F}$ of solutions which lie on the trajectory which takes $SM$ to $M$ at leading order, and hence agree with a $k<0$ Friedmann spacetime at order $n=1$ of the STV-ODE. We do not consider the $k>0$ Friedmann spacetimes, but observe that these exit the first quadrant of our coordinate system at $w_0=0$, the time of maximal expansion.

\begin{figure}
	\centering
    \caption{The phase portrait for the $2\times2$ system.}
    \includegraphics[width=0.8\textwidth]{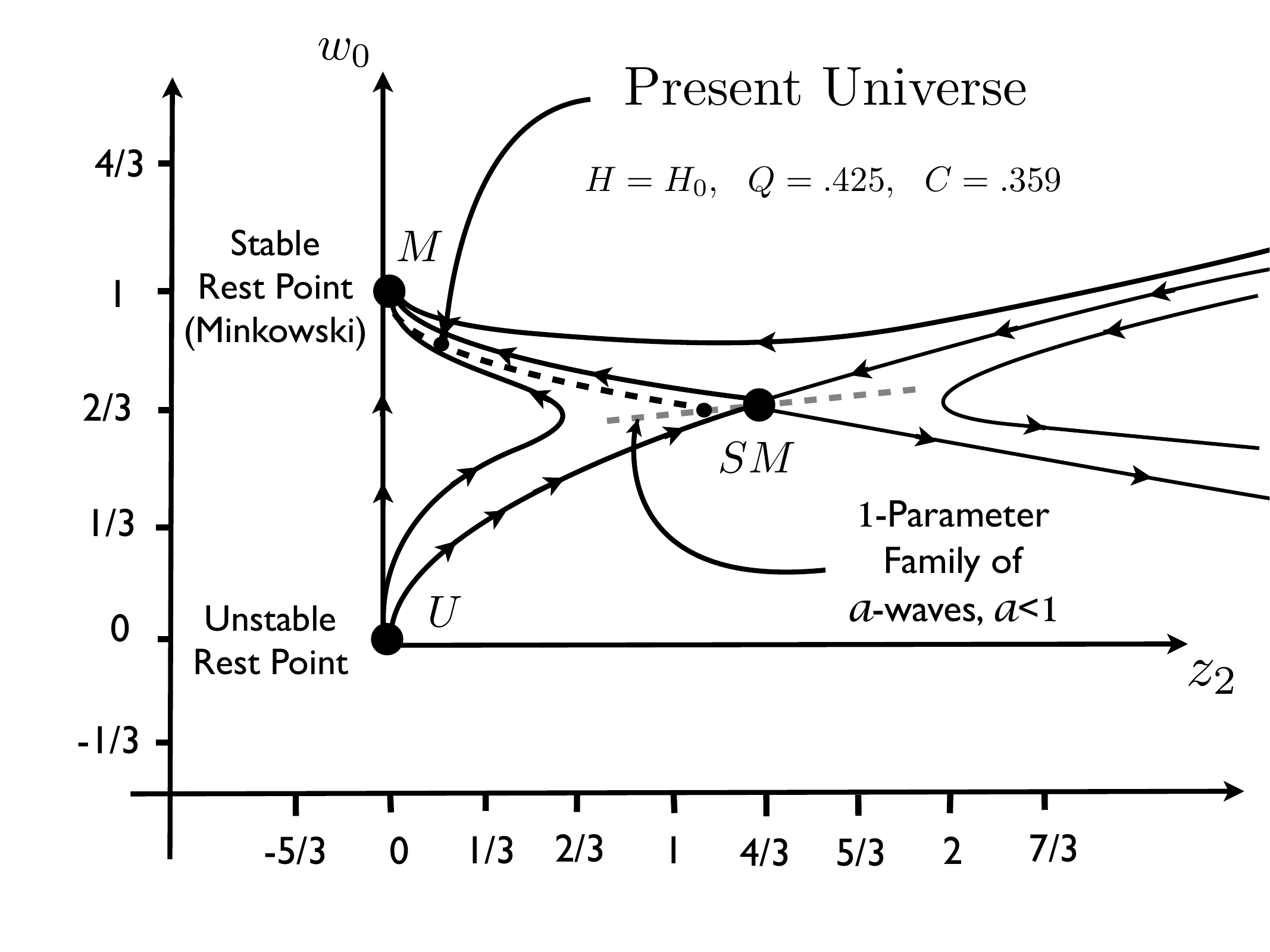}\label{Figure1}
\end{figure}

\begin{figure}
    \caption{The space $\mathcal{F}$ of solutions which decay to $M$.}
    \includegraphics[width=\textwidth]{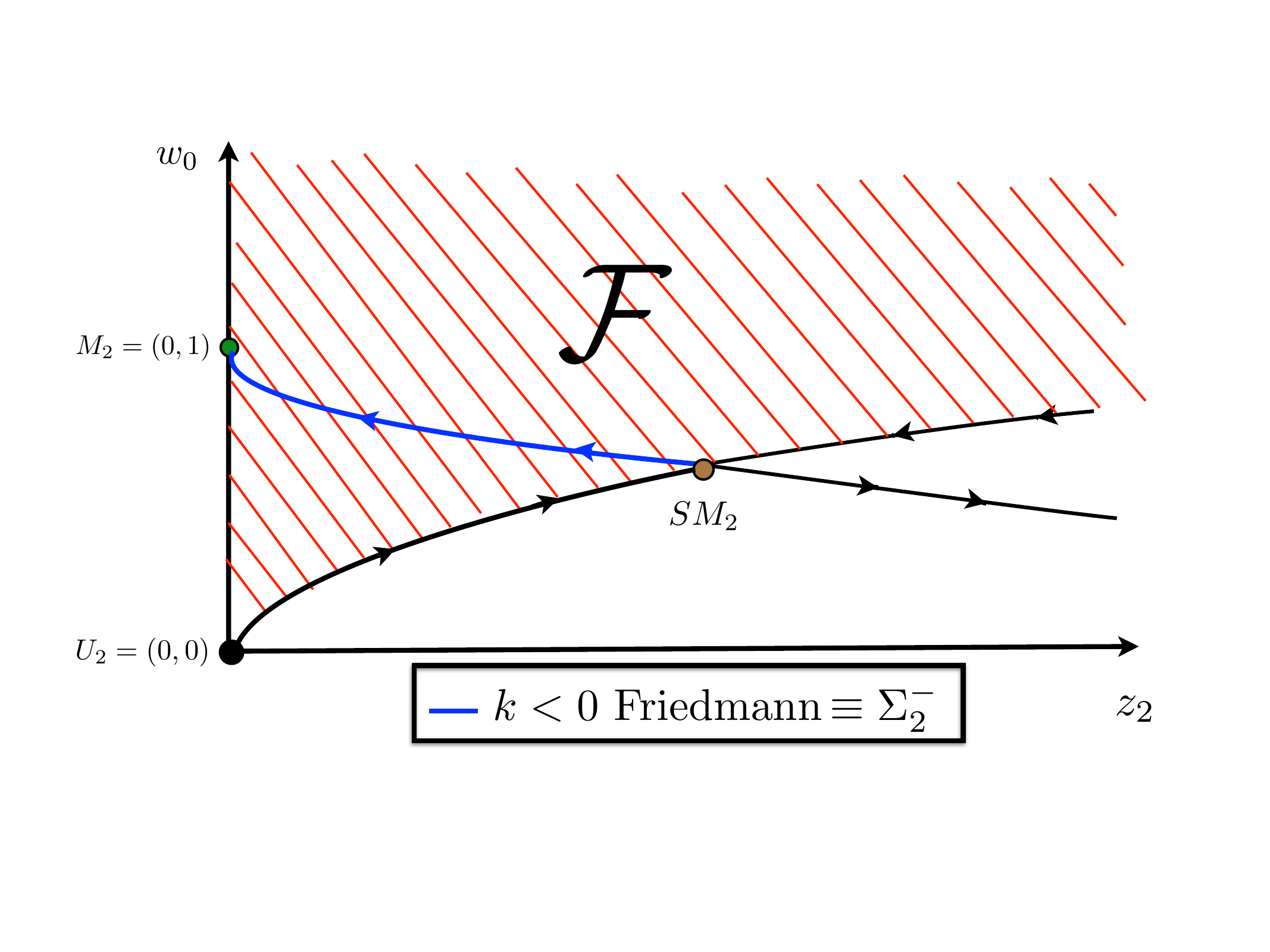}\label{Figure2}
\end{figure} 

The rest point $M$ describes the time asymptotic decay of solutions in $\mathcal{F}$ to Minkowski space as $t\to\infty$. A calculation shows that $M$ is a degenerate stable rest point with repeated eigenvalue $\lambda_M=-1$ and single eigenvector $\boldsymbol{R}_M=(0,1)^T$. Thus solutions in $\mathcal{F}$ decay to $M$ time asymptotically along the $w_0$-axis as $t\to\infty$. Moreover, as is standard for degenerate stable rest points with the character of $M$, $z_2(t)$ and $w_0(t)$ decay to $M$ at leading order like $O(t^{-1})$ and $O(t^{-1}\ln t)$ respectively. Thus, assuming solutions in $\mathcal{F}$ agree with Friedmann at leading order, but diverge at higher orders with errors estimated by Taylor's theorem, we can use the $n=1$ phase portrait of the STV-ODE alone to conclude that, to leading order as $t\to\infty$, every solution in $\mathcal{F}$ decays to $w=\frac{v}{\xi}=1$ and $z=0$ at the same rate for fixed $\xi$ and decays to Friedmann faster than to Minkowski for fixed $r$ (since $\xi\to0$). Theorem \ref{DegenerateM} below establishes that $M$ is a degenerate stable rest point in the phase portrait of the STV-ODE at every order $n\geq1$, exhibiting the same negative eigenvalue $\lambda_{M}=-1$ with a single eigenvector $\boldsymbol{R}_M$ at each order. The degenerate structure of rest point $M$ at all orders implies that the estimate for the discrepancy between a solution in $\mathcal{F}$ and the Friedmann spacetime it agrees with at leading order, is estimated by the discrepancy at second order as $t\to\infty$. We conclude that one would see perfect alignment between solutions in $\mathcal{F}$ and $k<0$ Friedmann at order $n=1$ and the error between them tends to zero by a factor $O(t^{-1})$ faster as $t\to\infty$ than what you would see without taking account of the decay of solutions to rest point $M$ at higher orders. The result, which uses standard rates of decay at degenerate stable rest points with the character of $M$, is recorded in the following theorem.

\begin{Thm}\label{decaytoM2}
    Imposing time since the Big Bang, each solution in $\mathcal{F}$ agrees exactly with a single $k<0$ Friedmann spacetime at leading order in the STV-ODE, Moreover, as $t\to\infty$:
    \begin{align}
        z &= z_2(t)\xi^2 + O(t^{-1})\xi^4,\\
        w &= w_0(t) + O(t^{-1}\ln t)\xi^2.
    \end{align}
    Using $z=\kappa\rho r^2+O(\xi^2)$ and $v=w\xi$, we conclude for a general solution in $\mathcal{F}$:
    \begin{align}
        \kappa\rho &= z_2(t)t^{-2} + O(t^{-5})r^2,\\
        v &= w_0(t)rt^{-1} + O(t^{-4}\ln t)r^3.
    \end{align}
    Thus, using the fact that for all solutions in $\mathcal{F}$, $w_0(t)$ and $z_2(t)$ agree at leading order with a Friedmann solution, the discrepancy between a general solution in $\mathcal{F}$ and the Friedmann solution it agrees with at leading order is estimated by:
    \begin{align}
        |\rho-\rho_F| &\leq O(t^{-3})\xi^2,\\
        |v-v_F| &\leq O(t^{-1}\ln t)\xi^3,
    \end{align}
    that is, exhibiting a faster decay rate by $O(\xi^2)$ with $\xi=\frac{r}{t}$ in the density than in the velocity at each fixed $r>0$ as $t\to\infty$. Furthermore, the rate of decay to Minkowski space is estimated by the rate of decay to $M$ at leading order, which is given by:
    \begin{align}
        |\rho| &\leq O(t^{-3})\\
        |w-1| &\leq O(t^{-1}\ln t),
    \end{align}
    as $t\to\infty$.
\end{Thm}

Note that in Theorem \ref{decaytoM2} the extra factor $t^{-1}$ in the density gives the faster rate of decay of $k<0$ Friedmann over the $O(t^{-2})$ decay rate known for the $k=0$ Friedmann spacetime. From this we conclude faster decay to Friedmann than to Minkowski and faster decay in the density than in the velocity.

The trajectory taking $SM$ to $M$ at order $n=1$ can be defined implicitly, which provides a canonical leading order evolution shared by all underdense solutions in $\mathcal{F}$, including $k<0$ Friedmann spacetimes. This spacetime is discussed in Subsection \ref{S8}. The rest point $SM$ persists to every order because the the critical Friedmann spacetime is self-similar at every order of the STV-ODE. Somewhat surprisingly, the crucial behavior of solutions in $\mathcal{F}$ emerges at order $n=2$.

\subsection{The STV-ODE Phase Portrait of Order $n\geq2$}\label{I5}

The nested property of the STV-ODE in (\ref{nbynsystemFirst}) implies that lower order eigenvalues of $SM$ persist to higher orders. We prove that two distinct additional eigenvalues always emerge at $SM$ in going from the STV-ODE of order $n-1$ to the STV-ODE of order $n$ for all $n\geq2$. These are given by the formulas:
\begin{align}
    \lambda_{An} &= \frac{2n}{3}, & \lambda_{Bn} &= \frac{1}{3}(2n-5).\label{eigenvaluesSM}
\end{align}
From (\ref{eigenvaluesSM}) we conclude that both eigenvalues $\lambda_{An}$ and $\lambda_{Bn}$ are positive except at orders $n=1$ and $n=2$. We argued above that $\lambda_{B1}$ is eliminated by changing to time since the Big Bang, an assumption equivalent to assuming a solution agrees with Friedmann at leading order $n=1$. At order $n=2$, $\lambda_{A2}=\frac{4}{3}>0$ and $\lambda_{B2}=-\frac{1}{3}<0$. A calculation also shows that at second order, the $k<0$ Friedmann solutions continue to lie on the trajectory associated with the leading order eigenvalue $\lambda_{A1}=\frac{2}{3}$. Recall that assuming time since the Big Bang eliminates the leading order negative eigenvalue by transforming the solution space to solutions which agree at order $n=1$ with either $SM$ or a trajectory in the unstable manifold of $SM$. The existence of the negative eigenvalue $\lambda_{B2}$ implies that $SM$ is an unstable saddle rest point with a one-dimensional stable manifold and a codimension one unstable manifold at each order $n\geq2$. Also recall that we let $\mathcal{F}_n'\subset\mathcal{F}_n$ denote the subset of solutions with trajectories in the unstable manifold of $SM$ identified as an $n-1$ dimensional space of trajectories taking $SM$ to $M$ in the phase portrait of the STV-ODE of order $n\geq2$. The appearance of one positive and one negative eigenvalue at order $n=2$, and only positive eigenvalues at higher orders, immediately implies the following theorem.

\begin{Thm}
    The unstable manifold $\mathcal{F}_n'\subset\mathcal{F}_n$ of $SM$ forms a codimension one set of trajectories in the STV-ODE at each order $n\geq2$ and a trajectory lies in $\mathcal{F}_n'$ at every order of the STV-ODE if and only if it lies in $\mathcal{F}_2'$. Moreover, trajectories in $\mathcal{F}'$ take $SM$ to $M$ at all orders of the STV-ODE, agree with a $k<0$ Friedmann spacetime at order $n=1$ in the limits $t\to0$ and $t\to\infty$ but are generically distinct from, and hence accelerate away from, $k<0$ Friedmann spacetimes at intermediate times in the phase portrait of the STV-ODE at every order $n\geq2$.
\end{Thm}

It follows from the theory of non-degenerate hyperbolic rest points that all solution trajectories in $\mathcal{F}_n'$ emerge tangent to the eigendirection associated with the smallest positive eigenvalue at $SM$. Of course at order $n=1$ the leading order eigenvalue associated with the solution trajectory of the $k<0$ Friedmann spacetime is the smallest eigenvalue and this remains the smallest positive eigenvalue at $n=2$. However, an eigenvalue smaller than this emerges at order $n=3$ namely, $\lambda_{B3}=\frac{1}{3}$. Thus solutions in $\mathcal{F}_n'$ generically emerge tangent to the leading order eigendirection of the $k<0$ Friedmann spacetimes only up to order $n=2$, but all solutions in $\mathcal{F}_n'$ which have a component of $\lambda_{B3}$ will emerge from $SM$ tangent to its eigenvector $\boldsymbol{R}_{B3}$, which is not tangent to the Friedmann direction $\boldsymbol{R}_{A1}$, corresponding to the leading order eigenvalue $\lambda_{A1}=\frac{2}{3}$ at $SM$. From this we establish that although all the solutions in $\mathcal{F}_n$ tend to rest point $M$ as $t\to\infty$, solutions in $\mathcal{F}_n$ in the complement of $\mathcal{F}_n'$, that is, those that miss $SM$ in backward time, follow the backward stable manifold of $SM$, consistent with the standard phase portrait picture of a non-degenerate saddle rest point. Since the only negative eigenvalue of $SM$ enters at order $n=2$, we can conclude that any solution that lies in the unstable manifold of $SM$ in the phase portrait of the STV-ODE at order $n=2$, also lies in the unstable manifold of $SM$ at all higher orders $n\geq3$. In this sense, the unstable manifold of $SM$ is characterized at order $n=2$. Note that all solutions of the STV-ODE in $\mathcal{F}$ which lie in the unstable manifold of $SM$, take $SM$ to $M$ in the phase portrait of the STV-ODE at all orders, and hence are bounded for all time $0\leq t\leq\infty$. Trajectories not in the unstable manifold of $SM$ miss $SM$ in backwards time in every STV-ODE of order $n\geq2$ and tend in backward time instead to the stable manifold associated with the unique negative eigenvalue, that is, a single trajectory. Thus $\mathcal{F}$, which consists of the domain of attraction of $M$ at every order, contains trajectories which do not emanate from $SM$, and hence correspond to a Big Bang at $t=0$ which is qualitatively different from the self-similar blow-up $t\to0$ at $SM$, and hence qualitatively different from the Big Bang observed in Friedmann spacetimes. An immediate conclusion of this analysis is a rigorous characterization the self-similar nature of the Big Bang in general spherically symmetric smooth solutions to the Einstein field equations when $p=0$.

\begin{Thm}\label{SSBigBang}
    When time is taken to be time since the Big Bang, solutions in $\mathcal{F}$ always exhibit self-similar blow-up in the $n=1$ phase portrait of the STV-ODE but will generically exhibit non-self-similar blow-up at all higher orders $n\geq2$.
\end{Thm}

Since all eigenvalues which emerge at $SM$ at orders above $n=2$ are positive, the unstable manifold of $SM$, and the entire phase portrait of the STV-ODE at higher orders, is determined from the $STV-ODE$ phase portrait at order $n=2$. In particular, the unstable manifold of $SM$ is determined at order $n=2$ in the sense that a solution in $\mathcal{F}$ lies in the unstable manifold of $SM$ at all orders $n\geq1$ if and only if it lies in the unstable manifold of $SM$ at order $n=2$. We now describe the phase portrait of the STV-ODE at order $n=2$ in detail.

\subsection{The STV-ODE Phase Portrait of Order $n=2$}\label{S1.6}

The corrections to $k<0$ Friedmann accounted for by solutions in $\mathcal{F}$ at order $n=2$ are most important as this is the order in which a negative eigenvalue emerges at $SM$ and also the leading order in which divergence from Friedmann spacetimes is observed. The order $n=2$ is important because it determines $w_2(t)\xi^2$, which provides the third order correction to redshift vs luminosity, the correction at the order of the anomalous acceleration of the galaxies which are purportedly accounted for by dark energy in the standard $\Lambda$CDM model of Cosmology \cite{SmolTeVo}.

The STV-ODE of order $n=2$ is the $4\times4$ system:\footnote{Note that this corrects an error in \cite{SmolTeVo} in the terms involving $z_4$, a mistake at fourth order in $\xi$ which did not affect the conclusions.}
\begin{align}
    t\dot{z}_2 &= 2z_2 - 3z_2w_0,\label{z2first2}\\
    t\dot{w}_0 &= -\frac{1}{6}z_2 + w_0 - w_0^2,\label{w0first2}\\
    t\dot{z}_4 &= \frac{5}{12}z_2^2w_0 - 5w_0z_4 + 4z_4 - 5z_2w_2,\label{z4first2}\\
    t\dot{w}_2 &= -\frac{1}{24}z_2^2 + \frac{1}{4}z_2w_0^2 - \frac{1}{10}z_4 - 4w_0w_2 + 3w_2.\label{w2first2}
\end{align}
Note first that the STV-ODE of order $n=1$ appears as the subsystem (\ref{z2first2})--(\ref{w0first2}). Viewed as a $4\times4$ autonomous system, (\ref{z2first2})--(\ref{w2first2}) admits the three rest points: $U=(0,0,0,0)$, $M=(0,1,0,0)$ and $SM=(\frac{4}{3},\frac{2}{3},\frac{40}{27},\frac{2}{9})$. Imposing time since the Big Bang restricts the solution space to solutions in the unstable manifold of $SM$ at order $n=1$ and this eliminates $U$ from the solution space. A calculation gives eigenvalues and eigenvectors of rest point $M$ and $SM$ in the STV-ODE of order $n=4$ as:
\begin{align*}
    \lambda_M &= -1, & \boldsymbol{R}_M &= (0,1,0,1)^T;\\
    \lambda_{A1} &= \frac{2}{3}, & \boldsymbol{R}_1 := \boldsymbol{R}_{A1} &= \bigg(-9,\frac{3}{2},-\frac{10}{3},-1\bigg)^T;\\
    \lambda_{B1} &= -1, & \boldsymbol{R}_{B1} &= \bigg(4,1,\frac{80}{9},1\bigg)^T;\\
    \lambda_{A2} &= \frac{4}{3}, & \boldsymbol{R}_3 := \boldsymbol{R}_{A2} &= (0,0,-10,1)^T;\\
    \lambda_{B2} &= -\frac{1}{3}, & \boldsymbol{R}_{B2} &= \bigg(0,0,\frac{20}{3},1\bigg)^T,
\end{align*}
where $\lambda_M$ and $\boldsymbol{R}_M$ correspond to the eigenvalue and eigenvector of $M$ and the rest to $SM$. The phase portrait for solutions of (\ref{z2first2})--(\ref{w2first2}) is depicted in Figure \ref{Figure3}. Note that $SM$ and $M$ are referred to as $SM_2$ (and $SM_4$) and $M_2$ (and $M_4$) in Figure \ref{Figure3} respectively to indicate that the fixed points are those for the $2\times2$ (and $4\times4$) system.

\begin{figure}
    \caption{The phase portrait for the $4\times4$ system.}
    \includegraphics[width=\textwidth]{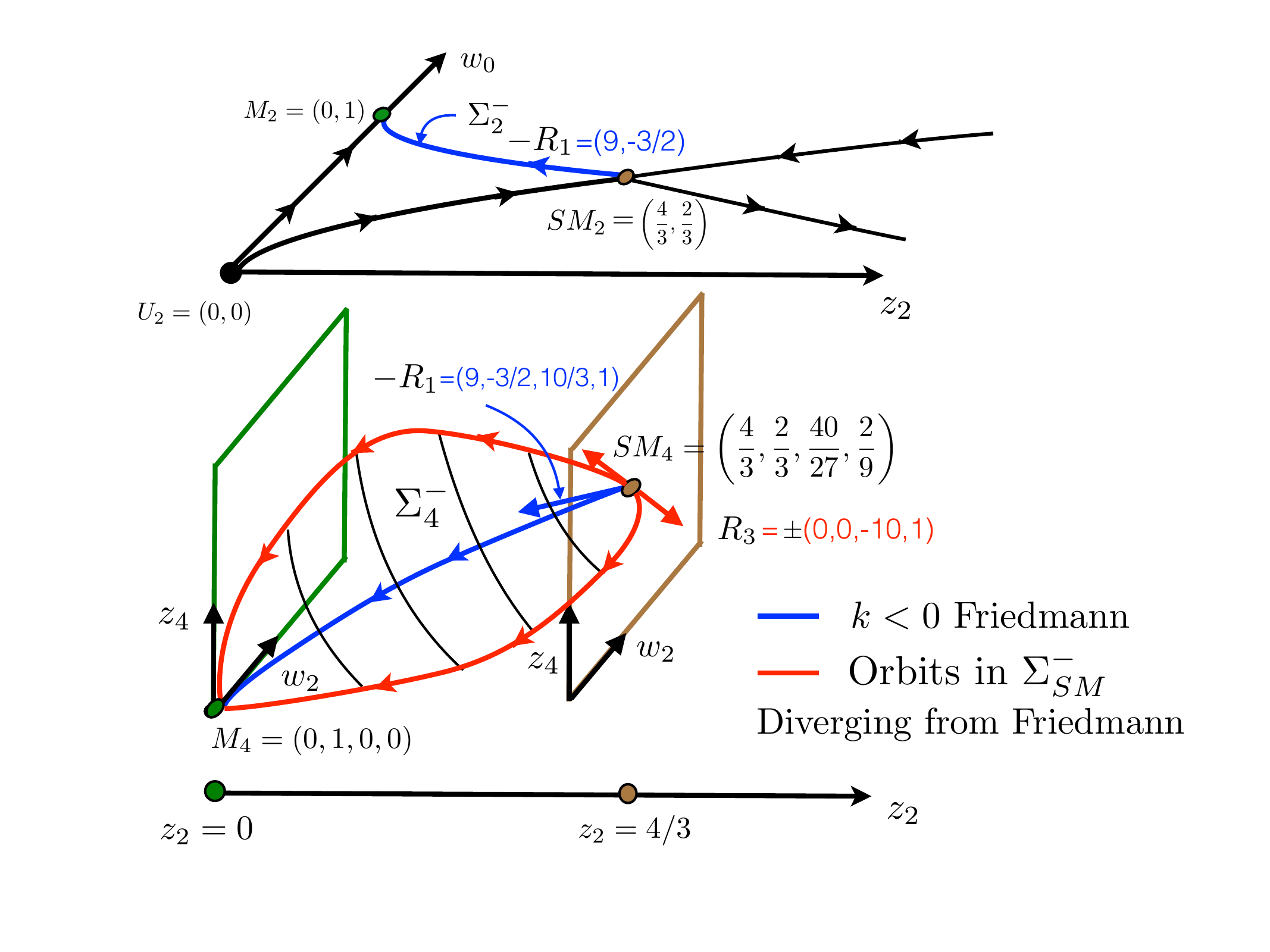}\label{Figure3}
\end{figure}

The time since the Big Bang gauge choice is assumed in Figure \ref{Figure3}, so all elements of $\mathcal{F}$ agree with $k<0$ Friedmann at leading order, that is, they lie on the unstable trajectory taking $SM$ to $M$ in the leading order phase portrait associated with (\ref{z2first2})--(\ref{w0first2}). This is denoted by $\Sigma_{SM}^-$ in Figure \ref{Figure3}, with $\Sigma_2^-$ and $\Sigma_4^-$ specifying the unstable manifold for the $2\times2$ and $4\times4$ systems respectively. The phase portraits of Figures \ref{Figure1} and \ref{Figure3} are consistent because the first two components of $-\boldsymbol{R}_1$ give the direction of the trajectory which connects $SM_2$ to $M_2$ at level $n=1$ and the second two components represent the higher order corrections. The projections of the $k<0$ Friedmann solutions onto solutions of the STV-ODE of orders $n=1$ and $n=2$ are represented by the blue curves in Figure \ref{Figure3}. Note that the presence of a second positive eigenvalue $\lambda_{A2}=\frac{4}{3}$ implies the unstable manifold of $SM$ intersects $\mathcal{F}$ in a two dimensional space of trajectories emanating from $SM$. We prove that only the eigensolution of $\lambda_{A1}$ corresponds to the Friedmann spacetime at order $n=2$. This immediately implies that there exist solutions in the unstable manifold of $SM$ which agree with a $k<0$ Friedmann spacetime at order $n=1$ but diverge from Friedmann at intermediate times. Since all solutions in $\mathcal{F}$ decay to $M$ as $t\to\infty$, this implies that solutions in the unstable manifold of $SM$ agree with $k<0$ Friedmann in the limits $t\to0$, $t\to\infty$ and at leading order $n=1$, but which diverge, and hence introduce accelerations away from Friedmann, at intermediate times. By the Hartman--Grobman Theorem, nonlinear solutions correspond to linearized solutions in a neighborhood of a rest point, so solutions in the unstable manifold of $SM$ are determined by their limiting eigendirection $\boldsymbol{R}=a\boldsymbol{R}_1+b\boldsymbol{R}_3$ at $SM$, and hence we conclude that the magnitude of the acceleration away from Friedmann is measured by $\frac{b}{a}$, which can be arbitrarily large. We state this precisely in the following theorem.

\begin{Thm}[Partial statement of Theorem \ref{UnstableManifoldOfSM-}]\label{UnstableManifoldOfSMIntro}
    All solutions in the unstable manifold of $SM$ at order $n=2$ leave $SM$ tangent to
    \begin{align}
        \boldsymbol{U}(\tau) = a e^{\lambda_{A1}(\tau-\tau_0)}\boldsymbol{R}_{A1} + b e^{\lambda_{A2}(\tau-\tau_0)}\boldsymbol{R}_{A2},
    \end{align}
    where
    \begin{align}
        \lambda_{A1} &= \frac{2}{3}, & \boldsymbol{R}_{A1} &= \left(\begin{array}{c}
        9\\
        -\frac{3}{2}\\
        \frac{10}{3}\\
        1
        \end{array}\right); & \lambda_{A2} &= \frac{4}{3}, & \boldsymbol{R}_{A2} &= \left(\begin{array}{c}
        0\\
        0\\
        -10\\
        1\end{array}\right);
    \end{align} 
    are the eigenpairs spanning the unstable manifold of the linearization of the $4\times4$ system of STV-ODE (\ref{z2first2})--(\ref{w2first2}) about the rest point $SM$ and $a$ and $\tau_0$ are fixed constants determined by the unique $k\neq0$ Friedmann spacetime at $n=1$. The constant $b$ is then a second free parameter which describes the instability of the $k\neq 0$ Friedmann spacetime at $SM$ at order $n=2$.
\end{Thm}

Note that the smallest positive eigenvalue to emerge at any order at $SM$ is $\lambda_{B3}=\frac{1}{3}$ and since all trajectories in $\mathcal{F}$ have a non-zero component of $\boldsymbol{R}_{A1}$ by definition, we can further conclude that all trajectories in the unstable manifold of $SM$ are tangent to $\boldsymbol{R}_{A1}$ in backward time at $SM$ in the phase portrait of the STV-ODE of order $n=2$ but come in tangent to $\boldsymbol{R}_{B3}$ at $SM$ in the portraits of the STV-ODE at all higher orders. We show in Theorem \ref{ThmPure} that the $k<0$ Friedmann solution has no components in direction $\boldsymbol{R}_{B3}$ and that by the nested structure of the STV-ODE, $\boldsymbol{R}_{Bn}$ has non-zero components in only leading order entries.

The existence of a unique negative eigenvalue $\lambda_{B2}=-\frac{1}{3}$ at order $n=2$ implies that $SM$ is an unstable saddle rest point, but not an unstable source. From this we conclude that not all solutions in $\mathcal{F}$ lie in the unstable manifold of $SM$, even though they decay time asymptotically to $M$ as $t\to\infty$. Since $SM$ is a saddle rest point, backward time trajectories in $\mathcal{F}$ starting near $SM$ will not typically tend to $SM$, but rather, indicative of the standard phase portrait of a saddle rest point, will generically follow the backward time trajectory of the stable manifold at $SM$. This implies the Big Bang is self-similar like the critical ($k=0$) Friedmann spacetime only at leading order $n=1$ but generically not self-similar at higher orders, as recorded in Theorem \ref{SSBigBang} above. We conclude that the time evolution of perturbations of $SM$ becomes indistinguishable from $k<0$ Friedmann solutions at late times after the Big Bang, agrees exactly with the same $k<0$ Friedmann solution in the leading order phase portrait, including the limits $t\to0$ and $\infty$, but introduce anomalous accelerations away from $k<0$ Friedmann spacetimes at intermediate times, starting at order $n=2$. This provides a rigorous mathematical framework and mechanism for determining and explaining the source of the corrections to redshift vs luminosity computed numerically in \cite{SmolTeVo}, that is, created by the instability of $SM$.

The new parameter $\beta$ associated with $(\lambda_{A2},\boldsymbol{R}_{A2})$ naturally introduces accelerations away from Friedmann spacetimes at order $\xi^2$ in $w$, and hence order $\xi^3$ in the velocity $v$. These mimic the effects of a cosmological constant at third order in redshift factor $\rm{z}$ vs luminosity distance $d_\ell$, as measured from the center \cite{SmolTeVo}. This is the order of the discrepancy between the prediction of Friedmann spacetimes with a cosmological constant and Friedmann spacetimes without one. According to Figure \ref{Figure3}, at late times after the Big Bang we should expect to observe spacetimes close to $k<0$ Friedmann, but not $k=0$. Moreover, perturbations from $k<0$ Friedmann spacetimes, including perturbations of $SM$ on the $k<0$ side of $\mathcal{F}$ at early times after the Big Bang, diverge from $k<0$ Friedmann spacetimes at intermediate times before they decay back to $k<0$ Friedmann at late times. Regarding the intermediate times, the new free parameter $\beta$ associated with the unstable manifold of the critical $k=0$ Friedmann spacetime ($SM$) is not present in pure $k<0$ Friedmann spacetimes and this effect appears to mimic the effects of a cosmological constant at the order (third order in redshift factor looking out from the center) at which the predictions of a cosmological constant diverge from the predictions of the Friedmann spacetimes without one. Said differently, this theory identifies a one parameter family of corrections to Friedmann at order $n=2$, with further corrections to Friedman determined by the positive eigenvalues of $SM$ at higher orders, the higher the order the smaller the correction near the center. 

More generally, it was proven in \cite{SmolTeVo} that the order in redshift factor in the relation between redshift and luminosity looking out from the center of a spherically symmetric spacetime, is at the same order as $\xi$ in our theory here. The Hubble constant and the quadratic correction to redshift vs luminosity is determined at order $n=1$ from $v=w_1\xi$ and $z_2\xi^2$ respectively, and hence $w_2\xi^3$ determines the third order correction in red-shift factor with $z_4\xi^4$ determining the fourth order term. Since one requires values of $w_2$ and $z_4$ to determine whether a solution trajectory lies in $\mathcal{F}'$, it follows that the fourth order correction to red-shift vs luminosity would be required to determine whether or not a cosmology lies in the unstable manifold of $SM$, that is, to determine whether the Big Bang is self-similar like $SM$ at all orders, or whether it diverges from self-similarity at order $n=2$. We conclude that the family $\mathcal{F}$ extends the $k<0$ Friedmann spacetimes to a stable family of spacetimes, closed under small perturbation, which characterizes the instability of the critical ($k=0$) and underdense ($k<0$) Friedmann spacetimes to smooth radial perturbations.

\subsection{The Canonical Spacetime at Order $n=1$}

When time since the Big Bang is imposed, every trajectory of the STV-ODE of order $n=1$ reduces to $SM$ or to a trajectory in its unstable manifold. In the underdense case, this is the unique trajectory which takes $SM=(\frac{4}{3},\frac{2}{3})$ to $M=(0,1)$ in the limit $t\to\infty$ at order $n=1$. This trajectory, which we label $(z_2^F(t),w_0^F(t))$, and its log-time translations provide a canonical leading order evolution shared by all underdense solutions in $\mathcal{F}$, including $k<0$ Friedmann spacetimes. The resulting evolution $\kappa\rho(t)=z^F_2(t)t^{-2}$, $v(t)=w^F_0(t)rt^{-1}$ is therefore an explicit spacetime which is in a sense more fundamental than $k<0$ Friedmann spacetimes because it is shared, under log-time translation, by all underdense solutions to leading order. In \cite{SmolTeVo} we proved that $\Delta_0$ and present time $t_0$ in this leading order evolution are sufficient to uniquely determine the Hubble constant $H_0$ and quadratic correction $Q$ in the relationship between redshift factor $\rm{z}$ vs luminosity distance $d_\ell$,
\begin{align*}
    H_0^{-1}d_\ell = \rm{z} + Q\rm{z}^2,
\end{align*}
as measured at the center of the spacetime, so long as $0.25\leq Q\leq0.5$. Serendipitously, this latter constraint was shown in \cite{SmolTeVo} to be consistent with $k=0$ Friedmann augmented with 70\% dark energy, the assumption of the $\Lambda CDM$ model. Because of its fundamental nature, it is useful to have an explicit formula for $(z_2^F(t),w_0^F(t))$, which is provided by Theorem \ref{Friedmannknotzeroexpansion} by extracting the leading order evolution from a known implicit formula for $k<0$ Friedmann spacetimes \cite{gronhe}. The result is restated in the following theorem.

\begin{Thm}[Partial statement of Theorem \ref{Friedmannknotzeroexpansion}]\label{2x2connectingOrbitIntro} 
    Define $\theta:(0,\infty)\to(0,\infty)$ by
    \begin{align*}
        \theta(s) &= f^{-1}(s),
    \end{align*}
    where
    \begin{align*}
        s = f(\theta) = \frac{1}{2}(\sinh2\theta-2\theta).
    \end{align*}
    Then the functions
    \begin{align}
        z_2^F(t) &= \tilde{z}_2(\theta(t)) = \frac{6\big(\sinh2\theta(t)-2\theta(t)\big)^2}{\big(\cosh2\theta(t)-1\big)^3},\\
        w_0^F(t) &= \tilde{w}_0(\theta(t)) = \frac{\big(\sinh2\theta(t)-2\theta(t)\big)\sinh2\theta(t)}{\big(\cosh2\theta(t)-1\big)^2},
    \end{align}
    provide exact formulas for the particular solution of the $2\times2$ system (\ref{z2first1})--(\ref{w0first1}) which traverses the trajectory connecting $SM$ to $M$ in Figure \ref{Figure1}, that is, 
    \begin{align*}
        \lim_{t\to0}\big(z_2^F(t),w_0^F(t)\big) &= \bigg(\frac{4}{3},\frac{2}{3}\bigg) = SM, & \lim_{t\to\infty}\big(z_2^F(t),w_0^F(t)\big) &= (0,1) = M.\\
    \end{align*}
    Moreover, $(z_2^F(t),w_0^F(t))$ is the leading order term in expansion (\ref{zfirst1})--(\ref{wfirst1}) of the $k<0$ Friedmann solution assuming $t$ is time since the Big Bang and $\Delta_0=\frac{4}{9}$, where $\Delta_0=\frac{\kappa}{3}\rho_0R_0^{3}$ parameterizes the $k\neq0$ Friedmann solutions in their standard formulation. Furthermore, the corresponding formula for the leading order part of a Friedmann spacetime in terms of general $\Delta_0>0$ is then given by: 
    \begin{align}
        z_2^{F}\bigg(\frac{t}{\Delta_0}\bigg) &= \tilde{z}_2\bigg(\theta\bigg(\frac{t}{\Delta_0}\bigg)\bigg), & w_0^F\bigg(\frac{t}{\Delta_0}\bigg) &= \tilde{w}_0\bigg(\theta\bigg(\frac{t}{\Delta_0}\bigg)\bigg).
    \end{align}
\end{Thm}

The approximate solution of the STV-PDE which corresponds to the $n=1$ trajectory $(z_2^F(t),w_0^F(t))$ is
\begin{align*}
    \boldsymbol{U}_F(t,\xi) = (z_F(t,\xi),w_F(t,\xi)) \approx \big(z_2^F(t)\xi^2,w_0^F(t)\big).
\end{align*}
Substituting $z=\kappa\rho r^2$ and $v=w\xi$ gives the equivalent SSC approximate solution 
\begin{align*}
    \boldsymbol{W}_F(t,r) = (\kappa\rho_F(t,r),v_F(t,r)) \approx \big(z_2^F(t)r^2t^{-2},w_0^F(t)rt^{-1}\big).
\end{align*}
The approximate solutions $\boldsymbol{U}_F(t,\xi)$ and $\boldsymbol{W}_F(t,r)$ describe the time asymptotics of solutions in $\mathcal{F}$ as recorded in the following corollary.

\begin{Corollary}
    Let $\boldsymbol{U}(t,\xi)=(z(t,\xi),w(t,\xi))$ be a solution in $\mathcal{F}$ which determines $\boldsymbol{W}(t,r)=(\kappa\rho(t,r),v(t,r))$ through $z=\kappa\rho r^2$ and $v=w\xi$. Then there exists a $\Delta_0>0$ such that
    \begin{align*}
        \boldsymbol{W}(t,r) \to W_F\bigg(\frac{t}{\Delta_0},r\bigg)
    \end{align*}
    as $t\to\infty$ at each fixed $r>0$, with errors $O(t^{-5}r^{-2})$ and $O(t^{-3}r^3)$ in $\rho$ and $v$ respectively; and
    \begin{align*}
        \boldsymbol{W}(t,r) \to W_F\bigg(\frac{t}{\Delta_0},r\bigg)
    \end{align*}
    as $r\to0$ at each fixed $t>0$, with errors $O(t^{-4}r^2)$ and $O(t^{-3}r^3)$ in $\rho$ and $v$ respectively. Moreover, $\boldsymbol{U}_F(\frac{t}{\Delta_0},r)$ agrees to the same orders with the unique $k<0$ Friedmann spacetime determined by $\Delta_0$.
\end{Corollary}

\subsection{Conclusions}

The Friedmann spacetimes in the limit $p=0$, with or without a cosmological constant, have been the accepted large scale model for late stage Big Bang Cosmology since Hubble's measurement of the expanding Universe in 1929. In the modern theory of Cosmology, the zero pressure Friedmann model applies after the time when the pressure drops precipitously to zero, some 10,000 years after the Big Bang, about an order of magnitude before the decoupling of radiation and matter gives rise to the microwave background radiation \cite{Peacock}. The widely accepted $\Lambda CDM$ standard model for the large scale expansion of the Universe is a critically expanding $k=0$ Friedmann spacetime with dark energy modeled by a positive cosmological constant, which is negligible relative to the energy density at early times. In this model, dark energy accounts for approximately 70\% of the energy density of the Universe at present time \cite{SmolTeVo}. We propose $\mathcal{F}$ as an extension of the $k\leq0$ Friedmann spacetimes to a stable family of cosmological models which reduce to $k<0$ Friedmann in the time asymptotic limit $t\to\infty$ (for fixed $r$) but naturally introduce \emph{anomalous} accelerations relative to the Friedmann spacetimes at early and intermediate times into the dynamics of solutions of the Einstein field equations, without recourse to a cosmological constant. This confirms mathematically that a direct consequence of Einstein's original theory of General Relativity, without a cosmological constant or dark energy, is that one can expect to observe a close approximation to non-critical ($k\neq0$) Friedmann spacetimes at late times after the Big Bang, but not critical ($k=0$) Friedmann spacetimes.\footnote{The density is too large for a cosmological constant of the current observed magnitude to influence the stability of $SM$ early on during the Big Bang at the onset of the instability when the pressure drops to zero \cite{SmolTeVo}.}

The presence of solutions in the stable manifold $\mathcal{F}'$ of $SM$ that are different from $k<0$ Friedmann spacetimes tells us that general perturbations of $k<0$ Friedmann spacetimes in $\mathcal{F}'$, as well as underdense perturbations of the $k=0$ Friedmann spacetime, at early times after the Big Bang produce accelerations away from Friedmann solutions before they decay back to $k<0$ Friedmann as $t\to\infty$ (for fixed $r$). Thus anomalous accelerations away from $k\leq0$ Friedmann spacetimes are not a violation, but a prediction of Einstein's original theory of General Relativity without a cosmological constant, and such does not change the picture during the early epoch when the cosmological constant is negligible relative to the energy density.\footnote{As in \cite{SmolTeVo}, it is interesting to consider whether this might explain some of the conundrums with the Standard Model of Cosmology, such as the variable cosmological constant, the flatness problem or the uniform temperature problem. The present paper focuses only on the mathematics, the intention is to address the aforementioned problems in future publications.}

The definitive description of the instability of the critical Friedmann spacetime in terms of the family of spacetimes $\mathcal{F}$ set out here provides new insights for exploring the hypothesis that the observed anomalous acceleration of the galaxies might be explained within Einstein's original theory of General Relativity without the cosmological constant, that is, the possibility that the Universe on the largest scale has evolved from a smooth perturbation of $SM$ shortly after the Big Bang, such that the resulting solution lies within the family $\mathcal{F}$, with our galaxy near the center of that expansion. This possibility appears more intriguing for two reasons: First, the instability of $SM$ to perturbations at every order makes the $k=0$ Friedmann spacetime implausible as a physically observable model, with or without dark energy, and second, our theory here establishes that solutions in $\mathcal{F}$, the space of solutions into which underdense perturbations of $SM$ will evolve, generically admit solutions which accelerate away from $k<0$ Friedmann spacetimes as they evolve away from $SM$ and before they decay back to $k<0$ Friedmann solutions as $t\to\infty$ at all orders of the STV-ODE, qualitatively rich enough to mimic the effects of dark energy. Moreover, as demonstrated in \cite{SmolTeVo}, the acceleration at the order of the quadratic term in lies within the narrow range $0.25\leq Q\leq0.5$, consistent with the effects of a cosmological constant $\Lambda\approx0.7$.

We comment that in \cite{smolteMemoirs} the authors also proposed a \emph{wave model} alternative to dark energy in which a self-similar solution of the $p\neq0$ perfect fluid Einstein field equations, interpreted as a local time-asymptotic wave pattern at the end of the Radiation Dominated Epoch,\footnote{When the pressure drops precipitously to $p=0$ at about 10,000 years after the Big Bang, an order of magnitude before the uncoupling of matter and radiation around 300,000 years after the Big Bang.} induces an underdensity which triggers an instability in $SM$ when the pressure drops to zero. This was modeled as a mechanism for creating the observed anomalous acceleration of the galaxies observed at present time. For this, authors in \cite{smolteMemoirs} identified a one parameter family of self-similar solutions of the perfect fluid Einstein field equations, parameterized by the so-called \emph{acceleration parameter} $a$, such that $a=1$ is the $k=0$ Friedmann spacetime with equation of state $p=\frac{c^2}{3}\rho$.\footnote{This is the equation of state for the state of matter known as \emph{pure radiation}, as well as the equation of state for the extreme relativistic limit of free particles \cite{wein}.} The authors proposed solutions in this family as candidates for time asymptotic wave patterns at the end of the Radiation Dominated Epoch of the Big Bang. These self-similar solutions produce perturbations of $SM$ at the end of the Radiation Dominated Epoch and the authors introduced and employed a self-similar formulation of the SSC equations to numerically evolve the resulting perturbations up through the Matter Dominated Epoch to present time. As a result of this, a unique value of the acceleration parameter was identified which produced the correct Hubble constant $H_0$ and quadratic correction $Q$ to redshift vs luminosity at present time. From this, a prediction was made at the third order $C_3$ in the redshift factor in the expansion of redshift vs luminosity \cite{smolteMemoirs},
\begin{align*}
    H_0^{-1}d_\ell = \rm{z} + Q\rm{z}^2 + C_3\rm{z}^3 + C_4\rm{z}^4 + \dots,
\end{align*}
where $\rm{z}$ is the redshift factor and $d_\ell$ is the luminosity. This was compared with the predictions of dark energy. The main principle is that an observer looking outward from the center at time $t=t_0$ into a spacetime evolving as a solution in the family $\mathcal{F}$ will measure the $n^{th}$ order correction to redshift vs luminosity as a function of the coefficient of $\xi^n$ among $v_{2n-2}(t_0)\xi^{2n-1}=w_{2n-2}(t_0)\xi^{2n-2}$ and $z_{2n}(t_0)\xi^{2n}$. Thus, in principle, there are the same number of parameters in an expansion of $H_0^{-1}d_\ell$ in powers of redshift $\rm{z}$ as there are initial conditions which can be freely assigned to determine a solution in $\mathcal{F}$. Thus the first through fourth order corrections are determined by $z_2$, $w_0$, $z_4$ and $w_2$ respectively, all determined by the STV-ODE of order $n=2$.

Note that the unstable manifold $\mathcal{F}'\subset\mathcal{F}$ at order $n=2$ of the STV-ODE is a two parameter surface in which the $k<0$ Friedmann solutions account for only one of the two parameters, so $k<0$ Friedmann spacetimes can only account for $H_0$ and $Q$, and then $C_3$ is determined from these. The freedom to allow a two-parameter unstable manifold at order $n=2$ allows one to freely assign $C_3$, but this then constrains the value of $C_4$. Finally, the freedom to assign $z_4$ and $w_2$ as two free parameters, in addition to $z_2$ and $w_0$, is the right number of initial conditions to determine a general solution in $\mathcal{F}$ at order $n=2$. This then allows for the freedom to assign $C_4$ as well. Although we confine ourselves here to the mathematics, our intention is to further explore the thesis proposed in \cite{SmolTeVo}, that is, that the instability of $SM$ alone might account for experimental incongruencies, like a variable cosmological constant, associated with the observed anomalous acceleration of the galaxies, within Einstein's original theory, without recourse to a cosmological constant.

As a final comment, we note that the STV-ODE describe the evolution of solutions along each line $\xi=\xi_0$, with $\xi_0$ constant, and determine the time asymptotics of solutions implicitly from initial data starting from arbitrary initial time $t_0>0$, so any boundary condition at infinity is free to be imposed. Indeed, for solutions $\mathcal{F}'$ starting at time $t=t_0$ on the underdense side of the stable manifold of $SM$ in the $n=1$ phase portrait in Figure \ref{Figure2}, the limit at $t=\infty$ is the rest point $M$. The rest point $M$ emerges implicitly from the analysis in SSCNG coordinates and we surmise that this boundary condition would be difficult to guess ahead of time to impose as a boundary condition in different coordinate systems. For example, decay to rest point $M$ implies that the velocity $v$ aligns with $\xi=\frac{r}{t}$ as the density tends to zero in smooth solutions of the Einstein field equations. It is interesting to note that when $v\approx\xi$, the SSCNG time coordinate diverges from comoving time except at $r=0$, so the time asymptotics of the velocity would be difficult to guess from knowledge of solutions given in a comoving coordinate system alone. The effect of imposing any other boundary condition at infinity in a different coordinate system, like Lemaître--Tolman--Bondi coordinates, would necessarily break the smoothness condition at $r=0$, leading to a singularity at the origin \cite{roma}. Moreover, the perturbations in $\mathcal{F}$ do not represent simple under-densities relative to the critical $k=0$ Friedmann solution because setting $k=0$ as a boundary condition at infinity would not in general be consistent with solutions in $\mathcal{F}$. For one thing, solutions in the unstable manifold of $SM$ which decay back to $M$ as $t\to\infty$ remain aligned with $k<0$ Friedmann at leading order but incur an advancing or retarding of $w_2$ which will be significant at intermediate times before the decay of the trajectory to $M$ takes over. In fact, this effect is exactly at the order of the measured anomalous acceleration \cite{smolteMemoirs}. The purpose of the present paper is to complete the mathematical theory of the pressureless self-similar Einstein field equations, give a definitive characterization of the instability of the critical Friedmann spacetime and to give a global description of underdense perturbations of $SM$ in terms of the family $\mathcal{F}$. Authors will return to the problem of modeling redshifts in a subsequent publication.

\subsection{Summary}

In summary, the family $\mathcal{F}$ characterizes the instability of the critical $k=0$ Friedmann spacetime and the accelerations away from Friedmann are described quantitatively at every order by the STV-ODE. We identify a new positive and negative eigenvalue at $SM$ in the phase portrait of the STV-ODE of order $n=2$, different from the leading order eigenvalue associated with $k<0$ Friedmann spacetimes. The positive eigenvalue produces a new free parameter which generates accelerations away from $k<0$ Friedmann spacetimes within the unstable manifold $\mathcal{F}'$ of $SM$ and the negative eigenvalue at $SM$ produces accelerations away from $k<0$ Friedmann spacetimes in $\mathcal{F}$, outside the unstable manifold $\mathcal{F}'$. This directs us to an underlying mechanism which produces a consequential third order correction to redshift vs luminosity, the order of the discrepancy associated with dark energy, both within $\mathcal{F}'$ and its complement $\mathcal{F}\setminus\mathcal{F}'$, different from the predictions made by $k=0$ Friedmann spacetimes. In particular, this provides a deeper mathematical understanding of the source of the third order correction computed numerically in \cite{SmolTeVo}. Such accelerations away from $k\leq0$ Friedmann spacetimes are shown to be triggered by arbitrarily small perturbations of $SM$, a significant consequence of the instability of the $k=0$ Friedmann spacetime to the subject of Cosmology. The authors intend to address the physical redshift vs luminosity problem quantitatively from this point of view in a forthcoming paper.

\subsection{Outline of Paper}

In Section \ref{S2} we summarize the main results in this paper. In Section \ref{S3} we explain the condition for a spherically symmetric solution of the Einstein field equations to be smooth at the center of symmetry. Our requirement for smoothness at the center is simply that the nonzero terms in the expansion of the solution in powers of $\xi$ should contain only even powers $\xi^{2n}$. In particular, this implies that smooth solutions solve the STV-ODE at each order $n$. In Section \ref{TSBB} we demonstrate that the time since the Big Bang gauge forces every trajectory at order $n=1$ to agree with a $k\neq0$ Friedmann spacetime. In Section \ref{S4} we review the Friedmann spacetimes in comoving coordinates and derive simple general formulas for coordinate transformations which take spherically symmetric metrics to SSCNG. In Section \ref{S5} we give a new simpler proof of the self-similarity of the $k=0$ Friedmann spacetime in SSCNG coordinates for equations of state of the form $p=\sigma\rho$. In Section \ref{S6} we present a new derivation of the STV-PDE, which includes the case of non-zero pressure by incorporating the equation of state $p=\sigma\rho$ into the equations. In Section \ref{S7} we derive the STV-ODE of order $n=2$ by expanding the STV-PDE in even powers of $\xi$. In Section \ref{S8} we discuss the STV-ODE at orders $n=1$ and $n=2$ and we incorporate the $k\leq0$ Friedmann solutions into these systems. In Section \ref{S9} we characterize the unstable manifolds of $SM$ at orders $n=1$ and $n=2$ and identify a new free parameter $\beta$ in the unstable manifold of $SM$ at order $n=2$, not accounted for by $k<0$ Friedmann solutions. In Section \ref{S11} we discuss the higher order STV-ODE and prove that all solutions which tend to the rest point $M$ at order $n=1$ also converge to $M$ at all higher orders $n>1$. In Section \ref{S11.5} we prove that the $k<0$ Friedmann spacetimes are pure eigensolutions of the STV-ODE up to order $n=3$. The first appendix, Section \ref{Appendix1} are where some of the more technical proofs are given. The second appendix, Section \ref{S12}, is where we make the connections between the Theorems stated in \cite{SmolTeVo}, which used a different notation and were stated without proof, and the results in this paper. Finally, in the third appendix, Section \ref{Appendix2}, we discuss Lemaître--Tolman--Bondi coordinates.

\section{Statement of Results}\label{S2}

We start by recording the following result, which asserts that self-similar coordinates are valid out to approximately the Hubble radius.

\begin{Thm}[Informal statement of Theorems \ref{ThmLog-time translation} and \ref{solvablexi}]
    The mapping $(t,r)\to(t,\xi)$ is a regular one-to-one mapping of the SSC $(t,r)$ to self-similar coordinates $(t,\xi)$ for all $|\xi|\leq\xi_0\approx0.816$.
\end{Thm}

In the theorem below we derive our most general version of the STV-PDE, that is, for perfect fluid spacetimes with equation of state $p=\sigma\rho$ with $\sigma$ constant.

\begin{Thm}[Partial statement of Theorem \ref{thmsigma}]\label{thmsigmainitial2}
    Assume the equation of state $p=\sigma\rho$ with constant $\sigma$. Then for $|\xi|<\xi_0$, the perfect fluid Einstein field equations with an SSC metric are equivalent to the following four equations in unknowns $A(t,\xi)$, $D(t,\xi)$, $z(t,\xi)$ and $w(t,\xi)$:
    \begin{align}
        \xi A_\xi &= -z + (1-A),\\
        \xi D_\xi &= \frac{D}{2A}\bigg(2(1-A)-(1-\sigma^2)\frac{1-v^2}{1+\sigma^2 v^2}z\bigg),\\
        tz_t + \xi\big((-1+Dw)z\big)_\xi &= -Dwz,\\
        tw_t + (-1+Dw)\xi w_\xi &- w + Dw^2 - \frac{(1+\sigma^2)\sigma^2}{\xi z}\bigg(D\frac{(1-v^2)v^2}{(1+\sigma^2v^2)^2}z\bigg)_\xi\notag\\
        &+ \frac{\sigma^2\xi}{z}\bigg(D\frac{1-v^2}{1+\sigma^2v^2}\frac{z}{\xi^2}\bigg)_\xi = \text{RHS},
    \end{align}
    where
    \begin{align*}
        \text{RHS} = -\frac{1}{\xi^2}\frac{1-v^2}{1+\sigma^2v^2}\frac{D}{2A}\bigg((1-\sigma^2)(1-A)+2\sigma^2\frac{1-v^2}{1+\sigma^2v^2}z\bigg)
    \end{align*}
    and
    \begin{align}
        w = \frac{1+\sigma^2}{1+\sigma^2v^2}\frac{v}{\xi}.
    \end{align}
\end{Thm}

Theorem \ref{thmsigmainitial2} reduces to Theorem \ref{STV-PDE} of the Introduction when $p=\sigma=0$, applicable to late time Big Bang Cosmology \cite{Peacock}, the setting of this paper.

The authors' original motivation to formulate a self-similar version of the Einstein field equations was the discovery that, when $p=\sigma\rho$ and the usual gauge of proper time at $r=0$ with the Big Bang at $t=0$ is employed, the $k=0$ Friedmann spacetime in SSC has the property that all of the variables $A$, $D$, $z$ and $w$ are functions of $\xi=\frac{r}{t}$ alone, and hence represent a time independent solution, or \emph{rest point}, of the STV-PDE, suggesting to the authors that such a PDE would be useful in studying the stability properties of the Standard Model of Cosmology. To state this precisely, we begin with the exact expression for the $p=\sigma\rho$, $k=0$ Friedmann spacetimes in comoving coordinates $(t,r)$, which is given by Theorem 2 on page 88 of \cite{smolteMAA}.

\begin{Thm}[Partial statement of Theorem \ref{ThmCosWithShockMAA}]
    In comoving coordinates $(t,r)$, the $k=0$ Friedmann metric with equation of state $p=\sigma\rho$ takes the form
    \begin{align}
        ds^2 = -dt^2 + R(t)^2\big(dr^2+r^2d\Omega^2\big),\label{metricFriedmann}
    \end{align}
    where:
    \begin{align}
        R(t) &= t^{\frac{2}{3(1+\sigma)}},\\
        H(t) &= \frac{2}{3(1+\sigma)t},\\
        \rho(t) &= \frac{4}{3\kappa(1+\sigma)^2t^2},
    \end{align}
    and the four velocity $\vec{u}$ satisfies
    \begin{align}
        \vec{u} = (u^0,u^1,u^2,u^3) = (1,0,0,0).
    \end{align}
\end{Thm}

To describe the mapping $(t,r)\to (\bar{t},\bar{r})$, which we prove takes (\ref{metricFriedmann}) to SSC self-similar form, we note first that the SSC radial coordinate must be $\bar{r}=Rr$ to match the spheres of symmetry, that is,
\begin{align}
    \bar{r} = R(t)r = t^{\frac{\alpha}{2}}r,\label{define-r}
\end{align}
where
\begin{align*}
    \alpha = \frac{4}{3(1+\sigma)}.
\end{align*}
Next, we introduce the auxiliary variable $\eta=\frac{\bar{r}}{t}$, so that $\bar{r}=\eta t$, and define the SSC time variable $\bar{t}$ in terms of $\eta$ by
\begin{align}
    \bar{t} = F(\eta)t,
\end{align}
where
\begin{align}
    \mathcal{F}(\eta) = \bigg(1+\frac{\alpha(2-\alpha)}{4}\eta^2\bigg)^{\frac{1}{2-\alpha}}.\label{defineF}
\end{align}
Then (\ref{define-r})--(\ref{defineF}) define a mapping $(t,r)\to(\bar{t},\bar{r})$ by
\begin{align}
    (\bar{t},\bar{r}) = (F(\eta)t,\eta t) = \big(F\big(rt^{\frac{\alpha}{2}-1}\big)t,rt^{\frac{\alpha}{2}}\big) =: \Phi(t,r).\label{ComovToSSC}
\end{align}
The following theorem establishes that the mapping $\Phi:(t,r)\to(\bar{t},\bar{r})$ converts the $p=\sigma\rho$, $k=0$ Friedmann spacetime to self-similar SSC form.

\begin{Thm}[Informal statement of Theorem \ref{SelfSimExpandxi}]
    Assume $|\xi|<\xi_0\approx0.816$. Then $\Phi$ in (\ref{ComovToSSC}) defines a regular coordinate mapping and takes the $p=\sigma\rho$, $k=0$ Friedmann spacetime to $SSC$ metric form
    \begin{align}
        ds^2 = -B_\sigma dt^2 + \frac{1}{A_\sigma}dr^2 + r^2d\Omega^2,
    \end{align}
    such that the metric components $A_\sigma$, $B_\sigma$, the density variable $\kappa\rho_\sigma r^2$ and velocity
    \begin{align}
        v_\sigma = \frac{1}{\sqrt{A_\sigma B_\sigma}}\frac{\bar{u}^1_\sigma}{\bar{u}_\sigma^0}
    \end{align}
    are functions of the single variable $\eta$ according to:
    \begin{align}
        A_\sigma &= 1 - \left(\frac{\alpha\eta}{2}\right)^2,\label{Aformetafirst}\\
        B_\sigma &= \frac{\left(1+\frac{\alpha(2-\alpha)}{4}\eta^2\right)^{\frac{2-2\alpha}{2-\alpha}}}{1-\left(\frac{\alpha\eta}{2}\right)^2},\label{Bformetafirst}\\
        \kappa\rho_\sigma r^2 &= \frac{3}{4}\alpha^2\eta^2,\label{zformetafirst}
        \\
        v_\sigma &= \frac{\alpha}{2}\eta.\label{vformetafirst}
    \end{align}
    Moreover, $\eta$ is given implicitly as a function of $\xi$ by the the relation
    \begin{align*}
        \xi = \frac{r}{t} = \frac{\eta}{\mathcal{F}(\eta)}.
    \end{align*}
\end{Thm}

Restricting to the case $p=0$ and imposing the NG time gauge, solutions of the STV-PDE smooth at the center of symmetry admit the following formal expansion in even powers of $\xi$. We include here the metric coefficients $A$ and $D=\sqrt{AB}$ as well as the fluid variables $z$ and $w$:
\begin{align}
    A(t,\xi) - 1 &= A_2(t)\xi^2 + A_4(t)\xi^4 + \dotsc + A_{2n}(t)\xi^{2n} + \dots,\label{AansatzIntro}\\
    D(t,\xi) - 1 &= D_2(t)\xi^2 + D_4(t)\xi^4 + \dotsc + D_{2n}(t)\xi^{2n} + \dots,\label{DansatzIntro}\\
    z(t,\xi) &= z_2(t)\xi^2 + z_4(t)\xi^4 + \dotsc + z_{2n}(t)\xi^{2n} + \dots,\label{zansatzIntro}\\
    w(t,\xi) &= w_0(t) + w_2(t)\xi^2 + \dotsc + w_{2n-2}(t)\xi^{2n-2} + \dots,\label{wansatzIntro}
\end{align}
with:
\begin{align}
    A_0 &= 1, & D_0 &= 1.\label{AzeroDzeroequal1}
\end{align}
The STV-ODE are derived by substituting (\ref{AansatzIntro})--(\ref{wansatzIntro}) into the STV-PDE and collecting like powers of $\xi$. The equations close at every order $n\geq1$ and we name the resulting systems the STV-ODE of order $n$. Carrying this procedure out to order $n=3$ is the subject of the following theorem.\footnote{The STV-ODE of order $n=3$ appear adequate for modeling redshift vs luminosity relations in Cosmology \cite{SmolTeVo}.}
 
\begin{Thm}[Partial statement of Corollary \ref{CorForNequal3}]\label{EqnsToOrder3}
    The STV-ODE computed up to order $n=3$ are equivalent to the following system, which closes at every order:
    \begin{align}
        t\dot{z}_2 &= 2z_2 - 3z_2w_0,\label{z2final}\\
        t\dot{w}_0 &= -\frac{1}{6}z_2 + w-w_0^2,\label{w0final}\\
        t\dot{z}_4 &= 4z_4 - 5(z_4w_0+z_2w_2+z_2w_0D_2),\label{z4final}\\
        t\dot{w}_2 &= -\frac{1}{10}z_4 + 3w_2 - 4w_0w_2 - \frac{1}{2}w_0^2A_2 - \frac{1}{2}A_2^2 + \frac{1}{2}A_2D_2-w_0^2D_2,\label{w2final}\\
        t\dot{z}_6 &= 6z_6 - 7(z_6w_0+z_2w_4+z_2w_0D_4+z_2w_2D_2+z_4w_0D_2+z_4w_2),\label{z6final}\\
        t\dot{w}_4 &= -\frac{1}{14}z_6 + 5w_4 - 6w_0w_4 - \frac{1}{2}w_0^2(A_4-A_2^2) - \frac{1}{2}w_0^2A_2D_2\notag\\
        &- w_0w_2A_2 + \frac{1}{2}A_2D_4 + \frac{1}{2}(A_4-A_2^2)D_2 - A_2A_4 + \frac{1}{2}A_2^3\notag\\
        &- w_0^2D_4 - 4w_0w_2D_2 - 3w_2^2.\label{w4final}
    \end{align}
    Moreover,
    \begin{align}
        A_2 &= -\frac{1}{3}z_2, & A_4 &= -\frac{1}{5}z_4, & A_6 &= -\frac{1}{7}z_6,\label{As}
    \end{align}
    and:
    \begin{align}
        D_2 &= -\frac{1}{12}z_2\label{D2-},\\
        D_4 &= -\frac{3}{40}z_4 + \frac{1}{8}z_2w_0^2 - \frac{1}{96}z_2^2,\label{D4}\\
        D_6 &= \frac{1}{12}\bigg(z_4w_0^2+2z_2w_0w_2+\frac{5}{24}z_2^2w_0^2-\frac{23}{120}z_2z_4-\frac{7}{288}z_2^3-\frac{5}{7}z_6\bigg).\label{D6}
    \end{align}
\end{Thm}

Consistent with (\ref{nbynsystemFirst}), the STV-ODE are nested the sense that the equations of order $n-1$ are a closed subsystem of the STV-ODE of order $n$. The following describes this nested structure of the STV-ODE in general.

\begin{Thm}[Partial statement of Theorem \ref{Thmgeneralexpansion}]\label{ThmSTVequations}
    Assume a smooth solution $(A(t,\xi),D(t,\xi),z(t,\xi),w(t,\xi))$ of (\ref{AeqnxiIntro})--(\ref{weqnxiIntro}) is expanded in even powers of $\xi$ as in (\ref{AansatzIntro})--(\ref{wansatzIntro}). Then $A_{2k}$ can be re-expressed in terms of $z_{2k}$, and $D_{2k}$ can be re-expressed in terms of $z_{2k}$ and $w_0,\dots,w_{2k-2}$, to form, at each order $n\in\mathbb{N}$, a $2n\times2n$ system of ODE
    \begin{align}
        t\dot{\boldsymbol{U}} = \boldsymbol{F}_n(\boldsymbol{U})\label{nbynsystemIntro}
    \end{align}
    in unknowns
    \begin{align*}
        \boldsymbol{U} = (\boldsymbol{v}_1,\dots,\boldsymbol{v}_n)^T,
    \end{align*}
    where
    \begin{align*}
        \boldsymbol{v}_{k}=(z_{2k},w_{2k-2})^T.
    \end{align*}
    Moreover, system (\ref{nbynsystemIntro}) takes the component form
    \begin{align}
        \frac{d}{d\tau}\left(\begin{array}{c}
        \boldsymbol{v}_1\\
        \boldsymbol{v}_2\\
        \vdots\\
        \boldsymbol{v}_n\end{array}\right) = \left(\begin{array}{c}
        P_1\boldsymbol{v}_1+\boldsymbol{q}_1\\
        P_2\boldsymbol{v}_2+\boldsymbol{q}_2\\
        \vdots\\
        P_n\boldsymbol{v}_n+\boldsymbol{q}_n
        \end{array}\right),\label{nbynsystemCompsIntro}
    \end{align}
    where
    \begin{align}
        P_k = P_k(\boldsymbol{v}_1) = \left(\begin{array}{cc}
        (2k+1)(1-w_0)-1 & -(2k+1)z_2\\
        -\frac{1}{4k+2} & (2k+2)(1-w_0)-1
        \end{array}\right)\label{PkIntro}
    \end{align}
    depends only on $\boldsymbol{v}_1=(z_2,w_0)^T$ and $\boldsymbol{q}_k$ depends only on lower order terms:
    \begin{align*}
        \boldsymbol{q}_1 &= \boldsymbol{0},\\
        \boldsymbol{q}_k &= \boldsymbol{q}_k(\boldsymbol{v}_1,\dots,\boldsymbol{v}_{k-1}),
    \end{align*}
    for each $k=2,\dots,n$.
\end{Thm}

The STV-ODE of order $n$ all admit the rest points $SM$ and $M$. The coordinates of the rest point $SM$ are obtained by expanding the self-similar formulation of the $k=0$ Friedmann spacetime in even powers of $\xi$ about the center. Alternatively, since the $k=0$ Friedmann spacetime is a time independent solution of the STV-PDE, it follows that the resulting expansion gives the coordinates of the rest point $SM$ at every order. The rest point $M$, which represents Minkowski spacetime in the limit $\rho\to0$, $v\to\xi$, has zeros in all entries except for a $1$ in the (second) $w_0$-entry, that is, $M=(0,1,0,\dots,0)$. It is easy to verify $M$ is a rest point by induction using (\ref{nbynsystemCompsIntro}). We record the rest points $SM$ and $M$, together with their eigenpairs as follows.

\begin{Thm}\label{DegenerateM}
    Each STV-ODE of order $n\geq1$ admit the rest points:
    \begin{align}
        M &= (0,1,0,\dots,0)\in\mathbb{R}^{2n}, & SM &= \bigg(\frac{4}{3},\frac{2}{3},\frac{40}{27},\frac{2}{9},\dots\bigg).
    \end{align}
    The rest point $M$ is a degenerate stable rest point with one eigenvalue and one eigenvector given by:
    \begin{align}
        \lambda_M &= -1, & \boldsymbol{R}_M &= (0,1,0,1,\dots)^T.\label{eiganM}
    \end{align}
    The eigenvalues of $SM$, computable directly from (\ref{PkIntro}), are given by:
    \begin{align}
        \lambda_{An} &= \frac{2n}{3}, & \lambda_{Bn} &= \frac{1}{3}(2n-5).
    \end{align}
    The corresponding eigenvectors up to order $n=2$ are given by:
    \begin{align}
        \lambda_{A1} &= \frac{2}{3}, & \boldsymbol{R}_1 := \boldsymbol{R}_{A1} &= \bigg(-9,\frac{3}{2},-\frac{10}{3},-1\bigg)^T;\label{eigenSM1}\\
        \lambda_{B1} &= -1, & \boldsymbol{R}_{B1} &= \bigg(4,1,\frac{80}{9},1\bigg)^T;\\
        \lambda_{A2} &= \frac{4}{3}, & \boldsymbol{R}_3 := \boldsymbol{R}_{A2} &= (0,0,-10,1)^T;\\
        \lambda_{B2} &= -\frac{1}{3}, & \boldsymbol{R}_{B2} &= \bigg(0,0,\frac{20}{3},1\bigg)^T.\label{eigenSM4}
    \end{align}
\end{Thm}

Theorem \ref{DegenerateM} follows from the calculations given in Section \ref{S8} below. A few comments are in order. To verify that $M$ is a rest point of the STV-ODE at every order $n\geq1$ is a straightforward induction argument based on setting the right hand side of the STV-ODE (\ref{nbynsystemCompsIntro}) to zero and applying induction on $n$. That it is a degenerate stable rest point with one eigenvalue and one eigenvector given by (\ref{eiganM}) can be verified again by induction based on computing $d\boldsymbol{F}_n(M)$ from (\ref{PkIntro}) and using the nested property of equations (\ref{nbynsystemCompsIntro}). The components of the rest point $SM$ can be computed two different ways: First, by setting the right hand side of the STV-ODE (\ref{nbynsystemCompsIntro}) to zero and computing the zeros of $\boldsymbol{F}_n(\boldsymbol{U})$ in (\ref{PkIntro}), assuming $\boldsymbol{F}_{n-1}(SM)=0$, and second, they can be computed by expanding the self-similar form of the $k=0$ Friedmann spacetime (\ref{Aformetafirst})--(\ref{vformetafirst}) in even powers of $\xi$, such as is done in Section \ref{S11.5}.

The SSCNG gauge, imposed in (\ref{AzeroDzeroequal1}), still leaves open one last freedom in the SSCNG coordinate ansatz, namely, the invariance associated with time translation $t\to t-t_0=\tilde{t}$, a transformation which preserves proper time at $r=0$. This represents a redundancy of physical solutions of the STV-PDE and STV-ODE. We fix this gauge freedom for each solution separately by defining \emph{time since the Big Bang}, that is, so the leading order solution agrees with $SM$ or one of the trajectories in its unstable manifold, and hence agrees with a unique Friedmann spacetime in the $2\times2$ STV-ODE of order $n=1$. From this, the $STV-ODE$ of higher order $n\geq2$ then characterize accelerations away from Friedmann spacetimes. The change of gauge to time since the Big Bang is developed carefully in Section \ref{TSBB}. The general result is stated in the following lemma.

\begin{Lemma}\label{Lemmatimetranslation}
    Let $\boldsymbol{U}(t,\xi)=(A(t,\xi),D(t,\xi),z(t,\xi),w(t,\xi))$ denote an arbitrary outgoing smooth solution of the STV-PDE (\ref{AeqnxiIntro})--(\ref{weqnxiIntro}) in SSCNG coordinates and let $\tilde{\boldsymbol{U}}(\tilde{t},\tilde{\xi})$ denote the transformed solution of (\ref{AeqnxiIntro})--(\ref{weqnxiIntro}) obtained by making the NG gauge transformation:
    \begin{align*}
        t &\to t - t_0 = \tilde{t}, & \xi &= \bigg(\frac{\tilde{t}}{\tilde{t}+t_0}\bigg)\tilde{\xi},
    \end{align*}
    so that
    \begin{align*}
        \tilde{\boldsymbol{U}}(\tilde{t},\tilde{\xi}) = \boldsymbol{U}\bigg(\tilde{t}+t_0,\bigg(\frac{\tilde{t}}{\tilde{t}+t_0}\bigg)\tilde{\xi}\bigg).
    \end{align*}
    Then for any given $\boldsymbol{U}$, there exists a unique time translation $t\to t-t_*=\tilde{t}$ such that the leading order part of $\tilde{\boldsymbol{U}}$ is a solution which lies on the trajectories corresponding to the unstable manifold of $SM$ or else agrees with $SM$ itself.
\end{Lemma}

The proof of Lemma \ref{Lemmatimetranslation} follows from calculations given in Section \ref{TSBB}. Lemma \ref{Lemmatimetranslation} implies that the transformation $t\to t-t_*$ maps solutions to solutions, and hence maps trajectories in $\mathcal{F}$ at order $n$ to trajectories in $\mathcal{F}$ for every $n\geq1$. Note that the scaling $\tau=\ln t$, which converts the STV-ODE to an autonomous system, only exists for $t>0$, so in this sense the mapping $t\to t-t_*$ does not in general map the entire solution on one trajectory to the entire solution on another, as represented in the phase portrait of the autonomous system described in Figure \ref{Figure1}, but rather transformed trajectories end at the rest point $U$. Nevertheless, after we accomplish the transformation to time since the Big Bang, we recover the whole Friedmann solution, so in the end this does not represent a real problem for this theory. Lemma \ref{Lemmatimetranslation} also tells us that imposing time since the Big Bang places the $n=1$ trajectory of a solution in $\mathcal{F}$ at $SM$, or on one of the two trajectories in the unstable manifold of $SM$ at order $n=1$. To study underdense perturbations of $SM$, we now always assume time since the Big Bang is imposed and restrict to the space of solutions $\mathcal{F}$ of the STV-ODE whose leading order trajectory is the connecting orbit that takes $SM$ to $M$ in the leading order phase portrait diagrammed in Figure \ref{Figure2}. In other words, we restrict to $\boldsymbol{U}\in\mathcal{F}$ which satisfy
\begin{align*}
    \boldsymbol{v}_1(t) = (z_2(t),w_0(t)) = \bigg(z_2^F\bigg(\frac{t}{\Delta_0}\bigg),w_0^F\bigg(\frac{t}{\Delta_0}\bigg)\bigg)
\end{align*}
for some $\Delta_0$, where $\Delta_0$ determines the $k<0$ Friedmann spacetime to which it agrees at leading order. The next theorem tells us that trajectories in $\mathcal{F}$ tend to the rest point $M$ at all orders of the STV-ODE and provides a rate of decay.

\begin{Thm}[Informal statement of Corollary \ref{globalexistencewithdecay}]\label{DecayToM}
    Let $\boldsymbol{U}(t)=(\boldsymbol{v}_1(t),\dots,\boldsymbol{v}_n(t))$ be a solution of the STV-ODE (\ref{nbynsystemCompsIntro}) of order $n$ such that
    \begin{align*}
        \lim_{t\to\infty}\boldsymbol{v}_1(t) = (0,1) = M.
    \end{align*}
    Then 
    \begin{align*}
        \lim_{t\to\infty}\boldsymbol{U}(t) = M
    \end{align*}
    as a solution of the STV-ODE at every higher order $n>1$. Moreover, there exists constants $(C_1,\dots,C_n)$ such that:
    \begin{align}
        \|\boldsymbol{v}_1(\cdot)-(0,1)\|_{sup} &\leq C_1\frac{\ln t}{t}, & k &= 1,\label{estimate1Intro}\\
        \|\boldsymbol{v}_k(\cdot)\|_{sup} &\leq C_k\frac{\ln t}{t}, & k &= 2,\dots,n,\label{estimatekIntro}
    \end{align}
    where for each $k\in\{1,\dots,n\}$, $C_k$ depends only on initial data assigned at $t_0>0$,
    \begin{align*}
        \boldsymbol{v}_1(t_0) &= \boldsymbol{v}_1^0,\\
        &\vdots\\
        \boldsymbol{v}_k(t_0) &= \boldsymbol{v}_k^0,
    \end{align*}
    that is, $C_k$ depends only on the initial data up to order $k$.
\end{Thm}

In other words, Theorem \ref{DecayToM} tells us that if a trajectory tends to $M$ at leading order, then it tends to $M$ at all orders. Note that if the $C_n$ in (\ref{estimate1Intro})--(\ref{estimatekIntro}) are bounded by a uniform constant $C$ for every $n\geq1$, then we can sum the geometric series and obtain an error in the approximation over all orders $n$ (assuming $|\xi|<1$):
\begin{align}
    z(t,\xi) = \sum_{k=1}^\infty z_{2k}\xi^{2k} = \sum_{k=1}^n z_{2k}\xi^{2k} + Error,
\end{align}
where
\begin{align}
    |Error| = \bigg|\sum_{k=n+1}^\infty z_{2k}\xi^{2k}\bigg| \leq C\frac{\ln t}{t}\frac{\xi^{2k+2}}{1-\xi^2}.\label{Error2}
\end{align}
Summing the geometric series in (\ref{Error2}) then gives us a formula for the rate at which an underlying solution of the STV-PDE decays to $M$.

It follows from Lemma \ref{Lemmatimetranslation} that to characterize underdense perturbations of the $k=0$ Friedmann spacetime, we can assume the time since the Big Bang gauge and define the family $\mathcal{F}$ as the set of all solutions of the STV-PDE which lie on the trajectory that takes $SM$ to $M$ in the leading order STV-ODE of order $n=1$. In this gauge, $t$ measures time since the Big Bang in the sense that $\lim_{t\to0^+}R(t)=\infty$, where $R(t)$ is the scale factor associated with the Friedmann spacetime it agrees with at leading order \cite{ABS}.

Recall that $\mathcal{F}_n$, also referred to as $\mathcal{F}$ at order $n$, is the set of solutions of the STV-ODE of order $n$ which satisfy the property
\begin{align*}
    \boldsymbol{v}_1(t) = \bigg(z_2^F\bigg(\frac{t}{\Delta_0}\bigg),w_0^F\bigg(\frac{t}{\Delta_0}\bigg)\bigg)
\end{align*}
for some $\Delta_0>0$, that is, solutions which take $SM$ to $M$ at order $n=1$. Recall also that the family $\mathcal{F}_n'\subset\mathcal{F}_n$, referred to as the set of trajectories in $\mathcal{F}'$ at order $n$, is the subset of $\mathcal{F}_n$ in the unstable manifold of $SM$ at order $n$. The following theorem describes how solutions in $\mathcal{F}_n$ generically accelerate away from Friedmann spacetimes at every order $n\geq1$ by characterizing the global dynamics of solutions in terms of the eigenvalues of $SM$ and its unstable manifold $\mathcal{F}_n'$. Since all smooth radial underdense perturbations of $SM$ evolve within the space of trajectories $\mathcal{F}_n$, we interpret this as a quantitative characterization of the instability of the $p=0$, $k=0$ Friedmann spacetime to smooth radial underdense perturbations.

\begin{Thm}
    Let $\boldsymbol{U}(t)=(\boldsymbol{v}_1(t),\dots,\boldsymbol{v}_n(t))$ denote a solution of the STV-ODE (\ref{nbynsystemIntro}) of order $n$ in the family $\mathcal{F}_n$ so that
    \begin{align*}
        \boldsymbol{v}_1(t) = \bigg(z_2^F\bigg(\frac{t}{\Delta_0}\bigg),w_0^F\bigg(\frac{t}{\Delta_0}\bigg)\bigg)
    \end{align*}
    for some $\Delta_0>0$. Then:
    \begin{enumerate}
        \item[(i)] Trajectories in the unstable manifold $\mathcal{F}_2'$ of $SM$ generically diverge from Friedmann spacetimes in the phase portrait of the STV-ODE of order $n=2$, and hence, by the nested property of the STV-ODE, they generically diverge from Friedmann at all orders $n>2$ as well.
        \item[(ii)] The character of the unstable manifold $\mathcal{F}_n'$ of $SM$ is determined at order $n=2$ in the sense that a solution $\boldsymbol{U}(t)$ is in the unstable manifold $\mathcal{F}_n'$ at all orders $n\geq2$ if and only if it is in $\mathcal{F}_2'$, that is, if and only if $(\boldsymbol{v}_1(t),\boldsymbol{v}_2(t))$ is in the unstable manifold $\mathcal{F}_2'$ of $SM$ at order $n=2$. This implies the unstable manifold $\mathcal{F}_n'$ is a codimension one set of trajectories in $\mathcal{F}_n$ at every order $n\geq2$ of the STV-ODE.
        \item[(iii)] By definition, solutions in $\mathcal{F}$ all satisfy $\lim_{t\to\infty}\boldsymbol{v}_1(t)=SM$ but generically $\lim_{t\to\infty}\boldsymbol{v}_2(t)\neq SM$.
        \item[(iv)] The smallest positive eigenvalue at $SM$ emerges at order $n=3$ in $\lambda_{A3}$ (followed by $\lambda_{A1}$ at $n=1$). This implies that solution trajectories in $\mathcal{F}_n'$ enter tangent to the Friedmann trajectory at order $n=1$ (and thus also $n=2$) but enter tangent to the eigenvector of $\lambda_{B3}$ at all higher orders $n\geq3$.
    \end{enumerate}
\end{Thm}

For point (i), this follows directly from the presence of two distinct positive eigenvalues $\lambda_{A1}$ and $\lambda_{A2}$ at $SM$ in the STV-ODE of order $n=2$. These determine two independent eigendirections in the unstable manifold $\mathcal{F}_2'$ of $SM$ at order $n=2$, with only the eigendirection of $\lambda_{A1}$ corresponding to $k<0$ Friedmann spacetimes.

For point (ii), this follows directly from the fact that there exists only a single negative eigenvalue $\lambda_{B2}$ at $SM$ at order $n=2$ and all higher order eigenvalues $\lambda_{An}$ and $\lambda_{Bn}$ for $n\geq3$ are positive.

For point (iii), this follows directly from the presence of the single negative eigenvalue $\lambda_{B2}$ above level $n=1$. We interpret this as establishing that, unlike Friedmann spacetimes, solutions in $\mathcal{F}$ exhibit a self-similar Big Bang only at leading order $n=1$ but generically do not at higher orders.

\section{Smoothness at the Center of Spherically Symmetric Spacetimes in SSCNG Coordinates}\label{S3}

Our goal is to characterize the instability of the $p=0$, $k=0$ Friedmann spacetime to perturbations within the class of smooth solutions. Since $r=0$ is a singular value in radial coordinates, we need a condition characterizing smoothness at the center ($r=0$) in SSC (\ref{SSCintro}). The results of this paper rely on the validity of approximating solutions by finite Taylor expansions about the center of symmetry, so the main issue is to guarantee that solutions are indeed smooth in a neighborhood of the center. 

The Universe is not smooth on small scales, so our assumption is that the center is not special regarding the regularity assumed in the large scale approximation of the Universe. Smoothness, by which we mean derivatives of all orders can be taken, at a point $P$ in a spacetime manifold is determined by the atlas of coordinate charts defined in a neighborhood of $P$. The regularity of tensors is identified with the regularity of tensor components expressed in the coordinate systems of the given atlas. Now spherically symmetric solutions given, in say, Lemaître--Tolman--Bondi (LTB) or SSC employ spherical coordinates $(r,\phi,\theta)$ for the spacelike surfaces at constant time. The subtly here is that $r=0$ is a coordinate singularity in spherical coordinates and functions are defined only for the radial coordinate $r\geq0$, however, a coordinate system must be specified in a neighborhood of $r=0$ to impose the conditions for smoothness at the center. Of course, once we have the metric represented as smooth in coordinate system $\vec{x}$ on an initial data surface in a neighborhood of $r=0$, the local existence theorem giving the smooth evolution of solutions from smooth initial data for the Einstein field equations would not alone suffice to obtain our smoothness condition, as one would still have to prove that this evolution preserved the metric ansatz.

Following \cite{SmolTeVo}, we begin by showing that this issue can be resolved relatively easily in SSC because the SSC are precisely the spherical coordinates associated with Euclidean coordinate charts defined in a neighborhood of $r=0$. Based on this, we show below that the condition for smoothness of metric components and functions in SSC is simply that all odd order derivatives should vanish at $r=0$.
 
Consider in more detail the problem of representing a smooth, spherically symmetric perturbation of a $k\leq0$ Friedman spacetime. To start, assume the existence of a solution of Einstein's field equations representing a large, smooth underdense region of spacetime that expands from the end of the Radiation Dominated Epoch out to present time. For smooth perturbations, there should exist a coordinate system in a neighborhood of the center of symmetry, in which the solution is represented as smooth. Assume we have such a coordinate system $(t,\boldsymbol{x})\in\mathbb{R}\times\mathbb{R}^3$, with $\boldsymbol{x}=0$ at the center, and use the notation
\begin{align*}
    \vec{x} = (x^0,x^1,x^2,x^3) = (t,x,y,z) =(t,\boldsymbol{x}).
\end{align*}
Spherical symmetry makes it convenient to represent the spatial Euclidean coordinates $\boldsymbol{x}\in \mathbb{R}^3$ in spherical coordinates $(r,\theta,\phi)$, with $r=|\boldsymbol{x}|$. Since generically, any spherically symmetric metric can be transformed locally to SSC form \cite{smolteMemoirs}, we assume the spacetime represented in the coordinate system $(t,r,\theta,\phi)$ takes the SSC form (\ref{SSCintro}). This is equivalent to the metric in Euclidean coordinates $\boldsymbol{x}$ taking the form
\begin{align*}
    ds^2 = -B(|\boldsymbol{x}|,t)dt^2 + \frac{dr^2}{A(|\boldsymbol{x}|,t)} + |\boldsymbol{x}|^2d\Omega^2,
\end{align*}
where:
\begin{align}
    r^2 &= x^2 + y^2 + z^2,\notag\\
    rdr &= xdx + ydy + zdz,\notag\\
    r^2dr^2 &= x^2dx^2 + y^2dy^2 + z^2dz^2 + 2xydxdy + 2xzdxdz + 2yzdydz,\label{A3}
\end{align}
and
\begin{align}
    dx^2 + dy^2 + dz^2 = dr^2 + r^2d\Omega^2.\label{A4}
\end{align}
To guarantee the smoothness of our perturbation of Friedman at the center, we assume a gauge in which:
\begin{align*}
    B(t,r) &= 1 + O(r^2), & A(t,r) &= 1 + O(r^2),
\end{align*}
so that also
\begin{align*}
    \frac{1}{A(t,r)} = 1 + O(r^2) =: 1 + \hat{A}(t,r)r^2,
\end{align*}
where the smoothness of $A$ is equivalent to the smoothness of $\hat{A}$ for $r>0$. This sets the SSC time gauge to proper time at $r=0$ and makes the SSC locally inertial at $r=0$ and $t>0$, a first step in guaranteeing that our spherical perturbations of Friedman are smooth at the center. Keep in mind that without this gauge the SSC form is invariant under arbitrary transformation of time, so we are free to choose proper time at $r=0$. The locally inertial condition at $r=0$ simply imposes that the corrections to Minkowski at $r=0$ are second order in $r$, in particular, the SSC metric (\ref{SSC}) tends to Minkowski as $r\to0$. These assumptions make physical sense and their consistency is guaranteed by reversing the steps in the argument to follow. We now ask what conditions on the metric functions $A$ and $B$ are imposed by assuming the SSC metric be smooth when expressed in our original Euclidean coordinate chart $(t,\boldsymbol{x})$ defined in a neighborhood of a point at $r=0$, $t>0$.

To transform the SSC metric (\ref{SSC}) to $(t,\boldsymbol{x})$ coordinates, use (\ref{A4}) to eliminate the $r^2d\Omega^2$ term and (\ref{A3}) to eliminate the $dr^2$ term to obtain
\begin{align}
    ds^2 &= -B(|\boldsymbol{x}|,t)dt^2 + dx^2 + dy^2 + dz^2\label{A7}\\
    &+ \hat{A}(|\boldsymbol{x}|,t)\big(x^2dx^2+y^2dy^2+z^2dz^2+2xydxdy+2xzdxdz+2yzdydz\big).\notag
\end{align}
The smoothness of $\hat{A}$ is equivalent to the smoothness of $A$, and the smoothness of $A$ and $B$ for $r>0$ guarantees the smoothness of the Euclidean spacetime metric (\ref{A7}) in $(t,\boldsymbol{x})$ coordinates everywhere except at $\boldsymbol{x}=0$. For smoothness at $\boldsymbol{x}=0$, we impose the condition that the metric components in (\ref{A7}) should be smooth functions of $(t,\boldsymbol{x})$ at $\boldsymbol{x}=0$ as well. Note again that imposing smoothness in $(t,\boldsymbol{x})$ coordinates at $\boldsymbol{x}=0$ is correct in the sense that it is preserved by the Einstein evolution equations. We now show that smoothness at $\boldsymbol{x}=0$ in this sense is equivalent to requiring that the metric functions $A$ and $B$ satisfy the condition that all odd $r$-derivatives vanish at $r=0$. To see this, observe that a function $f(r)$ represents a smooth spherically symmetric function of the Euclidean coordinates $\boldsymbol{x}$ at $r=|\boldsymbol{x}|=0$ if and only if the function
\begin{align*}
    g(x) = f(|\boldsymbol{x}|)
\end{align*}
is smooth at $\boldsymbol{x}=0$. Assuming $f$ is smooth for $r\geq0$ (by which we mean $f$ is smooth for $r>0$, and one sided derivatives exist at $r=0$) and taking the $n^{th}$ derivative of $g$ from the left and right and setting them equal gives the smoothness condition
\begin{align*}
    f^{n}(0) = (-1)^nf^{n}(0).
\end{align*}
We state this formally in the following lemma (see \cite{SmolTeVo}).

\begin{Lemma}
    A function $f(r)$ of the radial coordinate $r=|\boldsymbol{x}|$ represents a smooth function of the Euclidean coordinates $\boldsymbol{x}$ if and only if $f$ is smooth for $r\geq0$ and all odd derivatives vanish at $r=0$. Moreover, if any odd derivative $f^{(n+1)}(0)\neq0$, then $f(|\boldsymbol{x}|)$ has a jump discontinuity in its $n+1$ derivative, and hence a kink singularity in its $n^{th}$ derivative at $r=0$.
\end{Lemma}

As an immediate consequence, we obtain the condition for smoothness of SSC metrics at $r=0$, given in the following corollary.

\begin{Corollary}
    The SSC metric (\ref{SSC}) is smooth at $r=0$ in the sense that the metric components in (\ref{A7}) are smooth functions of the Euclidean coordinates $(t,\boldsymbol{x})$ if and only if the component functions $A(t,r)$, $B(t,r)$ are smooth in time and smooth for $r>0$, all odd one-sided $r$-derivatives vanish at $r=0$ and all even $r$-derivatives are bounded at $r=0$.
\end{Corollary}

To conclude, solutions of the Einstein field equations in SSC have four unknowns: The metric components $A$ and $B$, the density $\rho$ and the scalar velocity $v$. It is easy to show that if the SSC metric components satisfy the condition that all odd order $r$-derivatives vanish at $r=0$, then the components of the unit four-velocity vector $u^\mu$ associated with smooth curves that pass through $r=0$ will have the same property.\footnote{This implies that the coordinates are smooth functions of arc-length along curves passing through $r=0$.} Moreover, the scalar velocity $v$ will have the property that all even derivatives vanish at $r=0$ because $v$ is an outward velocity which picks up a change of sign when represented in $\vec{x}$ coordinates. Thus smoothness of SSC solutions at $r=0$ at fixed time is equivalent to requiring that the metric components satisfy the condition that all odd $r$-derivatives vanish at $r=0$. These then give conditions on SSC solutions equivalent to the condition that the solutions are smooth in the ambient Euclidean coordinate system $\vec{x}$. Theorem \ref{ThmsmoothAgain} of Section \ref{S6} proves that smoothness in the coordinate system $\vec{x}$ at $r=0$ at each $t>0$ in this sense is preserved by the Einstein evolution equations for SSC metrics when $p=0$. In particular, this demonstrates that our condition for smoothness of SSC metrics at $r=0$ is equivalent to the well-posedness of solutions in the ambient Euclidean coordinates defined in a neighborhood of $r=0$. Thus we obtain the condition for smoothness of SSC metrics at $r=0$ based on the Euclidean coordinate systems associated with SSC and show this is preserved by the evolution of the Einstein field equations. Since smoothness of the SSC metric components in this sense is equivalent to smoothness of the $\vec{x}$-coordinates with respect to arc-length along curves passing through $r=0$, in this sense, our condition for smoothness is geometric.

\section{Time Since the Big Bang}\label{TSBB}

The SSC-PDE self-similar form of the Einstein field equations for spherically symmetric dust ($p=0$) spacetimes has the advantage that the time translation freedom of the SSC metric ansatz enables one to scale the time so that the Big Bang singularity for a general smooth solution agrees with a Friedmann spacetime at leading order for some value of $k$. We establish this directly now with an argument based on the STV-ODE of order $n=1$, namely:
\begin{align}
    t\dot{z}_2 &= 2z_2 - 3z_2w_0,\label{Eq1again}\\
    t\dot{w}_0 &= -\frac{1}{6}z_2 + w_0 - w_0^2,\label{Eq2again}
\end{align}
where:
\begin{align*}
    w &= \frac{v}{\xi}, & z &= \frac{\rho r^2}{1-v^2}.
\end{align*}
Solutions $(z_2(t),w_0(t))$ of (\ref{Eq1again})--(\ref{Eq2again}) give the leading order approximation:
\begin{align*}
    w(t,\xi) &= w_0(t) + O(\xi^2),\\
    z(t,\xi) &= z_2(t)\xi^2 + O(\xi^4). 
\end{align*}
Consider now the effect of a time translation $\hat{t}=t-t_0$ and set
\begin{align*}
    \hat{\xi} = \frac{r}{\hat{t}},
\end{align*}
so that:
\begin{align*}
    \hat{w} &= \frac{v}{\hat{\xi}} = \left(\frac{t-t_0}{t}\right)\frac{v}{\xi} = \left(\frac{t-t_0}{t}\right)w,\\
    \hat{z} &= \hat{z}_2(\hat{t})\hat{\xi}^2 + O(\hat{\xi}^4) = \left(\frac{t}{t-t_0}\right)^2\hat{z}_2(\hat{t})\xi^2 + O(\xi^4).
\end{align*}
Thus it makes sense at leading order to define:
\begin{align*}
    \hat{w}_0 &:= \left(\frac{t-t_0}{t}\right)w_0, & \hat{z}_2 &:= \left(\frac{t-t_0}{t}\right)^2z_2.
\end{align*}
Given that the SSC metric form is invariant under time translation and the SSC-PDE and SSC-ODE faithfully represent the SSC solutions in $(t,\xi)$ coordinates, we should expect that $(z_2(t),w_0(t))$ should solve the leading order equations (\ref{Eq1again})--(\ref{Eq2again}) if and only if $(\hat{z}_2(\hat{t}),\hat{w}_0(\hat{t}))$ do. It suffices to verify that if $(z_2(t),w_0(t))$ solve (\ref{Eq1again})--(\ref{Eq2again}), then $(\hat{z}_2(\hat{t}),\hat{w}_0(\hat{t}))$ do. To this end, assuming a solution $(z_2(t),w_0(t))$ and substituting $(\hat{z}_2(\hat{t}),\hat{w}_0(\hat{t}))$ into (\ref{Eq1again})--(\ref{Eq2again}), we obtain
\begin{align*}
    \hat{t}\dot{\hat{z}}_2 &= (t-t_*)\left[\left(\frac{t-t_*}{t}\right)^2\dot{z}_2 + 2\left(\frac{t_*}{t^2}\right)\left(\frac{t-t_*}{t}\right)z_2\right]\\
    &= (t-t_*)\left[-\frac{3}{t}\left(\frac{t-t_*}{t}\right)^2z_2\left(w_0-\frac{2}{3}\right) + 2\left(\frac{t_*}{t^2}\right)\left(\frac{t-t_*}{t}\right)z_2\right]\\
    &= -3\hat{z}_2\left[\frac{t-t_*}{t}\left(w_0-\frac{2}{3}\right) - \frac{2}{3}\left(\frac{t_*}{t}\right)\right]\\
    &= 2\hat{z}_2 - 3\hat{z}_2\hat{w}_0
\end{align*}
and
\begin{align*}
    \hat{t}\dot{\hat{w}}_0 &= (t-t_*)\left[\left(\frac{t-t_*}{t}\right)\dot{w}_0 + \frac{t_*}{t^2}w_0\right]\\
    &= (t-t_*)\left[\frac{1}{t}\left(\frac{t-t_*}{t}\right)\left(-\frac{1}{6}z_2+w_0-w_0^2\right) + \frac{t_*}{t^2}w_0\right]\\
    &= -\frac{1}{6}\hat{z}_2 - \hat{w}_0^2 + \left(\frac{t-t_*}{t}\right)\hat{w}_0 + \frac{t_*}{t}\hat{w}_0\\
    &= -\frac{1}{6}\hat{z}_2 + \hat{w}_0 - \hat{w}_0^2.
\end{align*}
We conclude that equations (\ref{Eq1again})--(\ref{Eq2again}) are invariant under the transformation:
\begin{align}
    \hat{t} &\to t - t_*, & \hat{w}_0 &\to \left(\frac{t-t_*}{t}\right)w_0, & \hat{z}_2 &\to \left(\frac{t-t_*}{t}\right)^2z_2.\label{invarianceNequalOne}
\end{align}
Using (\ref{invarianceNequalOne}) we can give a rigorous proof that every solution of the STV-ODE agrees with a Friedmann solution at leading order $n=1$. For this it suffices to prove that for each solution $(z_2(t),w_0(t))$ of (\ref{Eq1again})--(\ref{Eq2again}), the STV-ODE of order $n=1$, there exists a time translation $t_0=t_*$, that is, \emph{time since the Big Bang}, such that (\ref{invarianceNequalOne}) transforms $(z_2(t),w_0(t))$ to $(\hat{z}_2(\hat{t}),\hat{w}_0(\hat{t}))$, where the latter lies on the trajectory corresponding to the point $SM$ or to one of the two trajectories in the unstable manifold of $SM$, see Figure \ref{Figure1}. But this follows directly from (\ref{invarianceNequalOne}) by simply verifying for each solution $(z_2(t),w_0(t))$ that there exists a value $t_0=t_*$ such that
\begin{align*}
    \lim_{\hat{t}\to0}(\hat{z}_2(\hat{t}),\hat{w}_0(\hat{t})) = SM.
\end{align*}
This alone implies the transformed solution lies in the unstable manifold of $SM$ and hence agrees with a Friedmann solution for some value of $k$, at the level of the STV-ODE of order $n=1$. To verify this, it suffices to argue conversely by assume a solution is $SM$, or in the unstable manifold of $SM$, and finding a time translation sufficient to impose any non-singular initial condition. For example, starting with $SM$, it is easy to see one can translate to $t+t_*$ to obtain any initial condition on the trajectory taking $U$ to $SM$ (see Figure \ref{Figure1}), implying, more generally, that time since the Big Bang takes the stable manifold of $SM$ to the rest point $SM$ itself. Similarly, starting with a solution on a trajectory connecting $SM$ to $M$, it is not difficult to find a time translation $t+t_*$ sufficient to set any initial condition on any trajectory on the underdense side of the two trajectories in the stable manifold of $SM$ and similarly on the underdense side. We conclude that since any initial condition away from a rest point can be imposed by some time translation of a trajectory in the unstable manifold of $SM$, the converse is true, that an inverse time translation will transform an arbitrary non-rest point trajectory to $SM$ or one of the two trajectories in the unstable manifold of $SM$ at order $n=1$. From this it follows that imposing the time translation gauge \emph{time since the Big Bang} is equivalent to assuming solutions $(z_2(t),w_0(t))$ lie on $SM$ or traverse one of the trajectories in its unstable manifold. The space $\mathcal{F}$ of smooth solutions underdense with respect to the $k=0$ Friedmann spacetime, identified in this paper, defined by the condition that solutions lie on the trajectory which takes rest point $SM$ to rest point $M$ in the phase portrait of the STV-ODE of order $n=1$, automatically imposes time since the Big Bang because this is the underdense trajectory in the unstable manifold of $SM$. Note that the nested structure of the STV-ODE implies that initial conditions for variables at higher order can still be freely assigned.

\section{The Friedmann Spacetimes in SSCNG}\label{S4}

In this section we review the Friedmann spacetimes of Cosmology. In Section \ref{Cos1} we review the Friedmann spacetimes and their cosmological interpretation. In Section \ref{Cos2} we discuss the instability of the $k=0$ Friedmann metric within the space of Friedmann metrics for general $k$, in the case $p=\sigma\rho$ with $\sigma=constant$ and $0\leq\sigma\leq1$, and record the exact formulas for $k=-1,0,+1$ Friedmann solutions we employ in the analysis to follow. In Section \ref{Cos3} we derive formulas for the unique coordinate transformation which takes a general Friedmann metric given in comoving coordinates to SSCNG coordinates.

\begin{Remark}
	To keep the notation to a minimum, in Sections \ref{S4} to \ref{S6} we change our notation and let $(t,r)$ denote the standard comoving coordinate system for Friedmann spacetimes and use barred coordinates $(\bar{t},\bar{r})$ for SSCNG systems. In Section \ref{S7} we begin the analysis of general solutions to the Einstein field equations in SSCNG coordinates, and from that point on, do not refer to comoving coordinates. Thus from Section \ref{S7} on, we return to the notation of Sections \ref{S1} to \ref{S3} in which unbarred coordinates denote SSCNG, that is, coordinates in which a metric takes the form (\ref{SSCintro}). The exception is Section \ref{S11.5}, where $(\bar{t},\bar{r})$ again denote SSCNG.
\end{Remark}

\subsection{The Friedmann Spacetimes in Cosmology}\label{selfsimilarformulas}\label{Cos1}

The Friedmann metric with curvature parameter $k\in\mathbb{R}$ in comoving coordinates $(t,r)$ takes the form
\begin{align}
    ds^2 = -dt^2 + \frac{R(t)^2}{1-kr^2}dr^2 + \bar{r}^2d\Omega^2,\label{Friedmann}
\end{align}
where $R$ is the cosmological scale factor, $k$ is the curvature parameter, $r=constant$ gives the radial geodesics and $\bar{r}=Rr$ measures arc-length distance at fixed $r$ \cite{wein}. Recall that (\ref{Friedmann}) is invariant under the scaling:
\begin{align}
    \hat{r} &= \sqrt{a}r, & \hat{R}(t) &= \frac{1}{\sqrt{a}}R(t), & \hat{k} &= \frac{k}{a},\label{rescale}
\end{align}
for any $a>0$. Note that $H$ (defined below) and $\bar{r}$ are invariant under rescaling but $R$ and $r$ are not. Taking $a=|k|$ rescales the Friedmann metric (with arbitrary $k$) into its standard form, that is, in which $k=-1,0,+1$ and with the metric taking the form
\begin{align}
    ds^2 = -dt^2 + R(t)^2\bigg(\frac{dr^2}{1-\text{sign}(k)r^2}+r^2d\Omega^2\bigg),\label{kFriedmannkequalone}
\end{align}
where $\text{sign}(k)\in\{-1,0,1\}$. That is, for any given $R(t)$ and $k$, the Friedmann spacetime is equivalent to one of the three forms (\ref{kFriedmannkequalone}), but a given $R(t)$ depends on the initial conditions for the Einstein field equations. The Einstein field equations for Friedmann metrics (\ref{Friedmann}) take the form
\begin{align}
    \dot{R}^2 &= \frac{\kappa}{3}\rho R^2 - k,\label{kFriedmannequation1}\\
    \dot{\rho} & = -3(\rho+p)H,\label{kFriedmannequation2}
\end{align}
where
\begin{align*}
    H = \frac{\dot{R}}{R},
\end{align*}
is the Hubble constant, a function which evolves in time.
In a cosmological model, solutions of (\ref{kFriedmannequation1})--(\ref{kFriedmannequation2}) are to be determined from the measurable quantities at present time in the Universe, namely:
\begin{align}
    H(t_0) &= H_0, & \rho(t_0) &= \rho_0,\label{ICHrho}
\end{align}
where $t_0$ is present time. The age of the Universe is $t_0-t_*$, where $H(t_*)=\infty$ and $R(t_*)=0$ is the Big Bang. The problem then is to determine $(R(t),\rho(t),k,t_*)$ from (\ref{ICHrho}). Assuming an equation of state $p=p(\rho)$, this is done formally as follows: First, solving for $H$ in (\ref{kFriedmannequation1}) and substituting into (\ref{kFriedmannequation2}) yields the system
\begin{align}
    \dot{R} &= \sqrt{\frac{\kappa}{3}\rho R^2 - k},\label{kFriedmannauton1}\\
    \dot{\rho} &= -\frac{3(\rho+p)}{R}\sqrt{\frac{\kappa}{3}\rho R^2-k},\label{kFriedmannauton2}
\end{align}
a $2\times2$ autonomous system of ODE for each fixed $k$, admitting the scaling law (\ref{rescale}) which preserves solutions. We can account for the scaling law in the solution of the initial value problem formally in one of two ways.

For the first way, we scale $k$ into $\text{sign}(k)=-1,0,+1$ and consider the initial value problem for:
\begin{align}
    \dot{R} &= \sqrt{\frac{\kappa}{3}\rho R^2-\text{sign}(k)},\label{kFriedmannauton11}\\
    \dot{\rho} &= -\frac{3(\rho+p)}{R}\sqrt{\frac{\kappa}{3}\rho R^2-\text{sign}(k)}.\label{kFriedmannauton22}
\end{align}
We then use (\ref{kFriedmannauton1}) to determine $\text{sign}(k)$ from $H_0$ and $\rho_0$ by 
\begin{align}
    \text{sign}\Big(H_0^2-\frac{\kappa}{3}\rho_0\Big) = -\text{sign}(k).
\end{align}
Once $\text{sign}(k)$ is fixed, (\ref{kFriedmannauton11})--(\ref{kFriedmannauton22}) is again a fixed autonomous system of ODE which has a unique solution $(R(t'),\rho(t'))$ for initial conditions $R(t'_0)=R_0$ and $\rho(t'_0)=\rho_0$ (we introduce the variable $t'$ here only to later set $t=t'-t_*$). Moreover, being autonomous, solution trajectories are distinct, time translation preserves solutions and time translation suffices to meet all initial conditions on each trajectory. In the cosmological problem, given $H_0$ and $\rho_0$, we use $\dot{R}_0=H_0R_0$ in equation (\ref{kFriedmannauton1}) to solve for $R_0$. Then $(R_0,\rho_0)$ determines a unique solution $(R(t'),\rho(t'))$ of (\ref{kFriedmannauton1})--(\ref{kFriedmannauton2}) for any given time $t'_0$. The time of the Big Bang, $t'=t_*$, is the time when $H(t_*)=\infty$ and the age of the Universe is $t_0=t'_0-t_*$. Setting $t=t'-t_*$, our solution $(H(t),\rho(t))$ as a function of time since the Big Bang $t$, is given in terms of our original solutions by making the time translation $t+t_*\to t$. Obtaining solutions of the initial value problem this way, it is clear that there is a unique cosmological model for each $H_0$ and $\rho_0$, but it is difficult to see that the solution, and age of the Universe, depend continuously on $H_0$ and $\rho_0$ because $\text{sign}(k)$ is discontinuous at $k=0$. For this we can view it a second way.

For the second way, we keep the free parameter $k$ in system (\ref{kFriedmannauton1})--(\ref{kFriedmannauton2}) so that we can continuously take $k\to0$. To start, fix an arbitrary starting time $t'_0$ and impose initial conditions $H(t'_0)=H_0$, $\rho(t'_0)=\rho_0$ and $R(t'_0)=1$. Using this in (\ref{kFriedmannauton1}) determines $k$ by
\begin{align}
    k = -\text{sign}(k)\Big(H_0^2-\frac{\kappa}{3}\rho_0\Big).
\end{align}
Once $k$ is fixed, (\ref{kFriedmannauton1})--(\ref{kFriedmannauton2}) is a fixed autonomous system of ODE which has a unique solution $(R(t'),\rho(t'))$ for initial conditions $R(t'_0)=R_0$ and $\rho(t'_0)=\rho_0$. Again, being autonomous, solution trajectories are distinct, time translation preserves solutions and time translation suffices to meet all initial conditions on each trajectory. Letting $t_*$ be the time when $H(t_*)=0$, we can let $t=t'-t_*$ and make the time translation $H(t+t_*)\to H(t)$ and $\rho(t+t_*)\to\rho(t)$ to obtain solutions as functions of $t$. Then $H(t_0)=H_0$, $\rho(t_0)=\rho_0$, $t$ measures \emph{time since the Big Bang} and $t_0$ gives the age of the Universe. The advantage of this second way to view the initial value problem, is that the right hand side of system (\ref{kFriedmannauton1})--(\ref{kFriedmannauton2}) is a smooth function $\boldsymbol{F}(\rho,R,k)$, and hence solutions depend continuously on $k$. Thus, the Friedmann solutions $(\rho(t),H(t))$ constructed as above to satisfy $(H(t_0),\rho(t_0))=(H_0,\rho_0)$ have the property that $(H(t),\rho(t))$ and $t_0$ all depend continuously on $k$ at fixed $t>0$. 

Although Friedmann solutions depend continuously on $k$ at each time $t$, the $k=0$ Friedmann solution is unstable within the Friedmann family of spacetimes with arbitrary $k\in\mathbb{R}$. The purpose of this paper is to characterize the instability of $k=0$ Friedmann within the general class of spherically symmetric solutions of the Einstein field equations which are smooth at the center. To incorporate the Friedmann family into this more general framework, we will adopt the first approach outlined above, that is, the approach based on rescaling $k$ to $k=-1,0,1$.

\subsection{Instability of $k=0$ Friedmann Within the Friedmann Family}\label{Cos2}

We now derive exact solutions of the Friedmann equations (\ref{kFriedmannauton1})--(\ref{kFriedmannauton2}) assuming the equation of state $p=\sigma\rho$ with $\sigma$ constant and $0\leq\sigma\leq1$. Note that the case $\sigma=\frac{1}{3}c^2$ corresponds to a radiation-dominated universe and $\sigma=0$ corresponds to a (pressureless) matter-dominated universe. It is also worth noting that under the assumptions of spherical symmetry and self-similarity, such as the case for the critical Friedman spacetime, a generic barotropic equation of state $p = p(\rho)$ is restricted to the form $p=\sigma\rho$ for some constant $\sigma$ \cite{cahita}. For such an equation of state, the Friedmann equations take the form:
\begin{align}
    \dot{R} &= \sqrt{\frac{\kappa}{3}\rho R^2-k},\label{kFriedmannzerop1}\\
    \dot{\rho} &= -3(1+\sigma)\rho\frac{\dot{R}}{R},\label{kFriedmannzerop2}
\end{align}
noting that we only consider the case $\dot{R}>0$. Equation (\ref{kFriedmannzerop2}) implies
\begin{align}
    \frac{d\rho}{\rho} = -3(1+\sigma)\frac{dR}{R},
\end{align}
which integrates to
\begin{align}
    \rho R^{3(1+\sigma)} = \rho_0R_0^{3(1+\sigma)},
\end{align}
so $\rho R^{3(1+\sigma)}$ is constant along solutions. Following \cite{ABS}, we set the constant to
\begin{align}
    \Delta_0 = \frac{\kappa}{3}\rho_0R_0^{3(1+\sigma)}.\label{DeltaDef}
\end{align}
As noted in \cite{ABS}, in the case $p=0$ and $k=+1$, taking units $\kappa=8\pi$, we see that
\begin{align}
    M_0 = \frac{4\pi}{3}R_0^3\rho_0 = \frac{1}{2}\Delta_0
\end{align}
has the physical interpretation as the total mass of the Universe. The case $p=0$ and $k=0$ is the mass of the ball of radius $R$. Note also that
\begin{align*}
    \Delta_0 = 2M_0,
\end{align*}
is a formal expression for the Schwarzschild radius.

Using (\ref{DeltaDef}) gives
\begin{align}
    \rho = \frac{3\Delta_0}{\kappa}R^{-3(1+\sigma)},
\end{align}
and using this in (\ref{kFriedmannzerop1}) gives the scalar equation
\begin{align}
    \dot{R}^2 = \Delta_0R^{-(1+3\sigma)} - k.
\end{align}
The acceleration parameter $q_0$ (which determines the quadratic correction to redshift vs luminosity) is then given by
\begin{align}
    q_0 = -\frac{\ddot{R}_0R_0}{\dot{R}_0^2}=\frac{(1+3\sigma)\Delta_0}{\Delta_0-kR_0^{1+3\sigma}}.
\end{align}

We can now discuss, formally, the instability of the critical $k=0$ Friedmann spacetimes within the space of $k\neq0$ Friedmann spacetimes when $p=\sigma\rho$. For this, we note that (\ref{kFriedmannzerop1}) gives
\begin{align*}
    1 &= \frac{\kappa\rho}{3H^2}\left(1-\frac{3k}{\kappa\rho R^2}\right)\\
    &= \frac{\kappa\rho}{3H^2}\left(1-\frac{3k}{\kappa\rho_0R_0^{3(1+\sigma)}R^{-3(1+\sigma)}R^2}\right)\\
    &= \frac{\kappa\rho}{3H^2}\left(1-\frac{k}{\Delta_0}R^{1+3\sigma}\right)
\end{align*}
or
\begin{align}
    \Omega = 1 - \frac{k}{\Delta_0}R^{1+3\sigma}\label{unstableFriedmann}
\end{align}
where
\begin{align*}
    \Omega(t) = \frac{3H^2(t)}{\kappa\rho(t)}.
\end{align*}
Thus, if at a given time the Universe is near critical expansion ($k=0$), then $\Omega(t)\approx1$. Therefore by (\ref{unstableFriedmann}), when $k<0$, $\Omega(t)\to\infty$ in positive time, and in the case $k>0$, $\Omega(t)\to0$ at the maximum value of $R$ \cite{ABS}. This gives a formal expression to the instability of critical expansion within the Friedmann family of spacetimes.

Our goal now is to express the instability of the $k=0$ Friedmann spacetime rigorously within a phase portrait, which we do by first transforming the Friedmann spacetimes to SSCNG coordinates. Recall that SSCNG coordinates are coordinates in which the metric takes the SSC form (\ref{SSC}) and employs a special normalized gauge (NG). In SSCNG coordinates, the instability can be expressed simply and rigorously in a phase portrait based on the self-similar variable $\xi=\frac{\bar{r}}{\bar{t}}$ associated with SSCNG coordinates $(\bar{t},\bar{r})$. The result is a rigorous characterization of the instability of the $k=0$ Friedmann spacetime to smooth spherically symmetric perturbations in the cosmologically significant case $p=0$. We show that, in this phase portrait, the $k\neq0$ Friedmann spacetimes (\ref{kminusone}) and (\ref{kplusone}) correspond to two trajectories in the unstable manifold of the rest point $SM$ (corresponding to $k=0$ Friedmann), but the unstable manifold has one extra degree of freedom over and above perturbations which are Friedmann solutions. This extra degree of freedom naturally replaces the one degree of freedom offered by the cosmological constant in predicting redshift vs luminosity observations.

To accomplish this, we use the following well known formulas for exact solutions of the Friedmann equations when $p=0$ and $k=-1,0,1$ (see \cite{ABS} pages 433--437).

\begin{Thm}\label{exactFriedmannsolutions}
    The following formulas provide exact solutions to the Friedmann equations (\ref{kFriedmannequation1})--(\ref{kFriedmannequation2}) when $p=\sigma=0$ and $k=-1,0,1$.

    Case $k=-1$:
    \begin{align}
        t &= \frac{\Delta_0}{2}(\sinh2\theta-2\theta),\label{kminusone}\\
        R &= \frac{\Delta_0}{2}(\cosh2\theta-1) = \Delta_0\sinh^2\theta.\label{kminusoneR}
    \end{align}

    Case $k=0$:
    \begin{align}
        R &= \big(\sqrt{\Delta_0}t\big)^{\frac{2}{3}},\\
        \rho &= \frac{4\kappa}{3t^2}.
    \end{align}

    Case $k=+1$:
    \begin{align}
        t &= \frac{\Delta_0}{2}(2\theta-\sin2\theta),\label{kplusone}\\
        R &= \frac{\Delta_0}{2}(1-\cos2\theta).\label{kplusoneR}
    \end{align}
\end{Thm} 

\subsection{Transforming Friedmann to SSCNG Coordinates}\label{Cos3}

We now consider the Friedmann metrics (\ref{Friedmann}) and derive the explicit coordinate transformation that puts them into SSC form
\begin{align}
    ds^2 = -B(\bar{t},\bar{r})d\bar{t}^2 + \frac{d\bar{r}^2}{A(\bar{t},\bar{r})} + \bar{r}^2d\Omega^2,\label{SSC}
\end{align}
such that they meet the normalized gauge condition $B(\bar{t},0)=1$. Note that the SSC metric form has the gauge freedom $\bar{t}\to F(\bar{t})$ for any smooth invertible function $F$, so specifying proper time at $\bar{r}=0$ fixes the functions $A(\bar{t},\bar{r})$ and $B(\bar{t},\bar{r})$ uniquely. Note also that the comoving metric form (\ref{Friedmann}) is invariant under the transformation:
\begin{align*}
    R(t) &\to \frac{R(t)}{\sqrt{|k|}}, &  r &\to \sqrt{|k|}r,
\end{align*}
which preserves $\bar{r}=Rr$. It follows that without loss of generality we can assume $k=\pm1$.

Our strategy for finding the change of variables which takes (\ref{Friedmann}) to (\ref{SSC}) with normalized gauge $B(\bar{t},0)=1$ is as follows. We first find an explicit formula for a unique coordinate transformation taking (\ref{Friedmann}) to (\ref{SSC}) of the separable solvable form
\begin{align*}
    \hat{t} = \Phi(t,r) = f(t)g(r),
\end{align*}
and then we apply a change of gauge
\begin{align*}
    \bar{t} = F(\hat{t}) = F(\Phi(t,r)),
\end{align*}
which fixes the normalized gauge condition $B(\bar{t},0)=1$. This is because the total time change $\bar{t}(t,r)$ to SSCNG is not separable. Now the SSC metric form is invariant under arbitrary changes of time and thus it follows that the transformation:
\begin{align}
    \bar{t} &= F(h(t)g(r)), & \bar{r} &= R(t)r,\label{Ftrans}
\end{align}
will also take the Friedmann metric (\ref{Friedmann}) to SSC form (\ref{SSC}). Applying $F$ is thus an arbitrary gauge transformation. We now identify the gauge transformation $F(y)$ such that $B(\bar{t},0)=1$. We have that when $F(y)=1$, 
\begin{align*}
    B(\bar{t},0) = \frac{1}{h'(t)^2}.
\end{align*}
It is straightforward to derive the condition on $F$ so that the transformation (\ref{Ftrans}) puts Friedmann in SSC with normalized gauge for every $k$, namely
\begin{align*}
    B(\bar{t},0) = \frac{1}{F'(h(t)h'(t))^2} = 1.
\end{align*}
Thus the condition is 
\begin{align*}
    \frac{d}{dt}F(h(t)) = 1,
\end{align*}
or
\begin{align*}
    F(h(t)) = t.
\end{align*}
Therefore, letting $y=h(t)$ and assuming the invertibility of $h$, gives $t=h^{-1}(y)$, so we conclude
\begin{align*}
    F(y) = h^{-1}(y).
\end{align*}
We can now state and prove the main theorem of this section, which provides an explicit formula for the coordinate transformation taking Friedmann metrics in comoving coordinates to Friedmann metrics in SSCNG coordinates.

\begin{Thm}\label{SSCStandardGauge}
    Define the coordinate transformation:
    \begin{align}
        \bar{t} &= F(h(t)g(r)), & \bar{r} &= R(t)r,\label{FtransFinal}
    \end{align}
    where:
    \begin{align}
        h(t) &= e^{\lambda\int_0^t\frac{d\tau}{\dot{R}(\tau)R(\tau)}},\\
        g(r) &= \begin{cases}
        (1-kr^2)^{-\frac{\lambda}{2k}}, & k\neq0,\\
        e^{\frac{\lambda}{2}r^2}, & k=0,
        \end{cases}\label{DefghThm}\\
        F(y) &= h^{-1}(y),\label{Fdefine}
    \end{align}
    and $R(t)$ and $k$ are the cosmological scale factor and curvature parameter, respectively, of a Friedmann metric (\ref{Friedmann}). Then for any $\lambda>0$ (we take $\lambda=\frac{1}{2}$ below), (\ref{FtransFinal}) transforms the Friedmann metric (\ref{Friedmann}) over to SSC form (\ref{SSC}) with normalized gauge condition
    \begin{align*}
        B(\bar{t},0) = 1
    \end{align*}
    and the transformed SSCNG metric components are given by:
    \begin{align}
        A &= 1 - kr^2 - H^2\bar{r}^2,\label{Aformulaone}\\
        B &= \frac{1}{F'(\Phi)^2}\hat{B} = \frac{1}{(F'(\Phi)\Phi_t)^2}\frac{1-kr^2}{1-kr^2-H^2\bar{r}^2},\label{Bformulaone}
    \end{align}
    where $\Phi(t,r) = f(t)g(r)$. Moreover, we have:
    \begin{align}
        \sqrt{AB} &= \frac{\sqrt{1-kr^2}}{\frac{\partial\bar{t}}{\partial t}(t,r)},\label{sqrtABbest}\\
        v &= \frac{\dot{R}r}{\sqrt{1-kr^2}},\label{barv}
    \end{align}
    where $v$ is the SSCNG coordinate fluid velocity.\footnote{Note that (\ref{Aformulaone}) and (\ref{Bformulaone}) agree with equation (2.19) of \cite{smolteMAA}.}
\end{Thm}

For the proof of Theorem \ref{SSCStandardGauge}, see Section \ref{Appendix1A} below.

We now use Theorem \ref{SSCStandardGauge} to write the Friedmann spacetimes in SSCNG coordinates $(\bar{t},\xi)$. We consider first the case $k\neq0$. Formulas (\ref{kminusone})--(\ref{kplusoneR}) give implicit formulas for the $k\neq0$ Friedmann spacetimes in comoving coordinates $(t,r)$, where we recall
\begin{align*}
    \Delta_0 = \frac{\kappa}{3}\rho R^3.
\end{align*}
Keep in mind that when $k\neq0$, different values of $\Delta_0$ do not correspond to a gauge transformation, but instead describe distinct Friedmann solutions. The variable $t$ in (\ref{kminusone})--(\ref{kplusoneR}) represents proper time at fixed $r$ for all values of $\Delta_0$ and $k$.

To display the dependence of the $k\neq0$ Friedmann spacetimes on $\Delta_0$, we use the notation:
\begin{align*}
    \chi &= \frac{t}{\Delta_0}, & \bar{\chi} &= \frac{\bar{t}}{\Delta_0}, & \xi &= \frac{\bar{r}}{\bar{t}}.
\end{align*}
In the derivations below, we work to express the functions $A$, $B$, $v$ and $\rho r^2$ of the $k\neq0$ Friedmann solutions in SSCNG coordinates as functions of $(\chi,\xi)$. This is accomplished in Theorem \ref{ThmLog-time translation}. The result shows that the dependence of a $k\neq0$ Friedmann solution on the variables $\bar{t}$, $\bar{r}$ and $\Delta_0$ is through $(\chi,\xi)$. Most importantly, because $A$, $B$, $v$ and $\rho r^2$ depend on $(\bar{t},\bar{r})$ only through $(\bar{\chi},\xi)$ in SSCNG, it follows that the free parameter $\Delta_0$ in $k\neq0$ Friedmann spacetimes corresponds to the mapping of smooth solutions to smooth solutions implemented by replacing $\bar{t}$ by $\bar{\chi}$, holding $\xi$ fixed. Anticipating what is to come next, we call this \emph{log-time translation}, because changing $\Delta_0$ corresponds to the log-time translation $\ln\bar{t}\to\ln\bar{t}+\ln\Delta_0$. In particular, we use this below to establish that the $k\neq0$ Friedmann solutions each lie on a unique trajectory of the STV-ODE (derived below) at every order $n\geq1$, respectively, by showing that this corresponds to time translation in an autonomous system obtained by using log-time $\tau=\ln\bar{t}$ in place of $\bar{t}$.

We begin in the next section by establishing the domain of validity of the SSCNG coordinate transformation. From this point onward, we focus on the main case of interest to this paper, the case $k=-1$. Analogous formulas follow for the case $k=+1$ with straightforward modification. When $k=-1$, the Friedmann spacetimes in comoving coordinates $(t,r)$ are described by the exact formulas (\ref{kminusone}) and (\ref{kminusoneR}).

\subsection{The SSCNG Coordinate System for $k=-1$}

In this subsection we transform (\ref{kminusone})--(\ref{kminusoneR}) over to SSCNG coordinates and characterize the region of validity of the transformation. For this we find a simple expression for
\begin{align}
    \bar{t} = h^{-1}(h(t)g(r))
\end{align}
in (\ref{FtransFinal}), where $h$ and $g$ are given by (\ref{DefghThm}) to be:
\begin{align}
    h(t) &= e^{\lambda\int_0^t\frac{d\tau}{\dot{R}(\tau)R(\tau)}},\label{hoft2again}\\
    g(r) &= (1+r^2)^{\frac{\lambda}{2}},\label{goft2again}
\end{align}
and where $R(t)$ is defined implicitly by (\ref{kminusone})--(\ref{kminusoneR}). We start by recording the following expressions, which follow directly from (\ref{kminusone})--(\ref{kminusoneR}):
\begin{align}
    \frac{dt}{d\theta} &= \Delta_0(\cosh2\theta-1),\\
    \dot{R} &= \Delta_0(\sinh2\theta)\frac{d\theta}{dt} = \frac{\sinh2\theta}{\cosh2\theta-1} = \coth\theta,\label{tstep}\\
    \frac{d\dot{R}}{d\theta} &= -\frac{2}{\cosh2\theta-1} = -\csch^2\theta,\\
    \ddot{R} &= \frac{d\dot{R}}{d\theta}\frac{d\theta}{dt} = -\frac{\csch^22\theta}{\Delta_0(\cosh2\theta-1)},\\
    \dot{H} &= \frac{\ddot{R}R-\dot{R}^2}{R^2} = -\frac{4(\sinh^22\theta+\cosh2\theta-1)}{\Delta_0^2(\cosh2\theta-1)^4}.
\end{align}

We next record that by (\ref{tstep}), (\ref{kminusone}) and (\ref{kminusoneR}), we have
\begin{align*}
    \dot{R}R = \frac{\Delta_0}{2}\sinh2\theta,
\end{align*}
and using this in (\ref{tstep}) gives a formula for $\theta$ in terms of $t$, namely
\begin{align*}
    \theta = \frac{1}{\Delta_0}(R\dot{R}-t).
\end{align*}
 Using this in (\ref{hoft2again}), we obtain
\begin{align*}\nonumber
    \int_0^t\frac{d\tau}{\dot{R}(\tau)R(\tau)} = 2\int_0^t\frac{d\tau}{\Delta_0\sinh2\theta(\tau)}.
\end{align*}
Now let
\begin{align*}
    \tau = \frac{\Delta_0}{2}(\sinh2\theta-2\theta),
\end{align*}
so that
\begin{align*}
    d\tau = \Delta_0(\cosh2\theta-1)d\theta.
\end{align*}
Substitution yields
\begin{align*}
    \int_0^t\frac{d\tau}{\dot{R}(\tau)R(\tau)} &= 2\int_{t=0}^t\frac{\cosh2\theta-1}{\sinh2\theta}d\theta=2\int_{t=0}^t\frac{2\sinh^2\theta}{2\sinh\theta\cosh\theta} d\theta\\
    &= 2\int_{t=0}^t\frac{\sinh\theta}{\cosh\theta} d\theta = 2\ln|\cosh\theta|_{t=0}^t = \ln\cosh^2\theta(t)
\end{align*}
and using this in (\ref{hoft2again}) gives
\begin{align*}
    h(t) = e^{\lambda\ln\cosh^2\theta(t)} = \cosh^{2\lambda}\theta(t).
\end{align*}
To make the transformation as simple as possible, from here on we assume
\begin{align}
    \lambda = \frac{1}{2},\label{setlambda}
\end{align}
which gives:
\begin{align}
    h(t) &= \cosh\theta(t),\label{hoft2}\\
    g(r) &= \sqrt[4]{1+r^2},\label{goft2}\\
    h^{-1}(y) &= \theta^{-1}\circ\cosh^{-1}(y).\label{hinverse}
\end{align}
We can now use (\ref{hoft2})--(\ref{hinverse}) to obtain a formula for $\bar{t}$ as a function of $(t,r)$ using (\ref{kminusone}) in the form
\begin{align*}
    \bar{t} = h^{-1}(\Phi),
\end{align*}
where $\Phi(t,r) = h(t)g(r)$. But at this stage, in order to connect $\Delta_0$ to log-time translation of the STV-PDE derived below, it is important to make clear the dependence of the coordinates $t$ and $\bar{t}$ on $\Delta_0$. We thus define:
\begin{align*}
    \chi &= \frac{t}{\Delta_0}, & \bar{\chi} &= \frac{\bar{t}}{\Delta_0}.
\end{align*}
Using this notation, (\ref{kminusone}) becomes
\begin{align*}
    \chi = \frac{1}{2}(\sinh2\theta-2\theta),
\end{align*}
which inverts to
\begin{align*}
    \theta(t) = \Theta\Big(\frac{t}{\Delta_0}\Big).
\end{align*}
With this notation, (\ref{hoft2}) takes the form
\begin{align*}
    h(t) = \rm{h}(\chi) = \cosh\Theta(\chi),
\end{align*}
so (\ref{hinverse}) becomes:
\begin{align*}
    h^{-1}(y) &= \Delta_0\rm{h}^{-1}(\chi),\\
    \bar{\chi} &= \frac{\bar{t}}{\Delta_0} = \rm{h}^{-1}(\Phi) = \Theta^{-1}\bigg(\cosh^{-1}\Big(\sqrt[4]{1+r^2}\cosh\Theta(\chi)\Big)\bigg),
\end{align*}
or
\begin{align*}
    \cosh\Theta(\bar{\chi}) = \sqrt[4]{1+r^2}\cosh\Theta(\chi).
\end{align*}
Thus in summary, $\bar{t}=F(h(t)g(r))$ with $F(y)=h^{-1}(y)$ gives the transformation from $(t,r)\to(\bar{t},\bar{r})$ at each value of $\Delta_0$, as:
\begin{align}
    \frac{\bar{t}}{\Delta_0} &= \Theta^{-1}\circ\cosh^{-1}\bigg(\sqrt[4]{1+r^2}\cosh\Theta\Big(\frac{t}{\Delta_0}\Big)\bigg),\label{transkequalminusone1}\\
    \bar{r} &= R(t)r.\label{transkequalminusone2}
\end{align}
By (\ref{transkequalminusone1})--(\ref{transkequalminusone2}), we have
\begin{align*}
    \cosh\Theta\Big(\frac{\bar{t}}{\Delta_0}\Big) = \sqrt[4]{1+\frac{\bar{t}^2\xi^2}{R^2(t)}}\cosh\Theta\Big(\frac{t}{\Delta_0}\Big).
\end{align*}
Finally, using (\ref{kminusoneR}), we obtain the following fundamental relation between $\bar{t}$ and $t$ at each value of $\xi$ and $\Delta_0$
\begin{align}
    \cosh\Theta(\bar{\chi}) = \sqrt[4]{1+\frac{\bar{\chi}^2\xi^2}{\sinh^4\Theta(\chi)}}\cosh\Theta(\chi).\label{formulaFort}
\end{align}
We now show that equation (\ref{formulaFort}) defines $\chi=\chi(\bar{\chi},\xi)$ for all points \emph{inside the black hole} \cite{smolte1}, which includes all points within the Hubble radius. This is made precise in the following lemma.

\begin{Lemma}\label{TimeByChi1}
    For the $k=-1$ Friedmann spacetimes, equation (\ref{formulaFort}) uniquely defines
    \begin{align}
        \chi = \chi(\bar{\chi},\xi)\label{SSCNGinChi}
    \end{align}
    if and only if
    \begin{align}
        r < \sinh\Theta(\chi),\label{solvable}
    \end{align}
    and (\ref{solvable}) is equivalent to the condition that the SSCNG metric component $A$ satisfies
    \begin{align}
        A > 0.\label{solvableA}
    \end{align}
\end{Lemma}

\begin{proof}
Note that (\ref{formulaFort}) gives $\Theta=\bar{\Theta}$ at $\xi=0$, so by continuity, we can solve for $\Theta=\Theta(\bar{\Theta},\xi)$ in a neighborhood of $\xi=0$. To determine how far this extends, we write $\Theta=\Theta(\chi)$ and $\bar{\Theta}=\Theta(\bar{\chi})$ so that formula (\ref{formulaFort}) is equivalent to
\begin{align*}
    0 = f(\Theta,\bar{\chi},\xi) :=& -\cosh^4\bar{\Theta} + \cosh^4\Theta + \cosh^4\Theta\frac{4\bar{\chi}^2\xi^2}{(\cosh2\Theta-1)^2}\notag\\
    =& -\cosh^4\bar{\Theta} + \cosh^4\Theta + \coth^4\Theta\bar{\chi}^2\xi^2.
\end{align*}
To extend the solubility near $\xi=0$ by virtue of the implicit function theorem, it suffices to obtain a condition for $\frac{\partial f}{\partial\Theta}>0$. Thus we compute
\begin{align}
    \frac{\partial f}{\partial\Theta}(\Theta,\bar{\chi},\xi) &= 4\cosh^3\Theta\sinh\Theta - 4\frac{\cosh^3\Theta}{\sinh^5\Theta}\bar{\chi}^2\xi^2\notag\\
    &= 4\frac{\cosh^3\Theta}{\sinh^5\Theta}\big(\sinh^6\Theta-\bar{\chi}^2\xi^2\big)\notag\\
    &= 4\frac{\cosh^3\Theta}{\sinh^5\Theta}\left(\frac{R^3(t)}{\Delta_0^3}-\frac{\bar{t}^2}{\Delta_0^2}\frac{R(t)^2r^2}{\bar{t}^2}\right)\notag\\
    &= 4\frac{\cosh^3\Theta}{\sinh^5\Theta}\frac{R(t)^2}{\Delta_0^2}\left(\frac{R(t)}{\Delta_0}-r^2\right)\notag\\
    &= 4\frac{\cosh^3\Theta}{\sinh\Theta}(\sinh\Theta+r)(\sinh\Theta-r),\label{solveF}
\end{align}
where we have used the identity
\begin{align*}
    \frac{1}{2}(\cosh2\Theta-1) = \sinh^2\Theta
\end{align*}
together with (\ref{kminusoneR}). By (\ref{solveF}), the condition $\frac{\partial f}{\partial\Theta}>0$ is equivalent to $r<\sinh\Theta$, and since $\theta=\Theta(\chi)$ is a monotone function of $\chi$, this together with the implicit function theorem establishes (\ref{solvable}).

Consider now the metric component $A$ given in (\ref{Afinalform}). Using
\begin{align}
    \dot{R}^2 + k = \frac{\kappa}{3}\rho R^2 = \frac{\Delta_0}{R},\label{dotR^2}
\end{align}
we can write $A$ as
\begin{align}
    A = 1 - kr^2 - H^2\bar{r}^2 = 1 - \frac{\kappa}{3}\rho\bar{r}^2 = 1 - \frac{\Delta_0}{R^3}\bar{r}^2 = 1 - \frac{\Delta_0\bar{t}^2}{R^3}\xi^2.\label{Afirst}
\end{align}
By (\ref{Afirst}), the condition $A=0$ is equivalent to
\begin{align*}
    \frac{\Delta_0\bar{t}^2}{R^3}\xi^2=1,
\end{align*}
which by the substitutions $\xi=\frac{\bar{r}}{\bar{t}}$ and $R=\Delta_0\sinh^2\theta$, is equivalent to $r=\sinh\theta$. This establishes (\ref{solvableA}).
\end{proof}

We now prove the main theorem of this subsection, which shows that the transformation from comoving coordinates $(t,r)$ to SSCNG coordinates is a transformation of the form $(\chi,r)\to(\bar{\chi},\xi)$, where $\chi=\frac{t}{\Delta_0}$, $\bar{\chi}=\frac{\bar{t}}{\Delta_0}$ and $\Delta_0$ is the free parameter (\ref{DeltaDef}). The significance of this is that the free parameter $\Delta_0$ is incorporated into the SSCNG coordinate system as a rescaling of time, holding both $r$ and $\xi$ fixed. As a consequence, when $k\neq0$, a change of $\Delta_0$ is not a coordinate gauge transformation, but describes a transformation between physically different solutions (note that we consider the case $k=-1$, a similar result for the case $k=+1$ can be obtained similarly).

\begin{Thm}\label{ThmLog-time translation}
    In the case of the $k=-1$ Friedmann spacetime, the $\Delta_0$ dependent mapping $(t,r)\to(\bar{t},\bar{r})$ from comoving $(t,r)$ coordinates to SSCNG coordinates $(\bar{t},\bar{r})$ is determined by the transformation\footnote{That is, $(t,r)\to(\chi,r)\to(\bar{\chi},\xi)\to(\bar{t},\xi)\to(\bar{t},\bar{r})$.}
    \begin{align}
        (\chi,r)\to(\bar{\chi},\xi),
    \end{align}
    which is a regular $1-1$ coordinate mapping uniquely given implicitly by:
    \begin{align}
        \cosh\Theta(\bar{\chi}) &= \sqrt[4]{1+\frac{\bar{\chi}^2\xi^2}{\sinh^4\Theta(\chi)}}\cosh\Theta(\chi),\label{formulaFortMappingTime}\\
        r &= \frac{\bar{\chi}\xi}{\sinh^2\Theta(\chi)},\label{formulaFortMappingXi}
    \end{align}
    under the solubility condition 
    \begin{align}
        r < \sinh\Theta(\chi).
    \end{align}
    That is, (\ref{formulaFortMappingTime}) determines $\chi=\chi(\bar{\chi},\xi)$ by (\ref{SSCNGinChi}), where $\chi=\frac{t}{\Delta_0}$, $\bar{\chi}=\frac{\bar{t}}{\Delta_0}$, $\Delta_0$ is the free parameter (\ref{DeltaDef}) and the function $\Theta(\chi)$ is determined by inverting
    \begin{align}
        \chi = \frac{1}{2}(\sinh2\Theta-2\Theta)
    \end{align}
    according to the $k=-1$ formula (\ref{kminusone}). Moreover, the $k=-1$ Friedmann solution in SSCNG coordinates is given by the formulas:
    \begin{align}
        1 - A &= A_{F}(\bar{\chi},\xi) := \frac{\kappa}{3}\rho\bar{r}^2 = \frac{8\bar{\chi}^2\xi^2}{(\cosh2\Theta(\chi)-1)^3},\label{Apm}\\
        \sqrt{AB} &= D_{F}(\bar{\chi},\xi) := \frac{\sqrt{1+r^2}}{\frac{\partial\bar{t}}{\partial t}(t,r)} = \frac{\sqrt{1+\frac{\bar{\chi}^2\xi^2}{\sinh^4\Theta}}}{\frac{\partial\bar{\chi}}{\partial\chi}(\chi,r)},\label{Dpm}\\
        v &= v_{F}(\bar{\chi},\xi) := \frac{\dot{R}r}{\sqrt{1+r^2}} = \frac{\frac{\bar{\chi}\xi}{\sinh^2\Theta}\coth\Theta}{\sqrt{1+\frac{\bar{\chi}^2\xi^2}{\sinh^4\Theta}}}.\label{vpm}
    \end{align}
\end{Thm}

The main point of Theorem \ref{ThmLog-time translation} is that it shows the $k=-1$ Friedmann solution for $\Delta_0>0$ is obtained from the solution for $\Delta_0=\frac{4}{9}$ by simply making the transformation $t\to\frac{t}{\Delta_0}=\chi$ and $\bar{t}\to\frac{\bar{t}}{\Delta_0}=\bar{\chi}$, holding $r$ and $\xi$ fixed. That is, $A_{F}$, $D_{F}$ and $v_{F}$ are explicit functions which describe the dependence of the $k=-1$ Friedmann solution on the variables $(\bar{\chi},\xi)$. As a consequence, the dependence on $(\bar{t},\bar{r})$ is through $(\bar{\chi},\xi)$ in (\ref{Apm})--(\ref{vpm}) and this implies that changing the free parameter $\Delta_0$ from $\frac{4}{9}\to\Delta_0>0$ in $k\neq0$ Friedmann spacetimes corresponds to the mapping of smooth solutions to smooth solutions implemented by replacing $\bar{t}$ by $\frac{\bar{t}}{\Delta_0}$ holding $\xi$ fixed. We call this \emph{log-time translation}, since changing $\Delta_0$ from $\frac{4}{9}\to\Delta_0>0$ corresponds to the log-time translation $\ln\bar{t}\to\ln\bar{t}+\Delta_0$. We use this to establish that the $k=-1$ Friedmann solutions each lie on a unique trajectory of the STV-ODE (derived below) at every order $n\geq1$, respectively, by showing that this corresponds to time translation in an autonomous system obtained by using log-time $\ln\bar{t}$ in place of $\bar{t}$.

\begin{proof}[Proof of Theorem \ref{ThmLog-time translation}]
In the case $k=-1$ the formulas (\ref{kminusone})--(\ref{kminusoneR}) describe the Friedmann spacetimes in comoving coordinate $(t,r)$. Equation (\ref{formulaFortMappingTime}) is (\ref{formulaFort}), so Lemma \ref{TimeByChi1} implies that $\chi=\chi(\bar{\chi},\xi)$ is uniquely determined by (\ref{formulaFortMappingTime}) when $r<\sinh\Theta(\chi)$. Equation (\ref{formulaFortMappingXi}) comes from solving 
\begin{align*}
    \xi = \frac{\bar{r}}{\bar{t}} = \frac{Rr}{\bar{t}} = \frac{r\sinh^2\Theta(\chi)}{\bar{\chi}}
\end{align*}
for $r$, which gives $r=r(\bar{\chi},\xi)$ since $\chi=\chi(\bar{\chi},\xi)$ is determined already by (\ref{formulaFortMappingTime}) alone. It follows that (\ref{formulaFortMappingTime})--(\ref{formulaFortMappingXi}) determine $\chi=\chi(\bar{\chi},\xi)$ and $r=r(\bar{\chi},\xi)$ uniquely, and these invert to give the regular mapping $(\chi,r)\to(\bar{\chi},\xi)$ when $r<\sinh\Theta(\bar{\chi})$. To verify (\ref{Apm})--(\ref{vpm}), consider first the metric component $A$ given in (\ref{Afinalform}). Using (\ref{dotR^2}, we can write (\ref{Afinalform}) as (\ref{Afirst}). Thus by (\ref{kminusone}) and (\ref{SSCNGinChi}), we have
\begin{align*}
    A = 1-\frac{8\bar{\chi}^2\xi^2}{(\cosh2\Theta(\chi)-1)^3} = 1-A_{F}(\bar{\chi},\xi).
\end{align*}
Regarding equations (\ref{Dpm}) and (\ref{vpm}), note that the first equalities in (\ref{Dpm}) and (\ref{vpm}) follow directly from (\ref{sqrtABbest}) and (\ref{barv}) upon setting $\lambda=\frac{1}{2}$, as assumed in (\ref{setlambda}), and the remaining equalities in (\ref{Dpm}) and (\ref{vpm}) follow directly from these.
\end{proof}

\subsection{The Hubble Radius Relative to the SSCNG Coordinate System}

In this section we characterize the condition (\ref{solvable}) for the validity of the SSCNG coordinate system in terms of the SSCNG variable $\xi$ for the $k=-1$ Friedmann spacetimes. We make use of this in subsequent sections when we expand solutions in powers of $\xi$ at fixed SSCNG time $\bar{t}$. The main result shows that the SSCNG coordinate system is valid for $\xi<\xi_{max}(\bar{\chi})$, where $\xi_{max}$ increases from $\xi_{max}=\sqrt{\frac{2}{3}}\approx0.816$ to $\xi_{max}=1$ as $\bar{t}$ (and $\bar{\chi}$) go from $0\to\infty$, and the region $\xi<\xi_{max}$ includes the physical extent of the $k=-1$ Friedmann spacetime to out beyond the Hubble radius. We begin with the following lemma.

\begin{Lemma}\label{SizeSSCNG}
    Letting $\bar{r}_H=ct$ denote the Hubble radius at time $t$ (a standard measure of the size of the visible Universe) and $\bar{r}_B$ the distance to $A=0$ at time $t$ in comoving coordinates $(t,r)$, we have
    \begin{align}
        \bar{r}_B \leq \bar{r}_H \leq 3\bar{r}_B\label{solvableB}
    \end{align}
    for all $t>0$, so (\ref{solvable}) guarantees a mapping between $(t,r)\to(\bar{t},\bar{r})$ at all points within the Hubble radius.
\end{Lemma}

\begin{proof}
First, to convert $r_B$ into a distance, we write $\bar{r}_B=R(t)r_B$. We loosely call this the \emph{distance to the black hole} and view it as a measure of the distance from the origin to the point where $A=0$ in SSC at fixed time $t$. We now compare this distance to the usual Hubble radius $\bar{r}_H=ct$, a measure of the size of the visible Universe. From (\ref{kminusone})--(\ref{kminusoneR}), using $\theta=\theta(t)=\Theta(\chi)$, we have:
\begin{align*}
    \frac{d\bar{r}_H}{d\theta} &= \Delta_0\sinh^2\theta,\\
    \frac{d\bar{r}_B}{d\theta} &= 3\Delta_0\sinh^2\theta\cosh\theta \leq 3\Delta_0\sinh^2\theta,
\end{align*}
so
\begin{align}
    \frac{d\bar{r}_H}{d\theta} \leq \frac{d\bar{r}_B}{d\theta}.\label{dbarHdtheta1}
\end{align}
Since both $\bar{r}_B=0$ and $\bar{r}_H=0$ at $\theta=0$, this directly implies $\bar{r}_H<\bar{r}_B$ for all $\theta$ and $t$. Integrating (\ref{dbarHdtheta1}) then gives
\begin{align*}
    \bar{r}_B-\bar{r}_H &= \int_0^\theta\frac{d}{d\theta}(\bar{r}_B-\bar{r}_H)\ d\theta\\
    &\leq \int_0^\theta2\Delta_0\sinh^2\theta\ d\theta\\
    &\leq 2\Delta_0\left(\frac{1}{2}(\sinh2\theta-2\theta)\right) = 2ct = 2\bar{r}_{H}.
\end{align*}
From this we conclude
\begin{align*}
    \bar{r}_B \leq 3\bar{r}_H,
\end{align*}
as claimed in (\ref{solvableB}).
\end{proof}

We now prove the following theorem, which implies that it is valid to expand the $k=-1$ Friedmann solution in powers of $\xi$ at each fixed $\bar{t}$, out to $\xi\leq\xi_{max}$, where $\xi_{max}$ tends to $\xi_{max}=\sqrt{\frac{2}{3}}\approx0.816$ as $\bar{t}\to0$ and $\xi_{max}$ increases monotonically to $\xi_{max}=1$ as $\bar{t}\to\infty$.

\begin{Thm}\label{solvablexi}
    The SSCNG coordinate system defines a regular transformation of the $k=-1$ Friedmann spacetimes so long as $A>0$, and this is equivalent to the condition
    \begin{align}
        0 \leq \xi < \xi_{max}(\bar{\chi}),
    \end{align}
    where
    \begin{align}
        \xi_{max}(\bar{\chi}) = \Pi(\bar{\chi}) := \frac{\left(\cosh^{\frac{4}{3}}\bar{\Theta}(\bar{\chi})-1\right)^{\frac{3}{2}}}{\bar{\chi}}
    \end{align}
    and $\bar{\Theta}(\bar{\chi})$ is defined by inverting
    \begin{align}
        \bar{\chi} = \frac{1}{2}\big(\sinh2\bar{\Theta}-2\bar{\Theta}\big).
    \end{align}
    Moreover, $\Pi(\bar{\chi})$ is an increasing function
    \begin{align}
        \Pi'(\bar{\chi}) &> 0, & 0 &< \bar{\chi} < \infty,\label{thmxisolve4}
    \end{align}
    and satisfies:
    \begin{align}
        \lim_{\bar{\chi}\to\infty}\Pi(\bar{\chi}) &= 1,\label{thmxisolve6}\\
        \lim_{\bar{\chi}\to0}\Pi(\bar{\chi}) &= \sqrt{\frac{2}{3}}\approx0.816,\label{thmxisolve5}
    \end{align}
    so $\xi_{max}=\Pi(\bar{\chi})$ increases continuously from $\xi_{max}=\sqrt{\frac{2}{3}}$ to $\xi_{max}=1$ as $\bar{\chi}$ increases from $0$ to $\infty$.
\end{Thm}

For the proof of Theorem \ref{solvablexi}, we begin by establishing (\ref{thmxisolve6}) and (\ref{thmxisolve5}) in two separate lemmas. The first lemma establishes that the limit of validity of the SSCNG coordinate system, $A>0$, extends out to $\xi_{max}=1$ in the limit $\bar{t}\to\infty$. This implies that it is possible to Taylor expand solutions in powers of $\xi$ out to any $\xi<1$ at arbitrarily late times after the Big Bang. The second lemma establishes $\xi_{max}\to\sqrt{\frac{2}{3}}\approx0.816$ as $\bar{t}\to0$, which, after establishing (\ref{thmxisolve4}), is the global lower bound for $\xi_{max}$.

\begin{Lemma}\label{TimeByChi3}
    The relationship between $\bar{\chi}$ and $\chi$, determined by (\ref{formulaFort}), tends to the condition 
    \begin{align}
        \bar{\chi} = \frac{1}{1-\xi^2}\chi
    \end{align}
    in the limit $t\to\infty$. This implies that the SSCNG coordinate system is valid for
    \begin{align}
        \xi < 1,\label{xilessthan1}
    \end{align}
    in the limit $t\to\infty$.
\end{Lemma}

\begin{proof}
By (\ref{kminusone}), we have that $\theta\to\infty$ as $t\to\infty$, and
\begin{align*}
    \chi = \frac{1}{2}(\sinh2\theta-\theta) \sim \frac{1}{4}e^{2\theta},
\end{align*}
since
\begin{align*}
    \sinh\theta = \frac{1}{2}\big(e^{\theta}-e^{-\theta}\big) \sim \frac{1}{2}e^{\theta},
\end{align*}
where, for this argument, we use the convention $f\sim g$ if they agree up to lower order terms as $t\to\infty$. Using this in (\ref{formulaFort}) gives
\begin{align*}
    \cosh\Theta(\bar{\chi}) = \sqrt[4]{1+\frac{\bar{\chi}^2\xi^2}{\sinh^4\Theta}}\cosh\Theta(\chi) \sim \sqrt[4]{1+\frac{\bar{\chi}^2}{\chi^2}\xi^2}\sqrt{\chi},
\end{align*}
which yields
\begin{align*}
    \cosh^4\Theta(\bar{\chi}) \sim \chi^2+\bar{\chi}^2\xi^2,
\end{align*}
or
\begin{align}
    \chi \sim \sqrt{\cosh^4\bar{\Theta}(\bar{\chi})-\bar{\chi}^2\xi^2}.\label{formulaFortAgain1}
\end{align}
Since $\bar{\chi}\to\infty$ as $t\to\infty$, and
\begin{align*}
    \bar{\chi} = \frac{1}{2}(\sinh2\bar{\theta}-\bar{\theta}) \sim \frac{1}{4}e^{2\bar{\theta}},
\end{align*}
as $t\to\infty$, we have
\begin{align*}
    \cosh^4\bar{\Theta}(\bar{\chi}) \sim \bar{\chi}^2,
\end{align*}
so (\ref{formulaFortAgain1}) becomes
\begin{align}
    \chi \sim \bar{\chi}\sqrt{1-\xi^2}.\label{formulaFortAgain2}
\end{align}
Thus in the limit $t\to\infty$, the solubility condition for $\chi$ as a function of $\bar{\chi}$ and $\xi$ becomes (\ref{formulaFortAgain2}), which is the condition $\xi<1$ as claimed in (\ref{xilessthan1}). This completes the proof of Lemma \ref{TimeByChi3}.
\end{proof}

The next lemma establishes that the limit of validity of the SSCNG coordinate system extends out to $\xi=\sqrt{\frac{2}{3}}\approx0.816$ asymptotically as $\bar{t}\to0$, implying that it makes sense to Taylor expand solutions in powers of $\xi$ out to any $\xi<\xi_{max}=0.816$ at sufficiently early times after the Big Bang.

\begin{Lemma}\label{TimeByChi2}
    In the limit $t\to0$, the relationship between $\bar{\chi}$ and $\chi$, determined by (\ref{formulaFort}), reduces to 
    \begin{align}
        \left(\frac{\chi}{\bar{\chi}}\right)^{\frac{4}{3}} - \left(\frac{\chi}{\bar{\chi}}\right)^2 = \frac{2}{9}\xi^2,
    \end{align}
    which has a unique solution
    \begin{align}
        \frac{\chi}{\bar{\chi}} \in \left(0,\frac{4}{27}\right)
    \end{align}
    for each
    \begin{align}
        \xi \in \left(0,\sqrt{\frac{2}{3}}\right).\label{lemmasolvet=0-3}
    \end{align}
\end{Lemma}

\begin{proof}
By (\ref{kminusone}), in the limit $t\to0$, the following asymptotic conditions hold:
\begin{align*}
    \chi &\sim \frac{2}{3}\theta^3, & \theta &\sim \left(\frac{3}{2}\chi\right)^{\frac{1}{3}}, & \sinh\theta &\sim \left(\frac{3}{2}\chi\right)^{\frac{1}{3}}, & \cosh\theta &\sim 1 + \frac{1}{2}\left(\frac{3}{2}\chi\right)^{\frac{2}{3}}.
\end{align*}
Taking these as exact, that is, neglecting higher order terms, (\ref{formulaFort}) becomes
\begin{align*}
    \cosh\bar{\theta} &= \sqrt[4]{1+\frac{\bar{\chi}^2\xi^2}{\sinh^4\theta}}\cosh\theta\\
    &\sim \Bigg(1+\frac{\bar{\chi}^2\xi^2}{\big(\frac{3}{2}\chi\big)^{\frac{4}{3}}}\Bigg)^{\frac{1}{4}}\Bigg(1+\frac{1}{2}\left(\frac{3}{2}\chi\right)^{\frac{2}{3}}\Bigg).
\end{align*}
Upon eliminating the radical, taking the leading order part on the left hand side and using
\begin{align*}
    \Bigg(1+\frac{1}{2}\left(\frac{3}{2}\chi\right)^{\frac{2}{3}}\Bigg)^4 \sim 1 + 2\left(\frac{3}{2}\chi\right)^{\frac{2}{3}},
\end{align*}
gives
\begin{align*}
    2\left(\frac{3}{2}\bar{\chi}\right)^{\frac{2}{3}} = 2\left(\frac{3}{2}\chi\right)^{\frac{2}{3}} + \frac{\bar{\chi}^2\xi^2}{\big(\frac{3}{2}\chi\big)^{\frac{4}{3}}}\Bigg(1+\frac{1}{2}\left(\frac{3}{2}\chi\right)^{\frac{2}{3}}\Bigg).
\end{align*}
Taking the leading order part (in $\chi$) of the expression on the right hand side and clearing fractions yields
\begin{align*}
    2\left(\frac{3}{2}\bar{\chi}\right)^{\frac{2}{3}}\left(\frac{3}{2}\chi\right)^{\frac{4}{3}} = 2\left(\frac{3}{2}\chi\right)^{\frac{2}{3}}\left(\frac{3}{2}\chi\right)^{\frac{4}{3}} + \bar{\chi}^2\xi^2,
\end{align*}
so dividing through by $\bar{\chi}^2$ results in
\begin{align}
    \left(\frac{\chi}{\bar{\chi}}\right)^{\frac{4}{3}} - \left(\frac{\chi}{\bar{\chi}}\right)^2 = \frac{2}{9}\xi^2.\label{lemmasolvet=0-8}
\end{align}
Equation (\ref{lemmasolvet=0-8}) is the leading order part of equation (\ref{formulaFort}) in the limit $t\to0$, and so implicitly gives $\chi$ as a function of $\bar{\chi}$ for each $\xi$ within the range of solubility. Now let
\begin{align*}
    W = \left(\frac{\chi}{\bar{\chi}}\right)^{\frac{2}3},
\end{align*}
so (\ref{lemmasolvet=0-8}) is equivalent to
\begin{align}
    W^2 - W^3 = \frac{2}{9}\xi^2.\label{lemmasolvet=0-10}
\end{align}
Now we know from (\ref{formulaFort}) that $\bar{\chi}\geq\chi$ and $\bar{\chi}=\chi$ at $\xi=0$. Thus the possible values of $\xi$ are restricted by the condition that there should exist a $W\in(0,1)$ such that (\ref{lemmasolvet=0-10}) holds. Continuity from $W=1$, $\xi=0$, implies that we should take the value of $W\in(0,1)$ which solves (\ref{lemmasolvet=0-10}) for the largest value of $\xi$. We have $W^2-W^3\geq0$ for $W\in(0,1)$, and the maximum value occurs at $W=\frac{2}{3}$, so
\begin{align*}
    W^2 - W^3 \leq \frac{4}{27}.
\end{align*}
It follows from (\ref{lemmasolvet=0-10}) that we can only solve for $W$ under the condition
\begin{align*}
    \frac{2}{9}\xi^2 \leq \frac{4}{27},
\end{align*}
or equivalently
\begin{align*}
    \xi\leq\sqrt{\frac{2}{3}}\approx0.816,
\end{align*}
as claimed in (\ref{lemmasolvet=0-3}). It follows that in the limit $t\to0$, the transformation $\chi\to\bar{\chi}$ is $\bar{\chi}=W\chi$ at each value of $\xi\in\left(0,\sqrt{\frac{2}{3}}\right)$, where $W$ is the constant determined by the solution of (\ref{lemmasolvet=0-10}). This completes the proof of Lemma \ref{TimeByChi2}.
\end{proof}

\begin{proof}[Proof of Theorem \ref{solvablexi}]
By (\ref{solvable}) of Lemma \ref{SSCNGinChi}, we have that the maximal radius $r_B$ of the SSCNG coordinate system, the region $A>0$, is described by
\begin{align}
    r < r_N = \sinh\theta.\label{thmxipf1}
\end{align}
Then at the maximum radius $r=r_B$, the relation (\ref{formulaFort}), which implicitly defines $\chi$ as a function of $(\bar{\chi},\xi)$, becomes
\begin{align*}
    \cosh\bar{\theta} &= \sqrt[4]{1+\frac{\bar{\chi}^2\xi^2}{\sinh^4\theta}}\cosh\theta\\
    &= \sqrt[4]{1+r_B^2}\cosh\theta\\
    &= \sqrt[4]{1+\sinh^2\theta}\cosh\theta,
\end{align*}
which reduces to 
\begin{align}
    \cosh\bar{\theta} = \cosh^{\frac{3}{4}}\theta.\label{thmxipf2}
\end{align}
Equation (\ref{thmxipf1}) gives the relationship between $\chi$ and $\bar{\chi}$ at the maximum radius $r_B$. Putting this in for $\cosh\bar{\theta}$ on the left hand side of (\ref{formulaFort}) gives
\begin{align*}
    \cosh^{\frac{3}{2}}\theta = \sqrt[4]{1+\frac{\bar{\chi}^2\xi^2}{\sinh^4\theta}}\cosh\theta,
\end{align*}
which simplifies to
\begin{align}
    \sinh^6\theta = \bar{\chi}^2\xi^2.\label{thmxipf3}
\end{align}
Using (\ref{thmxipf2}) to eliminate $\Theta(\chi)$ in favor of $\Theta(\bar{\chi})$ gives the relationship between $\xi$ and $\bar{\chi}$ at the maximum radius. Letting $\xi_B$ denote this maximum, (\ref{thmxipf3}) implies
\begin{align*}
    \xi_{max}(\bar{\chi}) = \Pi(\bar{\chi}) := \frac{\big(\cosh^{\frac{4}{3}}{\bar{\theta}}-1\big)^{\frac{3}{2}}}{\bar{\chi}},
\end{align*}
where $\bar{\theta}=\bar{\Theta}(\bar{\chi})$. We now claim
\begin{align}
    \Pi'(\bar{\chi}) > 0.\label{piprimepos}
\end{align}
Assuming (\ref{piprimepos}), the maximum and minimum values of $\Pi(\bar{\chi})$ in (\ref{thmxisolve5}) and (\ref{thmxisolve6}) then follow directly from Lemmas \ref{TimeByChi3} and \ref{TimeByChi2} by evaluating $\Pi(\bar{\chi})$ at $\bar{\chi}=0$ and $\bar{\chi}=\infty$. To prove (\ref{piprimepos}), we use the quotient rule, pull out a positive factor and use $\cosh\bar{\Theta}\geq1$ to obtain
\begin{align}
    \Pi'(\bar{\chi}) &= \frac{\big(\cosh^{\frac{4}{3}}\bar{\Theta}-1\big)^{\frac{1}{2}}}{\bar{\chi}^2}\Big(\cosh^{\frac{1}{3}}\bar{\Theta}\big(2\bar{\chi}\bar{\Theta}'\sinh\bar{\Theta}-\cosh\bar{\Theta}\big)+1\Big).\label{almostdone}
\end{align}
Now
\begin{align*}
    2\bar{\chi}\bar{\Theta}'\sinh\bar{\Theta} = \frac{\bar{\chi}}{\sinh\bar{\Theta}} = \frac{\sinh2\bar{\Theta}-2\Theta}{2\sinh\bar{\Theta}} = \cosh\bar{\Theta}-\frac{\bar{\Theta}}{\sinh\bar{\Theta}}
\end{align*}
and putting this into (\ref{almostdone}) gives $\Pi'(\bar{\Theta})\geq0$ for $\bar{\Theta}>0$ if and only if
\begin{align}
    1 - \bar{\Theta}\frac{\cosh^{\frac{1}{3}}\bar{\Theta}}{\sinh\bar{\Theta}} \geq 0.\label{onemore}
\end{align}
To finish the proof of the claim, it suffices to establish (\ref{onemore}) by proving: 
\begin{align*}
    \lim_{\theta\to0}f(\theta) &=0, & f'(\theta) &> 0,
\end{align*}
for $\theta>0$, where
\begin{align}
    f(\theta) = \frac{\sinh\theta}{\cosh^{\frac{1}{3}}\theta} - \theta.\label{onemore1}
\end{align}
Clearly $f(\theta)\to0$ as $\theta\to0$, and differentiating (\ref{onemore1}) we obtain:
\begin{align*}
    f'(\theta) &= \frac{\cosh^{\frac{1}{3}}\theta\cosh\theta-\frac{1}{3}\sinh^2\theta\cosh^{-\frac{2}{3}}\theta}{\cosh^{\frac{2}{3}}\theta} - 1\\
    &= \cosh^{\frac{2}{3}}\theta - \frac{1}{3}\big(\cosh^2\theta-1\big)\cosh^{-\frac{4}{3}}\theta - 1\\
    &= \frac{2}{3}\cosh^{\frac{2}{3}}\theta + \frac{1}{3}\cosh^{-\frac{4}{3}}\theta - 1\\
    &= \frac{2}{3}x + \frac{1}{3x^2} - 1\\
    &= \frac{1}{3x^2}(2x^3-3x^2+1),
\end{align*}
where we have set 
\begin{align*}
    x = \cosh^{\frac{2}{3}}\theta \geq 1.
\end{align*}
Now the minimum of $g(x)=2x^3-3x^2+1$ for $x\geq1$ is at $x=1$, at which $g(1)=0$, so $g(x)>0$ for $x>1$, implying that $f'(\theta)>0$ for $\theta>0$, as claimed. This completes the proof of Theorem \ref{solvablexi}.
\end{proof}

To get an upper bound on the size of the region in which our SSCNG coordinate system is nonsingular, we compare $\bar{r}_B$ to $\bar{r}_N$, that is, the distance at time $t$ of an outward radial light ray leaving the origin $r=0$ at $t=0$ in comoving coordinates in a $k=-1$ Friedmann spacetime. The null condition gives
\begin{align*}
    -dt^2 + \frac{R(t)^2}{1+r^2}dr^2 = 0,
\end{align*}
or
\begin{align*}
    \int_0^t\frac{dt}{R(t)} = \int_0^r\frac{dr}{\sqrt{1+r^2}}.
\end{align*}
The integral on the left is soluble by the substitution (\ref{kminusone}), giving $dt=2R(t)d\theta$. Using (\ref{kminusoneR}) then yields
\begin{align*}
    \int_0^t\frac{dt}{R(t)} = 2\Theta.
\end{align*}
The integral on the right has the exact solution
\begin{align*}
    \int_0^r\frac{dr}{\sqrt{1+r^2}} = \ln\big(\sqrt{1+r^2}+r\big) = \sinh^{-1}r.
\end{align*}
Thus the radial light ray is given by
\begin{align*}
    \sinh^{-1}r = 2\theta,
\end{align*}
or
\begin{align*}
    r_N = \sinh2\theta.
\end{align*}
Therefore we conclude that at time $t>0$,
\begin{align*}
    \bar{r}_N = R(t)r_N = \sin^2\theta\sinh2\theta = 2\sinh^3\theta\cosh\theta = 2\bar{r}_B\cosh\theta \geq \bar{r}_B.
\end{align*}
Hence we have the upper and lower bounds at each time $t>0$,
\begin{align*}
    \bar{r}_H \leq \bar{r}_B \leq \bar{r}_N.
\end{align*}
Finally, asymptotically as $t\to\infty$, we have
\begin{align*}
    \sinh2\theta &\sim 2\chi, & \sinh\theta &\sim \sqrt{2\chi}, & \cosh\theta &\sim \sqrt{2\chi},
\end{align*}
so
\begin{align*}
    \bar{r}_H &\sim \sqrt{\frac{\Delta_0}{2t}}\bar{r}_B, & \bar{r}_N &\sim \sqrt{\frac{2t}{\Delta_0}}, & \bar{r}_B &\sim \left(\frac{2t}{\Delta_0}\right)^{\frac{3}{2}}.
\end{align*}
Of course, $(t,r_N)$ would name a spacetime point well beyond the region of the Universe visible to an observer at $r=0$ at time $t>0$ in the $k<0$ Friedmann spacetime.

\section{Self-Similarity of the Standard Model in SSCNG}\label{S5}

The purpose of this section is to derive the self-similar form of the $p=0$, $k=0$ Friedmann spacetimes in SSCNG, that is, coordinates $(\bar{t},\bar{r})$ in which the metric takes the SSC form (\ref{SSC}) such that the time coordinate is normalized to proper time (recall this means $B(\bar{t},0)=1$). Even though we only fully address the instability of $SM$ in the case of zero pressure, $\sigma=0$, for completeness, and future reference, we derive self-similar forms of the $k=0$ Friedmann spacetimes in the more general case $p=\sigma\rho$ with $0\leq\sigma\leq1$. We show that the $k=0$ Friedmann metric is self-similar in the sense that, when written in SSCNG, the metric components and fluid variables $v$ and $\rho\bar{r}^2$ depend only on the self-similar variable $\xi=\frac{\bar{r}}{\bar{t}}$. In other words, as a function of $(\bar{t},\xi)$, the SSCNG components of the $k=0$ Friedmann spacetimes are time independent. Our goal is to study the stability of the $p=0$, $k=0$ Friedmann spacetimes by deriving the equations for spherically symmetric spacetimes in SSCNG coordinates expressed in variables $(t,\xi)$. We Taylor expand the equations about the origin $\xi=0$ and obtain a phase portrait with the knowledge ahead of time that the $k=0$ Friedmann spacetimes must show up as a rest point in these coordinates.

We begin by stating Theorem 2, on page 88 of \cite{smolteMAA} regarding the $k=0$ Friedmann spacetime with $p=\sigma\rho$ and $0\leq\sigma\leq1$. Note that typically $0\leq\sigma\leq\frac{1}{3}c^2$ and we set $c=1$ when convenient.

\begin{Thm}\label{ThmCosWithShockMAA}
    In comoving coordinates, the $k=0$ Friedmann metric takes the standard form
    \begin{align}
        ds^2 = -dt^2 + R(t)^2\big(dr^2+r^2d\Omega^2\big)\label{metricFriedmann2}
    \end{align}
    and the Einstein field equations
    \begin{align}
        G = \kappa T
    \end{align}
    for a perfect fluid
    \begin{align}
        T = (\rho+p)\Vec{u}\otimes\vec{u} + pg
    \end{align}
    reduce to the following system of ODE:
    \begin{align}
        H^2 &= \frac{\kappa}{3}\rho,\label{Friedmannstart1}\\
        \dot{\rho} &= -3(\rho+p)H,\label{Friedmannstart2}
    \end{align}
    where
    \begin{align*}
        H=\frac{\dot{R}}{R}
    \end{align*}
    is the Hubble constant, $p$ is the pressure, $\rho$ is the energy density and $\vec{u}$ is the fluid four-velocity.

    In the case $p=\sigma\rho$ with constant $0\leq\sigma\leq1$ (equivalently $\frac{2}{3}\leq\alpha\leq\frac{4}{3}$), assuming $\dot{R}>0$, $R(0)=0$ and $R(t_0)=1$, equations (\ref{Friedmannstart1})--(\ref{Friedmannstart2}) determine $R(t)$ to be
    \begin{align}
        R(t) = \left(\frac{t}{t_0}\right)^{\frac{2}{3(1+\sigma)}} = \left(\frac{t}{t_0}\right)^{\frac{\alpha}{2}},\label{Friedmannsoln1}
    \end{align}
    where
    \begin{align*}
        \alpha = \frac{4}{3(1+\sigma)}.
    \end{align*}
    The resulting solution of (\ref{Friedmannstart1})--(\ref{Friedmannstart2}) is given by:
    \begin{align}
        H(t) &= \frac{2}{3(1+\sigma)}\frac{1}{t},\label{Friedmannsoln3}\\
        \rho(t) &= \frac{4}{3\kappa(1+\sigma)^2}\frac{1}{t^2},\label{kapparho}
    \end{align}
    and by the comoving assumption, the four velocity $\vec{u}$ satisfies
    \begin{align*}
        \vec{u} = (u^0,u^1,u^2,u^3) = (1,0,0,0).
    \end{align*}
    Furthermore, the age of the Universe, $t_0$, and the infinite red shift limit, $r_{\infty}$, are given exactly by:
    \begin{align*}
        t_0 &= \frac{2}{3(1+\sigma)}\frac{1}{H_0}, & r_\infty &= \frac{2}{1+3\sigma}\frac{1}{H_0}.
    \end{align*}
\end{Thm}

The formulas in Theorem \ref{ThmCosWithShockMAA} follow directly from (\ref{kFriedmannzerop1})--(\ref{DeltaDef}) and one can easily verify that the free constant $t_0$ is related to the parameter $\Delta_0$ defined in (\ref{DeltaDef}) by
\begin{align}
    t_0 = \frac{2}{\sqrt{3\kappa(1+\sigma^2)\Delta_0}}.
\end{align}
In particular, when $p=0$ we have $\alpha=\frac{4}{3}$, so
\begin{align*}
    R(t) &= \left(\frac{t}{t_0}\right)^{\frac{2}{3}}, & \rho(t) &= \frac{4}{3\kappa}\frac{1}{t^2},
\end{align*}
and when $p=\frac{1}{3}\rho$ we have $\alpha=1$, so
\begin{align*}
    R(t) &= \left(\frac{t}{t_0}\right)^{\frac{1}{3}}, & \rho(t) &= \frac{3}{4\kappa}\frac{1}{t^2}.
\end{align*}
Our goal now is to prove the following theorem, which establishes the unique coordinate system $(\bar{t},\bar{r})$ in which the $p=\sigma\rho$, $k=0$ Friedmann metric (\ref{metricFriedmann2}) takes the SSC form
\begin{align*}
    ds^2 = -B(\bar{t},\bar{r})d\bar{t}^2 + \frac{d\bar{r}^2}{A(\bar{t},\bar{r})} + \bar{r}^2d\Omega^2,
\end{align*}
such that the metric components $A$ and $B$, together with appropriately scaled expressions for the density and velocity, all depend only on the self-similar variable $\xi=\frac{\bar{r}}{\bar{t}}$. For continuity, we employ the notation of \cite{SmolTeVo}.

\begin{Thm}\label{SelfSimExpandxi}
    Let $(t,r)$ denote comoving coordinates for the $k=0$ Friedmann metric (\ref{metricFriedmann2}) assuming the equation of state $p=\sigma\rho$ with constant $0\leq\sigma\leq1$, and define coordinates $(\bar{t},\bar{r})$ by:
    \begin{align}
        \bar{t} &= \mathcal{F}(\eta)t, & \bar{r} &= \eta t = t^{\frac{\alpha}{2}}r,\label{F1a}
    \end{align}
    where
    \begin{align}
        \eta &= \frac{\bar{r}}{t}, & \mathcal{F}(\eta) &= \left(1+\frac{\alpha(2-\alpha)}{4}\eta^2\right)^{\frac{1}{2-\alpha}}, & \alpha &= \frac{4}{3(1+\sigma)}.\label{eta0}
    \end{align}
    Then the Friedmann metric (\ref{metricFriedmann2}) written in $(\bar{t},\bar{r})$-coordinates is given by
    \begin{align}
        ds^2 = -B_\sigma d\bar{t}^2 + \frac{1}{A_\sigma}d\bar{r}^2 + \bar{r}^2d\Omega^2,
    \end{align}
    where the metric components $A_\sigma$, $B_\sigma$, $\kappa\rho_\sigma\bar{r}^2$ and
    \begin{align*}
        v_\sigma = \frac{1}{\sqrt{A_\sigma B_\sigma}}\frac{\bar{u}^1_\sigma}{\bar{u}_\sigma^0}
    \end{align*}
    are functions of the single variable $\eta$ according to:
    \begin{align}
        A_\sigma &= 1 - \left(\frac{\alpha\eta}{2}\right)^2,\label{Aformeta}\\
        B_\sigma &= \frac{\left(1+\frac{\alpha(2-\alpha)}{4}\eta^2\right)^{\frac{2-2\alpha}{2-\alpha}}}{1-\left(\frac{\alpha\eta}{2}\right)^2},\label{Bformeta}\\
        \kappa\rho_\sigma\bar{r}^2 &= \frac{3}{4}\alpha^2\eta^2,\label{zformeta}
        \\
        v_\sigma &= \frac{\alpha}{2}\eta.\label{vformeta}
    \end{align}
    Moreover, $\eta$ is given implicitly as a function of $\xi$ by the the relation
    \begin{align}
        \xi = \frac{\bar{r}}{\bar{t}} = \frac{\eta}{\mathcal{F}(\eta)},
    \end{align}
    which inverts to
    \begin{align}
        \eta^2 = \xi^2 + \frac{\alpha}{2}\xi^4 + \frac{\alpha^3}{16}\xi^6 + O(\xi^8)\label{etaexpand}
    \end{align}
    by expanding about $\xi=\eta=0$. This determines the following asymptotic expansions for $A_\sigma$, $B_\sigma$, $\kappa\rho_\sigma\bar{r}^2$ and $v_\sigma$ as functions of $\xi$ in a neighborhood of $\xi=\eta=0$, valid as $\xi\to0$:
    \begin{align}
        A_\sigma(\xi) &= 1 - \frac{\alpha^2}{4}\xi^2-\frac{\alpha^3}{8}\xi^4 + O(\xi^6),\label{Aform}\\
        B_\sigma(\xi) &= 1 + \frac{\alpha(2-\alpha)}{4}\xi^2 + O(\xi^4),\label{Bform}\\
        D_\sigma(\xi) &= \sqrt{A_\sigma B_\sigma} = 1 + \frac{\alpha(1-\alpha)}{4}\xi^2 + O(\xi^4),\label{Dform}\\
        \big(\kappa\rho_\sigma\bar{r}^2\big)(\xi) &= \frac{3}{4}\alpha^2\xi^2 + \frac{3}{8}\alpha^3\xi^4 + O(\xi^6),\label{kapparho1}\\
        v_\sigma(\xi) &= \frac{\alpha}{2}\xi\Big(1+\frac{\alpha}{4}\xi^2\Big) + O(\xi^4).\label{vform}
    \end{align}
\end{Thm}

We record two special cases in the following corollary.

\begin{Corollary}
    In the case $p=0$, $\alpha=\frac{4}{3}$ and we have:
    \begin{align}
        A_0(\xi) &= 1 - \frac{4}{9}\xi^2 - \frac{8}{27}\xi^4 + O(\xi^6),\\
        B_0(\xi) &= 1 + \frac{2}{9}\xi^2 + O(\xi^4),\\
        D_0(\xi) &= 1 - \frac{1}{9}\xi^2 + O(\xi^4),\\
        \big(\kappa\rho_0\bar{r}^2\big)(\xi) &= \frac{4}{3}\xi^2 + \frac{8}{9}\xi^4 + O(\xi^6),\\
        v_0(\xi) &= \frac{2}{3}\xi\Big(1+\frac{1}{3}\xi^2\Big) + O(\xi^4).
    \end{align}
    In the case $p=\frac{1}{3}\rho$, $\alpha=1$ and we have:
    \begin{align}
        A_{\frac{1}{3}}(\xi) &= 1 - \frac{1}{4}\xi^2 - \frac{1}{8}\xi^4 + O(\xi^6),\\
        B_{\frac{1}{3}}(\xi) &= 1 + \frac{1}{4} \xi^2 + O(\xi^4),\\
        D_{\frac{1}{3}}(\xi) &= 1 + O(\xi^4),\\
        \big(\kappa\rho_{\frac{1}{3}}\bar{r}^2\big)(\xi) &= \frac{3}{4}\xi^2 + \frac{3}{8}\xi^4 + O(\xi^6),\\
        v_{\frac{1}{3}}(\xi) &= \frac{1}{2}\xi\Big(1+\frac{1}{4}\xi^2\Big) + O(\xi^4).
    \end{align}
\end{Corollary}

The proof of Theorem \ref{SelfSimExpandxi} is given in Section \ref{Appendix1B} below.

\section{The STV-PDE}\label{S6}

In this section we derive a new form of the Einstein field equations for spherically symmetric spacetimes in SSCNG coordinates $(t,r)$, that is, coordinates in which the metric takes the form (\ref{SSCintro}). We do this by re-expressing the Einstein field equations in terms of self-similar variables $(t,\xi)$, where $\xi=\frac{r}{t}$. We call the resulting equations the STV-PDE.\footnote{These were introduced by Smoller, Temple and Vogler in \cite{SmolTeVo}.} Since from here on we only work with solutions in SSCNG, for ease of notation, and for the rest of the paper, we drop the bars from the SSC, the notation employed in Sections \ref{S4} to \ref{S6}. There should be no confusion when we refer back to Sections \ref{S4} to \ref{S6} in which $(t,r)$ refers to Friedmann comoving coordinates and where bars appear on the SSCNG coordinates (to distinguish them from comoving coordinates).

We start with the equations for spherically symmetric solutions of the Einstein field equations $G=\kappa T$ in SSC, now denoted $(t,r)$, derived in \cite{groate}, and look to express them in terms of $(t,\xi)$ as the independent variables. Note that this is not the same as writing the Einstein field equations in $(t,\xi)$-coordinates. We then study the subset of solutions of these equations which meet the further condition that solutions be smooth at the center ($r=0$).

In the unbarred notation, a time dependent metric taking the SSC form is given by
\begin{align*}
    ds^2 = -B(t,r)dt^2 + \frac{dr^2}{A(t,r)} + r^2d\Omega^2,
\end{align*}
where
\begin{align*}
    d\Omega^2 = d\theta^2 + \sin^2\theta d\theta^2
\end{align*}
is the usual line element on the unit sphere, see \cite{wein}. Then according to \cite{groate}, three of the four Einstein field equations determined by $G=\kappa T$ are first order and one is second order. The first order equations are equivalent to:\footnote{Metric entries $(A,B)$ are related to $(\hat{A},\hat{B})$ in \cite{groate} by $A=\frac{1}{\hat{B}}$, $B=\hat{A}$.}
\begin{align}
    -r\frac{A_r}{A} + \frac{1-A}{A} &= \frac{\kappa B}{A}T^{00}r^2 = \frac{\kappa}{A}T^{00}_Mr^2,\label{firstorder1}\\
    \frac{A_t}{A} &= \frac{\kappa B}{A}T^{01}r = \kappa\sqrt{\frac{B}{A}}T^{01}_Mr,\label{firstorder2}\\
    r\frac{B_r}{B} - \frac{1-A}{A} &= \frac{\kappa}{A^2}T^{11}r^2 = \frac{\kappa}{A}T^{11}_Mr^2,\label{firstorder3}
\end{align}
and the the two conservation laws, $\nabla\cdot T=0$, are equivalent to:
\begin{align}
    \big(T^{00}_M\big)_t + \big(\sqrt{AB}T^{01}_M\big)_r &= -\frac{2}{r}\sqrt{AB}T^{01}_M,\label{conservation_law00}\\
    \big(T^{01}_M\big)_t + \big(\sqrt{AB}T^{11}_M\big)_r &= -\frac{1}{2}\sqrt{AB}\bigg(\frac{4}{r}T^{11}_M + \frac{1}{r}\Big(\frac{1}{A}-1\Big)\big(T^{00}_M-T^{11}_M\big) + \frac{2\kappa r}{A}\big(T^{00}_MT^{11}_M-(T^{01}_M)^2\big) - 4rT^{22}\bigg),\label{conservation_law01}
\end{align}
where $T_M$ is the Minkowski stress tensor defined by (see \cite{groate}):
\begin{align*}
    T^{00}_M &= BT^{00}, & T^{01}_M &= \sqrt{\frac{B}{A}}T^{01}, & T^{11}_M &= \frac{1}{A}T^{11}, & T^{22}_M &= T^{22}.
\end{align*}
Equations (\ref{conservation_law00})--(\ref{conservation_law01}) are redundant because $\nabla\cdot T=0$ follows from $G=\kappa T$ by the Bianchi identities. Moreover, to close the equations we must impose an equation of state \cite{groate}. In \cite{groate} it is shown that the Einstein field equations $G=\kappa T$ for metrics in SSC are (weakly) equivalent to (\ref{firstorder1}), (\ref{firstorder3}), (\ref{conservation_law00}) and (\ref{conservation_law01}), and equation (\ref{firstorder2}) is derivable from these.

In this paper we assume the equation of state
\begin{align}
    p &= \sigma\rho\label{psigmarho}
\end{align}
with constant $0\leq\sigma\leq1$. With this equation of state, the Minkowski stress tensors become:
\begin{align}
    T^{00}_M &= c^2\rho\left(\frac{1+\frac{\sigma^2}{c^2}}{1-\frac{v^2}{c^2}}-\frac{\sigma^2}{c^2}\right) = \rho\frac{1+\sigma^2v^2}{1-v^2},\label{sigmastress1}\\
    T^{01}_M &= c^2\rho\frac{1+\frac{\sigma^2}{c^2}}{1-\frac{v^2}{c^2}}\frac{v}{c},\label{sigmastress2}\\
    T^{11}_M &= c^2\rho\left(\frac{1+\frac{\sigma^2}{c^2}}{1-\frac{v^2}{c^2}}\frac{v^2}{c^2}+\frac{\sigma^2}{c^2}\right) = \rho\frac{\sigma^2+v^2}{1-v^2},\label{sigmastress3}\\
    T^{22}_M &= pg^{22} = \rho\frac{\sigma^2}{r^2},\label{sigmastress4}
\end{align}
which imply:
\begin{align}
    T^{00}_MT^{11}_M - \big(T^{01}_M\big)^2 &= c^2\rho \bigg(\frac{1+\sigma^2}{1-v^2}-\sigma^2\bigg)\bigg(\frac{1+\sigma^2}{1-v^2}v^2+\sigma^2\bigg) - \bigg(\frac{1+\sigma^2}{1-v^2}\bigg)^2v^2\notag\\
    &= \sigma^2\bigg(\frac{1+\sigma^2}{1-v^2}(1-v^2)\bigg) - \sigma^4 = \sigma^2\rho^2,\\
    T^{00}_M - T^{11}_M &= \rho\left(\frac{1+\sigma^2}{1-v^2}(1-v^2)-2\sigma^2\right) = (1-\sigma^2)\rho,\\
    T^{01}_M &= \rho\frac{1+\sigma^2}{1-v^2}v = T^{00}_M\frac{1+\sigma^2}{1-\sigma^2v^2}v.
\end{align}
Finally, define the self-similar variable 
\begin{align*}
    \xi := \frac{r}{t},
\end{align*}
the metric variable
\begin{align}
    D := \sqrt{AB},
\end{align}
and the rescaled density and velocity variables:
\begin{align}
    z &:= \kappa T^{00}_Mr^2,\label{definez}\\
    w &:= \frac{\kappa T^{01}_Mr^2}{\xi z},\label{definew}
\end{align}
respectively. By (\ref{sigmastress1}) and (\ref{sigmastress2}), we have
\begin{align}
    w = \frac{\Xi}{\xi},
\end{align}
where
\begin{align*}
    \Xi = \frac{1+\sigma^2}{1+\sigma^2v^2}v,
\end{align*}
so $\Xi=v$ when $\sigma=0$, the case we restrict to below. Assuming the mapping $(t,r)\to(t,\xi)$ is regular, the following theorem gives four equations in independent variables $(t,\xi)$ which are equivalent\footnote{By equivalent, we mean equivalent for smooth solutions under regular transformations of the independent and dependent variables, and where constraints are satisfied subject to appropriate initial boundary data, see \cite{groate}.} to (\ref{firstorder1}), (\ref{firstorder3}), (\ref{conservation_law00}) and (\ref{conservation_law01}). 

\begin{Thm}\label{thmsigma}
    Assume equation of state (\ref{psigmarho}). Then equations (\ref{firstorder1}), (\ref{firstorder3}), (\ref{conservation_law00}), and (\ref{conservation_law01}) are equivalent to the following four equations in unknowns $A(t,\xi)$, $D(t,\xi)$, $z(t,\xi)$ and $w(t,\xi)$:
    \begin{align}
        \xi A_\xi &= -z + (1-A),\label{Axi-final}\\
        \xi D_\xi &= \frac{D}{2A}\bigg(2(1-A)-(1-\sigma^2)\frac{1-v^2}{1+\sigma^2 v^2}z\bigg),\label{Dxi-final}\\
        tz_t + \xi\big((-1+Dw)z\big)_\xi &= -Dwz,\label{zeqnxi-final}\\
        tw_t + (-1+Dw)\xi w_\xi &- w + Dw^2 - \frac{\sigma^2}{\xi z}\bigg(D\Xi^2\frac{1-v^2}{1+\sigma^2}z\bigg)_\xi\notag\\
        &+ \frac{\sigma^2\xi}{z}\bigg(D\frac{1-v^2}{1+\sigma^2v^2}\frac{z}{\xi^2}\bigg)_\xi = \text{RHS},\label{weqnxi-final}
    \end{align}
    where
    \begin{align*}
        \text{RHS} = -\frac{1}{\xi^2}\frac{1-v^2}{1+\sigma^2v^2}\frac{D}{2A}\bigg((1-\sigma^2)(1-A)+2\sigma^2\frac{1-v^2}{1+\sigma^2v^2}z\bigg).
    \end{align*}
    Moreover, equations (\ref{Axi-final})--(\ref{weqnxi-final}) imply the two equations:
    \begin{align}
        tA_t + \xi A_\xi &= wz,\label{At-final}\\
        \xi\frac{B_\xi}{B} &= \frac{1}{A}\frac{\sigma^2+v^2}{1+\sigma^2v^2}z + \frac{1-A}{A},\label{Bxi-final}
    \end{align}
    which are equivalent to (\ref{firstorder2}) and (\ref{firstorder3}) respectively.
\end{Thm}

We call system (\ref{Axi-final})--(\ref{weqnxi-final}) the \emph{self-similar Einstein field equations}. Note again that (\ref{Axi-final})--(\ref{weqnxi-final}) are not the Einstein field equations in $(t,\xi)$ coordinates, but rather the Einstein field equations in SSC $(t,r)$, expressed in terms of variables $(t,\xi)$. That is, $A$, $B$ and $v$ are the metric components and invariant velocity in SSC $(t,r)$ coordinates, not $(t,\xi)$ coordinates.

The case $\sigma=0$ of Theorem \ref{thmsigma}, which is the basis for this paper, is stated in the following Corollary.

\begin{Corollary}
    Assume equation of state (\ref{psigmarho}) with $\sigma=0$ ($p=0$). Then equations (\ref{firstorder1}), (\ref{firstorder3}), (\ref{conservation_law00}) and (\ref{conservation_law01}) are equivalent to the following four equations in unknowns $A(t,\xi)$, $D(t,\xi)$, $z(t,\xi)$ and $w(t,\xi)$:
    \begin{align}
        \xi A_\xi &= -z + (1-A),\label{Axi-final0}\\
        \xi D_\xi &= \frac{D}{2A}\big(2(1-A)-(1-\xi^2w^2)z\big),\label{Dxi-final0}\\
        tz_t + \xi\big((-1+Dw)z\big)_\xi &= -Dwz,\label{zeqnxi-final0}\\
        tw_t + \xi(-1+Dw)w_\xi &= w - D\bigg(w^2+\frac{1}{2\xi^2}(1-\xi^2w^2)\frac{1-A}{A}\bigg).\label{weqnxi-final0}
    \end{align}
    Moreover, equations (\ref{Axi-final0})--(\ref{weqnxi-final0}) imply the two equations:
    \begin{align}
        tA_t + \xi A_\xi &= wz,\\
        \xi\frac{B_\xi}{B} &= \frac{\xi^2w^2}{A}z + \frac{1-A}{A},
    \end{align}
    which are equivalent to (\ref{firstorder2}) and (\ref{firstorder3}) respectively.
\end{Corollary}

The proof of Theorem \ref{thmsigma} is given in Section \ref{Appendix1C}.

The next theorem justifies our proposal that the ambient Euclidean coordinate system $\vec{x}=(x^0,x^1,x^2,x^3)=(t,x,y,z)$ associated with our spherical SSC system $(t,r)$ provides the coordinate system that imposes the correct smoothness condition for SSC solutions in a neighborhood of $r=0$.

\begin{Thm}\label{ThmsmoothAgain}
    Assume $A(t,\xi)$, $D(t,\xi)$, $z(t,\xi)$ and $w(t,\xi)$ are a given smooth solution of the $p=0$ equations (\ref{Axi-final0})--(\ref{weqnxi-final0}) satisfying:
    \begin{align}
        A &= 1 + O(\xi^2), & D &= 1 + O(\xi^2), & z &= O(\xi^2), & w &= w_0(t) + O(\xi^2),\label{ansatzassumption}
    \end{align}
    and assume that at $t=t_*>0$ the solution agrees with initial data:
    \begin{align*}
        A(t_*,\xi) &= \bar{A}(\xi), & D(t_*,\xi) &= \bar{D}(\xi), & z(t_*,\xi) &= \bar{z}(\xi), & w(t_*,\xi) &= \bar{w}(\xi),
    \end{align*}
    such that each initial data function $\bar{A}(\xi)$, $\bar{D}(\xi)$, $\bar{z}(\xi)$ and $\bar{w}(\xi)$ satisfies the condition that all odd derivatives vanish at $\xi=0$. Then all odd derivatives of $A(t,0)$, $D(t,0)$, $z(t,0)$ and $w(t,0)$ vanish for all $t>t_*$.
\end{Thm}

The proof of Theorem \ref{ThmsmoothAgain} is given in Section \ref{Appendix1D}.
 
\section{The STV-ODE}\label{S7}

To describe the evolution of solutions of equations (\ref{Axi-final0})--(\ref{weqnxi-final0}) near $\xi=0$ in SSCNG, we assume the following asymptotic ansatz:
\begin{align}
    A(t,\xi) &= 1 + A_2(t)\xi^2 + A_4(t)\xi^4 + O(\xi^6),\label{Aansatz}\\
    D(t,\xi) &= 1 + D_2(t)\xi^2 + O(\xi^4),\label{Dansatz}\\
    z(t,\xi) &= z_2(t)\xi^2 + z_4(t)\xi^4 + O(\xi^6),\label{zansatz}\\
    w(t,\xi) &= w_0(t) + w_2(t)\xi^2 + O(\xi^4).\label{wansatz}
\end{align}
Note that starting the expansion of $A$ and $D$ at unity forces $B=1$ at $\xi=0$, so this imposes the normalized gauge condition. Note also that we have included only even powers of $\xi$, which is equivalent to the assumption that the solution is smooth at $\xi=r=0$, see \cite{SmolTeVo}. As a special case, we have from Section \ref{S5} the following expansion of the $p=0$, $k=0$ Friedmann solution:
\begin{align}
    A_F(\xi) &= 1 + A_2^F\xi^2 + A_4^F\xi^4 + O(\xi^6),\label{AF}\\
    D_F(\xi) &= 1 + D_2^F\xi^2 + O(\xi^4),\label{DF}\\
    z_F(\xi) &= z_2^F\xi^2 + z_4^F\xi^4 + O(\xi^6),\label{zF}\\
    w_F(\xi) &= w_0^F + w_2^F\xi^2 + O(\xi^4),\label{wF}
\end{align}
with
\begin{align}
    \boldsymbol{U}_F := (z_2^F,w_0^F,z_4^F,w_2^F) = \bigg(\frac{4}{3},\frac{2}{3},\frac{40}{27},\frac{2}{9}\bigg)\label{wFk}
\end{align}
and:
\begin{align}
    A_2^F &= -\frac{1}{3}z_2^F = -\frac{4}{9},\label{AFk}\\
    A_4^F &= -\frac{1}{5}z_4^F = -\frac{8}{27},\\
    D_2^F &= -\frac{1}{12}z_2^F = -\frac{1}{9}.\label{DFk}
\end{align}
The time independence of the coefficients of powers of $\xi$ in (\ref{AF})--(\ref{wF}) reflects the fact that the $p=0$, $k=0$ Friedmann spacetime is self-similar in SSCNG coordinates \cite{smoltePNAS}. To see that the equations for the ansatz (\ref{Aansatz})--(\ref{wansatz}) close at every even power of $\xi$, and to obtain the equations, we substitute (\ref{Aansatz})--(\ref{wansatz}) into equations (\ref{Axi-final0})--(\ref{weqnxi-final0}) and collect like powers of $\xi$. The result up to order $O(\xi^6)$ in $z$ and order $O(\xi^4)$ in $w$ (which is $O(\xi^5)$ in velocity $v=\xi w$) is stated in the following theorem.

\begin{Thm}\label{ODEsForCorrections}
    Putting the ansatz (\ref{Aansatz})--(\ref{wansatz}) into equations (\ref{Axi-final0})--(\ref{weqnxi-final0}) and equating like powers of $\xi$ leads to the following autonomous ODE for $A_i=A_i(t)$, $D_i=D_i(t)$, $z_i=z_i(t)$ and $w_i=w_i(t)$:
    \begin{align}
        t\dot{z}_2 &= 2z_2 - 3z_2w_0,\label{z2finalequation}\\
        t\dot{w}_0 &= -\frac{1}{6}z_2 + w_0 - w_0^2,\label{w0finalequation}\\
        t\dot{z}_4 &= \frac{5}{12}z_2^2w_0 - 5w_0z_4 + 4z_4 - 5z_2w_2,\label{z4finalequation}\\
        t\dot{w}_2 &= -\frac{1}{24}z_2^2 + \frac{1}{4}z_2w_0^2 - \frac{1}{10}z_4 - 4w_0w_2 + 3w_2,\label{w2finalequation}
    \end{align}
    together with:
    \begin{align}
        A_2 &= -\frac{1}{3}z_2, & A_4 &= -\frac{1}{5}z_4, & D_2 &= -\frac{1}{12}z_2.\label{metricfluidrelations}
    \end{align}
\end{Thm}

We refer to system (\ref{z2finalequation})--(\ref{w2finalequation}) as the STV-ODE of order $n=2$.

The proof of Theorem \ref{ODEsForCorrections} is given in Section \ref{Appendix1E}.

Note that relations (\ref{AFk}) and (\ref{DFk}) between metric components and fluid variables at each order in the expansion of the $k=0$ Friedmann solution anticipates (\ref{metricfluidrelations}), which holds at each order in the expansion of general smooth solutions in powers of $\xi$ in SSCNG coordinates. This simplifies the ODE for the corrections significantly, as it reduces the number of unknowns from seven, to the four unknowns $(z_2,w_0,z_4,w_2)$.

\section{The Phase Portrait for The STV-ODE}\label{S8}

Recall that letting $\tau=\ln t$ gives\footnote{We use a dot to denote $\frac{d}{dt}$ and a prime to denote $\frac{d}{d\tau}$.}
\begin{align*}
    \frac{d}{d\tau} = t\frac{d}{dt},
\end{align*}
and this turns the ODE (\ref{z2finalequation})--(\ref{w2finalequation}) into an autonomous system in $\tau$. Letting 
\begin{align*}
    \boldsymbol{U} = (z_2,w_0,z_4,w_4),
\end{align*}
system (\ref{z2finalequation})--(\ref{w2finalequation}) takes the form
\begin{align*}
    \frac{d}{d\tau}\boldsymbol{U} = \boldsymbol{F}(\boldsymbol{U}),
\end{align*}
where
\begin{align}
    \frac{d}{d\tau}\boldsymbol{U} = t\left(\begin{array}{c}
    \dot{z}_2\\
    \dot{w}_0\\
    \dot{z}_4\\
    \dot{w}_2
	\end{array}\right) = \left(\begin{array}{l}
	2z_2-3z_2w_0\\
    -\frac{1}{6}z_2+w_0-w_0^2\\
    \frac{5}{12}w_0z_2^2+4z_4-5w_0z_4-5z_2w_2\\
    -\frac{1}{24}z_2^2+\frac{1}{4}z_2w_0^2-\frac{1}{10}z_4+3w_2-4w_0w_2
	\end{array}\right) =: \boldsymbol{F}(\boldsymbol{U}).\label{4by4system}
\end{align}
Since (\ref{4by4system}) is autonomous, it can be described by a phase portrait, and the phase portrait is essentially determined by the structure of its rest points. Writing the ODE in the order $(z_2,w_0,z_4,w_2)$ and solving $\boldsymbol{F}(\boldsymbol{U})=0$ for $\boldsymbol{U}$ gives the rest points:
\begin{align}
    SM &= \boldsymbol{U}_F = \bigg(\frac{4}{3},\frac{2}{3},\frac{40}{27},\frac{2}{9}\bigg), & M &= (0,1,0,0), & U &= (0,0,0,0).\label{restpoints}
\end{align}
Note that the first two equations in (\ref{4by4system}) close in variables $(z_2,w_0)$ to form the $2\times2$ system
\begin{align}
    \frac{d}{d\tau}\left(\begin{array}{c}
    z_2\\
    w_0
	\end{array}\right) = \left(\begin{array}{l}
    2z_2-3z_2w_0\\
    -\frac{1}{6}z_2+w_0-w_0^2
	\end{array}\right) =: \left(\begin{array}{c}f(z_2,w_0)\\
    g(z_2,w_0)\end{array}\right) =: \boldsymbol{f}(z_2,w_0),\label{2by2system}
\end{align}
and the rest points of (\ref{2by2system}) are the restriction of the rest points $SM$, $M$ and $U$ of the $4\times4$ system to the first two components. Regarding the rest point $SM$, we know from Section \ref{S5} that the $p=0$, $k=0$ Friedmann metric (the Standard Model) is self-similar in SSCNG, so we know ahead of time that its expansion in powers of $\xi$ in (\ref{AF})--(\ref{wF}) implies $\boldsymbol{U}_F$ must determine a rest point of system (\ref{4by4system}), with this rest point being $SM$. Thus $SM$ is the solution of (\ref{4by4system}) corresponding to the first two terms in the expansion of the Standard Model in even powers of $\xi$, with the evolution of perturbations of $SM$ described by nearby solutions of (\ref{4by4system}). As for the rest point $M$, we observe that $z=0$ and $w=1$ ($v=\xi$) solves the self-similar Einstein field equations (\ref{Axi-final0})--(\ref{weqnxi-final0}) exactly with $A=B=D=1$, so a rest point $M$ satisfying $w_0=1$ with all other coefficients being null must also be a rest point at every level of expansion of solutions of the self-similar Einstein field equations in even powers of $\xi$. We will see at the level of (\ref{2by2system}) and (\ref{4by4system}) that $SM$ is an unstable saddle rest point and $M$ is a stable rest point, although $M$ is also a stable rest point for all higher levels of approximation too. Moreover, the underdense side of the unstable manifold of $SM$ contains trajectories which connect $SM$ to $M$, one of which, together with its time translations $\tau\to\tau-\Delta_0$, corresponds to the one parameter family of $k<0$ Friedmann spacetimes. In Theorem \ref{Thmgeneralexpansion}, we prove that all trajectories that tend to $M$ in the $(z_2,w_0)$-plane also tend to $M$ at every order of expansion of smooth solutions of (\ref{Axi-final0})--(\ref{weqnxi-final0}). In the next section we begin the description of the phase portrait of the $2\times2$ system (\ref{2by2system}). This will play a basic role in the description of the phase portrait for (\ref{4by4system}), which will be discussed in the section after. Without confusion, we use $SM$, $M$ and $U$ to label the three rest points in the $2\times2$, $4\times4$ and general $2n\times2n$ phase portraits.

\subsection{Phase Portrait for the $2\times2$ System}

In this section our notation is to use $\boldsymbol{u}=(z_2,w_0)$, $\boldsymbol{v}=(z_4,w_2)$ and
\begin{align*}
    \boldsymbol{U} = (\boldsymbol{u},\boldsymbol{v}) = (z_2,w_0,z_4,w_2).
\end{align*}
Now equations (\ref{z2finalequation})--(\ref{w2finalequation}) close at each even order of $\xi$, so to begin, we describe the phase portrait for the two equations (\ref{z2finalequation}) and (\ref{w0finalequation}) in $\boldsymbol{u}=(z_2,w_0)$, given by (\ref{2by2system}). Solutions of $2\times2$ autonomous ODE are characterized by their phase portrait. For this, observe first that
\begin{align*}
    f(z_2,w_0) = 2z_2 - 3z_2w_0 = 0
\end{align*}
gives the $z_2$ contour as $w_0=\frac{2}{3}$ or $z_2=0$, and
\begin{align*}
    g(z_2,w_0) = -\frac{1}{6}z_2 + w_0 - w_0^2 = 0
\end{align*}
gives the $w_0$ contour as $z_2=6w_0(1-w_0)$. Thus in the $(z_2,w_0)$-plane, the contours intersect in the three rest points of the system:
\begin{align*}
    SM &= \bigg(\frac{4}{3},\frac{2}{3}\bigg), & M &= (0,1), & U &= (0,0).
\end{align*}
Here $(z_2^F,w_0^F)=(\frac{4}{3},\frac{2}{3})$ are the first two components of $\boldsymbol{U}_F$, so for notational convenience (and to be consistent with \cite{SmolTeVo}), we denote by $SM$ the rest point corresponding to $\boldsymbol{U}_F$ in both the $2\times2$ and $4\times4$ systems (\ref{2by2system}) and (\ref{4by4system}) respectively. We now show that $SM$ (for Standard Model) is an unstable saddle rest point, $M$ (for Minkowski) is a stable rest point and $U$ is a fully unstable rest point.

The Jacobian, $d\boldsymbol{f}$, of $\boldsymbol{f}(z_2,w_0)$ is given by
\begin{align*}
    d\boldsymbol{f}(z_2,w_0) = \left(\begin{array}{cc}
    2-3w_0 & -3z_2\\
    -\frac{1}{6} & 1-2w_0
	\end{array}\right).
\end{align*}
At the rest point $SM$, the Jacobian is
\begin{align*}
    d\boldsymbol{f}\left(\frac{4}{3},\frac{2}{3}\right) = \left(\begin{array}{cc}0 & -4\\
    -\frac{1}{6} & -\frac{1}{3}\end{array}\right)
\end{align*}
and the eigenpairs are given by:
\begin{align*}
    \lambda_{A1} &= \frac{2}{3}, & \boldsymbol{R}_{A1} &= \left(\begin{array}{c}9\\
    -\frac{3}{2}\end{array}\right); & \lambda_{B1} &= -1, & \boldsymbol{R}_{B1} &= \left(\begin{array}{c}4\\
    1\end{array}\right).
\end{align*}
At the rest point $M$, the Jacobian is
\begin{align*}
    d\boldsymbol{f}(0,1) = \left(\begin{array}{cc}-1 & 0\\
    -\frac{1}{6} & -1\end{array}\right).
\end{align*}
The rest point $M$ has the double eigenvalue $\lambda_M=-1$, which has a resonant Jordan normal form with a one-dimensional eigenspace:
\begin{align*}
    \lambda_M &= -1, & \boldsymbol{R}_M &= \left(\begin{array}{c}0\\
    1\end{array}\right).
\end{align*}
Finally, at the rest point $U$, the Jacobian is
\begin{align*}
    d\boldsymbol{f}(0,0) = \left(\begin{array}{cc}2 & 0\\
    -\frac{1}{6} & 1\end{array}\right)
\end{align*}
and the eigenpairs are given by:
\begin{align*}
    \lambda_{U1} &= 1, & \boldsymbol{R}_{U1} &= \left(\begin{array}{c}0\\
    1\end{array}\right); & \lambda_{U2} &= 2, & \boldsymbol{R}_{U2} &= \left(\begin{array}{c}-6\\
    1\end{array}\right).
\end{align*}
The phase portrait for system (\ref{2by2system}) is depicted in Figure \ref{Figure1}. The main feature is that the rest point $SM$, which corresponds to the Standard Model $p=0$, $k=0$ Friedmann spacetime, is an unstable saddle rest point. Note first that $z_2=0$ is a solution trajectory of (\ref{2by2system}), so $z_2\geq0$ is an invariant region, since solution trajectories never cross in autonomous systems. Thus the solutions in the stable manifold of $M$ consist of all trajectories to the left of the two trajectories in the stable manifold of $SM$ in the $(z_2,w_0)$-plane, and to the right of $z_2=0$. These include all smooth radial underdense perturbations of $SM$, and all of these solutions enter $M$ asymptotically from below, along the eigendirection $\boldsymbol{R}_M=(0,1)^T$, that is, parallel to the $w_0$-axis. The unstable manifold of $SM$ thus has two components: The trajectories that connect the rest point $SM$ to the stable rest point $M$ on the underdense (smaller $z_2$) side of $SM$, and the trajectories on the overdense (larger $z_2$) side of $SM$, which continue to larger values of $z$ until they hit $w_0=0$. We show in the next section that these two components of the unstable manifold of $SM$ correspond to the $p=0$, $k\neq0$ Friedmann spacetimes to leading order in $\xi$. To complete the picture, the stable manifold of $SM$ on the underdense side is a single trajectory connecting the rest point $U$ to $SM$, and the stable manifold of $SM$ on the overdense side is a trajectory which goes off to infinity.

\subsection{The $2\times2$ Unstable Manifold of $SM$ is $k=-1$ Friedmann}

The $k=-1$ Friedmann family of solutions is parameterized by $\Delta_0>0$, and each one solves the STV-PDE exactly. It follows that the leading order term in the expansion of the $k=-1$ Friedmann solutions, that is, the projection of this family of solutions onto the $(z_2,w_0)$-plane, will produce a family of exact solutions $(z_2(t),w_0(t))$ of system (\ref{2by2system}), also parameterized by $\Delta_0$. The next theorem gives an implicit expression for these solutions. From this expression we find that each $k=-1$ Friedmann solution moves from the rest point $SM$ to the rest point $M$ as $\tau$ ranges from $-\infty$ to $+\infty$. It follows that this solution provides an exact expression for the portion of the unstable manifold of $SM$ consisting of the trajectory which connects $SM$ to $M$, and the parameter $\Delta_0>0$ represents time translation $\tau\to\tau-\tau_0$ with $\tau_0=\Delta_0$. In this way the $k=-1$ Friedmann family provides an exact formula for the portion of the unstable manifold of $SM$ given by the trajectory which connects $SM$ to $M$ in the phase portrait of the $2\times2$ system (\ref{2by2system}).

\begin{Thm}\label{2x2connecting orbit} 
    The expansion of the $k=-1$ Friedmann solution (with parameter $\Delta_0>0$) in even powers of $\xi$ produces the following implicit formulas for $z_2(t)$ and $w_0(t)$ in terms of $\theta=\theta(t)$:
    \begin{align}
        z_2 &= \tilde{z}_2(\theta) = \frac{6(\sinh2\theta-2\theta)^2}{(\cosh2\theta-1)^3} = -3A_2,\label{z2FinalAgain2by2}\\
        w_0 &= \tilde{w}_0(\theta) = \frac{(\sinh2\theta-2\theta)\sinh2\theta}{(\cosh2\theta-1)^2},\label{w0FinalAgain2by2}
    \end{align}
    where $\theta\geq0$ is defined as a function of $t\geq0$ through the relation
    \begin{align}
        \frac{t}{\Delta_0} = \frac{1}{2}(\sinh2\theta-2\theta).\label{Theta22by2}
    \end{align}
    Moreover, equation (\ref{Theta22by2}) inverts to define the inverse function
    \begin{align*}
        \Theta:(0,\infty) &\to (0,\infty), & \theta(t) &= \Theta\Big(\frac{t}{\Delta_0}\Big),
    \end{align*}
    and in terms of $\Theta$ defined by (\ref{Theta2}), the $p=0$, $k=-1$ expansion of Friedmann for general $\Delta_0>0$ is given by:
    \begin{align}
        z_2^F(t) &= \tilde{z}_2\bigg(\Theta\Big(\frac{t}{\Delta_0}\Big)\bigg),\label{twithDelta12by2}\\
        w_0^F(t) &= \tilde{w}_0\bigg(\Theta\Big(\frac{t}{\Delta_0}\Big)\bigg).\label{twithDelta22by2}
    \end{align}
    Furthermore, for every $\Delta_0>0$ we have:
    \begin{align}
        \lim_{t\to0}\big(z_2^F(t),w_0^F(t)\big) &= SM = \bigg(\frac{4}{3},\frac{2}{3}\bigg),\label{limitsFriedmann2by2-1}\\
        \lim_{t\to\infty}\big(z_2^F(t),w_0^F(t)\big) &= M =(0,1).\label{limitsFriedmann2by2-2}
    \end{align}
\end{Thm}

Note that by (\ref{z2FinalAgain2by2}) and (\ref{w0FinalAgain2by2}), $\big(\tilde{z}_2(\Theta(t)),\tilde{w}_0(\Theta(t))\big)$ represents the $k=-1$ Friedmann solution in the case $\Delta_0=\frac{4}{9}$.

\begin{proof}
    That (\ref{z2FinalAgain2by2}) and (\ref{w0FinalAgain2by2}) is an exact solution of system (\ref{2by2system}) follows directly from the fact that it agrees with the leading order expansion of an exact solution of the STV-PDE, namely, the $k=-1$ Friedmann solution. Also, that (\ref{z2FinalAgain2by2}) and (\ref{w0FinalAgain2by2}) describe the connecting orbit which takes $SM$ to $M$ follows from (\ref{limitsFriedmann2by2-1}) and (\ref{limitsFriedmann2by2-2}) by a simple calculation. Theorem \ref{2x2connecting orbit} is a special case of Theorem \ref{Friedmannknotzeroexpansion}, whose proof is given in Sections \ref{Appendix1F} and \ref{Appendix1G}.
\end{proof}

In summary, the phase portrait for system (\ref{2by2system}) consists of three rest points: $U$, $SM$ and $M$, and the connecting orbit between $SM$ and $M$ is the projection of the $k=-1$ Friedmann solution onto the leading order $(z_2,w_0)$-plane, with this trajectory described exactly by (\ref{twithDelta12by2}) and (\ref{twithDelta22by2}). The region $z_2\geq0$ is an invariant region because $z_2=0$ solves (\ref{2by2system}) and solution trajectories never cross in autonomous systems. The stable manifold of $M$ consists of all trajectories to the left of the two trajectories in the stable manifold of $SM$ in the $(z_2,w_0)$-plane and to the right of $z_2=0$. All of these solutions, including $(z_2^F(t),w_0^F(t))$, enter $M$ asymptotically, from below, along the eigendirection parallel to the $w_0$-axis. We now quantify this decay to $M$ with estimates, and prove that, under appropriate time translation, all solutions entering $M$ converge to $k=-1$ Friedmann at a faster rate than they converge to $M$ for each fixed $r$ (but not fixed $\xi$) as $t\to\infty$.

To study decay to $M=(0,1)$, let $x=z_2$, $y=1-w_0$ and write system (\ref{2by2system}) in the equivalent form
\begin{align*}
    \frac{d}{d\tau}\left(\begin{array}{c}x\\
    y\end{array}\right) = \left(\begin{array}{l}-x+3xy\\
    \frac{1}{6}x-y+y^2\end{array}\right) = \boldsymbol{f}(x,y).
\end{align*}
Separating the linear and nonlinear parts, we obtain the equivalent system in matrix form
\begin{align}
    \frac{d}{d\tau}\left(\begin{array}{c}x\\
    y\end{array}\right) = \left(\begin{array}{cc}-1 & 0\\
    \frac{1}{6} & -1\end{array}\right)\left(\begin{array}{c}x\\
    y\end{array}\right) + y\left(\begin{array}{cc}3 & 0\\
    0 & 1\end{array}\right)\left(\begin{array}{c}x\\
    y\end{array}\right).\label{2by2systemxy}
\end{align}
Note that the first matrix is $d\boldsymbol{f}(0,1)$, where $(x,y)=(0,0)$ is a non-degenerate rest point representing $M$ in the $(x,y)$-plane, so this rest point has the double eigenvalue $\lambda_M=-1$ and a resonant Jordan normal form with single eigenvector $\boldsymbol{R}_M=(0,1)^T$. By the Hartman--Grobman theorem, solutions of nonlinear autonomous systems entering a non-degenerate rest point, look asymptotically, to within quadratic errors $|\boldsymbol{U}|^2=|(x,y)|^2$, like the corresponding solution of the linear system.

Linearizing around the rest point $U=(0,0)$ of (\ref{2by2systemxy}), we obtain the linear system
\begin{align}
    \frac{d}{d\tau}\left(\begin{array}{c}x\\
    y\end{array}\right) = \left(\begin{array}{cc}-1 & 0\\
    \frac{1}{6} & -1\end{array}\right)\left(\begin{array}{c}x\\
    y\end{array}\right).\label{2by2systemxylin}
\end{align}
The solution $(\bar{x},\bar{y})$ of the linear system (\ref{2by2systemxylin}), satisfying initial data:
\begin{align}
    \bar{x}(\tau_*) &= x_*, & \bar{y}(\tau_*) &= y_*,\label{initialdata}
\end{align}
is
\begin{align}
    \bar{x}(\tau) &= x_*e^{-\tau}, & \bar{y}(\tau) &= \bigg(e^{\tau_*}y_*+\frac{1}{6}x_*(\tau-\tau_*)\bigg)e^{-\tau}.
\end{align}
Thus, assuming without loss of generality that $\tau_*\geq1$, by the Hartman--Grobman theorem, we can initially conclude that the solution $\boldsymbol{U}(\tau)=(x(\tau),y(\tau))$ of the nonlinear system (\ref{2by2systemxy}) satisfying the same initial condition (\ref{initialdata}), is given by:
\begin{align}
    x(\tau) &= a(\tau)e^{-\tau},\label{firstx}\\
    y(\tau) &= b(\tau)\tau e^{-\tau},\label{firsty}
\end{align}
where $a(\tau)$ and $b(\tau)$ are uniformly bounded, that is,
\begin{align}
    |a(\tau)| &\leq C, & |b(\tau)| &\leq C,\label{abbounds}
\end{align}
for some constant $C$ depending only on $x_*$ and $y_*$. That is, a smooth trajectory is bounded on any compact interval $[\tau_*,\tau]$, so the uniformity of $C$ over all $\tau\geq\tau_*$ follows from the decay to $M$ as $\tau\to\infty$. The assumption $\tau_*\geq1$, which is no loss of generality in light of the time translation invariance of autonomous systems of ODE, simplifies formulas by allowing us to bound constants by $\tau$. To obtain a more refined estimate, we use the Hartman--Grobman theorem, which implies that a solution $\boldsymbol{U}(\tau)=(x(\tau),y(\tau))$ of the nonlinear system (\ref{2by2systemxy}) satisfying the same initial condition (\ref{initialdata}), satisfies
\begin{align*}
    \boldsymbol{U} = \bar{\boldsymbol{U}} + O\big(|\boldsymbol{U}|^2\big).
\end{align*}
This, together with (\ref{firstx}) and (\ref{firsty}), translates into:
\begin{align}
    x(\tau) &= x_*e^{-\tau} + a^2(\tau)e^{-2\tau},\label{nextestx}\\
    y(\tau) &= \bigg(e^{\tau_*}y_* + \frac{1}{6}x_*(\tau-\tau_*)\bigg)e^{-\tau} + b^2(\tau)e^{-2\tau}.\label{nestesty}
\end{align}
Thus $y(\tau)$ decays to $M$ by a factor of $\tau$ slower than $x(\tau)$ due to the resonant double eigenvalue of the rest point $M$. The next lemma shows that the difference between two solutions, under appropriate time translation, decays faster. We use this to establish that, at leading order, solutions tending to $M$ decay to the $k=-1$ Friedmann solutions by a factor of $\tau$ faster than they decay to $M$.

\begin{Lemma}\label{diffxylemma}
    Let $\boldsymbol{U}_1(\tau)=(x_1(\tau),y_1(\tau))$ and $\boldsymbol{U}_2(\tau)=(x_2(\tau),y_2(\tau))$ be two solutions of (\ref{2by2systemxy}) satisfying (\ref{firstx})--(\ref{abbounds}) with initial conditions:
    \begin{align*}
        \boldsymbol{U}_1(\tau_*) &= \boldsymbol{U}_1^* = (x_1^*,y_1^*), & \boldsymbol{U}_2(\tau_*) &= \boldsymbol{U}_2^* = (x_2^*,y_2^*),
    \end{align*}
    respectively, where
    \begin{align*}
        x_1^* = x_2^* = x^*
    \end{align*}
    and, without loss of generality, assume $\tau_*\geq1$. Then there exists a constant $C$, depending only on $x^*$, $y_1^*$, $y_2^*$ and $\tau_*$, such that
    \begin{align}
        |\boldsymbol{U}_1-\boldsymbol{U}_2| \leq Ce^{-\tau}\label{Ufasterrate}
    \end{align}
    for all $\tau\geq\tau_*\geq1$.
\end{Lemma}

That is, by (\ref{Ufasterrate}), both $|x_2(\tau)-x_1(\tau)|\leq Ce^{-\tau}$ and $|y_2(\tau)-y_1(\tau)|\leq Ce^{-\tau}$, so comparing (\ref{Ufasterrate}) to (\ref{nestesty}), the two solutions converge to each other faster than they decay separately to the rest point $M$.

\begin{proof}
By (\ref{nextestx}),
\begin{align*}
    |x_2(\tau)-x_1(\tau)| &\leq \big(a_2(\tau)^2+a_1(\tau)^2\big)e^{-2\tau} \leq C_xe^{-1},\\
    |y_2(\tau)-y_1(\tau)| &\leq \big(e^{\tau_*}(|y_1^*|+|y_2^*|)\big)e^{-t} + \big(b_1(\tau)^2+b_2(\tau)^2\big)e^{-2\tau} \leq C_ye^{-1},
\end{align*}
for $C=\max\{C_x,C_y\}$, where by (\ref{abbounds}),
\begin{align*}
    C_x &\leq 2C_1^2, & C_y &\leq e^{\tau_*}(|y_1^*|+|y_2^*|)+C_1^2.
\end{align*}
\end{proof}

We use this Lemma to prove the following theorem, which implies that, for any given leading order solution $(z_2(\tau),w_0(\tau))$ which tends to the rest point $M$ as $\tau\to\infty$, there always exists a value of $\delta>0$ such that it decays faster asymptotically to the $k=-1$ Friedmann spacetime $(z_2^F(\tau),w_0^F(\tau))$ with $\Delta_0=\delta$, than it decays to the rest point $M$. Again, since smooth solutions of (\ref{2by2systemxy}) starting at $\tau_*\in\mathbb{R}$ are bounded on the compact interval $[\tau_*,\tau]$ for any $\tau_*\leq\tau<\infty$, and solutions are preserved under time translation $\tau\to\tau-\tau_*+1$ because system (\ref{2by2systemxy}) is autonomous, then to keep things simple, and without loss of generality, we state the following theorem in terms of solutions defined for $\tau\geq\tau_*$, assuming initial time $\tau_*\geq1$.

\begin{Thm}\label{2by2phaseportrait}
    Let $\tau_*\geq1$ and $\boldsymbol{U}(\tau)=(z_2(\tau),w_0(\tau))$ be any solution of the $2\times2$ system (\ref{2by2system}), with initial condition $\boldsymbol{U}(\tau_*)=\boldsymbol{U}_*=(z_2^*,w_0^*)\in\mathbb{R}_+^2$, which tends to the rest point $M=(0,1)$ as $\tau\to\infty$. Then there exists a constant $C>0$, depending only on the initial data $\boldsymbol{U}_*$ and $\tau_*$, such that:
    \begin{align}
        z_2(\tau) &= a(\tau)e^{-\tau},\label{z2bound}\\
        w_0(\tau) &= b(\tau)\tau e^{-\tau},\label{w0bound}
    \end{align}
    where
    \begin{align}
        |a(\tau)| &\leq C, & |b(\tau)| &\leq C.
    \end{align}
    Moreover, there exists a $k<-1$ Friedmann solution with leading order components $(z_2^F(\tau),w_0^F(\tau))$ and a constant $C_F>0$, depending only on the initial data $\boldsymbol{U}_*$ and $\tau_*$, such that:
    \begin{align}
        \big|z_2(\tau)-z_2^F(\tau)\big| &\leq C_Fe^{-\tau},\label{zbound2by2}\\
        \big|w_0(\tau)-w_0^F(\tau)\big| &\leq C_Fe^{-\tau},\label{wbound2by2}
    \end{align}
    or in terms of $t=e^{\tau}$,
    \begin{align}
        \big|z_2(\ln t)-z_2^F(\ln t)\big| &\leq \frac{C_F}{t},\\
        \big|w_0(\ln t)-w_0^F(\ln t)\big| &\leq \frac{C_F}{t}.
    \end{align}
\end{Thm}

\begin{proof}
Equations (\ref{z2bound}), (\ref{w0bound}) and (\ref{zbound2by2}) follow directly from (\ref{firstx})--(\ref{abbounds}). To apply Lemma \ref{diffxylemma} and obtain an improvement of factor $\tau$ over (\ref{w0bound}), note that since $z_2^*>0$, it follows that $z_2(\tau)$ takes on all values in the interval $(0,z_2^*)$ as $\tau$ ranges from $\tau_0$ to $\infty$. Thus, with the possible change of a constant, we can assume without loss of generality that $\tau_*\in(0,\frac{4}{3})$. Since $\tilde{z}_2$, defined in (\ref{z2FinalAgain2by2}), that is, the leading order $k=-1$ Friedmann solution with $\Delta_0=\frac{4}{9}$, ranges from the $SM$ value $z_2=\frac{4}{3}$ to the $M$ value $z_2=0$ as $\tau$ ranges from $-\infty$ to $\infty$, it follows that there must exist a time $\tau_F$ at which $\tilde{z}_2(\tau_F)=z_2(\tau_*)=z_2^*$. Then setting $\Delta_0=\tau_F-\tau_*$, we have
\begin{align*}
    z_2^F(\tau) = \tilde{z}_2(\tau-\Delta_0).
\end{align*}
Therefore Lemma \ref{diffxylemma} applies, and the improved rate (\ref{wbound2by2}) follows from (\ref{Ufasterrate}).
\end{proof}

We have proven that the coefficients $(z_2(\tau),w_0(\tau))$ of solutions which tend to the rest point $M$, such as the $k<0$ Friedmann solutions, actually converge to a $k<0$ Friedmann solution by a factor $\tau$ faster than they decay to $M$. Thus at the level of the $2\times2$ system, the $k<0$ family of spacetimes is a forward time global attractor for nearby solutions. Although these solutions which enter $M$ are stable in forward direction, all of them, including $k<0$ Friedmann solutions, are unstable in backward time. Thus one expects to see $k<0$ Friedmann solutions at late times, but we cannot assume solutions close to $k<0$ Friedmann at late times were also close to $k<0$ Friedmann solutions at early times. In the next section we show that for the $4\times4$ system (\ref{4by4system}), the $k<0$ Friedmann solutions represent just a single parameter in a two parameter family of solutions which span the unstable manifold of $SM$ at this order. This extra parameter induces third order effects in the velocity, which differentiate general perturbations of $SM$ that tend to $M$ as $t\to\infty$ from the $k=-1$ Friedmann spacetimes.\footnote{For example, the third order velocity term affects the third order correction to redshift vs luminosity, the term which differentiates the $k<0$ Friedmann spacetimes from the $k=0$ Friedmann spacetime with a cosmological constant, see \cite{SmolTeVo}.}

\subsection{Phase Portrait for the $4\times4$ System}

We now describe the phase portrait for the $4\times4$ system (\ref{4by4system}). We first analyze the rest points $SM$, $M$ and $U$ given in (\ref{restpoints}). Again observe that the rest points in the $2\times2$ system are the restriction of the rest points $SM$, $M$ and $U$ to the $\boldsymbol{U}=(z_2,w_0)$ plane, and thus we label them the same. The nonlinear function on the right hand side of (\ref{4by4system}) is
\begin{align*}
    \boldsymbol{F}(\boldsymbol{U}) = \left(\begin{array}{l}2z_2-3z_2w_0\\
    -\frac{1}{6}z_2+w_0-w_0^2\\
    \frac{5}{12}w_0z_2^2+4z_4-5w_0z_4-5z_2w_2\\
    -\frac{1}{24}z_2^2+\frac{1}{4}z_2w_0^2-\frac{1}{10}z_4+3w_2-4w_0w_2\end{array}
    \right)
\end{align*}
and the Jacobian of $\boldsymbol{F}$ at $\boldsymbol{U}=(z_2,w_0,z_4,w_2)$ is
\begin{align*}
    d\boldsymbol{F}(\boldsymbol{U}) = \left(\begin{array}{llll}-3w_0+2 & -3z_2 & 0 & 0\\
    -\frac{1}{6} & 1-2w_0 & 0 & 0\\
    \frac{5}{6}z_2w_0-5w_2 & \frac{5}{12}z_2^2-5z_4 & 4-5w_0 & -5z_2\\
    -\frac{1}{12}z_2+\frac{1}{4}w_0^2 & \frac{1}{2}z_2w_0-4w_2 & -\frac{1}{10} & 3-4w_0\end{array}\right).
\end{align*}

For the rest point $SM$, the Jacobian is
\begin{align}
    d\boldsymbol{F}\left(\frac{4}{3},\frac{2}{3},\frac{40}{27},\frac{2}{9}\right) = \left(\begin{array}{cccc}0 & -4 & 0 & 0\\
    -\frac{1}{6} & -\frac{1}{3} & 0 & 0\\
    -\frac{10}{27} & -\frac{20}{3} & \frac{2}{3} & -\frac{20}{3}\\
    0 & -\frac{4}{9} & -\frac{1}{10} & \frac{1}{3}\end{array}\right)\label{JacSM}
\end{align}
and the eigenpairs are given by:
\begin{align}
    \lambda_{A1} &=\frac{2}{3}, & \boldsymbol{R}_{A1} &= \left(\begin{array}{c}9\\
    -\frac{3}{2}\\
    \frac{10}{3}\\
    1\end{array}\right); & \lambda_{B1} &= -1, & \boldsymbol{R}_{B1} &= \left(\begin{array}{c}4\\
    1\\
    \frac{80}{9}\\
    1\end{array}\right);\label{epr12SM}\\
    \lambda_{A2} &= \frac{4}{3}, & \boldsymbol{R}_{A2} &= \left(\begin{array}{c}0\\
    0\\
    -10\\
    1\end{array}\right); & \lambda_{B2} &= -\frac{1}{3}, & \boldsymbol{R}_{B2} &= \left(\begin{array}{c}0\\
    0\\
    \frac{20}{3}\\
    1\end{array}\right).\label{epr34SM}
\end{align}

For the rest point $M$, the Jacobian is
\begin{align}
    d\boldsymbol{F}(0,1,0,0) = \left(\begin{array}{cccc}-1 & 0 & 0 & 0\\
    -\frac{1}{6} & -1 & 0 & 0\\
    0 & 0 & -1 & 0\\
    \frac{1}{4} & 0 & -\frac{1}{10} & -1\end{array}\right)\label{JacM}
\end{align}
and has a single repeated eigenvalue $\lambda=-1$, which has the two-dimensional eigenspace:
\begin{align}
    \lambda_M &= -1, & \boldsymbol{R}_{M1} &= \left(\begin{array}{c}0
    \\
    1\\
    0\\
    0\end{array}\right); & \lambda_M &= -1, & \boldsymbol{R}_{M2} &= \left(\begin{array}{c}0\\
    0\\
    0\\
    1\end{array}\right).\label{epr12M}
\end{align}

Finally, for the rest point $U$, the Jacobian is
\begin{align}
    d\boldsymbol{F}(0,0,0,0) = \left(\begin{array}{cccc}2 & 0 & 0 & 0\\
    -\frac{1}{6} & 1 & 0 & 0\\
    0 & 0 & 4 & 0\\
    0 & 0 & -\frac{1}{10} & 3\end{array}\right)\label{JacU}
\end{align}
and the eigenpairs are given by:
\begin{align}
    \lambda_{U1} &= 1, & \boldsymbol{R}_{U1} &= \left(\begin{array}{c}0\\
    1\\
    0\\
    0\end{array}\right); & \lambda_{U2} &= 2, & \boldsymbol{R}_{U2} &= \left(\begin{array}{c}-6\\
    1\\
    0\\
    0\end{array}\right);\label{epr12U}\\
    \lambda_{U3} &= 3, & \boldsymbol{R}_{U3} &= \left(\begin{array}{c}0\\
    0\\
    0\\
    1\end{array}\right); & \lambda_{U4} &= 4, & \boldsymbol{R}_{U4} &= \left(\begin{array}{c}0\\
    0\\
    -10\\
    1\end{array}\right).\label{epr34U}
\end{align}

For future use, we record the system in $(z_4,w_2)$ one obtains by substituting the $SM$ values for $(z_2,w_0)$ into the $4\times4$ system (\ref{4by4system}), namely, $(z_2,w_0)=(\frac{4}{3},\frac{2}{3})$. The result is the following $2\times2$ non-autonomous system in $(z_4,w_2)$
\begin{align}
    \frac{d}{d\tau}\left(\begin{array}{c}z_4\\
    w_2\end{array}\right) = \left(\begin{array}{c}F_3(\frac{4}{3},\frac{2}{3},z_4,w_2)\\
    F_4(\frac{4}{3},\frac{2}{3},z_4,w_2)\end{array}\right) = \left(\begin{array}{l}\frac{2}{3}z_4-\frac{20}{3}w_2+\frac{40}{81}\\
    -\frac{1}{10}z_4+w_2+\frac{2}{27}\end{array}\right).
\end{align}
 
We now prove that any trajectory whose restriction to the $(z_2,w_0)$-plane enters the stable rest point $M=(0,1)$ in forward time, must enter the stable rest point $M=(0,1,0,0)$ in the four-dimensional phase portrait $\boldsymbol{U}=(z_2,w_0,z_4,w_2)$ as well. Since smooth solutions of (\ref{4by4system}) starting at $\tau=\tau_*\in\mathbb{R}$ with $\tau_*\leq1$ are bounded on the compact interval $[\tau_*,1]$, and being autonomous, solutions are preserved under time translation $\tau\to\tau-\tau_*+1$, then to keep things simple, and without loss of generality, we state the following theorem in terms of solutions defined on the interval $1\leq\tau\leq\infty$.

\begin{Thm}\label{phaselimit} 
    Assume $(z_2(\tau),w_0(\tau),z_4(\tau),w_2(\tau))$ is a solution of the initial value problem for the $4\times4$ system (\ref{4by4system}) with initial data
    \begin{align*}
        \boldsymbol{U}(\tau_*) = \boldsymbol{U}_*.
    \end{align*}
    In addition, assume, without loss of generality, that $\tau_*=1$ and that the first two components (which solve (\ref{2by2system})) tend to the rest point $M$ in positive time, that is, assume
    \begin{align*}
        \lim_{\tau\to\infty}(z_2(\tau),w_0(\tau)) = (0,1).
    \end{align*}
    Then
    \begin{align}
        \lim_{\tau\to\infty}(z_2(\tau),w_0(\tau),z_4(\tau),w_2(\tau)) = (0,1,0,0) = M
    \end{align}
    and there exists a constant $C>0$, depending only on the system (\ref{4by4system}) and initial data $\boldsymbol{U}_*$, such that
    \begin{align}
        \left|\boldsymbol{U}(\tau)-M\right| \leq C\tau e^{-\tau}\label{higherorderlimitUest}
    \end{align}
    for all $\tau\geq1$. Furthermore, there exists a $\Delta_0>0$ and $C>0$, depending only on the system (\ref{4by4system}) and initial data $\boldsymbol{U}_*$, such that the associated $k<0$ Friedmann solution 
    $\boldsymbol{U}^F(\tau)=(z_2^F(\tau),w_0^F(\tau),z_4^F(\tau),w_2^F(\tau))$ 
    with parameter $\Delta_0$ satisfies
    \begin{align}
        \big|(z_2(\tau),w_0(\tau))-(z_2^F(\tau),w_0^F(\tau))\big| \leq Ce^{-\tau}
    \end{align}
    and
    \begin{align}
        \big|\boldsymbol{U}(\tau)-\boldsymbol{U}^F(\tau)\big| \leq C\tau e^{-\tau},
    \end{align}
    for all $\tau\geq1$.
\end{Thm}

\begin{proof}
This is the special $n=2$ case of the the more general theorem, Theorem \ref{ThmBigStability}, given in Section \ref{S11}.
\end{proof}

We now determine the possible backward time asymptotics of solutions whose projection in the $(z_2,w_0)$-plane is the unstable manifold of $SM=(\frac{4}{3},\frac{2}{3})$, which takes $SM$ to $M=(0,1)$.

\begin{Thm}\label{UnstableManifoldOfSM-}
    Assume $\boldsymbol{U}(\tau)=(z_2(\tau),w_0(\tau),z_4(\tau),w_2(\tau))$ is a solution of the $4\times4$ system (\ref{4by4system}) such that
    \begin{align*}
        \lim_{\tau\to-\infty}(z_2(\tau),w_0(\tau)) = \bigg(\frac{4}{3},\frac{2}{3}\bigg) = SM
    \end{align*}
    Then there are two cases:
    \begin{enumerate}
        \item[(i)] $\boldsymbol{U}(\tau)$ is not in the unstable manifold of $SM$.\\
        \item[(ii)] $\boldsymbol{U}(\tau)$ is a trajectory in the unstable manifold of $SM$, in which case
        \begin{align*}
            \lim_{\tau\to-\infty}\boldsymbol{U}(\tau)=SM.
        \end{align*}
    \end{enumerate}
    In Case (i), the solutions $(z_4(\tau),w_2(\tau))$ tend in backward time to the stable manifold of the $2\times2$ linear system
    \begin{align}
        \frac{d}{d\tau}\left(\begin{array}{c}u\\
        v\end{array}\right) = \left(\begin{array}{cc}\frac{2}{3} & -\frac{20}{3}\\
        -\frac{1}{10} & \frac{1}{3}
        \end{array}\right)\left(\begin{array}{c}u\\
        v\end{array}\right),\label{limitsystemU}
    \end{align}
    where $u=z_4-\frac{40}{27}$, $v=w_2-\frac{2}{9}$ and the stable manifold of system (\ref{limitsystemU}) is the line
    \begin{align}
        z_4-\frac{40}{27} = \frac{20}{3}\Big(w_2-\frac{2}{9}\Big).\label{backasymi}
    \end{align}
    Moreover, solutions of (\ref{limitsystemU}) are exact solutions of the fully nonlinear system (\ref{4by4system}), allowing for arbitrary initial conditions $(u_*,v_*)=(u(t_*),v(t_*))$.

    In Case (ii), for each $\boldsymbol{U}(\tau)$ there exists a unique value of $(a,b)$ such that in the limit $\tau\to-\infty$, $\boldsymbol{U}(\tau)$ becomes asymptotic to the linearized solution
    \begin{align}
        \boldsymbol{U}(\tau) \sim ae^{\frac{2}{3}\tau}\boldsymbol{R}_{A1} + be^{\frac{4}{3}\tau}\boldsymbol{R}_{A2},\label{unstablemanSM}
    \end{align}
    where
    \begin{align}
        \boldsymbol{R}_{A1} &= \left(\begin{array}{c}9\\
        -\frac{3}{2}\\
        \frac{10}{3}\\
        1\end{array}\right), & \boldsymbol{R}_{A2} &= \left(\begin{array}{c}0\\
        0\\
        -10\\
        1\end{array}\right).
    \end{align}
    It follows that all trajectories in the unstable manifold of $SM$ satisfying $a\neq0$ enter $SM$ in backward time along $\boldsymbol{R}_{A1}$, the eigendirection of the dominant eigenvalue.
\end{Thm}

\begin{proof}
In Case (i), since
\begin{align*}
    \lim_{\tau\to-\infty}(z_2(\tau),w_0(\tau)) = \bigg(\frac{4}{3},\frac{2}{3}\bigg),
\end{align*}
the solution components $(z_4(\tau),w_2(\tau))$ tend asymptotically to a solution of the $2\times2$ system obtained by substituting the constant values $(z_2,w_0)=(\frac{4}{3},\frac{2}{3})$ into (\ref{4by4system}). The result is the linear homogeneous constant coefficient $2\times2$ system in $(z_4,w_2)$, given by:
\begin{align}
    z_4' &= \frac{2}{3}z_4 - \frac{20}{3}w_2 + \frac{40}{81},\label{limitsystemzbar4}\\
    w_2' &= -\frac{1}{10}z_4 + \frac{1}{3}w_2 + \frac{2}{27},\label{limitsystemwbar2}
\end{align}
which is equivalent to (\ref{limitsystemU}) upon setting $u=z_4-\frac{40}{27}$ and $v=w_2-\frac{2}{9}$. Clearly the $SM$ values $z_4=\frac{40}{27}$ and $w_2=\frac{2}{9}$ give the unique rest point of system (\ref{limitsystemzbar4})--(\ref{limitsystemwbar2}). Moreover, since $(z_2(\tau),w_0(\tau))=(\frac{4}{3},\frac{2}{3})$ solve system (\ref{2by2system}) exactly, it follows that solutions of (\ref{limitsystemzbar4})--(\ref{limitsystemwbar2}) are exact solutions of the full $4\times4$ system (\ref{4by4system}). The eigenpairs of system (\ref{limitsystemU}) are
\begin{align*}
    \lambda_{A2} &= \frac{4}{3}, & \boldsymbol{R}^*_{A2} &= \left(\begin{array}{c}
    -10\\
    1
	\end{array}\right); &
    \lambda_{B2} &= -\frac{1}{3}, & \boldsymbol{R}^*_{B2} &= \left(\begin{array}{c}
    \frac{20}{3}\\
    1
	\end{array}\right);
\end{align*}
consistent with the $4\times4$ eigenpairs (\ref{epr34SM}). From this we conclude that the backward time asymptotics of solutions in Case (i) are given by the stable manifold of system (\ref{limitsystemU}), that is, the line (\ref{backasymi}). This completes the proof of Case (i).

Consider now Case (ii), the case when the $4\times4$ solution lies in the unstable manifold of the rest point $SM=(z_2,w_0,z_4,w_2)=(\frac{4}{3},\frac{2}{3},\frac{40}{27},\frac{2}{9})$. The eigenpairs of $SM$ are:
\begin{align}
    \lambda_{A1} &= \frac{2}{3}, & \boldsymbol{R}_{A1} &= \left(\begin{array}{c}9\\
    -\frac{3}{2}\\
    \frac{10}{3}\\
    1\end{array}\right);\label{epr1pre}\\
    \lambda_{B1} &= -1, & \boldsymbol{R}_{B1} &= \left(\begin{array}{c}4\\
    1\\
    \frac{80}{9}\\
    1\end{array}\right);\label{epr2pre}\\
    \lambda_{A2} &= \frac{4}{3}, & \boldsymbol{R}_{A2} &= \left(\begin{array}{c}0\\
    0\\
    -10\\
    1\end{array}\right);\label{epr3pre}\\
    \lambda_{B2} &= -\frac{1}{3}, & \boldsymbol{R}_{B2} &= \left(\begin{array}{c}0\\
    0\\
    \frac{20}{3}\\
    1\end{array}\right).\label{epr4pre}
\end{align}
Thus the positive eigenvalues $\lambda_{A1}=\frac{2}{3}$ and $\lambda_{A2}=\frac{4}{3}$ give the unstable directions $\boldsymbol{R}_{A1}$ and $\boldsymbol{R}_{A2}$ respectively. Now solutions of a nonlinear autonomous system are characterized by the solutions of the linearized system in a neighborhood of a rest point, and the linearized solutions which span the unstable manifold of $SM$ are
\begin{align*}
    U(\tau) = ae^{\lambda_{A1}\tau}\boldsymbol{R}_{A1} + be^{\lambda_{A2}\tau}\boldsymbol{R}_{A2}.
\end{align*}
This establishes the backward time asymptotics in Case (ii) and thus completes the proof.
\end{proof}

Since $SM$ is an unstable saddle rest point, we conclude that any perturbation of $SM$ which tends to the stable rest point $M$, tends to $M$ asymptotically along the unstable manifold of $SM$. Characterizing the unstable manifold of $SM$ characterizes the perturbations of $SM$, which give rise to underdense Friedmann-like solutions. This is the topic of the next section.

\section{The Unstable Manifold}\label{S9}

Recall that the rest point $SM$ of the autonomous system (\ref{4by4system}) corresponds to the $p=0$, $k=0$ Friedmann spacetime in the sense that it gives the first two terms of its expansion in even powers of $\xi$ in SSCNG. Also recall that $SM$ is represented as a rest point in system (\ref{4by4system}) because the $k=0$ Friedmann solution is self-similar in SSCNG coordinates, that is, the solution depends only on $\xi$. By (\ref{epr1pre})--(\ref{epr4pre}), $SM$ is an unstable saddle rest point with two negative and two positive eigenvalues. The main point, then, is that the $k<0$ Friedmann solutions represent only one parameter in a two-parameter family of solutions which lie in the unstable manifold of rest point $SM$ in the $4\times4$ system (\ref{4by4system}).\footnote{For example, in terms of an expansion of redshift vs luminosity about the center, this extra parameter plays the same role, at third order, as the extra parameter obtained by introducing a cosmological constant \cite{SmolTeVo}. This issue will be addressed in detail in a forthcoming paper.} By the Hartman--Grobman theorem, the unstable manifold of $SM$ in the nonlinear system (\ref{4by4system}) is parameterized by the asymptotic limits given by the unstable manifold of the system obtained by linearizing (\ref{4by4system}) about $SM$, that is, by the linearized system
\begin{align}
    \frac{d}{d\tau}(\boldsymbol{U}-\boldsymbol{U}_F) = d\boldsymbol{F}\bigg(\frac{4}{3},\frac{2}{3},\frac{40}{27},\frac{2}{9}\bigg)(\boldsymbol{U}-\boldsymbol{U}_F),\label{LinSysSM}
\end{align}
where
\begin{align}
    d\boldsymbol{F}\bigg(\frac{4}{3},\frac{2}{3},\frac{40}{27},\frac{2}{9}\bigg) = \left(\begin{array}{cccc}0 & -4 & 0 & 0\\
    -\frac{1}{6} & -\frac{1}{3} & 0 & 0\\
    -\frac{10}{27} & -\frac{20}{3} & \frac{2}{3} & -\frac{20}{3}\\
    0 & -\frac{4}{9} & -\frac{1}{10}&\frac{1}{3}\end{array}\right).
\end{align}
The unstable manifold of $SM$ in the nonlinear system (\ref{4by4system}) is therefore a family of solutions of (\ref{4by4system}) parameterized by the two parameters $(a,b)$ of the unstable manifold (\ref{unstablemanSM}) of (\ref{LinSysSM}), namely,
\begin{align}
    \boldsymbol{U}(\tau) = ae^{\lambda_{A1}\tau}\boldsymbol{R}_{A1} + be^{\lambda_{A2}\tau}\boldsymbol{R}_{A2}.\label{asymptoticlimits}
\end{align}
Assume now that $a\neq0$, which is equivalent to assuming that the trajectory in the $4\times4$ system (\ref{4by4system}) projects to the unstable manifold in the $2\times2$ system (\ref{2by2system}), and let $\Sigma_{SM}$ denote the (essential) subset of the unstable manifold of $SM$ in the nonlinear system (\ref{4by4system}) with asymptotic limits at $SM$ given by (\ref{asymptoticlimits}) when $a\neq0$. Since $a\neq0$, we can make the translation $\tau\to\tau-\tau_0$ and $b\to\beta$ to scale $a$ to $\pm1$. The two parameter family $\Sigma_{SM}$ is then parameterized by linearized solutions of the form
\begin{align*}
    \boldsymbol{U}(\tau) = \pm e^{\lambda_{A1}(\tau-\tau_0)}\boldsymbol{R}_{A1} + \beta e^{\lambda_{A2}(\tau-\tau_0)}\boldsymbol{R}_{A2}, 
\end{align*}
depending on $(\tau_0,\beta)$. Here $\pm$ determines the side of the unstable manifold of $SM$ in the $(z_2,w_0)$-plane, $\tau_0$ represents the time translation freedom of the autonomous system (\ref{4by4system}) and $\beta$ names the non-intersecting trajectories of the solutions in $\Sigma_{SM}$. In this section we prove that unique values $\beta=\beta_F^\pm$ determine unique trajectories on each side of the unstable manifold of $SM$ which correspond to the $k>0$ and $k<0$ Friedmann spacetimes respectively, that is, they correspond to the evolution of the first two terms in the expansion of the $k\neq0$ Friedmann solutions in powers of $\xi$ in SSCNG coordinates.\footnote{We prove in Section \ref{S11.5} that $\beta_F^-=0$. It is likely also the case that $\beta_F^+=0$.} More precisely, the $k<0$ Friedmann family of solutions is the unique trajectory in $\Sigma_{SM}$ corresponding to the linearized trajectory
\begin{align*}
    \boldsymbol{U}_-(\tau) = -e^{\lambda_{A1}(\tau-\tau_0)}\boldsymbol{R}_{A1} + \beta_F^-e^{\lambda_{A2}(\tau-\tau_0)}\boldsymbol{R}_{A2},
\end{align*}
and as a solution of the nonlinear system (\ref{4by4system}), the $k<0$ trajectory emanates from $SM$ on the smaller $z_2$ (underdense) side of $SM$ and tends to $M$ as $\tau\to\infty$; the $k>0$ Friedmann family of solutions corresponds to the unique trajectory in $\Sigma_{SM}$ corresponding to the linearized trajectory
\begin{align*}
    \boldsymbol{U}_+(\tau-\tau_0) = e^{\lambda_{A1}(\tau-\tau_0)}\boldsymbol{R}_{A1} + \beta_F^+e^{\lambda_{A2}(\tau-\tau_0)}\boldsymbol{R}_{A2},
\end{align*}
and as a solution of the nonlinear system (\ref{4by4system}), the $k>0$ trajectory emanates from $SM$ on the larger $z_2$ (overdense) side of $SM$ and intersects $w_0=0$ at some finite positive time. Moreover, the one degree of freedom associated with time translation $\tau\to\tau-\tau_0$ accounts for the gauge freedom $\Delta_0$ in $k\neq0$ Friedmann spacetimes through the relation $\tau_0=\Delta_0+\tau_F$ for some normalizing constant $\tau_F$, see (\ref{DeltaDef}). The fact that the $4\times4$ unstable manifold of $SM$ is a two parameter family of solutions containing the one parameter family of Friedmann solutions determined by $k$ implies that the unstable manifold of $SM$ allows for one extra degree of freedom than that accounted for by the Friedmann family, and by this, perturbations of $SM$ produce Friedmann-like solutions which admit one more degree of freedom than Friedmann in the redshift vs luminosity relation. Our results are summarized in the following theorem and the remainder of the section is devoted to the proof of this theorem.
 
\begin{Thm}\label{UnstableManifoldOfSM}
    Let $\Sigma_{SM}$ denote the subset of the unstable manifold of the rest point $SM$ in the $4\times4$ system (\ref{4by4system}) which consists of trajectories which project to the unstable manifold of $SM$ in the $2\times2$ nonlinear system (\ref{2by2system}). Then $\Sigma_{SM}$ is characterized by its backward time asymptotics given by the two parameter family of solutions of the linearized system (\ref{LinSysSM}),
    \begin{align}
        U(\tau) = \pm e^{\lambda_{A1}(\tau-\tau_0)}\boldsymbol{R}_{A1} + \beta e^{\lambda_{A2}(\tau-\tau_0)}\boldsymbol{R}_{A2},
    \end{align}
    parameterized by $(\tau_0,\beta)\in\mathbb{R}^2$, where:
    \begin{align}
        \lambda_{A1} &= \frac{2}{3}, & \boldsymbol{R}_{A1} &= \left(\begin{array}{c}9\\
        -\frac{3}{2}\\
        \frac{10}{3}\\
        1\end{array}\right); & \lambda_{A2} &= \frac{4}{3}, & \boldsymbol{R}_{A2} &= \left(\begin{array}{c}0\\
        0\\
        -10\\
        1\end{array}\right).
    \end{align}
    Moreover, the two sides of the unstable manifold of $SM$ in the $(z_2,w_0)$-plane correspond to the decomposition $\Sigma_{SM}=\Sigma_{SM}^-\cup\Sigma_{SM}^+$ where $\Sigma_{SM}^-$ corresponds to the linearized trajectories
    \begin{align*}
        U(\tau) = -e^{\lambda_{A1}(\tau-\tau_0)}\boldsymbol{R}_{A1} + \beta e^{\lambda_{A2}(\tau-\tau_0)}\boldsymbol{R}_{A2}
    \end{align*}
    and $\Sigma_{SM}^+$ corresponds to the linearized trajectories
    \begin{align*}
        U(\tau) = e^{\lambda_{A1}(\tau-\tau_0)}\boldsymbol{R}_{A1} + \beta e^{\lambda_{A2}(\tau-\tau_0)}\boldsymbol{R}_{A2}.
    \end{align*}
    Furthermore, trajectories in $\Sigma_{SM}^-$ leave the rest point $SM$ in the (underdense) direction of negative $z_2$, trajectories in $\Sigma_{SM}^+$ leave the rest point $SM$ in the (overdense) direction of positive $z_2$ and all trajectories in $\Sigma_{SM}^-$ tend to rest point $M$ as $\tau\to\infty$.
\end{Thm}

That the structure of the nonlinear manifold $\Sigma_{SM}=\Sigma_{SM}^-\cup\Sigma_{SM}^+$ is characterized by the linearized system (\ref{LinSysSM}) follows directly from the Hartman--Grobman theorem in light of the hyperbolic nature of rest point $SM$ in (\ref{epr1pre})--(\ref{epr4pre}). Moreover, the asymptotic limit $M$ for solutions in $\Sigma_{SM}^-$ follows directly from Theorem \ref{phaselimit} in light of the fact that when $a\neq0$, the projection of solutions in $\Sigma_{SM}^-$ lie on the unstable manifold of $SM$ in the $2\times2$ system (\ref{2by2system}), and this is the starting assumption of Theorem \ref{phaselimit}. The proof of Theorem \ref{UnstableManifoldOfSM} is thus a direct consequence of the following theorem which gives exact formulas for the expansion of the $p=0$, $k\neq0$ Friedmann solutions in even powers of $\xi$ in SSCNG coordinates.

\begin{Thm}\label{Friedmannknotzeroexpansion}
    The Taylor expansion of the $p=0$, $k=\pm1$ Friedmann solution in even powers of $\xi=\frac{r}{t}$, with coefficient functions of $t$, in SSCNG coordinates $(t,\xi)$ takes the form:
    \begin{align}
        z(t,\xi) &= z_2(t)\xi^2 + z_4(t)\xi^4 + O(\xi^6),\label{kneq0z}\\
        w(t,\xi) &= w_0(t) + w_2(t)\xi^2 + O(\xi^4),\label{kneq0w}\\
        A(t,\xi) &= 1 + A_2(t)\xi^2 + A_4(t)\xi^4 + O(\xi^6).\label{kneq0A}
    \end{align}
    In the case $k=-1$, $\Delta_0>0$, $(z_i(t),w_i(t))$ are given in terms of $\theta=\theta(t)$ by the formulas:
    \begin{align}
        z_2 = \tilde{z}_2(\theta) &= \frac{6(\sinh2\theta-2\theta)^2}{(\cosh2\theta-1)^3} = -3A_2,\label{z2FinalAgain}\\
        w_0 = \tilde{w}_0(\theta) &= \frac{(\sinh2\theta-2\theta)\sinh2\theta}{(\cosh2\theta-1)^2},\label{w0FinalAgain}\\
        z_4 = \tilde{z}_4(\theta) &= \frac{30(\sinh2\theta-2\theta)^4\cosh^2\theta}{(\cosh2\theta-1)^6} = -5A_4,\label{z4FinalAgain}\\
        w_2 = \tilde{w}_2(\theta) &= w_2^B-w_2^A\label{w2FinalAgain}\\
        &= \frac{(\sinh2\theta-2\theta)^3\cosh\theta}{2(\cosh2\theta-1)^5\sinh\theta}\Big(\sinh^22\theta+(1-2\sinh^2\theta)(\cosh2\theta-1)\Big),\notag
    \end{align}
    where $\theta\geq0$ is defined as a function of $t\geq0$ through the relation
    \begin{align}
        \frac{t}{\Delta_0} = \frac{1}{2}(\sinh2\theta-2\theta),\label{Theta2}
    \end{align}
    Equation (\ref{Theta2}) inverts to define the inverse function
    \begin{align*}
        \Theta:(0,\infty) &\to (0,\infty), & \theta(t) &= \Theta\bigg(\frac{t}{\Delta_0}\bigg),
    \end{align*}
    and in terms of $\Theta$ defined by (\ref{Theta2}), the expansion of $p=0$, $k=-1$ Friedmann for general $\Delta_0>0$ is given by:
    \begin{align}
        z_i(t) &= \tilde{z}_i\bigg(\Theta\bigg(\frac{t}{\Delta_0}\bigg)\bigg),\label{twithDelta1}\\
        w_j(t) &= \tilde{w}_j\bigg(\Theta\bigg(\frac{t}{\Delta_0}\bigg)\bigg),\label{twithDelta2}\\
        A_i(t) &= \tilde{A}_i\bigg(\Theta\bigg(\frac{t}{\Delta_0}\bigg)\bigg).\label{twithDelta3}
    \end{align}
    In the case $k=+1$, $\Delta_0>0$, $(z_i(t),w_i(t))$ are given in terms of $\theta=\theta(t)$ by the formulas:
    \begin{align}
        z_2 = \tilde{z}_2(\theta) &= \frac{6(2\theta-\sin2\theta)^2}{(1-\cos2\theta)^3} = -3A_2,\label{z2FinalAgaink1}\\
        w_0 = \tilde{w}_0(\theta) &= \frac{(2\theta-\sin2\theta)\sin2\theta}{(1-\cos2\theta)^2},\label{w0FinalAgaink1}\\
        z_4 = \tilde{z}_4(\theta) &= \frac{30(2\theta-\sin2\theta)^4\cos^2\theta}{(1-\cos2\theta)^6} = -5A_4,\label{z4FinalAgaink1}\\
        w_2 = \tilde{w}_2(\theta) &= \frac{(2\theta-\sin2\theta)^3\cos\theta}{2(1-\cos2\theta)^5\sin\theta}\Big(\sin^22\theta+(1-2\sin^2\theta)(1-\cos2\theta)\Big),\label{w2FinalAgaink1}
    \end{align}
    where $\theta\geq0$ is defined as a function of $t\geq0$ through the relation
    \begin{align}
        \frac{t}{\Delta_0} = \frac{1}{2}(2\theta-\sin2\theta).\label{Theta2k1}
    \end{align}
    Equation (\ref{Theta2k1}) inverts to define the inverse function,
    \begin{align*}
        \Theta:(0,\infty) &\to \Big(0,\frac{\pi}{2}\Big), & \theta &= \Theta\bigg(\frac{t}{\Delta_0}\bigg),
    \end{align*}
    and in terms $\Theta$ defined by (\ref{Theta2k1}), the expansion of $p=0$, $k=+1$ Friedmann for general $\Delta_0>0$ is again given by (\ref{twithDelta1})--(\ref{twithDelta3}). 
\end{Thm}

Note that although the formulas for $k=\pm1$ Friedmann spacetimes were given in terms of $\theta$ as a function of Friedmann time, in formulas (\ref{kneq0z})--(\ref{Theta2k1}) $t$ is now SSCNG time, not Friedmann time, as explained in the proof below. Note also that since $\Delta_0$ is the only free parameter in the Friedmann spacetimes, it follows from (\ref{twithDelta1})--(\ref{twithDelta3}) that the solutions of (\ref{4by4system}) determined by $p=0$, $k\neq0$ Friedmann solutions of the Einstein field equations are given, for general values of the free parameter $\Delta_0>0$, by
\begin{align}
    \boldsymbol{U}_-(\tau-\Delta_0) = \tilde{\boldsymbol{U}}\big(\Theta\circ\exp(\tau-\Delta_0)\big),\label{oneorbit}
\end{align}
where
\begin{align*}
    \tilde{\boldsymbol{U}}_\pm(\theta) = \big(\tilde{z}_2(\theta),\tilde{w}_0(\theta),\tilde{z}_4(\theta),\tilde{w}_2(\theta)\big),
\end{align*}
with the right hand side being defined separately by (\ref{z2FinalAgain})--(\ref{w2FinalAgain}) and (\ref{z2FinalAgain})--(\ref{w2FinalAgain}) in the cases $k=-1$ and $k=+1$ respectively. Equation (\ref{oneorbit}) shows that the $k\pm1$ Friedmann solutions for $\Delta_0\neq1$ are time-translations of the $k=\pm1$ solutions corresponding to $\Delta_0=\frac{4}{9}$, and hence the projection of $k\neq0$ Friedmann solutions onto solutions of (\ref{4by4system}) all lie on the same trajectory with $\Delta_0$ giving the time translation.

The cases $k=-1$ and $k=+1$ of Theorem \ref{Friedmannknotzeroexpansion} are proven separately in subsections \ref{Appendix1F} and \ref{Appendix1G} respectively. For the proof of Theorem \ref{UnstableManifoldOfSM}, we require the following corollary.

\begin{Corollary}\label{CorLimits}
    In the case $k=\pm1$, formulas (\ref{z2FinalAgain})--(\ref{w2FinalAgain}) and (\ref{z2FinalAgain})--(\ref{w2FinalAgain}) both imply the limits
    \begin{align}
        \lim_{\theta\to0}\tilde{\boldsymbol{U}}_\pm(\theta) = \boldsymbol{U}_F = \bigg(\frac{4}{3},\frac{2}{3},\frac{40}{27},\frac{2}{9}\bigg).\label{limitsUhat}
    \end{align}
    For the opposite limits, in the case $k=-1$, (\ref{z2FinalAgain})--(\ref{w2FinalAgain}) imply
    \begin{align}
        \lim_{\theta\to\infty}\tilde{\boldsymbol{U}}_-(\theta) = (0,1,0,0),\label{limitsUhatminus}
    \end{align}
    and in the case $k=+1$, (\ref{z2FinalAgain})--(\ref{w2FinalAgain}) imply
    \begin{align}
        \lim_{\theta\to\frac{\pi}{2}}\tilde{\boldsymbol{U}}_+(\theta) = \bigg(\frac{3}{4}\pi^2,0,0,0\bigg).\label{limitsUhatplus}
    \end{align}
\end{Corollary}

\begin{proof}
The limits (\ref{limitsUhat})--(\ref{limitsUhatplus}) can be confirmed numerically, but for completeness, we give a proof in the case $k=-1$ based on elementary asymptotics. The argument for $k=+1$ is of secondary interest to this paper and incorporated into the proof of Theorem \ref{Friedmannknotzeroexpansion} below. We first confirm the limits for $\theta\to0$, which by (\ref{Theta2}) is equivalent to $\bar{t}\to0$. For this it suffices to use the leading order expressions:
\begin{align*}
    \cosh\theta &= 1 + \frac{1}{2}\theta^2 + O(\theta^4), & \sinh\theta &= \theta + \frac{1}{6}\theta^3 + O(\theta^5),
\end{align*}
as $\theta\to0$. This then gives:
\begin{align*}
    \cosh2\theta - 1 &= 2\theta^2 + O(\theta^4), & \sinh2\theta - 2\theta &= \frac{4}{3}\theta^3 + O(\theta^5),
\end{align*}
as $\theta\to0$. Putting the leading order terms for $\theta\to0$ into (\ref{z2FinalAgain}) and using $A_2=-3z_2$ gives
\begin{align*}
    A_2 = -\frac{1}{3}z_2 = -2\frac{\big(\frac{4}{3}\theta^3\big)^2}{(2\theta^2)^3} + H.O.T. \to -\frac{4}{9},
\end{align*}
verifying the first limit in (\ref{limitsUhat}). Similarly, putting the leading order terms above into (\ref{z4FinalAgain}) for $\theta\to0$ gives
\begin{align*}
    A_4 = -\frac{1}{5}z_4 = -6\frac{\big(\frac{4}{3}\theta^3\big)^4}{(2\theta^2)^6} + H.O.T. \to -\frac{8}{27},
\end{align*}
verifying the third limit in (\ref{limitsUhat}). Again, putting the leading order terms above into (\ref{w0FinalAgain}) for $\theta\to0$ gives
\begin{align*}
    w_0 = 2\theta\frac{\frac{4}{3}\theta^3}{(2\theta^2)^2} + H.O.T. \to \frac{2}{3},
\end{align*}
verifying the second limit in (\ref{limitsUhat}). Finally, putting the leading order terms above into (\ref{w2FinalAgain}) for $\theta\to0$ gives
\begin{align*}
    w_2 = 3\theta\frac{\big(\frac{4}{3}\theta^3\big)^3}{(2\theta^2)^5} + H.O.T. \to \frac{2}{9},
\end{align*}
verifying the last limit in (\ref{limitsUhat}).

We now confirm the limits for $\theta\to\infty$, which by (\ref{Theta2}) is equivalent to $t\to\infty$. For this it suffices to use the leading order expressions:
\begin{align*}
    \cosh\theta &= \frac{1}{2}e^\theta + O(e^{-\theta}), & \sinh\theta &= \frac{1}{2}e^\theta + O(e^{-\theta}),
\end{align*}
as $\theta\to\infty$. This then gives the leading orders of $(\cosh2\theta-1)$ and $(\sinh2\theta-2\theta)$ as $\frac{1}{2}e^{2\theta}$ for $\theta\to\infty$. Putting these leading order expressions into (\ref{z2FinalAgain})--(\ref{w2FinalAgain}) easily confirms the limits on the right hand side of (\ref{limitsUhatminus}).
\end{proof}

\begin{proof}[Proof of Theorem \ref{UnstableManifoldOfSM}]
We begin by recalling Theorem \ref{Friedmannknotzeroexpansion}. The fact that Friedmann spacetimes are spherically symmetric solutions of the Einstein field equations which are smooth at the center for every $k$ implies that they all solve the self-similar equations (\ref{Axi-final0})--(\ref{weqnxi-final0}) in SSCNG coordinates. This justifies the expansions (\ref{kneq0z})--(\ref{kneq0A}) in even powers of $\xi$, with coefficient functions of $t$, and implies that the first and second terms in their expansion, given by (\ref{z2FinalAgain})--(\ref{w2FinalAgain}) in the case $k=-1$, $\Delta_0=\frac{4}{9}$, and by (\ref{z2FinalAgaink1})--(\ref{w2FinalAgaink1}) in the case $k=+1$, $\Delta_0=\frac{4}{9}$ respectively, produce unique exact solutions $(z_2,w_0,z_4,w_2)$ of the $4\times4$ system (\ref{4by4system}). Equation (\ref{oneorbit}) shows that the $k\neq0$ Friedmann solution for $\Delta_0\neq1$ is a time-translation of the $k\neq0$ Friedman solution corresponding to $\Delta_0=\frac{4}{9}$, and hence all of the $k<0$ Friedmann solutions lie on a single trajectory, with the same for $k>0$ Friedmann solutions but on a different trajectory. Therefore, since $\lim_{t\to0}\equiv\lim_{\theta\to0}\equiv\lim_{\tau\to-\infty}$, to prove that these Friedmann solutions of (\ref{4by4system}) lie in the unstable manifold of $SM$, it suffices to prove that $\lim_{\theta\to0}\tilde{\boldsymbol{U}}_\pm(\theta)=\boldsymbol{U}_F$ for $\tilde{\boldsymbol{U}}_\pm=(\tilde{z}_2,\tilde{w}_0,\tilde{z}_4,\tilde{w}_2)$ defined in (\ref{z2FinalAgain})--(\ref{w2FinalAgain}) and (\ref{z2FinalAgaink1})--(\ref{w2FinalAgaink1}) respectively. Given that this follows from Corollary \ref{CorLimits}, the proof of Theorem \ref{UnstableManifoldOfSM} is complete.
\end{proof}

\section{The Higher Order STV-ODE}\label{S11}

Consider now the case of the STV-ODE at orders $n\geq3$. In this light, assume a smooth solution $(z,w)$ of the STV-PDE (\ref{Axi-final0})--(\ref{weqnxi-final0}) is expanded asymptotically in even powers of $\xi$ as so:
\begin{align*}
    z(t,\xi) &= \sum_{n=1}^\infty z_{2n}(t)\xi^{2n}, & w(t,\xi) &= \sum_{n=0}^{\infty}w_{2n}\xi^{2n}.
\end{align*}
We introduce the notation
\begin{align*}
    \boldsymbol{v}_n = (z_{2n},w_{2n-2})
\end{align*}
for $n\geq1$, so in terms of our earlier notation,
\begin{align*}
    \boldsymbol{v}_1 = \boldsymbol{u} := (z_2,w_0).
\end{align*}
In the next theorem we establish that the system of ODE, obtained by substituting the above expansions into equations (\ref{Axi-final0})--(\ref{weqnxi-final0}) and collecting like powers of $\xi$, closes in $\boldsymbol{v}_1,\dots,\boldsymbol{v}_n$ at each order $n\geq1$. We then prove by induction that if
\begin{align*}
    \boldsymbol{U}_k := (\boldsymbol{v}_1,\dots,\boldsymbol{v}_k) \to M
\end{align*}
for $k\leq n-1$ in the $(n-1)\times(n-1)$ closed system of ODE, then also $\boldsymbol{U}_n\to M$ in the $n\times n$ system as well, where $M$ is the stable rest point in each closed $n\times n$ system given by
\begin{align*}
    M = (0,1,0,\dots,0).
\end{align*}
We note that at $M$ we have $\boldsymbol{v}_1=(0,1)$ and $\boldsymbol{v}_k=(0,0)$ for all $k\geq2$.

\begin{Thm}\label{Thmgeneralexpansion}
    Assume a smooth solution $(z,w)$ of system (\ref{Axi-final0})--(\ref{weqnxi-final0}) is expanded asymptotically in even powers of $\xi$ as so:
    \begin{align}
        z(t,\xi) &= \sum_{n=0}^\infty z_{2n}(t)\xi^{2n},\label{ansatzinfinityz}\\
        w(t,\xi) &= \sum_{n=0}^{\infty}w_{2n}(t)\xi^{2n},\label{ansatzinfinityw}\\
        A(t,\xi) &= \sum_{n=0}^{\infty}A_{2n}(t)\xi^{2n},\label{ansatzinfinityA}\\
        D(t,\xi) &= \sum_{n=0}^{\infty}D_{2n}(t)\xi^{2n},\label{ansatzinfinityD}
    \end{align}
    where
    \begin{align*}
        z_0 &= 0, & A_0 &= 1, & D_0 &= 1.
    \end{align*}
    Then substituting (\ref{ansatzinfinityz})--(\ref{ansatzinfinityD}) into (\ref{zeqnxi-final0})--(\ref{weqnxi-final0}) and collecting even powers of $\xi$ leads to the following system of equations: 
    \begin{align}
        t\dot{z}_{2n} &= (2n)z_{2n} - (2n+1)\sum_{i+j+k=n}D_{2i}z_{2j}w_{2k},\label{zneqn}\\
        t\dot{w}_{2n} &= (2n+1)w_{2n} + \frac{1}{2}\sum_{i+j+k=n}\hat{w}_{2i}a_{2j}D_{2k} - \sum_{i+j+k=n}(2i+1)w_{2i}w_{2j}D_{2k},\label{wneqn}
    \end{align}
    where $\hat{w}_{2k}$ and $a_{2k}$ are defined by
    \begin{align}
        \hat{w}_{2n} = \begin{cases}
            1, & n=0,\\
            -\sum_{i+j=n-1}w_{2i}w_{2j}, & n\geq1,
        \end{cases}\label{defwhat}
    \end{align}
    and
    \begin{align}
        \frac{A-1}{A} = \sum_{n=0}^{\infty}a_{2n}\xi^{2n}.\label{defan}
    \end{align}
    respectively. Note that (\ref{zneqn}) holds for $n\geq2$, (\ref{wneqn}) holds for $n\geq1$ and the case for $(z_2,w_0)$ is given in system (\ref{2by2system}). Moreover, substituting (\ref{ansatzinfinityz}) and (\ref{ansatzinfinityA}) into (\ref{Axi-final0}) and collecting like powers of $\xi$ determines $A_{2n}$ in terms of $z_{2n}$ according to
    \begin{align}
        A_{2n} &= -\frac{1}{2n+1}z_{2n},\label{Antozn}
    \end{align}
    for $n\geq1$. Furthermore, substituting (\ref{ansatzinfinityz})--(\ref{ansatzinfinityD}) into (\ref{Dxi-final0}) and collecting like powers of $\xi$ determines $D_{2n}$, as a function of $z_2,\dots,z_{2n}$ and $w_0,\dots,w_{2n-2}$ as so
    \begin{align}
        4nD_{2n} &= \sum_{i+j+k+l=n-1}z_{2i}w_{2j}w_{2k}D_{2l} - \sum_{\substack{i+j=n\\i\neq0}}(z_{2i}+2A_{2i})D_{2j} - \sum_{\substack{i+j=n\\j\neq n}}4jA_{2i}D_{2j}\notag\\
        &= -(z_{2n}+2A_{2n}) + \sum_{i+j+k+l=n-1}z_{2i}w_{2j}w_{2k}D_{2l} - \sum_{\substack{i+j= n\\i\neq0,n}}(z_{2i}+2A_{2i})D_{2j} - \sum_{\substack{i+j=n\\j\neq n}}4jA_{2i}D_{2j}\notag\\
        &= -\frac{2n-1}{2n+1}z_{2n} + \gamma_n(\boldsymbol{v}_1,\dots,\boldsymbol{v}_{n-1}),\label{Dntoznwn}
    \end{align}
    for $n\geq1$, where $\gamma_n$ is a smooth function of $(\boldsymbol{v}_1,\dots,\boldsymbol{v}_{n-1})$. Finally, substituting expressions (\ref{Antozn}) and (\ref{Dntoznwn}) into equations (\ref{zneqn})--(\ref{wneqn}) results in a system of $n$ $2\times2$ ODE in $\boldsymbol{v}_1,\dots,\boldsymbol{v}_n$ which closes for every $n\geq1$. For each $k=1,\dots,n$, the equation for $\boldsymbol{v}_k$ takes the form
    \begin{align}
        t\dot{\boldsymbol{v}}_k = P_k\boldsymbol{v}_k + \boldsymbol{q}_k\label{ktheqninbfv}
    \end{align}
    where $P_k$ is the $2\times2$ matrix
    \begin{align}
        P_k = P_k(\boldsymbol{v}_1) = \left(\begin{array}{cc}(2k+1)(1-w_0)-1 & -(2k+1)z_2\\
        -\frac{1}{2(2k+1)} & 2k(1-w_0)-1
        \end{array}\right)\label{Pk}
    \end{align}
    and:
    \begin{align*}
        \boldsymbol{q}_1 &= \boldsymbol{0},\\
        \boldsymbol{q}_k &= \boldsymbol{q}_k(\boldsymbol{v}_1,\dots,\boldsymbol{v}_{k-1}).
    \end{align*}
\end{Thm}

The proof of Theorem \ref{Thmgeneralexpansion} is given in Section \ref{Appendix1H}.

Since system (\ref{ktheqninbfv}) closes in $\boldsymbol{v}_1,\dots,\boldsymbol{v}_n$ for every $n\geq1$, it follows for each $n\geq1$ that equations (\ref{zneqn})--(\ref{wneqn}) together with (\ref{Antozn} and (\ref{Dntoznwn}) determine a closed system of $2n\times2n$ ODE in $z_2,\dots,z_{2n}$ and $w_0,\dots,w_{2n-2}$ as a function of $\tau=\ln t$. We write this as an $n\times n$ system in the compact form
\begin{align}
    \frac{d}{d\tau}\left(\begin{array}{c}\boldsymbol{v}_1\\
    \vdots\\
    \boldsymbol{v}_n\end{array}\right) = \left(\begin{array}{c}P_1\boldsymbol{v}_1 + \boldsymbol{q}_1\\
    \vdots\\
    P_n\boldsymbol{v}_n + \boldsymbol{q}_n\end{array}\right)\label{nbynsystem}
\end{align}
in unknowns $\boldsymbol{v}_1,\dots,\boldsymbol{v}_n$, where for each $1\leq k\leq n$, $P_k$ is a function of only $\boldsymbol{v}_1=(z_2,w_0)$ and $\boldsymbol{q}_k$ is a function of $\boldsymbol{v}_1,\dots\boldsymbol{v}_{k-1}$. We refer to system (\ref{nbynsystem}) as the $n^{th}$-order STV-ODE. Note that at order $n=1$ we recover the $2\times2$ system (\ref{2by2system}) in $(z_2,w_0)$ and at order $n=2$ we recover the $4\times4$ system (\ref{4by4system}) in $(z_2,w_0,z_4,w_2)$. As a consequence of the following generalization of Theorem \ref{2by2phaseportrait}, we establish the instability of the $p=0$, $k=0$ Friedmann spacetime, the stability of the rest point $M$ and the stability of the $p=0$, $k<0$ Friedmann spacetimes to perturbations of all orders. Now smooth solutions of (\ref{nbynsystem}) starting at $\tau_*\in\mathbb{R}$ are bounded on the compact interval $[\tau_*,\tau]$ for any $\tau_*\leq\tau<\infty$, and since (\ref{nbynsystem}) is autonomous, solutions are preserved under the time translation $\tau\to \tau-\tau_*+1$. To keep things simple, and without loss of generality, we state the following theorem in terms of solutions defined for $\tau\geq\tau_*$, assuming initial time $\tau_*\geq1$.
 
\begin{Thm}\label{ThmBigStability} 
    Let $\boldsymbol{U}(\tau)=(\boldsymbol{v}_1(\tau),\dots,\boldsymbol{v}_n(\tau))$ be a smooth solution of the initial value problem for the $2n\times2n$ system of ODE (\ref{nbynsystem}) starting from initial data
    \begin{align}
        \boldsymbol{U}(\tau_*) = \boldsymbol{U}_*\in\mathbb{R}^{2n}.\label{initialdataU*}
    \end{align}
    Since system (\ref{nbynsystem}) is autonomous, assume for convenience, and without loss of generality, that $\tau_*\geq1$. Assume further that $\boldsymbol{v}_1(\tau)=(z_2(t),w_0(t))$ is defined for all $\tau\geq\tau_*$ and 
    \begin{align}
        \lim_{\tau\to\infty}\boldsymbol{v}_1(\tau) = (0,1).\label{lowestorderlimit}
    \end{align}
    Then $\boldsymbol{U}(\tau)$ is defined for all $\tau\geq\tau_*$ and 
    \begin{align}
        \lim_{\tau\to\infty}\boldsymbol{v}_k = (0,0)\label{higherorderlimit}
    \end{align}
    for $k=2,\dots,n$, that is,
    \begin{align}
        \lim_{\tau\to\infty}\boldsymbol{U}(\tau) = (0,1,0,\dots,0) = M.\label{higherorderlimitU}
    \end{align}
    Moreover, there exists a constant $C>0$, depending only on the system (\ref{nbynsystem}) and initial data $(\boldsymbol{U}_*,\tau_*)$, such that
    \begin{align}
        |\boldsymbol{U}(\tau)-M| \leq C\tau^ne^{-\tau}
    \end{align}
    for all $\tau\geq\tau_*$. Furthermore, there exists a $\Delta_0>0$ and constant $C>0$ such that the associated $k<0$ Friedmann solution
    \begin{align*}
        \boldsymbol{U}^F(\tau) = (\boldsymbol{v}_1^F(\tau),\dots,\boldsymbol{v}_n^F(\tau))
    \end{align*}
    satisfies
    \begin{align}
        |\boldsymbol{v}_1(\tau)-\boldsymbol{v}_1^F(\tau)| \leq Ce^{-\tau}\label{higherorderlimitUF1}
    \end{align}
    and 
    \begin{align}
        |\boldsymbol{v}_k(\tau)-\boldsymbol{v}_k^F(\tau)| \leq C\tau^ke^{-\tau}\label{higherorderlimitUF2}
    \end{align}
    for $k=2,\dots,n$ and $\tau\geq\tau_*$.
\end{Thm}

We have the following rather remarkable corollary, which states that any solution of the $2\times2$ system (\ref{2by2system}) within the domain of attraction of $M=(0,1)$ can be extended arbitrarily to a global solution of (\ref{nbynsystem}) at every order $n$, with arbitrary higher order initial data. Moreover, all components higher order than $(z_2,w_0)$ decay to zero at the rate $\frac{\ln t}{t}$. This is a restatement of Theorem \ref{DecayToM} in the introduction.

\begin{Corollary}\label{globalexistencewithdecay}
    Assume $(z_2(t),w_0(t))$ is a solution of the $2\times2$ system (\ref{2by2system}) with initial data
    \begin{align}
        (z_2(t_*),w_0(t_*)) &= (z_2^*,w_0^*), & z_2^* &\geq 0, & t_* &> 0,\label{2by2initialdata}
    \end{align}
    such that
    \begin{align*}
        \lim_{t\to\infty}(z_2(t),w_0(t)) = (0,1) = M.
    \end{align*}
    Then the solution of the $2n\times2n$ system with initial data (\ref{2by2initialdata}) augmented with the arbitrary higher order initial data
    \begin{align}
        (z_{2k}(t_*),w_{2k-2}(t_*)) &= (z_{2k}^*,w_{2k-2}^*)\in\mathbb{R}^2, & k &\geq 2,
    \end{align}
    exists for all time. Moreover, there exists a constant $C>0$, depending only on the initial data and the equations, such that all higher order components satisfy
    \begin{align}
        \big|(z_{2k}(t),w_{2k-2}(t))-(z_{2k}^*,w_{2k-2}^*)\big| &\leq (C+1)\frac{\ln t}{t}, & k &\geq 2.\label{kbykdecay}
    \end{align}
\end{Corollary}

\begin{proof}
This is Theorem \ref{ThmBigStability} stated in terms of $t$ instead of $\tau=\ln t$, except that the presence of $(C+1)$ in (\ref{kbykdecay}) instead of $C$ is used to cover the bounded growth of a solution during the compact time interval $t\in(t_*,1)$.
\end{proof}

\begin{proof}[Proof of Theorem \ref{ThmBigStability}]
We prove this by induction. To begin, we know this holds in the $n=0$ case by Theorem \ref{2by2phaseportrait}, so we need only assume the result for cases $\leq n-1$ and show that this implies case $n$. Moreover, since the $k=-1$ Friedmann solutions $\boldsymbol{U}_F$ solve system (\ref{nbynsystem}), estimates (\ref{higherorderlimitUF1}) and (\ref{higherorderlimitUF2}) for general $n\geq1$ follow from the $n=0$ case once we establish (\ref{higherorderlimitU}) for $n\geq1$. Thus it remains only to prove (\ref{higherorderlimitUest}) for general $n$.

In this light, let $\boldsymbol{U}(\tau)=(\boldsymbol{v}_1(\tau),\dots,\boldsymbol{v}_n(\tau))$ be a solution of (\ref{nbynsystem}) satisfying initial data (\ref{initialdataU*}), assume (\ref{lowestorderlimit}) and assume for induction that (\ref{higherorderlimit}) and (\ref{higherorderlimitUest}) hold for $k\leq n-1$. Assuming this, the proof by induction is complete once we prove (\ref{higherorderlimit}) and (\ref{higherorderlimitUest}) hold for $k=n$. For this, define $\boldsymbol{U}_k(\tau)=(\boldsymbol{v}_1(\tau),\dots,\boldsymbol{v}_k(\tau))$. Since system (\ref{nbynsystem}) closes in $\boldsymbol{v}_1,\dots,\boldsymbol{v}_k$ for every $k\leq n-1$, it follows that $\boldsymbol{U}_k$ solves the $n=k$ version of system (\ref{nbynsystem}) for every $k=1,\dots,n-1$. Thus by the inductive assumption, (\ref{higherorderlimit}) and (\ref{higherorderlimitUest}) hold for $k\leq n-1$, that is, $\boldsymbol{U}_k(\tau)\to M$ and 
\begin{align*}
    |\boldsymbol{U}_k(\tau)-M| \leq C_k \tau^ke^{-\tau}
\end{align*}
holds for every $k\leq n-1$ for some constants $C_1,\dots,C_{n-1}$ with $C_k = C_k(\boldsymbol{v}_1^*,\dots,\boldsymbol{v}_k^*,\tau_*)$, that is, depending only on the equations and the initial data $(\boldsymbol{U}_{n-1}^*,\tau_*)$. Thus to prove the theorem, it suffices to prove there exists a constant $C_n$, depending only on $C_1,\dots,C_{n-1}$ and $(\boldsymbol{U}_{n-1}^*,\tau_*)$, such that
\begin{align*}
    |\boldsymbol{v}_n| \leq C_n\tau^{n}e^{-\tau}.
\end{align*}
So for induction, assume without loss of generality, that
\begin{align}
    z_2(\tau) &= a(\tau)\tau e^{-\tau},\label{thelimitz2gen}\\
    w_0(\tau) &= 1 + b(\tau)\tau e^{-\tau},\label{thelimitw0gen}\\
    \boldsymbol{U}_{n-1}(\tau) &= M + \boldsymbol{V}_{n-1}(\tau)\tau^{n-1}e^{-\tau},\label{thelimitvkgen}
\end{align}
where
\begin{align}
    |a(\tau)| &\leq C_{n-1}, & |b(\tau)| &\leq C_{n-1}, & |\boldsymbol{V}_{n-1}(\tau)| &\leq C_{n-1},\label{boundz2w0gen}
\end{align}
for all $\tau\geq\tau_*$. Putting (\ref{thelimitz2gen}) into system (\ref{nbynsystem}) and assuming (\ref{boundz2w0gen}), we find that the $n^{th}$ equation for $\boldsymbol{v}_n=(z_{2n},w_{2n-2})=(u,v)$ in system (\ref{nbynsystem}) takes the form
\begin{align}
    \frac{d}{d\tau}\left(\begin{array}{c} u\\v\end{array}\right) = \left(\begin{array}{cc}-1 & 0\\
    -\frac{1}{10} & -1\end{array}\right)\left(\begin{array}{c}u\\
    v\end{array}\right) + \tau^{n-1}e^{-\tau}A(\tau)\left(\begin{array}{c}u\\
    v\end{array}\right) + \tau^{n-1}e^{-\tau}B(\tau),\label{nonauto2genuv}
\end{align}
with
\begin{align*}
    u(1) &= u_*, & v(1) &= v_*,
\end{align*}
and where $u=z_{2n}$ and $v=w_{2n-2}$ are treated as unknowns. Moreover, $A$ and $B$ are $2\times2$ matrices determined by $\boldsymbol{v}_1,\dots,\boldsymbol{v}_{n-1}$ and assumed without loss of generality to satisfy the induction hypothesis in the form
\begin{align}
    \|A\| &\leq C_{n-1}, & \|B\| &\leq C_{n-1}.\label{estAB}
\end{align}

We obtain estimate (\ref{higherorderlimitUest}) from the theory of scalar first-order linear equations. For this we need a preliminary supnorm estimate on $|u|+|v|$ to derive an estimate for $|u|$, independent of $v$, from the first equation in (\ref{nonauto2genuv}) and an estimate for $|v|$, independent of $u$, from the second equation in (\ref{nonauto2genuv}). For this, (\ref{nonauto2genuv}) implies the two estimates:
\begin{align}
    \frac{d}{d\tau}|u| &\leq -|u| + \tau^{n-1}e^{-\tau}\|A\||v| + \tau^{n-1}e^{-\tau}\|B\|,\label{gronestu}\\
    \frac{d}{d\tau}|v| &\leq -|v| + \frac{1}{10}|u| + \tau^{n-1}e^{-\tau}\|B\|.\label{gronestv}
\end{align}
Adding positive terms to the right hand sides of (\ref{gronestu}) and (\ref{gronestv}) and then adding the equations gives
\begin{align*}
    \frac{d}{d\tau}(|u|+|v|) \leq -\frac{9}{10}(|u|+|v|) + \tau^{n-1}e^{-\tau}\|A\|(|u|+|v|) + \tau^{n-1}e^{-\tau}\|B\|,
\end{align*}
which in light of (\ref{estAB}) yields the linear Gr\"{o}nwall estimate
\begin{align}
    \frac{d}{d\tau}w \leq g(\tau)w + f(\tau),\label{weqn}
\end{align}
where:
\begin{align*}
    w &= |u| + |v|,\\
    g(\tau) &= -\frac{9}{10} + C_{n-1}\tau^{n-1}e^{-\tau},\\
    f(\tau) &= C_{n-1}\tau^{n-1}e^{-\tau}.
\end{align*}
Multiplying (\ref{weqn}) through by the integrating factor $e^{-\int_1^{\tau}g(s)ds}$, we obtain
\begin{align*}
    \frac{d}{d\tau}\Big(e^{-\int_1^{\tau}g(s)ds}w\Big) \leq e^{-\int_1^{\tau}g(s)ds}f(\tau),
\end{align*}
which integrates to
\begin{align}
    w(\tau) \leq e^{-\int_1^{\tau}g(z)dz}\Bigg(e^{\tau_*}w_*+\int_1^{\tau}f(z)\bigg(-\int_1^zg(s)ds\bigg)dz\Bigg).\label{estforw}
\end{align}
Now recall that
\begin{align*}
    I_k = \int_0^\infty s^ke^{-s}ds
\end{align*}
integrates by parts to
\begin{align*}
    I_k = kI_{k-1} = k(k-1)I_{k-2} = \dots = k!,
\end{align*}
thus
\begin{align*}
    \int_1^\tau g(s)ds \leq \int_1^\tau-\frac{9}{10}+C_{n-1}s^{n-1}e^{-s}\ ds \leq -\frac{9}{10}(\tau-1) + C_{n-1}(n-1)!.
\end{align*}
Using this in (\ref{estforw}) gives the estimate
\begin{align}
    w(\tau) &\leq e^{C_{n-1}(n-1)!}e^{-\frac{9}{10}(\tau-1)}\big(e^{\tau_*}w_*+C_{n-1}^2(n-1)!\big)\notag\\
    &\leq e^{C_{n-1}(n-1)!}\big(e^{\tau_*}w_*+C_{n-1}^2(n-1)!\big) =: \bar{C},\label{estforw1}
\end{align}
where $\bar{C}$ depends only on $C_{n-1}$ and the initial data.

Using the estimate (\ref{estforw1}) for $|v|$ in (\ref{gronestu}) gives
\begin{align*}
    \frac{d}{d\tau}|u| \leq -|u| + C_{n-1}(\bar{C}+1)\tau^{n-1}e^{-\tau},
\end{align*}
which by the integrating factor method yields
\begin{align*}
    \frac{d}{d\tau}(e^{\tau}|u|) \leq C_{n-1}(\bar{C}+1)\tau^{n-1},
\end{align*}
and integrates to
\begin{align*}
    |u(\tau)| \leq e^{-\tau}\bigg(e^{\tau_*}|u_*|+\frac{C_{n-1}(\bar{C}+1)}{n}\tau^n\bigg) \leq \bar{C}_u\tau^ne^{-\tau},
\end{align*}
where $\tau\geq\tau_*\geq1$ and 
\begin{align*}
    \bar{C}_u := e^{\tau_*}|u_*| + \frac{C_{n-1}(\bar{C}+1)}{n}.
\end{align*}
Alternatively, using (\ref{estforw1}) to estimate $u$ in the $v$-equation (\ref{gronestv}) gives
\begin{align*}
    \frac{d}{d\tau}|v| \leq -|v| + \bigg(\frac{1}{10}\bar{C}_u+C_{n+1}\bigg)\tau^{n-1}e^{-\tau},
\end{align*}
which integrates as above to
\begin{align*}
    |v(\tau)| \leq e^{-\tau}\bigg(e^{\tau_*}|v_*|+\frac{\frac{1}{10}\bar{C}_u+C_{n+1}}{n}\tau^n\bigg) \leq \bar{C}_v\tau^ne^{-\tau},
\end{align*}
where again $\tau\geq\tau_*\geq1$ and
\begin{align*}
    \bar{C}_v := e^{\tau_*}|v_*| + \frac{\frac{1}{10}\bar{C}_u+C_{n+1}}{n}.
\end{align*}
Now by setting
\begin{align*}
    C_n = \sqrt{\bar{C}_u^2+\bar{C}_v^2},
\end{align*}
we conclude
\begin{align*}
    |\boldsymbol{v}_n(\tau)| = \sqrt{|u(\tau)|^2+|v(\tau)|^2} \leq \sqrt{\bar{C}_u^2+\bar{C}_v^2}\ \tau^ne^{-\tau} = C_{n+1}\tau^ne^{-\tau}, 
\end{align*}
from which estimate (\ref{higherorderlimitUest}) follows. This completes the induction step and thereby completes the proof.
\end{proof}

In the next corollary of Theorem \ref{ThmBigStability}, we explicitly compute system (\ref{zneqn})--(\ref{wneqn}) for $n=1,2,3$.

\begin{Corollary}\label{CorForNequal3}
    Extracting the leading order terms from the right hand side of (\ref{zneqn}), we obtain the equivalent form (also see (\ref{leadinginz}) below)
    \begin{align}
        t\dot{z}_{2n} = \big((2n+1)(1-w_0)-1\big)z_{2n} - (2n+1)z_2w_{2n-2}\delta_1(n) - (2n+1)\sum_{\substack{i+j+k=n\\i\neq n,j\neq n-1}}z_{2i}w_{2j}D_{2k},\label{zneqnextract2nrepeat}
    \end{align}
    where
    \begin{align*}
        \delta_1(n) = \begin{cases}
            1, & n=1,\\
            0, & n\neq 1.
        \end{cases}
    \end{align*}
    From (\ref{zneqnextract2nrepeat}) we compute the equations for $n=1,2,3$:
    \begin{align}
        t\dot{z}_2 &= 2z_2 - 3z_2w_0,\label{z2neqnis1}\\
        t\dot{z}_4 &= 4z_4 - 5(z_4w_0+z_2w_2+z_2w_0D_2),\label{z4neqnis2}\\
        t\dot{z}_6 &= 6z_6 - 7(z_6w_0+z_2w_4+z_2w_0D_4+z_2w_2D_2+z_4w_0D_2+z_4w_2).\label{z6neqnis3}
    \end{align}
    Similarly, extracting the leading order terms from the right hand side of (\ref{wneqn}), we obtain the equivalent form
    \begin{multline}
        t\dot{w}_{2n} = -\frac{1}{2(2n+3)}z_{2n+2} + \big((2n+2)(1-w_0)-1\big)w_{2n} + \frac{1}{2}(a_{2n}-A_{2n+2})\\
        + \frac{1}{2}\sum_{\substack{i+j+k=n\\j\neq n}}\hat{w}_{2i}a_{2j}D_{2k} - \sum_{\substack{i+j+k=n\\i\neq n,j\neq n}}(2i+1)w_{2i}w_{2j}D_{2k}.\label{wneqnextract2nrepeat}
    \end{multline}
    Again, from (\ref{wneqnextract2nrepeat}) we compute the equations for $n=1,2,3$:
    \begin{align}
        t\dot{w}_0 &= -\frac{1}{6}z_2 + w-w_0^2,\label{w0eqnis1}\\
        t\dot{w}_2 &= -\frac{1}{10}z_4 + 3w_2 - 4w_0w_2 - \frac{1}{2}w_0^2A_2 - \frac{1}{2}A_2^2 + \frac{1}{2}A_2D_2-w_0^2D_2,\label{w2eqnis2}\\
        t\dot{w}_4 &= -\frac{1}{14}z_6 + 5w_4 - 6w_0w_4 - \frac{1}{2}w_0^2(A_4-A_2^2) - \frac{1}{2}w_0^2A_2D_2\notag\\
        &- w_0w_2A_2 + \frac{1}{2}A_2D_4 + \frac{1}{2}(A_4-A_2^2)D_2 - A_2A_4 + \frac{1}{2}A_2^3\notag\\
        &- w_0^2D_4 - 4w_0w_2D_2 - 3w_2^2.\label{w4eqnis3}
    \end{align}
    Moreover,
    \begin{align}
        A_2 &= -\frac{1}{3}z_2, & A_4 &= -\frac{1}{5}z_4, & A_6 &= -\frac{1}{7}z_6,\label{AsInPaper}
    \end{align}
    and:
    \begin{align}
        D_2 &= -\frac{1}{12}z_2,\label{D2InPaper}\\
        D_4 &= -\frac{3}{40}z_4 + \frac{1}{8}z_2w_0^2 - \frac{1}{96}z_2^2,\label{D4InPaper}\\
        D_6 &= \frac{1}{12}\bigg(z_4w_0^2+2z_2w_0w_2+\frac{5}{24}z_2^2w_0^2-\frac{23}{120}z_2z_4-\frac{7}{288}z_2^3-\frac{5}{7}z_6\bigg).\label{D6InPaper}
    \end{align}
\end{Corollary}

\begin{proof}
Equations (\ref{zneqnextract2nrepeat}) and (\ref{wneqnextract2nrepeat}) are derived from (\ref{leadinginz}) and (\ref{finalwgood}) respectively. To eliminate $\hat{w}_{2i}a_{2j}D_{2k}$ from (\ref{wneqnextract2nrepeat}), close the equations in $z_{2k}$, $w_{2k}$, $A_{2k}$ and $D_{2k}$ at orders $n=1,2,3$ and use the expression
\begin{align*}
    \sum_{i+j+k=n}\hat{w}_{2i}a_{2j}D_{2k} &= A_2 - w_0^2A_2 - A_2^2 + A_2D_2 + A_4 - w_0^2(A_4-A_2^2)\\
    &- w_0^2A_2D_2 - 2w_0w_2A_2 + A_2D_4 + A_4 - A_2^2\\
    &- 2A_2A_4 + A_2^3 + A_6.
\end{align*}
By this we arrive at (\ref{z2neqnis1})--(\ref{z6neqnis3}) and (\ref{w0eqnis1})--(\ref{w4eqnis3}).

Equations (\ref{AsInPaper}) and (\ref{D4InPaper}) follow from (\ref{Antozn}) and (\ref{Dntoznwn}). As an example, we consider the case of $D_6$. From (\ref{Dntoznwn}) we have
\begin{align*}
    12D_6 &= \sum_{(i,j,k,l)\in X}z_{2i}w_{2j}w_{2k}D_{2l} - \sum_{(i,j)\in Y}\Big((z_{2i}+2A_{2i})D_{2j} - 4jA_{2i}D_{2j}\Big)\\
    &= z_4w_0^2 + 2z_2w_0w_2 + z_2w_0^2D_2 - z_2D_4 - 10A_2D_4 - z_4D_2 + 6A_4D_2 - z_6 + 2A_6\\
    &= z_4w_0^2 + 2z_2w_0w_2 + z_2w_0^2D_2 + \frac{7}{3}z_2D_4 + \frac{1}{5}z_4D_2 - \frac{5}{7}z_6\\
    &= z_4w_0^2 + 2z_2w_0w_2 + \frac{5}{24}z_2^2w_0^2 - \frac{23}{120}z_2z_4 - \frac{7}{288}z_2^3 - \frac{5}{7}z_6,
\end{align*}
where
\begin{align*}
    X &= \{(2,0,0,0),(1,0,0,0),(1,0,1,0),(1,1,0,0)\},\\
    Y &= \{(1,2),(2,1),(3,0)\},
\end{align*}
which confirms (\ref{D6InPaper}).
\end{proof}

Equations (\ref{AsInPaper})--(\ref{D6InPaper}), recorded in (\ref{As})--(\ref{D6}) of the introduction, express the $A_i$'s and $D_j$'s in terms of $w_i$'s and $z_j$'s. Using these to eliminate the $A_i$'s and $D_j$'s from equations (\ref{z2neqnis1})--(\ref{z6neqnis3}) and (\ref{w0eqnis1})--(\ref{w4eqnis3}) leads, after simplification, to equations (\ref{z2final})--(\ref{w4final}) of the introduction. This, together with (\ref{AsInPaper})--(\ref{D6InPaper}), then establishes Theorem \ref{EqnsToOrder3}, our final result given in the introduction.

\section{Pure Eigenvalue Solutions}\label{S11.5}

Recall from Section \ref{I5} and (\ref{eigenvaluesSM}) that the eigenvalues of our expansion in $\xi$ about $SM$ take the form:
\begin{align*}
    \lambda_{An} &= \frac{2n}{3}, & \lambda_{Bn} &= \frac{1}{3}(2n-5).
\end{align*}
Each eigenvalue introduces a free parameter into the expansion, so there are two additional free parameters at each order. At orders $n=1$ and $n=2$ we know the negative eigenvalues do not appear in the expansion of a $k<0$ Friedmann solution, since this would mean the trajectory does not originate at the fixed point $SM$. Furthermore, we know from Section \ref{S8} that the $k<0$ Friedmann solutions at leading order ($n=1$) have a one-to-one correspondence with trajectories emanating from the unstable manifold of $SM$ on the underdense side. Thus the free parameter associated with the first positive eigenvalue, $\lambda_{A1}$, is related to the single $k<0$ Friedmann free parameter $\Delta_0$. This free parameter will appear at higher orders in the expansion as well, the question is whether the higher order coefficients of the expansion of the $k<0$ Friedmann solution are generated purely by $\Delta_0$, or whether the parameters that appear at higher orders introduce additional contributions of $\Delta_0$. Put another way, we know that the $k<0$ Friedmann solution has a single parameter freedom, but we do not know if the additional parameters that enter generically in the expansion are absent, making the $k<0$ Friedmann solution a pure eigenvector solution, or whether these additional parameters are functions of $\Delta_0$, implying the $k<0$ Friedmann solution is not generated by the single leading order parameter. We conjecture that it is the former that is true, with the following theorem confirming this up to order $n=3$.

\begin{Thm}\label{ThmPure}
    Let $(\bar{t},\xi)$ represent SSCNG coordinates, with $\xi=\frac{\bar{r}}{\bar{t}}$. Then smooth perturbations of $SM$ take the form:
    \begin{align}
    	A(\bar{t},\xi) &= 1 + A_2(\bar{t})\xi^2 + A_4(\bar{t})\xi^4 + A_6(\bar{t})\xi^6 + O(\xi^8),\label{PertA}\\
    	D(\bar{t},\xi) &= 1 + D_2(\bar{t})\xi^2 + D_4(\bar{t})\xi^4 + D_6(\bar{t})\xi^6 + O(\xi^8),\label{PertD}\\
        z(\bar{t},\xi) &= z_2(\bar{t})\xi^2 + z_4(\bar{t})\xi^4 + z_6(\bar{t})\xi^6 + O(\xi^8),\label{Pertz}\\
        w(\bar{t},\xi) &= w_0(\bar{t}) + w_2(\bar{t})\xi^2 + w_4(\bar{t})\xi^4 + O(\xi^6),\label{Pertw}
    \end{align}
    where:
    \begin{align*}
        A_2(\bar{t}) &= -\frac{1}{3}z_2(\bar{t}), & A_4(\bar{t}) &= -\frac{1}{5}z_4(\bar{t}), & A_6(\bar{t}) &= -\frac{1}{7}z_6(\bar{t}),
	\end{align*}
	\begin{align*}
        D_2(\bar{t}) &= -\frac{1}{12}z_2(\bar{t}), & D_4(\bar{t}) &= \frac{1}{8}\left(w_0^2(\bar{t}) - \frac{1}{12}z_2(\bar{t})\right)z_2(\bar{t}) - \frac{3}{40}z_4(\bar{t}),
    \end{align*}
    \begin{align*}
        D_6(\bar{t}) = \frac{1}{6}\left(w_0(\bar{t})w_2(\bar{t}) + \frac{5}{48}w_0^2(\bar{t})z_2(\bar{t}) - \frac{7}{576}z_2^2(\bar{t})\right)z_2(\bar{t}) + \frac{1}{12}\left(w_0^2(\bar{t}) - \frac{23}{120}z_2(\bar{t})\right)z_4(\bar{t}) - \frac{5}{84}z_6(\bar{t}),
    \end{align*}
    and:
    \begin{align}
        z_2(\bar{t}) = \frac{4}{3} - 6a\bar{t}^{\frac{2}{3}} + \frac{153}{7}a^2\bar{t}^{\frac{4}{3}} - \frac{1023}{14}a^3\bar{t}^2 + O(\bar{t}^{\frac{8}{3}}),\label{z2Series}
    \end{align}
    \begin{align}
        z_4(\bar{t}) = \frac{40}{27} - \frac{20}{9}a\bar{t}^{\frac{2}{3}} - \left(\frac{100}{21}a^2+10b\right)\bar{t}^{\frac{4}{3}} + \left(\frac{1000}{21}a^3+\frac{720}{7}ab\right)\bar{t}^2 + O(\bar{t}^{\frac{8}{3}}),\label{z4Series}
    \end{align}
    \begin{multline}
        z_6(\bar{t}) = \frac{448}{243} - \frac{196}{81}a\bar{t}^{\frac{2}{3}} - \left(\frac{58}{27}a^2+\frac{98}{9}b\right)\bar{t}^{\frac{4}{3}} + \left(\frac{77}{3}a^3+\frac{211}{2}ab-14c\right)\bar{t}^2\\
        + \frac{28}{3}\alpha\bar{t}^{\frac{1}{3}} - 84\alpha a\bar{t} + 465\alpha a^2\bar{t}^{\frac{5}{3}} + O(\bar{t}^{\frac{8}{3}}),\label{z6Series}
    \end{multline}
	\begin{align}
		w_0(\bar{t}) = \frac{2}{3} + a\bar{t}^{\frac{2}{3}} - \frac{39}{14}a^2\bar{t}^{\frac{4}{3}} + \frac{213}{28}a^3\bar{t}^2 + O(\bar{t}^{\frac{8}{3}}),\label{w0Series}
	\end{align}
	\begin{align}
		w_2(\bar{t}) = \frac{2}{9} - \frac{2}{3}a\bar{t}^{\frac{2}{3}} + \left(\frac{17}{14}a^2+b\right)\bar{t}^{\frac{4}{3}} - \frac{60}{7}ab\bar{t}^2 + O(\bar{t}^{\frac{8}{3}}),\label{w2Series}
	\end{align}
	\begin{align}
		w_4(\bar{t}) = \frac{13}{81} - \frac{35}{54}a\bar{t}^{\frac{2}{3}} + \left(\frac{295}{126}a^2-\frac{7}{18}b\right)\bar{t}^{\frac{4}{3}} - \left(\frac{70}{9}a^3-c\right)\bar{t}^2 + \alpha\bar{t}^{\frac{1}{3}} - \frac{11}{2}\alpha a\bar{t} + \frac{99}{4}\alpha a^2\bar{t}^{\frac{5}{3}} + O(\bar{t}^{\frac{8}{3}}),\label{w4Series}
	\end{align}
    where $a$, $b$, $c$ and $\alpha$ are constants. In particular, the $k=-1$ Friedmann spacetime satisfies $b=c=0$ and $\alpha=\beta=\gamma=0$, with
    \begin{align*}
		a^3 = \frac{2}{375}\Delta_0^{-2},
    \end{align*}
    that is, the $k=-1$ Friedmann spacetime is a pure eigenvalue solution up to order $\xi^6$.
\end{Thm}

\begin{proof}
Substituting series (\ref{PertA})--(\ref{Pertw}) into the STV PDE (\ref{Axi-final0})--(\ref{weqnxi-final0}), we immediately obtain the algebraic relations for $A_2$, $A_4$, $A_6$, $D_2$, $D_4$ and $D_6$. The remaining equations are then given by:
\begin{align*}
    \bar{t}\dot{w}_0 &= (1-w_0)w_0 - \frac{1}{6}z_2,\\
    \bar{t}\dot{z}_2 &= (2-3w_0)z_2,\\
    \bar{t}\dot{w}_2 &= (3-4w_0)w_2 - \frac{1}{10}z_4 + \frac{1}{4}\left(w_0^2 - \frac{1}{6}z_2\right)z_2,\\
    \bar{t}\dot{z}_4 &= (4-5w_0)z_4 - 5z_2w_2 + \frac{5}{12}z_2^2w_0,\\
    \bar{t}\dot{w}_4 &= (5-6w_0)w_4 - \frac{1}{14}z_6 + \frac{7}{40}\left(w_0^2-\frac{11}{42}z_2\right)z_4 - 3\left(w_2-\frac{2}{9}z_2w_0\right)w_2 - \frac{1}{8}\left(w_0^4 - \frac{1}{4}w_0^2z_2 + \frac{7}{72}z_2^2\right)z_2,\\
    \bar{t}\dot{z}_6 &= (6-7w_0)z_6 - 7z_2w_4 + \frac{7}{12}z_2^2w_2 - 7\left(w_2 - \frac{19}{120}z_2w_0\right)z_4 - \frac{7}{8}\left(w_0^2-\frac{1}{12}z_2\right)w_0z_2^2.
\end{align*}
We know from Theorem \ref{UnstableManifoldOfSM} that the leading order terms of $w_0$ and $z_2$ in the limit $\bar{t}\to0$ are proportional to
\begin{align*}
    e^{\lambda_{A1}(\tau)} = \bar{t}^{\lambda_{A1}},
\end{align*}
where $\lambda_{A1} = \frac{2}{3}$ is the positive eigenvalue of the $n=1$ system and $\tau=\ln\bar{t}$. We have already seen that the negative eigenvalue $\lambda_{B1}$ is eliminated by setting \emph{time since the Big Bang} and thus does not feature in the leading order analysis. For $n=2$, if we denote $U(\bar{t}) = (z_2,w_0,z_4,w_2)$, then the leading order behavior as $\bar{t}\to0$ becomes
\begin{align*}
    U(\bar{t}) = a\bar{t}^{\lambda_{A1}}\boldsymbol{R}_{A1} + b\bar{t}^{\lambda_{A2}}\boldsymbol{R}_{A2},
\end{align*}
where $\lambda_{A2}=\frac{4}{3}$. We note that $\lambda_{B2}$ is the only negative eigenvalue for $n>1$ since in general:
\begin{align*}
    \lambda_{An} &= \frac{2n}{3}, & \lambda_{Bn} &= \frac{1}{3}(2n-5).
\end{align*}
If we denote $U(\bar{t}) = (z_2,w_0,z_4,w_2,z_6,w_4)$, then the leading order behavior as $\bar{t}\to0$ is given in general by
\begin{align*}
    U(\bar{t}) &= a\bar{t}^{\lambda_{A1}}\boldsymbol{R}_{A1} + b\bar{t}^{\lambda_{A2}}\boldsymbol{R}_{A2} + c\bar{t}^{\lambda_{A3}}\boldsymbol{R}_{A3}\\
    &+ \alpha\bar{t}^{\lambda_{B3}}\boldsymbol{R}_{B3} + \beta\bar{t}^{\lambda_{B4}}\boldsymbol{R}_{B4} + \gamma\bar{t}^{\lambda_{B5}}\boldsymbol{R}_{B5} + O(\bar{t}^{\frac{7}{3}}),
\end{align*}
where $a$, $b$, $c$, $\alpha$, $\beta$ and $\gamma$ are free parameters. With this knowledge, we can compute the Taylor series of $w_0$, $w_2$, $w_4$, $z_2$, $z_4$ and $z_6$ to yield (\ref{z2Series})--(\ref{w4Series}). What remains is to show that the $k=-1$ Friedmann spacetime parameters satisfy $b=c=0$ and $\alpha=\beta=\gamma=0$, where
\begin{align*}
    a^3 = \frac{2}{375}\Delta_0^{-2}.
\end{align*}
Recall from Theorem \ref{ThmLog-time translation} that:
\begin{align}
    1-A &= \frac{\kappa}{3}\rho\bar{r}^2 = \frac{8\bar{\chi}^2\xi^2}{(\cosh2\Theta-1)^3},\label{AStart}\\
    \sqrt{AB} &= \frac{\sqrt{1+r^2}}{\frac{\partial\bar{t}}{\partial t}(t,r)} =
    \frac{\sqrt{1+\frac{\bar{\chi}^2\xi^2}{\sinh^4\Theta}}}{\frac{\partial\bar{\chi}}{\partial\chi}(\chi,r)},\label{DStart}\\
    v &= \frac{\dot{R}r}{\sqrt{1+r^2}} = \frac{\frac{\bar{\chi}\xi}{\sinh^2\Theta}\coth\Theta}{\sqrt{1+\frac{\bar{\chi}^2\xi^2}{\sinh^4\Theta}}}\label{vStart},
\end{align}
where:
\begin{align}
    \chi &= \frac{t}{\Delta_0} = \frac{1}{2}(\sinh2\Theta-2\Theta),\label{Theta1}\\
    \bar{\chi} &= \frac{\bar{t}}{\Delta_0},
\end{align}
and:
\begin{align}
    \cosh\Theta(\bar{\chi}) &= \sqrt[4]{1+\frac{\bar{\chi}^2\xi^2}{\sinh^4\Theta}}\cosh\Theta,\label{Theta3}\\
    r &= \frac{\bar{\chi}\xi}{\sinh^2\Theta}.
\end{align}
Note that $\Theta = \Theta(\chi)$ is the inverse of (\ref{Theta1}), whereas $\Theta(\bar{\chi})$ is the inverse of the same expression but with $\chi$ replaced with $\bar{\chi}$. Now with a nontrivial amount of algebra we can write (\ref{AStart})--(\ref{vStart}) as:
\begin{align}
    A &= 1-\frac{\xi^2\bar{\chi}^2}{\sinh^6\Theta},\label{ATheta}\\
    D &= \frac{\cosh\Theta(\bar{\chi})\sinh\Theta}{\sinh\Theta(\bar{\chi})\cosh\Theta},\label{DTheta}\\
    w &= \frac{\bar{\chi}\cosh^3\Theta}{\cosh^2\Theta(\bar{\chi})\sinh^3\Theta},\label{wTheta}
\end{align}
along with
\begin{align}
    z = \frac{3(1-A)}{1-\xi^2w^2}.\label{zTheta}
\end{align}
The strategy is to write $A$, $D$, $w$ and $z$ as a series of the form (\ref{PertA})--(\ref{Pertw}). From this, the series of $w_{2n-2}(\bar{t})$ and $z_{2n}(\bar{t})$ can be deduced and compared to the general series (\ref{z2Series})--(\ref{w4Series}). The first step is thus to write variables $A$, $D$, $w$ and $z$ as a series in $(\bar{\chi},\xi)$, which requires expanding $\Theta=\Theta(\chi)$ as a series in $(\bar{\chi},\xi)$. We can compute this using relation (\ref{Theta3}) to yield
\begin{align*}
    \Theta(\chi) &= \Theta(\bar{\chi})\\
    &- \frac{1}{4}\bar{\chi}^2\frac{\cosh\Theta(\bar{\chi})}{\sinh^5\Theta(\bar{\chi})}\xi^2\\
    &- \frac{1}{32}\bar{\chi}^4\frac{8\cosh\Theta(\bar{\chi})+\cosh3\Theta(\bar{\chi})}{\sinh^{11}\Theta(\bar{\chi})}\xi^4\\
    &- \frac{1}{384}\bar{\chi}^6\frac{116+75\cosh2\Theta(\bar{\chi})+4\cosh4\Theta(\bar{\chi})}{\sinh^{17}\Theta(\bar{\chi})}\cosh\Theta(\bar{\chi})\xi^6 + O(\xi^8).
\end{align*}
With this series, we can expand (\ref{ATheta})--(\ref{zTheta}) to take the form given by (\ref{PertA})--(\ref{Pertw}), recalling that $\bar{t}=\Delta_0\bar{\chi}$. Computing these series we obtain:
\begin{align}
    w_0(\bar{\chi}) &= \bar{\chi}\frac{\cosh\Theta(\bar{\chi})}{\sinh^3\Theta(\bar{\chi})},\label{w0chi}\\
    z_2(\bar{\chi}) &= 3\bar{\chi}^2\frac{1}{\sinh^6\Theta(\bar{\chi})},\label{z2chi}\\
    w_2(\bar{\chi}) &= \frac{3}{4}\bar{\chi}^3\frac{\cosh\Theta(\bar{\chi})}{\sinh^9\Theta(\bar{\chi})},\label{w2chi}\\
    z_4(\bar{\chi}) &= \frac{15}{2}\bar{\chi}^4\frac{\cosh^2\Theta(\bar{\chi})}{\sinh^{12}\Theta(\bar{\chi})},\label{z4chi}\\
    w_4(\bar{\chi}) &= \frac{3}{64}\bar{\chi}^5\frac{23\cosh\Theta(\bar{\chi})+3\cosh3\Theta(\bar{\chi})}{\sinh^{15}\Theta(\bar{\chi})},\label{w4chi}\\
    z_6(\bar{\chi}) &= \frac{21}{16}\bar{\chi}^6\frac{11+5\cosh2\Theta(\bar{\chi})}{\sinh^{18}\Theta(\bar{\chi})}\cosh^2\Theta(\bar{\chi}).\label{z6chi}
\end{align}
Now noting that
\begin{align*}
    \Theta(\bar{\chi}) = \left(\frac{3}{2}\right)^{\frac{1}{3}}\bar{\chi}^{\frac{1}{3}} - \frac{1}{10}\bar{\chi} + \frac{3}{175}\left(\frac{3}{2}\right)^{\frac{2}{3}}\bar{\chi}^{\frac{5}{3}} - \frac{1}{175}\left(\frac{3}{2}\right)^{\frac{1}{3}}\bar{\chi}^{\frac{7}{3}} + O(\bar{\chi}^3),
\end{align*}
we can thus expand (\ref{w0chi})--(\ref{z6chi}) as a series in $\bar{\chi}$ to obtain:
\begin{align*}
    w_0(\bar{\chi}) &= \frac{2}{3} + \frac{1}{5}\left(\frac{2}{3}\right)^{\frac{1}{3}}\bar{\chi}^{\frac{2}{3}} - \frac{13}{175}\left(\frac{3}{2}\right)^{\frac{1}{3}}\bar{\chi}^{\frac{4}{3}} + \frac{71}{1750}\bar{\chi}^2 + O(\bar{\chi}^{\frac{8}{3}}),\\
    z_2(\bar{\chi}) &= \frac{4}{3} - \frac{2}{5}(18)^{\frac{1}{3}}\bar{\chi}^{\frac{2}{3}} + \frac{51}{175}(12)^{\frac{1}{3}}\bar{\chi}^{\frac{4}{3}} - \frac{341}{875}\bar{\chi}^2 + O(\bar{\chi}^{\frac{8}{3}}),\\
    w_2(\bar{\chi}) &= \frac{2}{9} - \frac{2}{15}\left(\frac{2}{3}\right)^{\frac{1}{3}}\bar{\chi}^{\frac{2}{3}} + \frac{17}{175}(18)^{-\frac{1}{3}}\bar{\chi}^{\frac{4}{3}} + O(\bar{\chi}^{\frac{8}{3}}),\\
    z_4(\bar{\chi}) &= \frac{40}{27} - \frac{4}{9}\left(\frac{2}{3}\right)^{\frac{1}{3}}\bar{\chi}^{\frac{2}{3}} - \frac{4}{21}\left(\frac{2}{3}\right)^{\frac{2}{3}}\bar{\chi}^{\frac{4}{3}} + \frac{16}{63}\bar{\chi}^2 + O(\bar{\chi}^{\frac{8}{3}}),\\
    w_4(\bar{\chi}) &= \frac{13}{81} - \frac{7}{27}(12)^{-\frac{1}{3}}\bar{\chi}^{\frac{2}{3}} + \frac{59}{315}(18)^{-\frac{1}{3}}\bar{\chi}^{\frac{4}{3}} - \frac{28}{675}\bar{\chi}^2 + O(\bar{\chi}^{\frac{8}{3}}),\\
    z_6(\bar{\chi}) &= \frac{448}{243} - \frac{196}{405}\left(\frac{2}{3}\right)^{\frac{1}{3}}\bar{\chi}^{\frac{2}{3}} - \frac{58}{675}\left(\frac{2}{3}\right)^{\frac{2}{3}}\bar{\chi}^{\frac{4}{3}} + \frac{154}{1125}\bar{\chi}^2 + O(\bar{\chi}^{\frac{8}{3}}).
\end{align*}
Finally, we see that by identifying
\begin{align*}
    a^3 = \frac{2}{375}\Delta_0^{-2},
\end{align*}
we obtain (\ref{z2Series})--(\ref{w4Series}) with $b=c=0$ and $\alpha=\beta=\gamma=0$.
\end{proof}

\section{Appendix: Proofs of the Main Theorems}\label{Appendix1}

\subsection{Proof of Theorem \ref{SSCStandardGauge}: Transformation to SSCNG}\label{Appendix1A}

First note that by direct differentiation:
\begin{align}
    \Phi_t &= h'g = \frac{\lambda}{\dot{R}R}\Phi,\label{Phit}\\
    \Phi_r &= hg' = \frac{\lambda r}{1-kr^2}\Phi.\label{Phir}
\end{align}
Letting $\boldsymbol{x}=(t,r)$ and $\hat{\boldsymbol{x}}=(\hat{t},\bar{r})$, the inverse Jacobian is given by
\begin{align}
    J^{-1} = \frac{\partial\hat{x}^\nu}{\partial x^\mu} = \left(\begin{array}{cc}
    \Phi_t & \Phi_r\\
    \dot{R}r & R
    \end{array}\right)^\nu_\mu,\label{Jinverse}
\end{align}
so
\begin{align}
    J = \frac{\partial x^\mu}{\partial\hat{x}^\nu} = \frac{1}{|J^{-1}|}\left(\begin{array}{cc}
    R & -\Phi_r\\
    -\dot{R}r & \Phi_t
    \end{array}\right)^\mu_\nu,
\end{align}
where
\begin{align*}
    |J^{-1}| = |R\Phi_t-\dot{R}r\Phi_r|.
\end{align*}
Thus 
\begin{align}
    \hat{g} = J^TgJ &= \frac{1}{|J^{-1}|^2}\left(\begin{array}{cc}
    R & -\dot{R}r\\
    -\Phi_r & \Phi_t
    \end{array}\right)\left(\begin{array}{cc}
    -1 & 0\\
    0 & \frac{R^2}{1-kr^2}
    \end{array}\right)\left(\begin{array}{cc}
    R & -\Phi_r\\
    -\dot{R}r & \Phi_t
    \end{array}\right)\notag\\
    &= \frac{1}{|J^{-1}|^2}\left(\begin{array}{cc}
    -R & -\frac{\dot{R}R^2r}{1-kr^2}\\
    \Phi_r & \frac{R^2}{1-kr^2}\Phi_t
    \end{array}\right)\left(\begin{array}{cc}
    R & -\Phi_r\\
    -\dot{R}r & \Phi_t
    \end{array}\right)\notag\\
    &= \frac{1}{|J^{-1}|^2}\left(\begin{array}{cc}
    \frac{(\dot{R}Rr)^2}{1-kr^2} - R^2 & R\Phi_r - \frac{\dot{R}R^2r}{1-kr^2}\Phi_t\\
    R\Phi_r - \frac{\dot{R}R^2r}{1-kr^2}\Phi_t & \frac{R^2}{1-kr^2}\Phi_t^2 - \Phi_r^2
    \end{array}\right).\label{TransformedMetric}
\end{align}
We first verify that the middle term $\hat{g}_{01}$ in (\ref{TransformedMetric}) vanishes when $\Phi(t,r)=h(t)g(r)$, with $g$ and $h$ given in (\ref{DefghThm}). To see this we derive (\ref{DefghThm}) from the condition
\begin{align*}
    \hat{g}_{01} = R\Phi_r - \frac{\dot{R}R^2r}{1-kr^2}\Phi_t = 0,
\end{align*}
which gives
\begin{align*}
    \frac{1-kr^2}{r}\frac{g'}{g} = \lambda = \dot{R}R\frac{h'}{h}
\end{align*}
for some positive constant $\lambda$. Integrating then gives:
\begin{align*}
    h(t) &= e^{\lambda\int_0^t\frac{ds}{\dot{R}(s)R(s)}},\\
    g(r) &= e^{\lambda\int_0^r\frac{sds}{1-ks^2}},
\end{align*}
in agreement with (\ref{DefghThm}).

To establish (\ref{Aformulaone}) and (\ref{Bformulaone}), let $\hat{g}_{00}=-\hat{B}$ and $\hat{g}_{11}=\hat{A}^{-1}$, then straightforward substitutions give:
\begin{align}
    \frac{1}{\hat{A}(\hat{t},\bar{r})} &= \frac{R^2\Phi_t^2-(1-kr^2)\Phi_r^2}{(R\Phi_t-\dot{R}r\Phi_r)^2(1-kr^2)} = 1 - kr^2 - H^2\bar{r}^2,\label{Abar}\\
    \hat{B}(\hat{t},\bar{r}) &= -\frac{1}{|J^{-1}|^2}\bigg(\frac{(\dot{R}Rr)^2}{1-kr^2}-R^2\bigg)\notag\\
    &= \frac{(1-kr^2)R^2-(\dot{R}Rr)^2}{(R\Phi_t-\dot{R}r\Phi_r)^2(1-kr^2)}\notag\\
    &= \frac{1-kr^2}{1-kr^2-H^2\bar{r}^2}\bigg(\frac{R^2\dot{R}^2}{\lambda^2\Phi^2}\bigg).\label{Bbar}
\end{align}
This implies
\begin{align*}
    B = \frac{1}{F'(\Phi)^2}\hat{B} = \frac{1}{(F'(\Phi)\Phi_t)^2}\frac{1-kr^2}{1-kr^2-H^2\bar{r}^2},
\end{align*}
verifying (\ref{Bformulaone}). To obtain the expression for $A(\bar{t},\bar{r})$ for the Friedmann metrics in SSCNG, note that, assuming (\ref{FtransFinal}), equation (\ref{TransformedMetric}) gives the formula
\begin{align*}
    \frac{1}{A} = \frac{(F')^2}{|J^{-1}|^2}\bigg(\frac{R^2}{1-kr^2}\Phi_t^2-\Phi_r^2\bigg),
\end{align*}
with
\begin{align*}
    |J^{-1}| = |F'||R\Phi_t-\dot{R}r\Phi_r|,
\end{align*}
so
\begin{align}
    A = \frac{(R\Phi_t-\dot{R}r\Phi_r)^2}{\frac{R^2}{1-kr^2}\Phi_t^2-\Phi_r^2}.\label{Aformula}
\end{align}
Putting (\ref{Phit}) and (\ref{Phir}) into (\ref{Aformula}) gives
\begin{align*}
    A = \frac{\Big(\frac{\lambda R\Phi}{\dot{R}R}-\frac{\dot{R}r\lambda r\Phi}{1-kr^2}\Big)^2}{\frac{R^2}{1-kr^2}\Big(\frac{\lambda\Phi}{\dot{R}R}\Big)^2-\Big(\frac{\lambda r\Phi}{1-kr^2}\Big)^2} = \frac{(1-kr^2-\dot{R}^2r^2)^2}{1-kr^2-\dot{R}^2r^2},
\end{align*}
which takes the final form
\begin{align}
    A = 1 - kr^2 - H^2\bar{r}^2,\label{Afinalform}
\end{align}
agreeing with (\ref{Aformulaone}). Thus (\ref{Aformulaone}) and (\ref{Bformulaone}) are confirmed. 

It remains to establish (\ref{barv}). Since the four-velocity equals $\boldsymbol{e}_0=(1,0)$ and the fluid velocity vanishes in comoving coordinates $(t,r)$, the formula for $J^{-1}$ in (\ref{Jinverse}) gives
\begin{align*}
    \hat{u}^\nu = (J^{-1})^\nu_\mu u^\mu = \left(\begin{array}{c}
    \Phi_t\\
    \dot{R}r\end{array}\right),
\end{align*}
and (\ref{Abar})--(\ref{Bbar}) give
\begin{align*}
    \sqrt{AB} = \frac{\sqrt{1-kr^2}}{\frac{\partial\bar{t}}{\partial t}(t,r)},
\end{align*}
which verifies (\ref{sqrtABbest}). Putting this together with (\ref{Phit}) into the formula for $\hat{v}$ gives
\begin{align}
    \hat{v} = \frac{1}{\sqrt{\hat{A}\hat{B}}}\frac{\hat{u}^1}{\hat{u}^0} = \frac{\dot{R}r}{\sqrt{1-kr^2}},\label{ABbyhatAB}
\end{align}
where
\begin{align*}
    B &= \bar{B} = \frac{1}{(F'(\hat{t}))^2}\hat{B}, & A &= \bar{A} = \hat{A}.
\end{align*}
To verify (\ref{barv}), note that the mapping from $(\bar{t},\bar{r})=(F(\hat{t}),\hat{r})$ involves only a change in time. Thus, using the notation $\bar{\boldsymbol{u}}=\boldsymbol{u}$, we have $u^0=F'(\hat{t})\hat{t}$, $u^0=\hat{u}^0$ and (\ref{ABbyhatAB}), so plugging this into the formula $v=\frac{1}{\sqrt{AB}}\frac{u^1}{u^0}$, the two factors of $F'(\hat{t})$ cancel and we see that $v=\hat{v}$.

Finally, note that $B(\bar{t},0)=1$ determines $F$ in (\ref{FtransFinal}) to be (\ref{Fdefine}), as it must, because in the SSC gauge $B(\bar{t},0)=1$, SSC time $\bar{t}$ and comoving time $t$ both measure proper time at $r=0$, so $\bar{t}=F(h(t))=t$ and $F(y)=h^{-1}(y)$. This completes the proof of Theorem \ref{SSCStandardGauge}.\qed

\subsection{Proof of Theorem \ref{SelfSimExpandxi}: The Friedmann Spacetime in SSCNG Coordinates}\label{Appendix1B}

Using the $k=0$ version of (\ref{FtransFinal})--(\ref{Fdefine}) of Theorem \ref{SSCStandardGauge}, let $R(t)$ denote the cosmological scale factor and define the coordinate transformation:
\begin{align}
    \bar{t} &= F(h(t)g(r)), & \bar{r} &= R(t)r,\label{FtransFinalagain}
\end{align}
where:
\begin{align}
    h(t) &= e^{\lambda\int_0^t\frac{d\tau}{\dot{R}(\tau)R(\tau)}},\label{Defhagain}\\
    g(r) &= e^{\frac{\lambda}{2}r^2},\label{Defgagain}\\
    F(y) &= h^{-1}(y),\label{Fdefineagain}
\end{align}
and we use the notation
\begin{align*}
    y = \hat{t} = \Phi(t,r) = h(t)g(r).
\end{align*}
By Theorem \ref{SSCStandardGauge}, (\ref{FtransFinalagain}) transforms metric (\ref{Friedmann}) over to SSC form (\ref{SSC}), the normalized gauge condition $B(\bar{t},0)=1$ holds and:
\begin{align}
    A_\sigma &= 1 - H^2\bar{r}^2,\label{Aformulaonenew}\\
    B_\sigma &= \frac{1}{(F(\Phi)_t)^2(1-H^2\bar{r}^2)}.\label{Bformulaonenew}
\end{align}
We now use formulas (\ref{Friedmannsoln1})--(\ref{kapparho}) for $R(t)$, $H(t)$ and $\rho(t)$ of Theorem \ref{ThmCosWithShockMAA}, together with the comoving velocity condition, applicable to the $p=0$, $k=0$ Friedmann metric in comoving coordinates $(t,r)$. Starting with (\ref{Friedmannsoln1}),
\begin{align*}
    R(t) = \bigg(\frac{t}{t_0}\bigg)^{\frac{\alpha}{2}},
\end{align*}
where
\begin{align*}
    \alpha = \frac{4}{3(1+\sigma)}.
\end{align*}
Differentiating yields
\begin{align*}
    \dot{R}(t) = \frac{\alpha}{2t_0}\bigg(\frac{t}{t_0}\bigg)^{\frac{\alpha}{2}-1}
\end{align*}
and substituting this into (\ref{Defhagain}) and integrating gives
\begin{align*}
    h(t) = e^{\lambda\big(\frac{2t_0^{\alpha}}{\alpha(2-\alpha)}t^{2-\alpha}\big)}.
\end{align*}
From this we obtain
\begin{align*}
    \hat{t} = \Phi(t,r) = h(t)g(r) = e^{\lambda\big(1+\frac{\alpha(2-\alpha)}{4}\eta^2\big)\frac{2t_0^2}{\alpha(2-\alpha)}t^{2-\alpha}}
\end{align*}
and
\begin{align*}
    F(y) = h^{-1}(y) = \bigg(\frac{\alpha(2-\alpha)\ln(y)}{2t_0^\alpha\lambda}\bigg)^{\frac{1}{2-\alpha}}.
\end{align*}
Thus by simple algebra,
\begin{align*}
    \bar{t} = h^{-1}(h(t)g(r)) = h^{-1}(\hat{t}) = 1 + \frac{\alpha(2-\alpha)}{4}\eta^2.
\end{align*}
This confirms that the transformation (\ref{F1a}) is equivalent to (\ref{FtransFinalagain}).

Consider now the formula (\ref{Aformulaonenew}). Equation (\ref{Friedmannsoln3}) gives the Friedmann formula for the Hubble constant
\begin{align*}
    H = \frac{2}{3(1+\sigma)}\frac{1}{t} = \frac{\alpha}{2}\frac{1}{t}.
\end{align*}
Substituting this into (\ref{Aformulaonenew}) and simplifying confirms (\ref{Aformeta}). To confirm (\ref{Bformulaonenew}), we calculate
\begin{align*}
    F(\Phi(t,r))_t = \frac{dh^{-1}}{d\hat{t}}\frac{d\hat{t}}{dt} = \bigg(\frac{\ln\hat{t}}{\lambda}\frac{\alpha(2-\alpha)}{2t_0^\alpha}\bigg)^{-\frac{1-\alpha}{2-\alpha}}t^{1-\alpha},
\end{align*}
which simplifies to
\begin{align*}
    F(\Phi(t,r))_t = \bigg(1+\frac{\alpha(2-\alpha)}{4}\eta^2\bigg)^{-\frac{1-\alpha}{2-\alpha}}.
\end{align*}
Using the above formulas in (\ref{Bformulaonenew}) then gives
\begin{align*}
    B_\sigma = \frac{1}{(1-H\bar{r}^2)F(h(t)g(r))_t^2} = \frac{\Big(1+\frac{\alpha(2-\alpha)}{4}\eta^2\Big)^{-2\frac{\alpha-1}{\alpha-2}}}{1-\frac{\alpha^2}{4}\eta^2},
\end{align*}
confirming (\ref{Bformeta}).

To verify formula (\ref{kapparho1}) for $\kappa\rho_\sigma\bar{r}^2$, start with (\ref{kapparho}), that is,
\begin{align*}
    \rho_\sigma = \frac{4}{3\kappa(1+\sigma)^2}\frac{1}{t^2},
\end{align*}
to directly obtain
\begin{align*}
    \kappa\rho_\sigma\bar{r}^2 = \kappa\rho\bar{r}^2 = \frac{4}{3}\alpha^2\eta^2,
\end{align*}
confirming (\ref{kapparho1}).

Finally, to confirm formula (\ref{vformeta}) for $v_\sigma$, set $k=0$ in (\ref{barv}) to obtain
\begin{align*}
    v_\sigma = \dot{R}r = H\bar{r} = \frac{\alpha}{2t}\frac{\bar{r}}{2} = \frac{\alpha}{2}\eta,
\end{align*}
which confirms (\ref{vformeta}).

It remains to verify expansions (\ref{etaexpand})--(\ref{vform}). Now equations (\ref{F1a}) and (\ref{eta0}) determine $\xi$ in terms of $\eta$ as
\begin{align*}
    \xi = \frac{\bar{r}}{\bar{t}} = \frac{\eta}{\mathcal{F}(\eta)} = \eta\bigg(1+\frac{\alpha(2-\alpha)}{4}\eta^2\bigg)^{-\frac{1}{2-\alpha}}.
\end{align*}
Using this together with
\begin{align*}
    (1+\gamma\eta^2)^\beta = 1 + \gamma\beta\eta^2 + \frac{1}{2}\gamma^2\beta(\beta-1)\eta^4 + O(\eta^6),
\end{align*}
we have
\begin{align*}
    \xi = \eta\bigg(1+\gamma\beta\eta^2+\frac{1}{2}\gamma^2\beta(\beta-1)\eta^4+O(\eta^6)\bigg),
\end{align*}
with:
\begin{align*}
    \beta &= -\frac{1}{2-\alpha}, & \gamma &= \frac{\alpha(2-\alpha)}{4}.
\end{align*}
Squaring and substituting gives
\begin{align}
    \xi^2 = \eta^2 - \frac{\alpha}{2}\eta^4 + \frac{1}{16}\alpha^2(4-\alpha)\eta^6 + O(\eta^8).\label{xiintermsofeta}
\end{align}
Using the elementary facts: 
\begin{align*}
    y &= x + ax^2 + bx^3 + O(x^4) & &\iff & x &= y - ay^2 + (2a^2-b)y^3 + O(y^4)
\end{align*}
and
\begin{align*}
    \frac{1}{1-ax-bx^2+O(x^3)} = 1 + ax + (a^2+b)x^2 + O(x^3),
\end{align*}
it is straightforward to invert the series (\ref{xiintermsofeta}) in a neighborhood of $\xi=\eta=0$ to obtain
\begin{align}
    \eta^2 = \xi^2 + \frac{\alpha}{2}\xi^4 + \frac{\alpha^3}{16}\xi^6 + O(\xi^8).\label{etabyxi}
\end{align}
This verifies (\ref{etaexpand}). Using (\ref{etabyxi}) in (\ref{Aformeta})--(\ref{vformeta}) gives (\ref{Aform})--(\ref{kapparho1}) and:
\begin{align*}
    \eta &= \xi + \frac{\alpha}{4}\xi^3 + O(\xi^4), & \xi &= \eta - \frac{\alpha}{4}\eta^3 + O(\xi^4), & \mathcal{F}(\eta) &= 1 + \frac{\alpha}{4}\eta^2 + O(\eta^4).
\end{align*}
Thus we can compute
\begin{align*}
    v_\sigma = \frac{\alpha}{2}\eta = \frac{\alpha}{2}\xi\frac{\bar{t}}{t} = \frac{\alpha}{2}\mathcal{F}(\eta) = \frac{\alpha}{2}\xi\Big(1+\frac{\alpha}{4}\xi^2+\dots\Big),
\end{align*}
as claimed in (\ref{vform}). This completes the proof of Theorem \ref{SelfSimExpandxi}.\qed

\subsection{Proof of Theorem \ref{thmsigma}: Derivation of the STV-PDE}\label{Appendix1C}

By \cite{groate}, it suffices to show that (\ref{Axi-final})--(\ref{weqnxi-final}) are equivalent to (\ref{firstorder1}), (\ref{firstorder3}), (\ref{conservation_law00}) and (\ref{conservation_law01}). Neglecting bars, we first convert (\ref{firstorder1}) and (\ref{firstorder3}) to a system in $(t,\xi)$ to obtain (\ref{Axi-final}) and (\ref{Dxi-final}) respectively. Since we have:
\begin{align}
    f_r &= \frac{1}{t}f_\xi,\label{trans1}\\
    f_t(t,r) &= f_t(t,\xi) - \frac{\xi}{t}f_\xi(t,\xi),\label{trans2}
\end{align}
where $f_t(t,\xi)$ denotes the partial of $f$ with respect to $t$ holding $\xi$ fixed and $f_\xi(t,\xi)$ denotes the partial of $f$ with respect to $\xi$ holding $t$ fixed. Making these substitutions into (\ref{firstorder1}) and (\ref{firstorder3}) gives the equivalent equations:
\begin{align*}
    \xi A_\xi &= -z + (1-A), & \xi\frac{B_\xi}{B} &= \frac{1-A}{A} + \frac{1}{A}T^{11}_M,
    \end{align*}
respectively. To obtain the implied equation for $D=\sqrt{AB}$, write
\begin{align*}
    2\xi DD_\xi = \xi(AB)_\xi = \xi A_\xi B + \xi B_\xi A,
\end{align*}
so
\begin{align*}
    \xi D_\xi = \frac{\xi A_\xi B + \xi B_\xi A}{2D}
\end{align*}
and
\begin{align}
    \xi D_\xi &= \frac{1}{2D}\bigg((-z+(1-A))B+BA\frac{1-A}{A}+\frac{BA}{A}T^{11}_M\bigg)\notag\\
    &= \frac{B}{2D}\left(-z+2(1-A)+T^{11}_M\right)\notag\\
    &= \frac{D}{2A}\left(2(1-A)-z+T^{11}_M\right).\label{Dalmost}
\end{align}
Putting (\ref{sigmastress3}) into (\ref{Dalmost}) gives (\ref{Axi-final}). Reversing these steps verifies the equivalence of (\ref{Axi-final}) and (\ref{Dxi-final}) with (\ref{Axi-final}) and (\ref{Bxi-final}) and (\ref{firstorder1}) and (\ref{firstorder3}).

It remains now only to prove that (\ref{zeqnxi-final}) and (\ref{weqnxi-final}) are equivalent to (\ref{conservation_law00}) and (\ref{conservation_law01}) respectively, assuming (\ref{Axi-final}) and (\ref{Dxi-final}). To start, assume $c=1$ and multiply equations (\ref{conservation_law00}) and (\ref{conservation_law01}) through by $r^2$ to get:
\begin{align}
    \big(T^{00}_Mr^2\big)_t + r^2\big(\sqrt{AB}T^{01}_M\big)_r + 2r\sqrt{AB}T^{01}_M &= 0,\label{cl00}\\
    \big(T^{01}_Mr^2\big)_t + r^2\big(\sqrt{AB}T^{11}_M\big)_r + 2r\sqrt{AB}\big(T^{11}_M-T^{22}_Mr^2\big) &= -\frac{1}{2}r^2\sqrt{AB}\{\cdot\}_*,\label{cl01}
\end{align}
with
\begin{align*}
    \{\cdot\}_* = \bigg\{\frac{1}{r}\bigg(\frac{1}{A}-1\bigg)\big(T^{00}_M-T^{11}_M\big)+\frac{2\kappa r}{A}\big(T^{00}_MT^{11}_M-(T^{01}_M)^2\big)\bigg\}_*.
\end{align*}
Combining terms, (\ref{cl00}) becomes
\begin{align}
    \big(T^{00}_Mr^2\big)_t + \big(\sqrt{AB}T^{01}_Mr^2\big)_r = 0.\label{cl00a}
\end{align}
To achieve a similar simplification in (\ref{cl01}), add and subtract to get
\begin{align*}
    \big(T^{01}_Mr^2\big)_t + r^2\Big(\sqrt{AB}\big(T^{11}_M-T^{22}_Mr^2\big)\Big)_r + 2r\sqrt{AB}\big(T^{11}_M-T^{22}_Mr^2\big) + r^2\big(\sqrt{AB}T^{22}_Mr^2\big)_r = -\frac{1}{2}r^2\sqrt{AB}\{\cdot\}_*,
\end{align*}
so the second and third terms combine to give
\begin{align*}
    \big(T^{01}_Mr^2\big)_t + \Big(\sqrt{AB}\big(T^{11}_Mr^2-T^{22}_Mr^4\big)\Big)_r + r^2\big(\sqrt{AB}T^{22}_Mr^2\big)_r = -\frac{1}{2}r^2\sqrt{AB}\{\cdot\}_*,
\end{align*}
or equivalently
\begin{align}
    t\big(T^{01}_Mr^2\big)_t - \xi\big(T^{01}_Mr^2\big)_\xi + \Big(D\big(T^{11}_Mr^2-T^{22}_Mr^4\big)\Big)_\xi + r^2\big(DT^{22}_Mr^2\big)_\xi = -\frac{1}{2}tr^2D\{\cdot\}_*.\label{cl01b}
\end{align}
Then by (\ref{definez}) and (\ref{definew}), together with (\ref{sigmastress1})--(\ref{sigmastress4}), we have:
\begin{align*}
    \kappa T^{00}_Mr^2 &= \kappa\rho r^2\frac{1+\sigma^2v^2}{1-v^2} = z,\\
    \kappa T^{01}_Mr^2 &= \frac{1+\sigma^2}{1+\sigma^2v^2}vz = \bar{v}z, & \bar{v} &= \frac{1+\sigma^2}{1+\sigma^2v^2} = \xi zw,\\
    \kappa T^{11}_Mr^2 &= \frac{\sigma^2+v^2}{1+\sigma^2v^2}z,\\
    \kappa T^{22}r^2 &= \kappa\sigma^2\rho, & \kappa T^{22}r^4 &= \sigma^2\frac{1-v^2}{1+\sigma^2v^2}z,
\end{align*}
and
\begin{align*}
    \kappa T^{11}_Mr^2 - \kappa T^{22}r^4 &= \frac{\sigma^2+v^2}{1+\sigma^2v^2}z - \frac{\sigma^2-\sigma^2v^2}{1+\sigma^2v^2}z\\
    &= \frac{(1+\sigma^2)v^2}{1+\sigma^2v^2}z = \bar{v}vz\\
    &= \bar{v}^2z + \bar{v}(v-\bar{v})z\\
    &= \bar{v}^2\bigg(1-\sigma^2\frac{1-v^2}{1+\sigma^2}\bigg)z.
\end{align*}
Moreover,
\begin{align*}
    (\bar{v}z)_t + \bigg(\sqrt{AB}\bar{v}^2z\frac{\kappa T^{11}_Mr^2-\kappa T^{22}r^4}{\bar{v}^2z}\bigg)_r + r^2\big(\kappa\sqrt{AB}T^{22}_Mr^2\big)_r = -\frac{1}{2}\kappa r^2\sqrt{AB}\{\cdot\}_*.
\end{align*}
Using these in (\ref{cl00a}) and (\ref{cl01b}) yields the equivalent system:
\begin{align}
    z_t &+ \big(\sqrt{AB}\bar{v}z\big)_r = 0,\label{00-rt}\\
    (\bar{v}z)_t &+ \bigg(\sqrt{AB}\bar{v}^2\bigg(1-\sigma^2\frac{1-v^2}{1+\sigma^2}\bigg)z\bigg)_r + r^2\big(\kappa\sqrt{AB}T^{22}_Mr^2\big)_r = -\frac{1}{2}\kappa r^2\sqrt{AB}\{\cdot\}_*.\label{01-rt}
\end{align}
We take (\ref{00-rt}) as our final form for the $z$ equation in terms of $(t,r)$, but for (\ref{01-rt}), we use (\ref{00-rt}) to write the first term as
\begin{align*}
    (\bar{v}z)_t = z\bar{v}_t + \bar{v}z_t = z\bar{v}_t - \bar{v}(\sqrt{AB}\bar{v}z)_r
\end{align*}
and the second term in (\ref{01-rt}) as
\begin{align*}
    \bigg(\sqrt{AB}\bar{v}^2\bigg(1-\sigma^2\frac{1-v^2}{1+\sigma^2}\bigg)z\bigg)_r = \big(\sqrt{AB}\bar{v}^2z\big)_r - \sigma^2\bigg(\sqrt{AB}\bar{v}^2\bigg(\frac{1-v^2}{1+\sigma^2}\bigg)z\bigg)_r,
\end{align*}
so that they combine to form
\begin{align*}
    z\bar{v}_t + \sqrt{AB}\bar{v}z\bar{v}_r - \sigma^2\bigg(\sqrt{AB}\bar{v}^2\bigg(\frac{1+v^2}{1+\sigma^2}\bigg)z\bigg)_r.
\end{align*}
Substituting this into (\ref{01-rt}) yields
\begin{align*}
    z\bar{v}_t + \sqrt{AB}\bar{v}z\bar{v}_r - \sigma^2\bigg(\sqrt{AB}\bar{v}^2\bigg(\frac{1-v^2}{1+\sigma^2}\bigg)z\bigg)_r + r^2\big(\kappa\sqrt{AB}T^{22}_Mr^2\big)_r = -\frac{1}{2}\kappa r^2\sqrt{AB}\{\cdot\}_*.
\end{align*}
We conclude that in the case (\ref{psigmarho}), equations (\ref{00-rt}) and (\ref{01-rt}) are equivalent to:
\begin{align*}
    z_t &+ \big(\sqrt{AB}\bar{v}z\big)_r = 0,\\
    z\bar{v}_t &+ \sqrt{AB}\bar{v}z\bar{v}_r - \sigma^2\bigg(\sqrt{AB}\bar{v}^2\bigg(\frac{1-v^2}{1+\sigma^2}\bigg)z\bigg)_r + r^2\big(\kappa\sqrt{AB}T^{22}_Mr^2\big)_r = -\frac{1}{2}\kappa r^2\sqrt{AB}\{\cdot\}_*,
\end{align*}
and hence equivalent to (\ref{conservation_law00}) and (\ref{conservation_law01}).

We now convert (\ref{00-rt}) and (\ref{01-rt}) to a system in $(t,\xi)$. Substituting (\ref{trans1}) and (\ref{trans2}) into (\ref{00-rt}) directly gives (\ref{zeqnxi-final}). For (\ref{weqnxi-final}), we use (\ref{trans1}) and (\ref{trans2}) to obtain:
\begin{align}
    \bar{v}_r &= \frac{1}{t}(\xi w_\xi+w),\label{trans3}\\
    \bar{v}_t &= \frac{\xi}{t}(tw_t-\xi w_\xi-w),\label{trans4}
\end{align}
for $w(t,\xi)$. Substituting (\ref{trans3}) and (\ref{trans4}) into (\ref{01-rt}) gives
\begin{multline*}
    \frac{\xi z}{t}\Big(tw_t+(-1+\sqrt{AB}w)\xi w_\xi-w+\sqrt{AB}w^2\Big)\\
    -\frac{\sigma^2}{t}\bigg(\sqrt{AB}\bar{v}^2\bigg(\frac{1-v^2}{1+\sigma^2}\bigg)z\bigg)_\xi + \frac{r^2}{t}\big(\kappa\sqrt{AB}T^{22}_Mr^2\big)_\xi = -\frac{1}{2}\kappa r^2\sqrt{AB}\{\cdot\}_*,
\end{multline*}
which yields
\begin{multline}
    tw_t + (-1+\sqrt{AB}w)\xi w_\xi - w + \sqrt{AB}w^2\\
    -\frac{\sigma^2}{\xi z}\bigg(\sqrt{AB}\bar{v}^2\bigg(\frac{1-v^2}{1+\sigma^2}\bigg)z\bigg)_\xi + \frac{r^2}{\xi z}\big(\kappa\sqrt{AB}T^{22}_Mr^2\big)_\xi = -\frac{\kappa tr^2}{2\xi z}\sqrt{AB}\{\cdot\}_*,\label{01-xita}
\end{multline}
and where from (\ref{conservation_law01}) we have
\begin{align*}
    \{\cdot\}_* = \bigg\{\frac{1}{r}\bigg(\frac{1}{A}-1\bigg)\big(T^{00}_M-T^{11}_M\big)+\frac{2\kappa r}{A}\big(T^{00}_MT^{11}_M-(T^{01}_M)^2\big)\bigg\}_*.
\end{align*}
Now
\begin{align*}
    \frac{r^2}{\xi z}\big(\kappa\sqrt{AB}T^{22}_Mr^2\big)_\xi &= \frac{1}{t^2}\frac{r^2}{\xi z}\big(\kappa\sqrt{AB}T^{22}_Mr^2t^2\big)_\xi\\
    &= \frac{\xi}{z}\bigg(\sqrt{AB}\sigma^2\kappa \rho r^2\frac{1}{\xi^2}\bigg)_\xi\\
    &= \sigma^2\frac{\xi}{z}\bigg(\sqrt{AB}\frac{1-v^2}{1+\sigma^2v^2}\frac{z}{\xi^2}\bigg)_\xi,
\end{align*}
so putting this into (\ref{01-xita}) gives
\begin{multline}
    tw_t + (-1+\sqrt{AB}w)\xi w_\xi - w + \sqrt{AB}w^2 - \frac{\sigma^2}{\xi z}\bigg(\sqrt{AB}\bar{v}^2\bigg(\frac{1-v^2}{1+\sigma^2}\bigg)z\bigg)_\xi\\
    + \sigma^2\frac{\xi}{z}\bigg(\sqrt{AB}\frac{1-v^2}{1+\sigma^2v^2}\frac{z}{\xi^2}\bigg)_\xi = -\frac{\kappa tr^2}{2\xi z}\sqrt{AB}\{\cdot\}_*.\label{01-xit}
\end{multline}
Using:
\begin{align*}
    T^{00}_M - T^{11}_M &= (1-\sigma^2)\rho, & T^{00}_MT^{11}_M - (T^{01}_M)^2 &= \sigma^2\rho^2,
\end{align*}
we get
\begin{align*}
    \{\cdot\}_* = \bigg\{\frac{1}{r}\bigg(\frac{1}{A}-1\bigg)(1-\sigma^2)\rho+\frac{2\kappa r}{A}\sigma^2\rho^2\bigg\}_*,
\end{align*}
so
\begin{align*}
    \frac{\kappa tr^2}{\xi z}\{\cdot\}_* &= \bigg(\frac{1}{\xi^2}\bigg(\frac{1}{A}-1\bigg)(1-\sigma^2)\frac{\kappa\rho r^2}{z}
    +\frac{2\sigma^2}{\xi^2A}\frac{(\kappa\rho r^2)^2}{z}\bigg)\\
    &= \frac{1}{\xi^2}\frac{\kappa\rho r^2}{z}\bigg(\frac{1-A}{A}(1-\sigma^2)
    +\frac{2\sigma^2}{A}\kappa\rho r^2\bigg).
\end{align*}
Since
\begin{align*}
    \frac{\kappa\rho r^2}{z} = \frac{1-v^2}{1+\sigma^2v^2},
\end{align*}
we have
\begin{align*}
    \frac{\kappa tr^2}{\xi z}\{\cdot\}_* = \frac{1}{\xi^2}\frac{1-v^2}{1+\sigma^2v^2}\bigg(\frac{1-A}{A}(1-\sigma^2)+\frac{2\sigma^2}{A}\frac{1-v^2}{1+\sigma^2v^2}z\bigg).
\end{align*}
Thus the right hand side of equation (\ref{01-xit}) is
\begin{align*}
    -\frac{1}{2\xi^2}\frac{1-v^2}{1+\sigma^2v^2}\sqrt{AB}\bigg(\frac{1-A}{A}(1-\sigma^2)+\frac{2\sigma^2}{A}\frac{1-v^2}{1+\sigma^2v^2}z\bigg),
\end{align*}
which yields (\ref{weqnxi-final}).\qed

\subsection{Proof of Theorem \ref{ThmsmoothAgain}: Smoothness at the Center is Preserved Under Evolution}\label{Appendix1D}

We start with equations:
\begin{align}
    \xi A_\xi &= (1-A) - z,\label{A}\\
    \xi D_\xi &= \frac{D}{2A}\big(2(1-A)-z+\xi^2w^2z\big).\label{D}\\
    tz_t &= -\xi\big((-1+Dw)z\big)_\xi - Dwz,\label{z}\\
    tw_t &= -\xi(-1+Dw)w_\xi + w - D\bigg(w^2+\frac{1-A}{2\xi^2A}(1-\xi^2w^2)\bigg),\label{w}
\end{align}
First note that products and quotients of smooth functions that satisfy the condition that all odd derivatives vanish at $\xi=0$ also have this property. Now for a function $F(t,\xi)$, let $F^{(n)}_\xi(t)$ denote the $n^{th}$ partial derivative of $F$ with respect to $\xi$ at $\xi=0$. We prove this theorem by induction on $n$. For this, assume $n\geq1$ is odd and make the induction hypothesis that for all odd $k<n$, $F^{(k)}_\xi(t)=0$ for all $t\geq t_0$ and all functions $F=z,w,A,D$. We prove that $F^{(n)}_\xi(t)=0$ for $t>t_0$. For this, we employ the following simple observation: If the $n^{th}$ derivative of the product of $m$ functions
\begin{align*}
    \frac{\partial^n}{\partial\xi^n}(F_1\dots F_m)
\end{align*}
is expanded into a sum by the product rule, the only terms that will not have a factor containing an odd derivative of order less than $n$ are the terms in which all the derivatives fall on the same factor. This follows from the simple fact that if the sum of $k$ integers is odd, then one of them must be odd. Taking the $n^{th}$ derivative of (\ref{z}) and setting $\xi=0$ gives the ODE at $\xi=0$
\begin{align}
    t\frac{d}{dt}z^{(n)}_\xi = -n\frac{\partial^n}{\partial\xi^n}\big((-1+Dw)z\big) - \frac{\partial^n}{\partial\xi^n}(DWz).\label{zode}
\end{align}
Since all odd derivatives of order less than $n$ are assumed to vanish at $\xi=0$, we can apply the observation and the assumptions (\ref{ansatzassumption}), that $D=1$, $w=w_0(t)$ and $z=0$ at $\xi=0$, to see that only the $n^{th}$ order derivative $z^{(n)}_\xi$ survives on the right hand side of (\ref{zode}). That is, by the induction hypothesis, (\ref{zode}) reduces to
\begin{align}
    t\frac{d}{dt}z^{(n)}_\xi = \big(n-(n+1)w_0(t)\big)z^{(n)}_\xi.\label{zode1}
\end{align}
Since under the change of variable $t\to\ln(t)$, (\ref{zode1}) is a linear first order homogeneous ODE in $z^{(n)}_\xi(t)$ with $z^{(n)}_\xi(t_0)=0$, it follows by uniqueness of solutions that $z^{(n)}_\xi(t)=0$ for all $t\geq t_0$. This proves the theorem for the solution component $z(t,\xi)$.

Consider next equation (\ref{A}). Differentiating both sides $n$ times with respect to $\xi$ and setting $\xi=0$ gives
\begin{align*}
    (n+1)A^{(n)}_\xi(t) = -z^{(n)}_\xi(t) = 0
\end{align*}
for $t\geq t_0$, thus verifying 
\begin{align*}
    A^{(n)}_\xi(t) = 0
\end{align*}
for $t\geq t_0$, which verifies the theorem for component $A$.

Consider now equation (\ref{D}). Differentiating both sides $n$ times with respect to $\xi$, setting $\xi=0$ and applying the observation and the induction hypothesis yields
\begin{align*}
    nD^{(n)}_\xi &= \frac{\partial^n}{\partial\xi^n}\bigg(D\frac{1-A}{A}\bigg)\\
    &= D^{(n)}_\xi\frac{1-A}{A} + D\bigg(\frac{1-A}{A}\bigg)^{(n)}_\xi + \sum_{k<n\ \text{odd}}c_kD^{(k)}_\xi = 0
\end{align*}
for $t\geq t_0$ because $A=1$ at $\xi=0$, all lower order odd derivatives are assumed to vanish at $\xi=0$ and we have already verified the theorem for the component $A$. This proves
\begin{align*}
    D^{(n)}_\xi(t) = 0
\end{align*}
for $t\geq t_0$, verifying the theorem for component $D$.

Consider lastly equation (\ref{w}). Differentiating both sides $n$ times with respect to $\xi$, setting $\xi=0$ and applying our observation gives
\begin{align*}
    t\frac{d}{dt}w^{(n)}_\xi &= -n(-1+w_0(t))w^{(n)}_\xi + w^{(n)}_\xi - (w^2)^{(n)}_{\xi}\\
    &= -n(-1+w_0(t))w^{(n)}_\xi + w^{(n)}_\xi - 2ww^{(n)}_\xi\\
    &= (-n(-1+w_0(t))+1-2w)w^{(n)}_\xi
\end{align*}
for $t\geq t_0$ because all lower order odd derivatives are assumed to vanish at $\xi=0$. Thus $w^{(n)}_\xi(t)$ solves the first order homogeneous ODE
\begin{align*}
    t\frac{d}{dt}w^{(n)}_\xi &= (-n(-1+w_0(t))+1-2w)w^{(n)}_\xi
\end{align*}
starting from zero initial data at $t=t_0$, so again we conclude 
\begin{align*}
    w^{(n)}_\xi(t) = 0
\end{align*}
for $t\geq t_0$. This verifies the theorem for the final component $w$, thereby completing the proof of Theorem \ref{ThmsmoothAgain}.\qed

\subsection{Proof of Theorem \ref{ODEsForCorrections}: Derivation of the STV-ODE of Order 2}\label{Appendix1E}

We start with equations (\ref{z}) and (\ref{w}) in the form:
\begin{align}
    tz_t &= \xi z_\xi - \xi(Dwz)_\xi - Dwz\label{zagain}\\
    tw_t &= (1-Dw)(\xi w_\xi+w) - D\frac{1-A}{2\xi^2A}(1-\xi^2w^2).\label{wagain}
\end{align}
Applying (\ref{zansatz}) to terms in (\ref{zagain}) gives:
\begin{align}
    tz_t &= t\dot{z}_2\xi^2 + t\dot{z}_4\xi^4 + O(\xi^6),\label{zeq1}\\
    \xi z_\xi &= 2z_2\xi^2 + 4z_4\xi^4 + O(\xi^6),\label{zeq2}
\end{align}
and
\begin{align}
    Dwz = w_0z_2\xi^2 + (z_2w_2+w_0z_4+D_2z_2w_0)\xi^4 + O(\xi^6).\label{zeq3}
\end{align}
Thus putting ansatz (\ref{Aansatz})--(\ref{wansatz}) into (\ref{zagain}), using (\ref{zeq1})--(\ref{zeq3}) and collecting like powers of $\xi$ shows that the equations for the corrections close at every even power of $\xi$ and yield the following equations for the corrections at orders $\xi^2$ and $\xi^4$:
\begin{align}
    t\dot{z}_2 &= 2z_2 - 3z_2w_0,\label{z2eqn}\\
    t\dot{z}_4 &= -5D_2w_0z_2 + 4z_4 - 5w_0z_4 - 5w_2z_2.\label{z4eqn}
\end{align}

Applying the ansatz (\ref{Aansatz})--(\ref{wansatz}) to terms in equation (\ref{wagain}) gives:
\begin{align}
    tw_t &= t\dot{w}_0 + t\dot{w}_2\xi^2 + O(\xi^4),\label{weq1}\\
    1 - Dw &= (1-w_0) - (w_2+D_2w_0)\xi^2 + O(\xi^4),\label{weq2}\\
    \xi w_\xi + w &= w_0 + 3w_2\xi^2 + O(\xi^4),\label{weq3}\\
    \frac{1-A}{2\xi^2A} &= -\frac{1}{2}A_2 + \frac{1}{2}(A_2^2-A_4^2)\xi^2 + O(\xi^4),\label{weq5}\\
    D\frac{1-A}{2\xi^2A}(1-\xi^2w^2) &= -\frac{1}{2}A_2 + \frac{1}{2}\big(A_2w_0^2+A_2^2-A_4-A_2D_2\big)\xi^2 + O(\xi^4).\label{weq6}
\end{align}
Thus putting ansatz (\ref{Aansatz})--(\ref{wansatz}) into (\ref{wagain}), using (\ref{weq1})--(\ref{weq6}) and collecting like powers of $\xi$ shows that the equations for the corrections close at every even power of $\xi$, and this yields the following equations for the corrections at orders zero and $\xi^2$:
\begin{align}
    t\dot{w}_0 &= w_0 - w_0^2 + \frac{1}{2}A_2,\label{w0eqn}\\
    t\dot{w}_2 &= 3w_2 - 4w_0w_2 - D_2w_0^2 - \frac{1}{2}A_2w_0^2 - \frac{1}{2}A_2^2 + \frac{1}{2}A_4 + \frac{1}{2}A_2D_2.\label{w2eqn}
\end{align}

Substituting now the ansatz (\ref{Aansatz}) and (\ref{zansatz}) into (\ref{Axi-final0}) in the form
\begin{align*}
    z = -\big(\xi A_\xi+(A-1)\big)
\end{align*}
yields
\begin{align*}
    z_2\xi^2 + z_4\xi^4 + O(\xi^6) = -3A_2\xi^2 - 5A_4\xi^4 + O(\xi^6),
\end{align*}
so equating orders gives
\begin{align}
    A_2 &= -\frac{1}{3}z_2, & A_4 &= -\frac{1}{5}z_4.\label{A2A4}
\end{align}
Similarly, putting ansatz (\ref{Aansatz})--(\ref{wansatz}) into (\ref{Dxi-final0}) and keeping only $O(\xi^2)$ terms yields
\begin{align*}
    2D_2\xi^2 + O(\xi^4) = -\frac{1}{2}(2A_2+z_2)\xi^2 + O(\xi^4),
\end{align*}
so using (\ref{A2A4}) and equating orders gives
\begin{align}
    D_2 = -\frac{1}{12}z_2.\label{D2}
\end{align}
Finally, putting (\ref{A2A4}) and (\ref{D2}) into (\ref{z2eqn}) and (\ref{z4eqn}) and (\ref{w0eqn}) and (\ref{w2eqn}) yields (\ref{z2finalequation})--(\ref{w2finalequation}). This completes the proof of Theorem \ref{ODEsForCorrections}.\qed

\subsection{Proof of Theorem \ref{Friedmannknotzeroexpansion}: The Expansion of Friedmann in SSCNG Case $k=-1$}\label{Appendix1F}

Since the derivation of (\ref{z2FinalAgain})--(\ref{w2FinalAgain}) and (\ref{z2FinalAgaink1})--(\ref{w2FinalAgaink1}) is based on known formulas for the $k=\pm1$ Friedmann metric, for the proofs contained in this subsection and the next, we return to our earlier notation of letting unbarred coordinates $(t,r)$ denote Friedmann comoving coordinates, and barred coordinates $(\bar{t},\bar{r})$ denote SSC.

The solution of the $p=0$, $k=-1$ Friedmann equations (\ref{kFriedmannequation1})--(\ref{kFriedmannequation2}) is given by (\ref{kminusone})--(\ref{kminusoneR}) in Friedmann coordinates $(t,r)$. We write these formulas in the form:
\begin{align}
    \frac{t}{\Delta_0} &= \frac{1}{2}(\sinh2\theta-2\theta),\label{tagainkneg}\\
    \frac{R}{\Delta_0} &= \frac{1}{2}(\cosh2\theta-1),\label{Ragainkneg}
\end{align}
where by (\ref{DeltaDef}) the Friedmann free parameter $\Delta_0$ is taken to be
\begin{align*}
    \Delta_0 &= \frac{\kappa}{3}\rho R^3.
\end{align*}
To transform these coordinates over to SSCNG, we begin by finding a simple expression for $\bar{t}=h^{-1}(h(t)g(r))$ in (\ref{FtransFinal}), where $h$ and $g$ are given by (\ref{DefghThm}) to be:
\begin{align*}
    h(t) &= e^{\lambda\int_0^t\frac{d\tau}{\dot{R}(\tau)R(\tau)}}, & g(r) &= (1+r^2)^{\frac{\lambda}{2}},
\end{align*}
where $R(t)$ is defined implicitly by (\ref{kminusone}) and (\ref{kminusoneR}). To obtain expressions in terms of $\theta$, we calculate the following:
\begin{align}
    \frac{dt}{d\theta} &= \Delta_0(\cosh2\theta-1),\label{dtdtheta}\\
    \dot{R} &= \Delta_0\frac{d\theta}{dt}\sinh2\theta = \coth\theta,\label{tstep2}\\
    \frac{d\dot{R}}{d\theta} &= -\frac{2}{\cosh2\theta-1} = -\csch^2\theta,\\
    \ddot{R} &= \frac{d\dot{R}}{d\theta}\frac{d\theta}{dt} = -\frac{\csch^22\theta}{\Delta_0(\cosh2\theta-1)},\\
    \dot{H} &= \frac{\ddot{R}R-\dot{R}^2}{R^2} = -\frac{4(\sinh^22\theta+\cosh2\theta-1)}{\Delta_0^2(\cosh2\theta-1)^4}.\label{Hdot}
\end{align}

Note first that by (\ref{tstep2}), (\ref{tagainkneg}) and (\ref{Ragainkneg}), we have
\begin{align*}
    \dot{R}R = \frac{\Delta_0}{2}\sinh2\theta,
\end{align*}
and using this in (\ref{tstep2}) gives a formula for $\theta$ in terms of $t$, namely
\begin{align*}
    \theta = \frac{1}{\Delta_0}(R\dot{R}-t).
\end{align*}
Using this in (\ref{hoft2again}) gives
\begin{align*}
    \int_0^t\frac{ds}{\dot{R}R} = 2\int_0^t\frac{ds}{\Delta_0\sinh2\theta(s)}.
\end{align*}
Now let
\begin{align*}
    s = \frac{\Delta_0}{2}(\sinh2\theta-2\theta),
\end{align*}
so
\begin{align*}
    ds = \Delta_0(\cosh2\theta-1)d\theta,
\end{align*}
and substitution yields
\begin{align*}
    \int_0^t\frac{ds}{\dot{R}(s)R(s)} &= 2\int_{t=0}^t\frac{\cosh2\theta-1}{\sinh2\theta}d\theta = 2\int_{t=0}^t\frac{2\sinh^2\theta}{2\sinh\theta\cosh\theta}d\theta\\
    &= 2\int_{t=0}^t\frac{\sinh\theta}{\cosh\theta}d\theta = \Big[2\ln|\cosh\theta|\Big]_{t=0}^t = \cosh^2\theta.
\end{align*}
Using this in (\ref{hoft2again}) gives
\begin{align*}
    h(t) &= \Big[e^{\lambda\ln|\cosh^2\theta|}\Big]_{t=0}^t = \cosh^{2\lambda}\theta.
\end{align*}
Taking $\lambda=\frac{1}{2}$ now gives:
\begin{align*}
    h(t) &= \cosh\theta(t),\\
    g(r) &= \sqrt[4]{1+r^2},\\
    h^{-1}(y) &= \theta^{-1}\circ\cosh^{-1}(y).
\end{align*}
So since we started with (\ref{kminusone}) and $\Phi(t,r)=h(t)g(r)$, we have
\begin{align*}
    \bar{t} = h^{-1}(\Phi) = \theta^{-1}\Big(\cosh^{-1}\big(\sqrt[4]{1+r^2}\cosh\theta(t)\big)\Big),
\end{align*}
or equivalently
\begin{align*}
    \cosh\theta(\bar{t}) = \sqrt[4]{1+r^2}\cosh\theta(t).
\end{align*}
Thus in summary, $\bar{t}=F(h(t)g(r))$ with $F(y)=h^{-1}(y)$, so the coordinate transformation from $(t,r)\to(\bar{t},\bar{r})$ is:
\begin{align}
    \bar{t} &= \theta^{-1}\circ\cosh^{-1}\big(\sqrt[4]{1+r^2}\cosh\theta(t)\big),\label{transkequalminusone1-}\\
    \bar{r} &= R(t)r.\label{transkequalminusone2-}
\end{align}
Our goal is to derive (\ref{z2FinalAgain})--(\ref{w2FinalAgain}) by writing $A$, $z$ and $w$ as functions of $(\bar{t},\xi)$ for $0\leq\theta\leq\frac{\pi}{2}$ using the $k=-1$ Friedmann formulas and then determine the Taylor coefficients of the expansion in $\xi$.

We begin by deriving formulas for $A_2=-\frac{1}{3}z_2$ and $A_4=-\frac{1}{5}z_4$. Start with equation (\ref{Afirst}). We find $A_2(\bar{t})$ and $A_4(\bar{t})$ such that
\begin{align*}
    A(\bar{t},\xi) = A_2(\bar{t})\xi^2 + A_4(\bar{t})\xi^4 + O(\xi^6).
\end{align*}
Note that from (\ref{Afirst}) we have
\begin{align}
    A(\bar{t},\xi) = 1 - \frac{\Delta_0\bar{t}^2}{R(t)^3}\xi^2,\label{Adef}
\end{align}
where $t=t(\bar{t},\bar{r})$ needs to be expressed as a function of $(\bar{t},\bar{r})$. Since $t=\bar{t}+O(\xi^2)$, it follows immediately from (\ref{Adef}) that
\begin{align}
    A_2(\bar{t}) = -\frac{\Delta_0\bar{t}^2}{R^3(\bar{t})}.\label{A2tbar}
\end{align}
Now since $A_2$ is a function of $\bar{t}$ and $\bar{t}=t+O(\xi^2)$ to within the order we seek, we can write $A_2(\bar{t})$ as a function of $\theta$ at $\xi=0$ by identifying
\begin{align*}
    \bar{t} = t = \frac{\Delta_0}{2}(\sinh2\theta-2\theta)
\end{align*}
from (\ref{tagainkneg}). The result is
\begin{align}
    A_2(\bar{t}) = -\frac{2(\sinh2\theta-2\theta)^2}{(\cosh2\theta-1)^3},\label{A2theta}
\end{align}
which follows directly from (\ref{tagainkneg}) and (\ref{Ragainkneg}). This establishes (\ref{z2FinalAgain}).

We now find $A_4(\bar{t})$ in the case $k=-1$. Writing
\begin{align}
    A(\bar{t},\xi) = 1 - \frac{\Delta_0\bar{t}^2}{R^3(t(\bar{t},\xi))}\xi^2 =: 1 + a(\bar{t},t(\bar{t},\xi))\xi^2,\label{Aformulas}
\end{align}
together with
\begin{align*}
    A(\bar{t},\xi) = 1 + A_2(\bar{t})\xi^2 + A_4(\bar{t})\xi^4 + O(\xi^6),
\end{align*}
we derive $A_4(\bar{t})$ by Taylor expanding $a(\bar{t},t(\bar{t},\xi))$ in $\xi$, using a formula for
\begin{align*}
    t_{\xi\xi} := \frac{\partial^2}{\partial \xi^2}t(\bar{t},\xi),
\end{align*}
obtained implicitly from the transformation law (\ref{transkequalminusone1-})--(\ref{transkequalminusone2-}) by writing
\begin{align}
    \cosh\theta(\bar{t}) = \sqrt[4]{1+\frac{\bar{t}^2\xi^2}{R^2(t)}}\cosh\theta(t).\label{formulaFort-}
\end{align}
That is, we take (\ref{formulaFort-}) as the implicit definition for $t=t(\bar{t},\xi)$. Differentiating (\ref{Aformulas}) with respect to $\xi$ holding $\bar{t}$ fixed gives:
\begin{align*}
    A &= 1 + a\xi^2,\\
    A_\xi &= a_\xi\xi^2 + 2a\xi,\\
    A_{\xi\xi} &= a_{\xi\xi}\xi^2 + 4a_\xi\xi + 2a,\\
    A_{\xi\xi\xi} &= a_{\xi\xi\xi}\xi^2 + 6a_{\xi\xi}\xi + 6a_\xi,\\
    A_{\xi\xi\xi\xi} &= a_{\xi\xi\xi\xi}\xi^2 + 8a_{\xi\xi\xi}\xi + 12a_{\xi\xi}.
\end{align*}
Thus
\begin{align*}
    A_4 = \frac{1}{4!}A_{\xi\xi\xi\xi}(\bar{t},0) = \frac{1}{2}a_{\xi\xi}(\bar{t},0).
\end{align*}
Differentiating
\begin{align*}
    a(\bar{t},\xi) = -\frac{\Delta_0\bar{t}^2}{R^3(t)}
\end{align*}
with respect to $\xi$ holding $\bar{t}$ fixed yields:
\begin{align*}
    a_\xi &= 3\Delta_0\frac{\dot{R}(t)}{R^4(t)}t_\xi\bar{t}^2,\\
    a_{\xi\xi} &= 3\Delta_0\frac{\dot{R}(t)}{R^4(t)}t_{\xi\xi}\bar{t}^2 - 12\Delta_0\frac{\dot{R}(t)}{R^5(t)}t_\xi^2\bar{t}^2 + 3\Delta_0\frac{\dot{R}(t)}{R^4(t)}t_\xi^2\bar{t}^2.
\end{align*}
Setting $\xi=0$ and using $t_\xi(\bar{t},0)=0$ and $\bar{t}=t$ at $r=\xi=0$ gives
\begin{align*}
    a_{\xi\xi}(\bar{t},0) = 3\Delta_0\frac{\dot{R}(\bar{t})}{R^4(\bar{t})}t_{\xi\xi}(\bar{t},0)\bar{t}^2.
\end{align*}
It remains to find $t_{\xi\xi}$ and $\dot{R}(\bar{t})$ as functions of $(\bar{t},\xi)$ at $\xi=0$. First write (\ref{formulaFort-}) as
\begin{align}
    \cosh\theta(\bar{t}) = \sqrt[4]{b(\bar{t},\xi)}\cosh\theta(t)\label{bform}
\end{align}
with
\begin{align}
    b(\bar{t},\xi) = 1 + r^2 = 1 + \frac{\bar{t}^2}{R^2(t)}\xi^2.\label{bform1}
\end{align}
Differentiating (\ref{bform}) with respect to $\xi$ with $\bar{t}$ fixed gives
\begin{align*}
    \frac{1}{4}b^{-\frac{3}{4}}b_\xi\cosh\theta(t) + b^{\frac{1}{4}}t_\xi\theta'(t)\sinh\theta(t) = 0.
\end{align*}
Differentiating again, setting $\xi=0$ and using:
\begin{align*}
    t &= \bar{t} + O(\xi^2), & \dot{R}(t) &= \dot{R}(\bar{t}) + O(\xi),
\end{align*}
together with:
\begin{align*}
     b(\bar{t},0) &= 1, & b_{\xi} &:= b_{\xi}(\bar{t},0) = 0, & t_{\xi} := t_{\xi}(\bar{t},0) &= 0,
\end{align*}
(noting that we are changing definitions) gives
\begin{align*}
    \frac{1}{4}b_{\xi\xi}(\bar{t},0)\cosh\theta(\bar{t}) + t_{\xi\xi}(\bar{t},0)\theta'(\bar{t})\sinh\theta(\bar{t}) = 0.
\end{align*}
So from the definition of $b$ in (\ref{bform1}),
\begin{align*}
    b_{\xi\xi}(\bar{t},0) = \frac{2\bar{t}^2}{R^2(\bar{t})}.
\end{align*}
From this we conclude
\begin{align}
    t_{\xi\xi}(\bar{t},0) = -\frac{\cosh\theta(\bar{t})}{2\theta'(\bar{t})\sinh\theta(\bar{t})}\frac{\bar{t}^2}{R^2(\bar{t})}.\label{btwice}
\end{align}
By (\ref{tagainkneg}), since $\xi=0$, we have
\begin{align*}
    \theta'(\bar{t}) = \frac{1}{\Delta_0(\cosh2\theta(\bar{t})-1)},
\end{align*}
so putting this into (\ref{btwice}), and using $\theta=\theta(\bar{t})$, gives
\begin{align}
    t_{\xi\xi}(\bar{t},0) &= -\frac{\Delta_0(\cosh2\theta-1)\cosh\theta}{2\sinh\theta}\frac{\bar{t}^2}{R^2(\bar{t})}\label{btwice2-}\\
    &= -\frac{\Delta_0(\sinh2\theta-2\theta)^2\cosh\theta}{2(\cosh2\theta-1)\sinh\theta}.\label{btwice2}
\end{align}
Thus we have
\begin{align}
    A_4(\bar{t},0) = \frac{1}{2}a_{\xi\xi}(\bar{t},0) = \frac{3\Delta_0}{2}\frac{\dot{R}(\bar{t})}{R^4(\bar{t})}t_{\xi\xi}(\bar{t},0)\bar{t}^2.\label{A4start}
\end{align}
Note that (\ref{A4start}) holds for any $k=\pm1$ and will be used for the case $k=+1$ below as well. Using (\ref{A2tbar}) and (\ref{btwice2-}) in (\ref{A4start}) gives
\begin{align*}
    A_4(\bar{t},0) &= -\frac{3\Delta_0^2}{4}\bigg(\frac{\bar{t}^2}{R^3(\bar{t})}\bigg)^2\dot{R}(\bar{t})\frac{(\cosh2\theta-1)\cosh\theta}{\sinh\theta}\\
    &= -\frac{3}{4}A_2(\bar{t})^2\dot{R}(\bar{t})\frac{(\cosh2\theta-1)\cosh\theta}{\sinh\theta}.
\end{align*}
We summarize this as
\begin{align}
    A_4(\bar{t},0) &= -\frac{6(\sinh2\theta-2\theta)^4\cosh^2\theta}{(\cosh2\theta-1)^6}.\label{A4final}
\end{align}
Equation (\ref{A4final}) follows from (\ref{A2tbar}) and (\ref{tstep2}) because $A_4(\bar{t})$ is computed at $\xi=0$ where $\bar{t}=t$. This establishes (\ref{z4FinalAgain}).

We next obtain formulas for $w_0$ and $w_2$ by expanding $w=\frac{v}{\xi}$ in even powers of $\xi$. We begin with the formula (\ref{barv}) for $v$,
\begin{align}
    w = \frac{v}{\xi} = \frac{1}{\xi}\frac{\dot{R}(t)r}{\sqrt{1-kr^2}}.\label{wfirst}
\end{align}
Substituting $k=-1$, $\bar{r}=R(t)r$ and $\xi=\frac{\bar{r}}{\bar{t}}$ into (\ref{wfirst}) gives the exact formula
\begin{align}
    w = \frac{\dot{R}(t)\bar{t}}{R(t)}\frac{1}{\sqrt{1+\frac{\bar{t}^2}{R^2(t)}\xi^2}}.\label{w2}
\end{align}
Now we know $t=t(\bar{t},\xi)$ and $t$ agrees with $\bar{t}$ at $\xi=0$, so $t=t(\bar{t},0)$ with error order $O(\xi^2)$, so we again write
\begin{align*}
    t = \bar{t} + f(\bar{t})\xi^2 + O(\xi^4),
\end{align*}
where $2f(\bar{t})=t_{\xi\xi}(\bar{t},0)$ and, according to (\ref{btwice2}),
\begin{align*}
    t_{\xi\xi}(\bar{t},0) = -\frac{\Delta_0(\sinh2\theta-2\theta)^2\cosh\theta}{2(\cosh2\theta-1)\sinh\theta}.
\end{align*}
Thus
\begin{align*}
    w(\bar{t},\xi) &= H(t)\bar{t}\bigg(1-\frac{1}{2}\frac{\bar{t}^2}{R(t)^2}\xi^2\bigg) + O(\xi^4)\\
    &= H\big(\bar{t}+f(\bar{t})\xi^2\big)\bar{t}\bigg(1-\frac{1}{2}\frac{\bar{t}^2}{R(t)^2}\xi^2\bigg) + O(\xi^4)\\
    &= \bigg(H(\bar{t})+\frac{1}{2}H_{\xi\xi}\xi^2\bigg)\bar{t}\bigg(1-\frac{1}{2}\frac{\bar{t}^2}{R(t)^2}\xi^2\bigg) + O(\xi^4),\\
\end{align*}
where we have written $H=H(t(\bar{t},\xi))$, so
\begin{align*}
    H_\xi=\dot{H}t_\xi,
\end{align*}
and since $t_\xi=O(\xi)$ we have to leading even orders
\begin{align*}
    H_{\xi\xi}(t(\bar{t},0)) = H_{\xi\xi}(\bar{t}) = 2\dot{H}f(\bar{t}) = \dot{H}(\bar{t})t_{\xi\xi}(\bar{t},0).
\end{align*}
Continuing, we obtain
\begin{align*}
    w(\bar{t},\xi) &= \Big(H(\bar{t})+\dot{H}(\bar{t})t_{\xi\xi}(\bar{t},0)\xi^2\Big)\bar{t}\bigg(1-\frac{1}{2}\frac{\bar{t}^2}{R(\bar{t})^2}\xi^2\bigg) + O(\xi^4)\\
    &= H(\bar{t})\bar{t} + \bigg(\frac{1}{2}\dot{H}(\bar{t})t_{\xi\xi}(\bar{t},0)\bar{t} - \frac{1}{2}\frac{H(\bar{t})\bar{t}^3}{R(\bar{t})^2}\bigg)\xi^2 + O(\xi^4).
\end{align*}
From this we obtain the following exact formulas valid at $\xi=0$ for $k=-1$ and $k=+1$:
\begin{align}
    w_0 &= H(\bar{t})\bar{t}\label{w0first}\\
    w_2 &= \frac{1}{2}\dot{H}(\bar{t})t_{\xi\xi}(\bar{t},0)\bar{t} - \frac{H(\bar{t})\bar{t}^3}{2R(\bar{t})^2}.\label{w2first}
\end{align}
Using formulas (\ref{tagainkneg})--(\ref{Hdot}) together with $\bar{t}=t(\bar{t},0)$ we immediately obtain
\begin{align*}
    w_0 = \frac{(\sinh2\theta-2\theta)\sinh2\theta}{(\cosh2\theta-1)^2}.
\end{align*}
This establishes (\ref{w0FinalAgain}).

To express $w_2$ as a function of $\theta$, write (\ref{w2first}) as
\begin{align}
    w_2 = w_2^B - w_2^A,\label{wAB}
\end{align}
where
\begin{align}
    w_2^A &= \frac{H\bar{t}^3}{2R^2},\label{wA}\\
    w_2^B &= \frac{1}{2}\dot{H}t_{\xi\xi}\bar{t},\label{wB}
\end{align}
and where the functions $R$, $H$, $\dot{H}$ and $t_{\xi\xi}$ each take $\bar{t}$ as their only argument, since we evaluate at $\xi=0$. Substituting (\ref{tagainkneg}), (\ref{Ragainkneg}) and (\ref{Hdot}) into (\ref{wA}) and simplifying gives
\begin{align}
    w_2^A = \frac{(\sinh2\theta-2\theta)^3\cosh\theta}{(\cosh2\theta-1)^4}\sinh\theta,\label{w2Afinal}
\end{align}
which uses the expression
\begin{align*}
    \frac{\bar{t}^3}{R^2(\bar{t})} = \frac{\Delta_0(\sinh2\theta-2\theta)^3}{2(\cosh2\theta-1)^4}.
\end{align*}
To write $w_2^B$ as a function of $\theta$, substitute (\ref{tagainkneg}), (\ref{Ragainkneg}), (\ref{Hdot}) and (\ref{btwice2}) into (\ref{wB}) and simplify to obtain
\begin{align}
    w_2^B = \frac{(\sinh2\theta-2\theta)^3\cosh\theta}{2(\cosh2\theta-1)^5\sinh\theta}(\sinh^22\theta+\cosh2\theta-1).\label{w2Bfinal}
\end{align}
Putting formulas (\ref{w2Afinal}) and (\ref{w2Bfinal}) for $w_2^A$ and $w_2^B$ into (\ref{wAB}) establishes (\ref{w2FinalAgain}). This competes the proof of Theorem \ref{Friedmannknotzeroexpansion} in the case $k=-1$.\qed

\subsection{Proof of Theorem \ref{Friedmannknotzeroexpansion}: The Expansion of Friedmann in SSCNG Case $k=+1$}\label{Appendix1G}

We continue in this subsection with our original notation that unbarred coordinates $(t,r)$ denote Friedmann comoving coordinates and barred coordinates $(\bar{t},\bar{r})$ denote SSC. To establish (\ref{z2FinalAgaink1})--(\ref{w2FinalAgaink1}), we begin with formulas (\ref{kplusone})--(\ref{kplusoneR}) for the solutions of the Friedmann equations (\ref{kFriedmannequation1})--(\ref{kFriedmannequation2}), given in Friedmann coordinates $(t,r)$, for $p=0$, $k=+1$, written in the form:
\begin{align}
    \frac{t}{\Delta_0} &= \frac{1}{2}(2\theta-\sin2\theta),\label{tagaink1}\\
    \frac{R}{\Delta_0} &= \frac{1}{2}(1-\cos2\theta),\label{Ragaink1}
\end{align}
where again $\Delta_0$ is the Friedmann free parameter given by
\begin{align*}
    \Delta_0 = \frac{\kappa}{3}\rho R^3.
\end{align*}
As in the case $k=-1$, to transform these coordinates over to SSCNG, we begin by finding a simple expression for $\bar{t}=h^{-1}(h(t)g(r))$, where $f$ and $g$ are given in (\ref{DefghThm}) for case $k=+1$ by:
\begin{align}
    h(t) &= e^{\lambda\int_0^t\frac{d\tau}{\dot{R}(\tau)R(\tau)}},\label{hoftk1}\\
    g(r) &= (1-r^2)^{-\frac{\lambda}{2}}.
\end{align}
To obtain expressions in terms of $\theta$, we calculate the following:
\begin{align}
    \frac{dt}{d\theta} &= \Delta_0(1-\cos2\theta),\label{dtdthetak1}\\
    \dot{R} &= \frac{d}{dt}R(\theta(t)) = \frac{dR}{d\theta}\frac{d\theta}{dt} = \frac{\sin2\theta}{1-\cos2\theta} = 2\cot\theta,\label{tstepk1}\\
    \frac{d\dot{R}}{d\theta} &= -\frac{4\sin^2\theta}{(1-\cos2\theta)^2} = -\csc^2\theta,\\
    \ddot{R} &= \frac{d\dot{R}}{d\theta}\frac{d\theta}{dt} = -\frac{2}{\Delta_0(1-\cos2\theta)^2} = -\frac{1}{\Delta_0}\csc^4\theta,\label{Rddotk1}
\end{align}
and
\begin{align}
    \dot{H} &= \frac{\ddot{R}R-\dot{R}^2}{R^2} = -\frac{4(1-\cos2\theta+\sin^22\theta)}{\Delta_0(1-\cos2\theta)^4}.\label{Hdotk1}
\end{align}

Note first that by (\ref{Ragaink1}) and (\ref{tagaink1}) we have
\begin{align*}
    \dot{R}R = \frac{\Delta_0}{2}\sin2\theta,
\end{align*}
and using this in (\ref{tagaink1}) gives a formula for $\theta$ in terms of $t$, namely
\begin{align*}
    \theta = \frac{1}{\Delta_0}(R\dot{R}+t).
\end{align*}
Using this in (\ref{hoftk1}), we obtain
\begin{align*}
    \int_0^t\frac{ds}{\dot{R}(s)R(s)} = \int_0^t\frac{2ds}{\Delta_0\sin2\theta(s)}.
\end{align*}
Making the substitution:
\begin{align*}
    s &= \frac{\Delta_0}{2}(2\theta-\sin2\theta), & ds &= \Delta_0(1-\cos2\theta)d\theta,
\end{align*}
gives, noting (\ref{dtdthetak1}),
\begin{align*}
    \int_0^t\frac{ds}{\dot{R}(s)R(s)} = 2\int_{t=0}^t\frac{2(1-\cos2\theta)}{\sin2\theta}d\theta = 2\int_{t=0}^t\frac{\sin\theta}{\cos\theta}d\theta = -2\ln|\cos\theta(t)|.
\end{align*}
Using this in (\ref{hoftk1}) then gives
\begin{align*}
    h(t) = e^{-2\lambda\ln|\cos2\theta(t)|} = \sec^{2\lambda}\theta(t).
\end{align*}

Now taking $\lambda=\frac{1}{2}$ we obtain:
\begin{align}
    h(t) &= \sec\theta(t),\label{SepTrans1}\\
    g(r) &= \sqrt[4]{1+r^2},\label{SepTrans2}\\
    h^{-1}(y) &= \theta^{-1}\circ\sec^{-1}(y).\label{SepTrans3}
\end{align}
Thus by (\ref{SepTrans1})--(\ref{SepTrans3}) and $\Phi(t,r)=h(t)g(r)$, we have
\begin{align*}
    \bar{t} = h^{-1}(\Phi) = \theta^{-1}\Big(\sec^{-1}\big(\sqrt[4]{1+r^2}\sec\theta(t)\big)\Big)
\end{align*}
or
\begin{align}
    \sec\theta(\bar{t}) = \sqrt[4]{1+r^2}\sec\theta(t).\label{tbartotk1}
\end{align}
Thus in summary, $\bar{t}=F(h(t)g(r))$ with $F(y)=h^{-1}(y)$, so the coordinate transformation from $(t,r)\to(\bar{t},\bar{r})$ with normalized gauge is:
\begin{align}
    \bar{t} &= \theta^{-1}\circ\sec^{-1}\big(\sqrt[4]{1+r^2}\sec\theta(t)\big),\label{transkequalminusone1k1}\\
    \bar{r} &= R(t)r.\label{transkequalminusone2k1}
\end{align}
Note that the coordinate transformation (\ref{transkequalminusone1k1}) becomes singular at $\theta=\frac{\pi}{2}$, the point where the $k=+1$ Friedmann spacetime expands to its maximum. Thus for $0\leq\theta\leq\frac{\pi}{2}$, the transformation (\ref{transkequalminusone1k1})--(\ref{transkequalminusone2k1}) describes the $k=+1$ Friedmann spacetime in SSC from the Big Bang $\theta=t=0$, out to the maximum $\theta=\frac{\pi}{2}$, $\bar{t}=\Delta_0(\frac{\pi}{2}-1)$. Our goal is to write $A$, $z$ and $w$ as functions of $(\bar{t},\xi)$ for $0\leq\theta\leq\frac{\pi}{2}$ for the $k=+1$ Friedmann solution and confirm that the Taylor coefficients of the expansion in $\xi$ are given by (\ref{z2FinalAgaink1})--(\ref{w2FinalAgaink1}).

To derive $A_2(\bar{t})$ and $A_4(\bar{t})$ such that
\begin{align*}
    A(\bar{t},\xi) = A_2(\bar{t})\xi^2 + A_4(\bar{t})\xi^4 + O(\xi^6),
\end{align*}
we start with equation (\ref{Afirst}). Note that from (\ref{Afirst}) we have
\begin{align}
    A(\bar{t},\xi) = 1 - \frac{\Delta_0\bar{t}^2}{R(t)^3}\xi^2,\label{Adef2}
\end{align}
where again, $t=t(\bar{t},\bar{r})$ needs to be expressed as a function of $(\bar{t},\bar{r})$. Since $t=\bar{t}+O(\xi^2)$, it follows immediately from (\ref{Adef2}) that
\begin{align*}
    A_2(\bar{t}) = -\frac{\Delta_0\bar{t}^2}{R^3(\bar{t})}.
\end{align*}
Now since $A_2$ is a function of $\bar{t}$ and $\bar{t}=t+O(\xi^2)$ to within the order we seek, we can write $A_2(\bar{t})$ as a function of $\theta$ at $\xi=0$ by identifying
\begin{align*}
    \bar{t} = t = \frac{\Delta_0}{2}(2\theta-\sin2\theta)
\end{align*}
from (\ref{tagaink1}). The result is
\begin{align}
    A_2(\bar{t}) = -\frac{2(2\theta-\sin2\theta)^2}{(1-\cos2\theta)^3},\label{A2finalk1}
\end{align}
which follows directly from (\ref{tagaink1}) and (\ref{Ragaink1}). From (\ref{A2finalk1}) we obtain the limits:
\begin{align*}
    \lim_{\theta\to0}A_2 &= -\frac{4}{9}, & \lim_{\theta\to\frac{\pi}{2}}A_2 &= -\frac{\pi^2}{4}, & \lim_{\theta\to\pi^-}A_2 &= \infty.
\end{align*}
Equations (\ref{z2FinalAgaink1}) and the second limit in (\ref{limitsUhatplus}) follow from the identity $z_2=-3A_2$, see (\ref{A2A4}). We next find $w_0$. For this, we start with (\ref{w0first}), given by
\begin{align*}
    w_0 = H(\bar{t})\bar{t} = \frac{\dot{R}(\bar{t})}{R(\bar{t})}\bar{t}.
\end{align*}
Substituting the $k=+1$ formulas (\ref{tagaink1}), (\ref{Ragaink1}) and (\ref{tstepk1}) into the right hand side gives
\begin{align}
    w_0 = \frac{(2\theta-\sin2\theta)\sin2\theta}{(1-\cos2\theta)^2}.\label{w0fork1}
\end{align}
From (\ref{w0fork1}) we obtain the limits:
\begin{align*}
    \lim_{\theta\to0}w_0 &= \frac{2}{3}, & \lim_{\theta\to\frac{\pi}{2}}w_0 &= 0, & \lim_{\theta\to\pi^-}w_0 &= -\infty.
\end{align*}
This establishes the second limits in (\ref{limitsUhat}) and (\ref{limitsUhatplus}). To express $z_4$ and $w_2$ as functions of $\theta$ for the case $k=+1$, we need a formula for $t_{\xi\xi}:=t_{\xi\xi}(\bar{t},0)$. For this, we start with (\ref{transkequalminusone1k1}) in the form
\begin{align*}
    \sqrt[4]{b(\bar{t},\xi)}\sec\theta(\bar{t}) = \sec\theta(t),
\end{align*}
with
\begin{align}
    b(\bar{t},\xi) = 1 - \frac{\bar{t}^2}{R^2}\xi^2.\label{txixi2k1}
\end{align}
By (\ref{txixi2k1}),
\begin{align*}
    b_{\xi\xi} := b_{\xi\xi}(\bar{t},0) = -\frac{2\bar{t}^2}{R^2}.
\end{align*}
Differentiating (\ref{tbartotk1}) twice with respect to $\xi$, setting $\xi=0$ and using $b_\xi=t_\xi=0$ at $\xi=0$ gives
\begin{align*}
    \frac{1}{4}b_{\xi\xi} = t_{\xi\xi}\dot{\theta}\tan\theta,
\end{align*}
so by (\ref{tagaink1}), (\ref{Ragaink1}), (\ref{txixi2k1}) and (\ref{dtdthetak1}) we have
\begin{align}
    t_{\xi\xi} &= \frac{1}{4}b_{\xi\xi}\frac{dt}{d\theta}\cot\theta\notag\\
    &= \frac{\Delta_0}{4}(1-\cos2\theta)\left(-2\frac{\bar{t}^2}{R^2}\right)\frac{\cos\theta}{\sin\theta}\notag\\
    &= -\frac{\Delta_0(2\theta-\sin2\theta)^2\cos\theta}{2(1-\cos2\theta)\sin\theta}.\label{txixi3k1}
\end{align}
Consider now $A_4$ in the case $k=+1$. We start with (\ref{A4start}). Using (\ref{tagaink1}), (\ref{Ragaink1}), (\ref{txixi2k1}) and (\ref{tstepk1}) in (\ref{A4start}) gives
\begin{align}
    A_4 = \frac{3\Delta_0\bar{t}^2\dot{R}(\bar{t})t_{\xi\xi}}{2R^4(\bar{t})} = -\frac{6(2\theta-\sin2\theta)^4\cos^2\theta}{(1-\cos2\theta)^6}.\label{A4finalk1}
\end{align}
From (\ref{A4finalk1}) we obtain the limits:
\begin{align*}
    \lim_{\theta\to0}A_4 &= -\frac{8}{27}, & \lim_{\theta\to\frac{\pi}{2}}A_4 &= 0, & \lim_{\theta\to\pi^-}A_4 &= -\infty.
\end{align*}
Since $z_4=-5A_4$, this establishes the third limits in (\ref{limitsUhat}) and (\ref{limitsUhatplus}).

Finally, consider $w_2$ in the case $k=+1$. We start with (\ref{A4start}). Using (\ref{tagaink1}), (\ref{Ragaink1}), (\ref{txixi2k1}) and (\ref{tstepk1}) in (\ref{A4start}) gives
\begin{align}
    w_2 = w_2^B - w_2^A,\label{w2Aminusw2B}
\end{align}
where:
\begin{align}
    w_2^A &= \frac{H\bar{t}^3}{2R^2},\label{w2Afirstk1}\\
    w_2^B &= \frac{1}{2}\dot{H}t_{\xi\xi}\bar{t}.\label{w2Bfirstk1}
\end{align}
Substituting (\ref{tagaink1}), (\ref{Ragaink1}) and (\ref{tstepk1}) into (\ref{w2Afirstk1}) gives
\begin{align}
    w_2^A = \frac{\dot{R}\bar{t}^3}{2R^3} = \frac{(2\theta-\sin2\theta)^3\sin2\theta}{2(1-\cos2\theta)^4}\label{w2Ak1}
\end{align}
and substituting (\ref{tagaink1}), (\ref{Hdotk1}) and (\ref{txixi3k1}) into (\ref{w2Bfirstk1}) gives
\begin{align}
    w_2^B = \frac{\cos\theta(1-\cos2\theta+\sin^22\theta)(2\theta-\sin2\theta)^3}{2(1-\cos2\theta)^5\sin\theta}.\label{w2Bk1}
\end{align}
From (\ref{w2Ak1}) we obtain the limits:
\begin{align*}
    \lim_{\theta\to0}w_2^A &= 0, & \lim_{\theta\to\frac{\pi}{2}}w_2^A &= 0, & \lim_{\theta\to\pi^-}w_2^A &= -\infty,
\end{align*}
and from (\ref{A4finalk1}) we obtain the limits:
\begin{align*}
    \lim_{\theta\to0}w_2^B &= 0, & \lim_{\theta\to\frac{\pi}{2}}w_2^B &= 0, & \lim_{\theta\to\pi^-}w_2^B &= -\infty.
\end{align*}
Finally, putting (\ref{w2Ak1}) and (\ref{w2Bk1}) into (\ref{w2Aminusw2B}) gives the formula
\begin{align}
    w_2 = \frac{(2\theta-\sin2\theta)^3\cos\theta}{2(1-\cos2\theta)^5\sin\theta}\big(\sin^22\theta+(1-2\sin^2\theta)(1-\cos2\theta)\big).\label{w2finalk1}
\end{align}
From (\ref{w2finalk1}) we obtain the limits
\begin{align*}
    \lim_{\theta\to0}w_2 &= \frac{2}{9}, & \lim_{\theta\to\frac{\pi}{2}}w_2 &= 0, & \lim_{\theta\to\pi^-}w_2 &= -\infty.
\end{align*}
This establishes (\ref{w2FinalAgaink1}), as well as (\ref{limitsUhat}) and (\ref{limitsUhatplus}), and thus completes the proof of Theorem \ref{Friedmannknotzeroexpansion} and Corollary \ref{CorLimits} in the case $k>0$, thereby completing the proof of Theorem \ref{UnstableManifoldOfSM}.\qed

\subsection{Proof of Theorem \ref{Thmgeneralexpansion}: Derivation of the STV-ODE of Order $n$}\label{Appendix1H}

To show (\ref{zneqn}), we start with equation 
(\ref{zeqnxi-final0}) in the form
\begin{align*}
    tz_t = \xi z_\xi - \xi(zwD)_\xi - zwD
\end{align*}
and expand the terms separately. We first note that
\begin{align*}
    tz_t &= \sum_{n=1}^\infty\dot{z}_{2n}\xi^{2n}, & \xi z_\xi &= \sum_{n=0}^\infty2nz_{2n}\xi^{2n},
\end{align*}
and
\begin{align*}
    zwD = \sum_{n=1}^\infty\Bigg(\sum_{i+j+k=n}z_{2i}w_{2j}D_{2k}\Bigg)\xi^{2n}.
\end{align*}
Thus
\begin{align*}
    \xi(zwD)_\xi = \sum_{n=1}^\infty2n\Bigg(\sum_{i+j+k=n}z_{2i}w_{2j}D_{2k}\Bigg)\xi^{2n}
\end{align*}
and
\begin{align}
    \xi z_\xi - \xi(zwD)_\xi - zwD = \sum_{n=1}^\infty\Bigg(2nz_{2n}-(2n+1)\sum_{\substack{i+j+k=n\\i\neq n,\,j\neq n-1}}z_{2i}w_{2j}D_{2k}\Bigg)\xi^{2n}.\label{z2neqn}
\end{align}
Using this in (\ref{zeqnxi-final0}) gives (\ref{zneqn}). Extracting the leading order term first in equation (\ref{z2neqn}) we obtain
\begin{align}
    \sum_{n=1}^\infty t\dot{z}_{2n}\xi^{2n} &= \sum_{n=1}^\infty\Big(\big((2n+1)(1-w_0)-1\big)z_{2n}-(2n+1)z_2w_{2n-2}\Big)\xi^{2n}\notag\\
    &- \sum_{n=1}^\infty\Bigg((2n+1)\sum_{\substack{i+j+k=n\\i\neq n,\,j\neq n-1}}z_{2i}w_{2j}D_{2k}\Bigg)\xi^{2n},\label{leadinginz}
\end{align}
and so this confirms the first row of the matrix (\ref{Pk}).

To show (\ref{wneqn}), we write equation (\ref{weqnxi-final0}) in the form
\begin{align}
    tw_t = \frac{1}{2}(1-w^2\xi^2)\frac{A-1}{\xi^2A}D + (\xi w_\xi+w)(1-wD)\label{startw}
\end{align}
and again expand the terms separately. To expand the first term on the right hand side of (\ref{startw}), we write
\begin{align*}
    \frac{1}{2}(1-w^2\xi^2)\frac{A-1}{\xi^2A}D = \sum_{n=0}^\infty\bigg(\frac{1}{2}\hat{w}_{2i}a_{2j}D_{2k}\bigg)\xi^{2n},
\end{align*}
where, noting (\ref{defwhat}), we use
\begin{align*}
    1 - w^2\xi^2 &= 1 - w_0^2\xi^2 - 2w_0w_2\xi^2 -(2w_0w_4+w_2^2)\xi^6 + \dots\\
    &= 1 - \sum_{m=0}^\infty\sum_{i+j=m}w_{2i}w_{2j}\xi^{2m+2}\\
    &= \sum_{n=0}^\infty\hat{w}_n\xi^{2n},
\end{align*}
and, noting (\ref{defan}),
\begin{align*}
    \frac{A-1}{\xi^2A} &= \frac{1}{\xi^2}\sum_{n=1}^\infty(-1)^{n+1}(A-1)^n\\
    &= \sum_{n=1}^\infty(-1)^{n+1}\Bigg(\sum_{m=n}\Bigg(\sum_{i_1+i_2+\dotsc+i_n=m}A_{2i_1}A_{2i_2}\cdots A_{2i_n}\Bigg)\xi^{2m-2}\Bigg)\\
    &= A_2 + (A_4-A_2^2)\xi^2 + (A_6-2A_2A_4+A_2^3)\xi^4 + \dots\\
    &= \sum_{n=0}^\infty a_n\xi^{2n}.
\end{align*}
Note that the highest order term $A_{2k}$, which appears in $a_{2n}$, is $A_{2n+2}$. 

To expand the second term on the right hand side of (\ref{startw}), we write:
\begin{align*}
    \xi w_\xi + w &= \sum_{n=0}^\infty(2n+1)w_{2n}\xi^{2n},\\
    1 - wD &= 1 - \sum_{n=0}^\infty\Bigg(\sum_{j+k=n}w_{2j}D_{2k}\Bigg)\xi^{2n},
\end{align*}
so that
\begin{align*}
    (\xi w_\xi+w)(1-wD) = \sum_{n=0}^\infty\Bigg((2n+1)w_{2n}-\sum_{i+j+k=n}(2i+1)w_{2i}w_{2j}D_{2k}\Bigg)\xi^{2n},
\end{align*}
and in the case $n\neq0$, we obtain
\begin{align*}
    (\xi w_\xi+w)(1-wD) &= \sum_{n=0}^\infty\big((2n+1)w_{2n}-(2n+1)w_0w_{2n}-w_0w_{2n}\big)\xi^{2n}\\
    &- \sum_{n=0}^\infty\Bigg(\sum_{\substack{i+j+k=n\\i\neq n,\,j\neq n}}(2i+1)w_{2i}w_{2j}D_{2k}\Bigg)\xi^{2n}\\
    &= \sum_{n=0}^\infty\big((2n+2)(1-w_0)-1\big)w_{2n}\xi^{2n}\\
    &- \sum_{n=0}^\infty\Bigg(\sum_{\substack{i+j+k=n\\i\neq n,\,j\neq n}}(2i+1)w_{2i}w_{2j}D_{2k}\Bigg)\xi^{2n}.
\end{align*}
Thus
\begin{align*}
    \sum_{n=0}^\infty t\dot{w}_{2n}\xi^{2n} &= \sum_{n=0}^\infty\Bigg(\big((2n+2)(1-w_0)-1\big)w_{2n}+\sum_{i+j+k=n}\frac{1}{2}\hat{w}_{2i}a_{2j}D_{2k}\Bigg)\xi^{2n}\\
    &- \sum_{n=0}^\infty\Bigg(\sum_{\substack{i+j+k=n\\i\neq n,\,j\neq n}}(2i+1)w_{2i}w_{2j}D_{2k}\Bigg)\xi^{2n}
\end{align*}
or
\begin{align}
    \sum_{n=0}^\infty t\dot{w}_{2n}\xi^{2n} &= \sum_{n=0}^\infty\bigg(\big((2n+2)(1-w_0)-1\big)w_{2n}-\frac{1}{2(2n+3)}z_{2n+2}\bigg)\xi^{2n}\notag\\
    &+ \frac{1}{2}\sum_{n=0}^\infty\Bigg(a_{2n}-A_{2n+2}+\sum_{\substack{i+j+k=n\\j\neq n}}\hat{w}_{2i}a_{2j}D_{2k}\Bigg)\xi^{2n}\notag\\
    &- \sum_{n=0}^\infty\Bigg(\sum_{\substack{i+j+k=n\\i\neq n,\,j\neq n}}w_{2i}w_{2j}D_{2k}\Bigg)\xi^{2n},\label{finalwgood}
\end{align}
where we have used the fact that the leading order term in
\begin{align*}
    \sum_{i+j+k=n}\frac{1}{2}\hat{w}_{2i}a_{2j}D_{2k}
\end{align*}
is
\begin{align*}
    \frac{1}{2}\hat{w}_0a_{2n}D_0 = \frac{1}{2}a_{2n} - \frac{1}{2}A_{2n+2} + \frac{1}{2}A_{2n+2}.
\end{align*}
The first line in (\ref{finalwgood}) confirms (\ref{wneqn}) and the penultimate line displays the leading order terms, which establish the second row of the matrix $P_n$ in (\ref{Pk}). The next to leading order terms in (\ref{leadinginz}) and (\ref{finalwgood}) confirm that $P_n\boldsymbol{v}_n$ really does give the leading terms in (\ref{zneqn}) and (\ref{wneqn}), thereby confirming that $\boldsymbol{q}_n$ involves only lower order terms. This completes the proof of Theorem \ref{Thmgeneralexpansion}. 

To verify (\ref{Dntoznwn}), we start with the STV self-similar equation (\ref{D}) for $D$, which we write as
\begin{align}
    2\xi AD_\xi = D\big(2(1-A)-z+zw^2\xi^2\big).\label{Dproof1}
 \end{align}
We now substitute the ansatz (\ref{ansatzinfinityz})--(\ref{ansatzinfinityD}) into (\ref{Dproof1}) and collect like powers of $\xi$. First, we use:
\begin{align*}
    \xi D_\xi &= \sum_{j=0}^\infty 2jD_{2j}\xi^{2j}, & A &= \sum_{i=0}^\infty A_{2i}\xi^{2i}, & A_0 &= 1,
\end{align*}
to obtain
\begin{align}
    2\xi AD_\xi &= \sum_{n=0}^\infty\Bigg(\sum_{i+j=n}4jA_{2i}D_{2j}\Bigg)\xi^{2n}\notag\\
    &= 4D_2\xi^2 + \sum_{n=2}^\infty\Bigg(4nD_{2n}+\sum_{\substack{i+j=n\\j\neq n}}4jA_{2i}D_{2j})\Bigg)\xi^{2n}\label{Dproof2}\\
    &= 4D_2\xi^2 + (4A_2D_2+8D_4)\xi^4 + (4A_4D_2+8A_2D_4+12D_6)\xi^6 + \dots\notag
\end{align}
By (\ref{ansatzinfinityz}) and (\ref{ansatzinfinityw}) we have
\begin{align*}
    zw^2\xi^2 = \sum_{n=2}^\infty\Bigg(\sum_{i+j+k=n-1}z_{2i}w_{2j}w_{2k}\Bigg)\xi^{2n},
\end{align*}
and by (\ref{ansatzinfinityA}),
\begin{align*}
    2(1-A) = -\sum_{n=1}^\infty2A_{2n}\xi^{2n},
\end{align*}
so
\begin{align*}
    -z + 2(1-A) = -(z_2+2A_2)\xi^2 - \sum_{n=2}^\infty(z_{2n}+2A_{2n})\xi^{2n},
\end{align*}
and
\begin{align*}
    zw^2\xi^2 - z + 2(1-A) &= -(z_2+2A_2)\xi^2 + \sum_{n=2}^\infty\Bigg(\sum_{i+j+k=n-1}z_{2i}w_{2j}w_{2k}-(z_{2n}+2A_{2n})\Bigg)\xi^{2n}\\
    &= -(z_2+2A_2)\xi^2 + \big(z_2w_0^2-(z_4+2A_4)\big)\xi^4\\
    &+ \big(2z_2w_0w_2+z_4w_0^2-(z_6+2A_6)\big)\xi^6 + \dots
\end{align*}
Putting these together we have
\begin{align}
    D\big(zw^2\xi^2-z+2(1-A)\big) &= -(z_2+2A_2)\xi^2 + \sum_{n=2}^\infty\Bigg(\sum_{i+j+k+l=n-1}z_{2i}w_{2j}w_{2k}D_{2l}\Bigg)\xi^{2n}\notag\\
    &- \sum_{n=2}^\infty\Bigg(\sum_{\substack{i+j=n\\i\neq0}}(z_{2i}+2A_{2i})D_{2j}\Bigg)\xi^{2n}.\label{Dproof7}
\end{align}
Equating (\ref{Dproof2}) to (\ref{Dproof7}) and setting $n=1$ gives
\begin{align*}
    D_2 = -\frac{1}{4}(z_2+2A_2).
\end{align*}
Equating (\ref{Dproof2}) to (\ref{Dproof7}) for $n\geq2$ then gives
\begin{align}
    4nD_{2n} + \sum_{\substack{i+j=n\\j\neq n}}A_{2i}D_{2j} &= \sum_{i+j+k+l=n-1}z_{2i}w_{2j}w_{2k}D_{2l} - \sum_{\substack{i+j=n\\i\neq0}}(z_{2i}+2A_{2i})D_{2j}.\label{Dproof9}
\end{align}
Finally, solving (\ref{Dproof9}) for $D_{2n}$, using the fact that the leading order term $z_{2n}+2A_{2n}$ comes from the last sum in the case $i=n$ and $j=0$, we obtain, for $n\geq2$,
\begin{align*}
    4nD_{2n} &= -(z_{2n}+2A_{2n}) + \sum_{i+j+k+l=n-1}z_{2i}w_{2j}w_{2k}D_{2l}\\
    &- \sum_{\substack{i+j=n\\i\neq0,n}}(z_{2i}+2A_{2i})D_{2j} - \sum_{\substack{i+j=n\\j\neq n}}4jA_{2i}D_{2j},
\end{align*}
from which (\ref{Dntoznwn}) follows immediately.\qed

\section{Appendix: Comparison with Results in \cite{SmolTeVo}}\label{S12}

We have proven that solutions of the $n\times n$ STV-ODE which lie on the underdense side of the unstable manifold of the rest point $SM$ at order $n=1$ all tend to the rest point $M=(0,1,0,0,\dots)$ as $t\to\infty$ in the $n\times n$ system as well. This means the higher order corrections $\boldsymbol{v}_k=(z_{2k},w_{2k-2})\to0$ for all $k\geq2$, and hence the leading order approximation $\boldsymbol{v}_1=(z_2,w_0)$ becomes the dominant part of the solution for fixed $r$ as $t\to\infty$. It follows that such perturbations of $SM$ produce approximately uniform expanding spacetimes which appear more and more like Friedmann spacetimes at late times at each fixed $r$ but expand at an apparently accelerated rate during intermediate times. This case was first made in \cite{SmolTeVo} using the theory up to order $n=2$, but the expansions were centered on $SM$, while here we centered the expansions on $\boldsymbol{v}_k=0$. We conclude this paper by making the connection between the formulas in \cite{SmolTeVo} and the formulas established in this paper.

To make the connections with \cite{SmolTeVo} at order $n=2$ as clear as possible, in this section we adopt the notation of \cite{SmolTeVo}, by which $(t,r)$ denote SSCNG coordinates, and calculations are based on corrections to $SM$, instead of corrections to zero. That is, as in \cite{SmolTeVo}, we now let $\boldsymbol{U}=(z_2,w_0,z_4,w_0)$ denote corrections to the $k=0$ Friedmann solution based on solutions $(z,w)$ expanded about $\boldsymbol{U}_F=(\frac{4}{3},\frac{2}{3},\frac{40}{27},\frac{2}{9})$ according to:
\begin{align}
    z(t,\xi) &= \left(\frac{4}{3}+z_2(t)\right)\xi^2 + \left(\frac{40}{27}+z_4(t)\right)\xi^4 + O(\xi^6),\label{ansatzz}\\
    w(t,\xi) &= \left(\frac{2}{3}+w_0(t)\right) + \left(\frac{2}{9}+w_2(t)\right)\xi^2 + O(\xi^4),\label{ansatzw}\\
    A(t,\xi) &= \left(-\frac{4}{9}+A_2(t)\right)\xi^2 + \left(-\frac{8}{27}+A_4(t)\right)\xi^4 + O(\xi^6),\label{ansatzA}\\
    D(t,\xi) &= \left(-\frac{1}{9}+D_2(t)\right)\xi^2 + O(\xi^4).\label{ansatzD}
\end{align}
So, to be clear, this is accomplished by defining $\bar{\boldsymbol{U}}=\boldsymbol{U}-\boldsymbol{U}_F$ in terms of $\boldsymbol{U}$ as defined before in (\ref{zF}) and (\ref{wFk}) and then dropping the bars. In this notation, $\boldsymbol{U}_0=(0,0,0,0)$ now corresponds the rest point $SM$ of the $k=0$ Friedmann solution, but to keep notation to a minimum, we continue to use $\boldsymbol{U}_F$ to record the values $\boldsymbol{U}_F=(\frac{4}{3},\frac{2}{3},\frac{40}{27},\frac{2}{9})$. Since the change $\bar{\boldsymbol{U}}\to\boldsymbol{U}$ is a simple translation of the unknowns, relations (\ref{metricfluidrelations}), that is:
\begin{align*}
    A_2 &= -\frac{1}{3}z_2, & A_4 &= -\frac{1}{5}z_4, & D_2 &= -\frac{1}{12}z_2,
\end{align*}
continue to hold with $z_i$, $w_i$, $A_i$ and $D_i$ defined in (\ref{ansatzz})--(\ref{ansatzD}), and the new system in $\bar{\boldsymbol{U}}\to\boldsymbol{U}$ has the same rest points $SM$, $M$ and $U$ as (\ref{4by4system}) up to translation by $\boldsymbol{U}_F$, and rest points have the same Jacobians and eigenpairs as system (\ref{4by4system}). The only change from (\ref{4by4system}) is the constants which appear in the equations, which we now record to achieve correspondence with the equations and results stated in \cite{SmolTeVo}.
 
\begin{Thm}
    Using (\ref{ansatzz})--(\ref{ansatzD}) to define the corrections $(z_2,w_0,z_4,w_2)$ to $SM$, and letting $t$ denote SSC time, the $4\times4$ system of equations for the corrections is given by:
    \begin{align}
        t\frac{dz_2}{dt} &= -4w_0 - 3z_2w_0,\label{z2finalsumeqn1}\\
        t\frac{dw_0}{dt} &= -\frac{1}{6}z_2 - \frac{1}{3}w_0 - w_0^2,\label{w0finalsumeqn1}\\
        t\frac{dz_4}{dt} &= -\frac{10}{27}z_2 - \frac{20}{3}w_0 + \frac{5}{18}z_2^2 + \frac{10}{9}z_2w_0 + \frac{5}{12}w_0z_2^2\notag\\
        &+ \frac{2}{3}z_4 - \frac{20}{3}w_2 - 5w_0z_4 - 5z_2w_2,\label{z4finalsumeqn1}\\
        t\frac{dw_2}{dt} &= -\frac{4}{9}w_0 - \frac{1}{24}z_2^2 + \frac{1}{3}z_2w_0 + \frac{1}{3}w_0^2 + \frac{1}{4}w_0^2z_2\notag\\
        &- \frac{1}{10}z_4 + \frac{1}{3}w_2 - 4w_0w_2.\label{w2finalsumeqn1}
    \end{align}
    Moreover, the rest points of system (\ref{z2finalsumeqn1})--(\ref{w2finalsumeqn1}) are:
    \begin{align}
        SM &= (0,0,0,0), & M &= \left(-\frac{4}{3},-\frac{2}{3},-\frac{40}{27},-\frac{2}{9}\right), & U &= -\boldsymbol{U}_F,
    \end{align}
    with unchanged Jacobians and eigenpairs given by (\ref{JacSM})--(\ref{epr34U}).
\end{Thm}

This corresponds to equations (3.31)--(3.34) in \cite{SmolTeVo}, except for one correction, which we record here. Equation (3.33) in \cite{SmolTeVo} is missing the term $\frac{2}{3}z_4$, which appears in (\ref{z4finalsumeqn1}). This error did not significantly effect further calculations in \cite{SmolTeVo} because the third order term in redshift vs luminosity used only $w_2$, not $z_4$, and this missing term in the $z_4$ equation did not significantly effect the right hand side of (3.34) for the numerics performed in \cite{SmolTeVo}.

\section{Appendix: Lemaître (1932) Tolman (1933) Bondi (1945) Spacetimes}\label{Appendix2}

In the papers [61*]-[65*] from Matt Visser's paper, the references suggested by the editors at RSPA, we see that inhomogeneous cosmologies modeled by the LTB spacetimes were employed in an attempt to model the anomalous acceleration without dark energy. A main issue addressed by these papers is the existence of \emph{central weak singularities} at $r=0$. To quote [65*]=\cite{roma} Romano, 5th paragraph:

\emph{In this paper [we] calculate the low redshift expansion of mn(z) and DL(z) for flat $\Lambda$CDM and matter dominated LTB. We then show how, if the conditions to avoid a central weak singularity are imposed, it is impossible to mimic dark energy with a LTB model without cosmological constant for both these observables, giving a general proof of the impossibility to give a local solution of the inversion problem for a smooth LTB model. This central singularity is rather mild, and is associated to linear terms in the energy density which lead to a divergence of the second derivative, so non smooth LTB models could still be viable cosmological models. It can be shown that the inversion problem [34] can be solved if the smoothness conditions we are imposing are relaxed. This implies that the numerical solutions of the inversion problem which have been recently proposed [11, 20] must contain such a weak central singularity.}

Thus the singularity \emph{is associated to linear terms in the energy density which lead to a divergence of the second derivative}. Now our solutions in SSC constructed within our asymptotic ansatz are singularity free. That is, the density $\rho$ and metric entries $A$ and $B$ invoke only even powers of $\xi$ (and hence $\bar{r}$), so these are smooth at $\bar{r}=0$. Moreover, $v$ is odd in $\xi$ and $r$, and since it is a derivative, this is the condition that $v=\frac{dx}{dt}$ is smooth in $x$ for $r=|x|$ at $x=0$. Thus all components of our solution are smooth and it gives the quadratic correction to redshift vs luminosity in line with dark energy, so it appears to contradict the statement above by Romano who claims that all such solutions must have a weak singularity at $r=0$ if expressed in LTB coordinates. This begs the following questions:

(1) Can every solution in SSC be expressed in LTB? That is, could our solution not be an LTB solution? We answer this in the negative by proving below that every SSC metric can be transformed to LTB when $p=0$.

(2) Could the transformation from SSC to LTB introduce a weak singularity at $r=0$ in LTB when no such singularity exists in SSC? We argue that there is a mechanism for this.

In regard to (2), we first prove that a function $\rho(t,r)$, with $r\geq0$, extends to a smooth function $f(t,x)=\rho(t,|x|)$, with $-\epsilon<x<\epsilon$, if an only if $\rho(t,r)$ is smooth and all odd derivatives vanish at $r=0$. Thus, for example, if the Taylor expansion of $\rho$ in powers of $r$ about $r=0$ contains an odd power term $f(t)r^n$, then $\frac{\partial^n\rho}{\partial r^n}=n!f(t)\neq0$ implies $\rho$ has a kink in the $(n-1)$ derivative at $r=0$.

In paper [65] it is argued that $p=0$ solutions constructed in Lemaître--Tolman--Bondi (LTB) coordinates that can account for the anomalous acceleration near the center also exhibit a \emph{central weak singularity} in the second derivative of the (scalar) density at $r=0$. This appears to be inconsistent with the fact that our solutions in SSC, including the density, are smooth with no singularity at the center. We show how our work here clarifies this issue, and there is actually no inconsistency, due essentially to the fact that spherical coordinates do not form a regular coordinate system at $r=0$.

To this end, recall that polar coordinates for $\boldsymbol{x}=(x^1,x^2,x^3)\in\mathbb{R}^3$ take the radial coordinate to be $r=|x|$, and a function given by $f(r)$, with $r\geq0$, represents a smooth spherically symmetric function of $\boldsymbol{x}$ precisely when $f$ is smooth and satisfies the condition that all odd derivatives of $f$ vanish at the origin $r=0$. That is, a function $f(r)$ represents a smooth spherically symmetric function of the Euclidean coordinates $\boldsymbol{x}$ at $r=0$ if and only if the function
\begin{align*}
    g(\boldsymbol{x}) = f(|\boldsymbol{x}|)
\end{align*}
is smooth at $\boldsymbol{x}=0$. Assuming $f$ is smooth for $r\geq0$, taking the $n^{th}$ derivative of $g$ from the left and right and setting them equal gives the smoothness condition $f^n(0)=(-1)^nf^n(0)$. Thus $f(r)$ represents a smooth function of the underlying coordinates $\boldsymbol{x}$ if and only if $f$ is smooth for $r\geq0$ and all odd derivatives vanish at $r=0$. Moreover, if any odd derivative $f^{(n+1)}(0)\neq0$, then $f(|\boldsymbol{x}|)$ has a jump discontinuity in its $(n+1)$ derivative, and hence a kink singularity in its $n^{th}$ derivative at $r=0$. Similarly, a spherically symmetric function $f(t,r)$ on a four-dimensional spacetime in spherical coordinates $(t,r,\phi,\theta)$ will represent a smooth function of the underlying Euclidean coordinates at $r=0$ if and only if $f(t,|\boldsymbol{x}|)$ is a smooth function at $\boldsymbol{x}=0$. In particular, if the Taylor expansion of $f(r)$ about $r=0$ contains a nonzero odd power of order $n+1$, so that $f^{n+1}(0)\neq0$, then the function has a kink singularity in its $n^{th}$ derivative at the origin in those coordinates. But since $r=0$ is a singular point of spherical coordinates, this may only be an \emph{apparent} coordinate singularity.

To characterize the problem for LTB coordinates, consider now a coordinate transformation that takes a $p=0$ metric from LTB coordinates $(\hat{t},\hat{r})$ over to SSC given by
\begin{align*}
    t &= t(\hat{t},\hat{r}), & \bar{r} &= \bar{r}(\hat{t},\hat{r}).
\end{align*}
Then by definition, the fluid is comoving with respect to $\hat{r}$, constant $\hat{r}$ are geodesics, $\hat{t}$ is proper time along constant $\hat{r}$ and $\bar{r}$ is arc-length distance along radial directions at constant $t$ \cite{gronhe}.

The following theorem characterizes when a smooth scalar density function $\rho(t,\bar{r})$ in SSC has a kink singularity in its second derivative at $\hat{r}=0$ when represented in LTB coordinates. The theorem is a direct consequence of the following lemma.

\begin{Lemma}\label{Lemmakink}
    Assume that $\rho(t,\bar{r})$ is a scalar density function which extends to a smooth function $\rho(t,|\boldsymbol{x}|)$ in SSC, so that it is given near $\bar{r}=0$ by
    \begin{align}
        \rho(t,\bar{r}) = f_0(t) + f_2(t)\bar{r}^2 + \dots,\label{rhosmooth}
    \end{align}
    where the dots indicate that the expansion includes only even powers of $\bar{r}$. Assume further that the mapping $(t,\bar{r})\to(\hat{t},\hat{r})$ from SSC to LTB coordinates is smooth, invertible on $r\geq0$ and meets the minimal regularity conditions that all derivatives of $\frac{\partial t}{\partial\hat{r}}(\hat{t},\hat{r})$ up to order three have continuous one-sided limits at $\hat{r}=0$, together with
    \begin{align}
        \lim_{\hat{r}\to0}\bar{r}(\hat{t},\hat{r}) = \bar{r}(\hat{t},0) = 0\label{reg1}
    \end{align}
    and
    \begin{align}
        \lim_{\hat{r}\to0}\frac{\partial t}{\partial\hat{r}}(\hat{t},\hat{r}) = \frac{\partial t}{\partial\hat{r}}(\hat{t},0) = 0.\label{reg2}
    \end{align}
    Finally, let
    \begin{align*}
        \hat{\rho}(\hat{t},\hat{r}) = \rho(t(\hat{t},\hat{r}),\bar{r}(\hat{t},\hat{r}))
    \end{align*}
    denote the representation of the function $\rho(t,\bar{r})$ in LTB coordinates. Then among odd order derivatives, the first partial derivative of $\hat{\rho}$ with respect to $\hat{r}$ always vanishes at $(\hat{t},0)$, but the third partial derivative of $\hat{\rho}$ with respect to $\hat{r}$ at $(\hat{t},0)$ is given by
    \begin{align}
        \frac{\partial^3\hat{\rho}}{\partial\hat{r}^3} = \frac{\partial\rho}{\partial t}\frac{\partial^3t}{\partial\hat{r}^3} + 3\frac{\partial^2\rho}{\partial\bar{r}^2}\frac{\partial\bar{r}}{\partial\hat{r}}\frac{\partial^2\bar{r}}{\partial\hat{r}^2}.\label{kinkcondtn}
    \end{align}
\end{Lemma}

\begin{proof}
To verify (\ref{kinkcondtn}), compute the partial derivatives of $\hat{\rho}$ with respect to $\hat{r}$ as so:
\begin{align}
    \frac{\partial\hat{\rho}}{\partial\hat{r}}(\hat{t},\hat{r}) &= \frac{\partial\rho}{\partial t}\frac{\partial t}{\partial\hat{r}} + \frac{\partial\rho}{\partial\bar{r}}\frac{\partial\bar{r}}{\partial\hat{r}},\label{firstterm}\\
    \frac{\partial^2\hat{\rho}}{\partial\hat{r}^2}(\hat{t},\hat{r}) &= \frac{\partial^2\rho}{\partial t^2}\bigg(\frac{\partial t}{\partial\hat{r}}\bigg)^2 + \frac{\partial^2\rho}{\partial t\partial\bar{r}}\frac{\partial\bar{r}}{\partial\hat{r}}\frac{\partial t}{\partial\hat{r}} + \frac{\partial\rho}{\partial t}\frac{\partial^2t}{\partial\hat{r}^2}\notag\\
    &+ \frac{\partial^2\rho}{\partial t\partial\bar{r}}\bigg(\frac{\partial\bar{r}}{\partial\hat{r}}\bigg)^2 + \frac{\partial^2\rho}{\partial\bar{r}^2}\bigg(\frac{\partial\bar{r}}{\partial\hat{r}}\bigg)^2 + \frac{\partial\rho}{\partial\bar{r}}\frac{\partial^2\bar{r}}{\partial\hat{r}^2}.\label{secondterm}
\end{align}
Now $\frac{\partial\rho}{\partial\bar{r}}=0$ at $(t,0)$ by (\ref{rhosmooth}) and $\frac{\partial t}{\partial\hat{r}}=0$ at $(\hat{t},0)$ by (\ref{reg2}), so these in (\ref{firstterm}) imply $\frac{\partial\hat{\rho}}{\partial\hat{r}}(\hat{t},\hat{r})=0$ at $(\hat{t},0)$ as claimed. For the third derivative, use (\ref{reg2}) together with the fact that by (\ref{rhosmooth}), all partial derivatives of $\rho(t,\bar{r})$ that are odd order in $\bar{r}$ vanish at $\bar{r}=0$. It is then straightforward to see that the only terms that survive under differentiation of (\ref{secondterm}) with respect to $\hat{r}$ upon setting $\bar{r}=\hat{r}=0$ are given by the right hand side of (\ref{kinkcondtn}).
\end{proof}

We conclude the condition that the third derivative (\ref{kinkcondtn}) be nonzero is necessary and sufficient for a density function $\rho$, smooth in SSC, to have a nonzero third order derivative with respect to $\hat{r}$ in LTB at $\hat{r}=0$, and hence is necessary and sufficient for the second $\hat{r}$-derivative of $\rho$ to have a kink singularity in LTB at $\hat{r}=0$.

\begin{Thm}\label{Thmkink}
    Assuming (\ref{rhosmooth})--(\ref{reg2}), the right hand side of (\ref{kinkcondtn}) is nonzero at $(\hat{t},0)$ if and only if the density function $\hat{\rho}(\hat{t},\hat{r})$ has a kink singularity in its second derivative at the point $(\hat{t},0)$ in the sense that the function $\hat{\rho}(\hat{t},|\boldsymbol{x}|)$ has a jump discontinuity in its second derivative in $\hat{r}$ at $\hat{r}=0$.
\end{Thm}

Note that since SSC is the coordinate system in which $\bar{r}$ is arc-length distance along radial curves at constant $t$, the condition of smoothness of $\rho$ in $\bar{r}$ at $\bar{r}=0$ is geometric smoothness, so the kink singularity in LTB should be treated as a coordinate singularity.

All of this raises an important and interesting unresolved issue with our work. Namely, our asymptotic ansatz is saying no more and no less than that the solution is smooth at $\bar{r}=0$. But our equations only \emph{close up} when we impose our gauge condition relating metric coefficients in the expansion to fluid coefficients. So assuming we can solve for solutions within the asymptotics, we have to wonder whether there are other solutions which are smooth but for which that ansatz does not close. However, there must be, because if we take one of our solutions and change the gauge by a transformation of $t$, it still will not change the even and odd powers in $r$, so the solutions will stay smooth in the new gauge, but in the new gauge, it will not solve our asymptotic equations. On the other hand, it will still solve the exact equations because they are gauge invariant. The question then is, are there smooth solutions of the Einstein field equations, expressible in even powers of $r$ and $\xi$, which are not gauge transformations of solutions expressible in our ansatz?

We can resolve this as follows. Every solution of the Einstein field equations (that admits a Taylor expansion at $\bar{r}=0$) in SSC has to have even powers of $\bar{r}$ and hence even powers of $\xi$ in that expansion. But it will not be in our gauge because we may have $B(t,0)\neq1$. However, a gauge transformation will transform the solution into our gauge. Now plugging into the Einstein field equations, our gauge condition $A_2=-\frac{1}{12}z_2$ and so on must hold in order for the Einstein field equations to hold. Not for the ansatz to close, but for the Einstein field equations themselves to hold on that solution. Thus, every solution that has a Taylor expansion at $\bar{r}=0$ and is smooth there must be a gauge transformation of one of our solutions.

To conclude, every spherically symmetric solution of the Einstein field equations that is smooth at $\bar{r}=0$ in SSC, is one of our solutions, and hence the solution space admits the phase portrait that we introduced. In particular, $SM$ is an unstable solution in that space.

Finally, since computing $\hat{t}(t,\hat{r})$ entails integrating arc-length along the geodesic particle path $\hat{r}=0$, the resulting formulas would not in general close up under even powers, because applying this integration to quadratics would result in cubics and so on.

It now remains to verify (1). We start with the following theorem.

\begin{Thm}\label{thmLTB}
    Assume that $r=r_0$ for constant $r_0$ are geodesics and
    \begin{align*}
        ds^2 = g_{ij}dx^idx^j = -B(t,r)dt^2 + \frac{1}{A(t,r)}dr^2 + r^2d\Omega^2
    \end{align*}
    is diagonal with $x^0=t$ and $x^1=r$. Then $B$ depends only on $t$.
\end{Thm}

\begin{proof}
The geodesic equation is
\begin{align*}
    \ddot{x}^i = \Gamma^i_{jk}\dot{x}^i\dot{x}^j.
\end{align*}
Thus for an $r=x^1=r_0$ (with $r_0$ constant) geodesic, we have $\dot{x}^0\neq0$ and $\dot{x}^1=0$, so
\begin{align*}
    \ddot{x}^1 = \Gamma^1_{00}\dot{x}^0\dot{x}^0 = 0.
\end{align*}
However,
\begin{align*}
    0 = \Gamma^1_{00}\dot{x}^0\dot{x}^0 = \frac{1}{2}g^{1\sigma}(-g_{00,\sigma}+2g_{\sigma0,0}) = \frac{1}{2}g^{11}(-g_{00,1}+2g_{01,0}) = \frac{1}{2}g^{11}g_{00,1},
\end{align*}
so
\begin{align*}
    g_{00,1} = \frac{\partial B}{\partial r}(t,r) = 0,
\end{align*}
and hence $B$ depends only on $t$.
\end{proof}

Now assume that we are given a general metric in SSC $(t,\bar{r})$,
\begin{align*}
    \left(\begin{array}{cc}
    B & 0\\
    0 & \frac{1}{A}
    \end{array}\right)_{(t,\bar{r})},
\end{align*}
where $\bar{r}$ is arc-length at each fixed $t$ and the subscript indicates the coordinates on which we assume the components depend. Assume further that $p=0$, so the particle paths are geodesics, and assume we know constant $\hat{r}$ describes the particle paths, so that subluminal velocities imply $\hat{r}$ is a space-like coordinate. We can thus transform to comoving coordinates by $(t,\bar{r})\to(t,\hat{r})$, producing
\begin{align*}
    \left(\begin{array}{cc}
    B & 0\\
    0 & \frac{1}{A}
    \end{array}\right)_{(t,\bar{r})} \to \left(\begin{array}{cc}
    \tilde{H} & \tilde{E}\\
    \tilde{E} & \tilde{F}
    \end{array}\right)_{(t,\hat{r})}.
\end{align*}
However, by Weinberg \cite{wein}, there always exists a time transformation $\tilde{t}=\hat{t}(t,\hat{r})$, obtained locally from an integrating factor, which eliminates the middle term $\tilde{E}$, while keeping the comoving radial coordinate $\hat{r}$, that is,
\begin{align*}
    \left(\begin{array}{cc}
    \tilde{H} & \tilde{E}\\
    \tilde{E} & \tilde{F}
    \end{array}\right)_{(t,\hat{r})} \to \left(\begin{array}{cc}
    H & 0\\
    0 & F
    \end{array}\right)_{(\tilde{t},\hat{r})}.
\end{align*}
We now conclude by Theorem \ref{thmLTB} that $H$ depends only on $\tilde{t}$. But now the time transformation $\tilde{t}=\phi(\hat{t})$ with $\phi'(\hat{t})=\frac{1}{\sqrt{H}}$ converts $H\to1$ and the final coordinates $(\hat{t},\hat{r})$ are LTB coordinates. This verifies (1).

\end{document}